\RequirePackage[l2tabu]{nag}		% Warns for incorrect (obsolete) LaTeX usage
\documentclass[a4paper,11pt,openany,oldfontcommands]{memoir} %add 'draft' to turn draft option on (see below)
\usepackage{datetime}
\usepackage{ifpdf}
\ifpdf
\fi
\ifdraftdoc 
	\usepackage{draftwatermark}				%Sets watermarks up.
	\SetWatermarkScale{0.3}
	\SetWatermarkText{\bf Draft: \today}
\fi
\newsubfloat{figure}
\newsubfloat{table}
\settrimmedsize{297mm}{210mm}{*}
\settypeblocksize{634pt}{448.13pt}{*} 
\setulmargins{4cm}{*}{*} 
\setlrmargins{*}{*}{1.5} 
\setmarginnotes{17pt}{51pt}{\onelineskip} 
\setheadfoot{\onelineskip}{2\onelineskip} 
\setheaderspaces{*}{2\onelineskip}{*} 
\checkandfixthelayout
\usepackage{fouriernc}
\OnehalfSpacing 
\setsecnumdepth{subsection} 
\maxsecnumdepth{subsubsection}
\usepackage{calc,soul,fourier}
\makeatletter 
\newlength\dlf@normtxtw 
\newsavebox{\feline@chapter} 
\newcommand\feline@chapter@marker[1][4cm]{%
	\sbox\feline@chapter{% 
		\resizebox{!}{#1}{\fboxsep=1pt%
			\colorbox{gray}{\color{white}\thechapter}% 
		}}%
		\rotatebox{90}{% 
			\resizebox{%
				\heightof{\usebox{\feline@chapter}}+\depthof{\usebox{\feline@chapter}}}% 
			{!}{\scshape\so\@chapapp}}\quad%
		\raisebox{\depthof{\usebox{\feline@chapter}}}{\usebox{\feline@chapter}}%
} 
\newcommand\feline@chm[1][4cm]{%
	\sbox\feline@chapter{\feline@chapter@marker[#1]}% 
	\makebox[0pt][c]{% aka \rlap
		\makebox[1cm][r]{\usebox\feline@chapter}%
	}}
\makechapterstyle{daleifmodif}{

	\renewcommand\printchapternum{\null\hfill\feline@chm[2.5cm]\par}

} 
\makeatother 
\chapterstyle{daleifmodif}
\makepagestyle{myvf} 
\makeoddfoot{myvf}{}{\thepage}{} 
\makeevenfoot{myvf}{}{\thepage}{} 
\makeheadrule{myvf}{\textwidth}{\normalrulethickness} 
\makeevenhead{myvf}{\small\textsc{\leftmark}}{}{} 
\makeoddhead{myvf}{}{}{\small\textsc{\rightmark}}
\makeindex
\usepackage{import}

\usepackage{pdfpages}
\usepackage{lipsum}					%Needed to create dummy text
\usepackage{amsfonts} 					%Calls Amer. Math. Soc. (AMS) fonts
\usepackage{amsmath}			%Writes maths centred down[centertags]
\usepackage{stmaryrd}					%New AMS symbols
\usepackage{amssymb}					%Calls AMS symbols
\usepackage{amsthm}					%Calls AMS theorem environment
\usepackage{amsmath}
\usepackage{newlfont}					%Helpful package for fonts and symbols
\usepackage{layouts}					%Layout diagrams
\usepackage{graphicx}					%Calls figure environment
\usepackage{longtable,rotating}			%Long tab environments including rotation. 
\usepackage[utf8]{inputenc}			%Needed to encode non-english characters 
\usepackage{colortbl}					%Makes coloured tables
\usepackage{wasysym}					%More math symbols
\usepackage{mathrsfs}					%Even more math symbols
\usepackage{float}						%Helps to place figures, tables, etc. 
\usepackage{verbatim}					%Permits pre-formated text insertion
\usepackage{upgreek }					%Calls other kind of greek alphabet
\usepackage{latexsym}					%Extra symbols
\usepackage[square,numbers,
		     sort&compress]{natbib}		%Calls bibliography commands 
\usepackage{url}						%Supports url commands
\usepackage{color}                    				%Creates coloured text and background
\usepackage[pdftex,plainpages=false,colorlinks=true,citecolor=blue,linkcolor=blue,urlcolor=blue,filecolor=green,bookmarksopen=true]{hyperref}           %Creates hyperlinks in cross references
\usepackage{memhfixc}					%Must be used on memoir document 
\usepackage{enumerate}					%For enumeration counter
\usepackage{footnote}					%For footnotes
\usepackage{microtype}					%Makes pdf look better.
\usepackage{rotfloat}					%For rotating and float environments as tables, 
\usepackage{alltt}						%LaTeX commands are not disabled in 
\usepackage[version=0.96]{pgf}			%PGF/TikZ is a tandem of languages for producing vector graphics from a 
\usepackage{tikz}						%geometric/algebraic description.
\usetikzlibrary{arrows,shapes,snakes,
		       automata,backgrounds,
		       petri,topaths,positioning}				%To use diverse features from tikz		
\newcommand{\pgftextcircled}[1]{                                                                    %Defines encircled text
    \setbox0=\hbox{#1}%
    \dimen0\wd0%
    \divide\dimen0 by 2%
    \begin{tikzpicture}[baseline=(a.base)]%
        \useasboundingbox (-\the\dimen0,0pt) rectangle (\the\dimen0,1pt);
        \node[circle,draw,outer sep=0pt,inner sep=0.1ex] (a) {#1};
    \end{tikzpicture}
}
\newcommand{\blackged}{\hfill$\blacksquare$}
\newcommand{\whiteged}{\hfill$\square$}
\newcounter{proofcount}
\renewenvironment{proof}[1][\proofname.]{\par
 \ifnum \theproofcount>0 \pushQED{\whiteged} \else \pushQED{\blackged} \fi%
 \refstepcounter{proofcount}
 \normalfont 
 \trivlist
 \item[\hskip\labelsep
       \itshape
   {\bf\em #1}]\ignorespaces
}{%
 \addtocounter{proofcount}{-1}
 \popQED\endtrivlist
}
\let\oldsqrt\sqrt
\def\sqrt{\mathpalette\DHLhksqrt}
\def\DHLhksqrt#1#2{%
\setbox0=\hbox{$#1\oldsqrt{#2\,}$}\dimen0=\ht0
\advance\dimen0-0.2\ht0
\setbox2=\hbox{\vrule height\ht0 depth -\dimen0}%
{\box0\lower0.4pt\box2}}
\newcommand{\mycaption}[2][\@empty]{
	\captionnamefont{\scshape}
	\changecaptionwidth
	\captionwidth{0.9\linewidth}
	\captiondelim{.\:} 
	\indentcaption{0.75cm}
	\captionstyle[\centering]{}
	\setlength{\belowcaptionskip}{10pt}
	\tiny\ifx \@empty#1 \caption{#2}\else \caption[#1]{#2}
}

\usepackage{lettrine}
\newcommand{\initial}[1]{%
	\lettrine[lines=3,lhang=0.33,nindent=0em]{
		\color{gray}
     		{\textsc{#1}}}{}}
\theoremstyle{plain}
\newtheorem{theo}{Theorem}[chapter]
\theoremstyle{plain}

\theoremstyle{plain}
\newtheorem{definition}{Definition}[chapter]
\theoremstyle{plain}

\theoremstyle{plain}

\theoremstyle{plain}

\usepackage{graphicx,epstopdf}
\usepackage{blindtext}
\usepackage{lipsum}
\usepackage{amsfonts}
\usepackage{bbm}
\usepackage{amssymb}
\usepackage{latexsym}
\usepackage{perpage}
\MakePerPage{footnote}
\usepackage{cancel}

\newcommand{\Acal}{\mathcal{A}}
\newcommand{\Bcal}{\mathcal{B}}

\newcommand{\Dcal}{\mathcal{D}}
\newcommand{\Ecal}{\mathcal{E}}
\newcommand{\Fcal}{\mathcal{F}}
\newcommand{\Gcal}{\mathcal{G}}
\newcommand{\Hcal}{\mathcal{H}}

\newcommand{\Lcal}{\mathcal{L}}
\newcommand{\Ncal}{\mathcal{N}}
\newcommand{\Ocal}{\mathcal{O}}
\newcommand{\Pcal}{\mathcal{P}}

\newcommand{\Zcal}{\mathcal{Z}}

\usepackage{aurical}
\DeclareMathAlphabet{\pazocal}{OMS}{zplm}{m}{n}
\newcommand{\Acalb}{\pazocal{A}}

\newcommand{\Dcalb}{\pazocal{D}}
\newcommand{\Ecalb}{\pazocal{E}}
\newcommand{\Fcalb}{\pazocal{F}}
\newcommand{\Gcalb}{\pazocal{G}}
\newcommand{\Hcalb}{\pazocal{H}}
\newcommand{\Icalb}{\pazocal{I}}

\newcommand{\Lcalb}{\pazocal{L}}
\newcommand{\Mcalb}{\pazocal{M}}

\newcommand{\Ocalb}{\pazocal{O}}
\newcommand{\Pcalb}{\pazocal{P}}

\newcommand{\Rcalb}{\pazocal{R}}
\newcommand{\Scalb}{\pazocal{S}}
\newcommand{\Tcalb}{\pazocal{T}}
\newcommand{\Ucalb}{\pazocal{U}}
\newcommand{\Vcalb}{\pazocal{V}}

\newcommand{\Zcalb}{\pazocal{Z}}

\newcommand{\1}{\mathbbm{1}}
\newcommand{\Lmath}{\mathbbm{L}}
\newcommand{\Cmath}{\mathbbm{C}}
\newcommand{\Rmath}{\mathbbm{R}}

\newcommand{\Zmath}{\mathbbm{Z}}

\newcommand{\dket}[1]{| #1 \rangle\rangle}
\newcommand{\dbra}[1]{\langle\langle #1 |}
\newcommand{\ket}[1]{| #1 \rangle}

\newcommand{\dinterpro}[2]{\langle \langle #1 | #2 \rangle \rangle}
\newcommand{\interpro}[2]{\langle #1 | #2 \rangle}
\newcommand{\bra}[1]{\langle #1 |}
\newcommand{\tr}[1]{ \text{Tr}\left[ #1 \right]}
\newcommand{\trs}[1]{ \text{Tr}[ #1 ]}
\newcommand{\trnone}[1]{ \text{Tr} #1 }
\newcommand{\dbar}{d\hspace*{-0.16em}\bar{}\hspace*{0.16em}}
\newcommand{\ketus}[0]{|\!\! \uparrow \rangle}
\newcommand{\ketds}[0]{|\!\! \downarrow \rangle}
\newcommand{\braus}[0]{\langle \uparrow \!\! |}
\newcommand{\brads}[0]{\langle \downarrow \!\! |}

\usepackage{multicol}

\begin{document}
% UoB guidlines:
%
% Preliminary pages
% 
% The five preliminary pages must be the Title Page, Abstract, Dedication
% and Acknowledgements, Author's Declaration and Table of Contents.
% These should be single-sided.
% 
% Table of contents, list of tables and illustrative material
% 
% The table of contents must list, with page numbers, all chapters,
 % sections and subsections, the list of references, bibliography, list of
% abbreviations and appendices. The list of tables and illustrations
% should follow the table of contents, listing with page numbers the
% tables, photographs, diagrams, etc., in the order in which they appear
% in the text.
% 
\frontmatter
\pagenumbering{roman}
%

%
% File: Title.tex
% Author: V?ctor Bre?a-Medina
% Description: Contains the title page
%
% UoB guidelines:
% 
% At the top of the title page, within the margins, the dissertation should give the title and, if 
% necessary, sub-title and volume number. If the dissertation is in a language other than English, the 
% title must be given in that language and in English. The full name of the author should be in the 
% centre of the page. At the bottom centre should be the words ?A dissertation submitted to the 
% University of Bristol in accordance with the requirements for award of the degree of ? in the 
% Faculty of ...?, with the name of the school and month and year of submission. The word count of 
% the dissertation (text only) should be entered at the bottom right-hand side of the page.
%
%
\begin{titlingpage}
	\begin{SingleSpace}
		\calccentering{\unitlength} 
		\begin{adjustwidth*}{\unitlength}{-\unitlength}
			\begin{center}
				\includegraphics[scale=0.21]{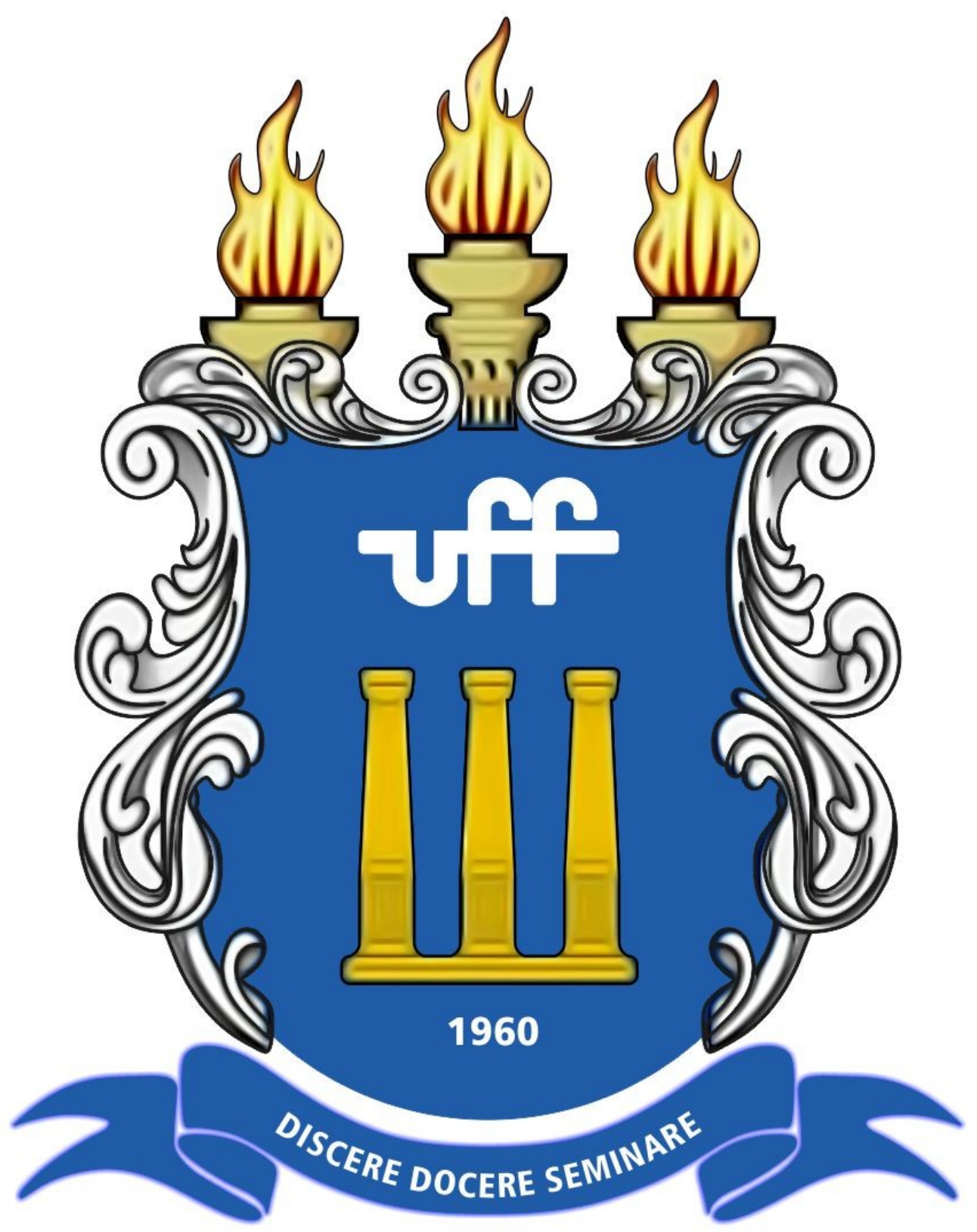}\\
				\rule[0.5ex]{\linewidth}{2pt}\vspace*{-\baselineskip}\vspace*{3.2pt}
				\rule[0.5ex]{\linewidth}{1pt}\\[\baselineskip]
				{\huge \textsc{Adiabatic Dynamics and}}\\[1.5mm]
				{\huge \textsc{Shortcuts to Adiabaticity}}\\[2mm]
				{\Large \textsc{Fundamentals and Applications}}\\[4mm]
				%{\Large \textit{\textsc{Applications to quantum information processing}}}\\
				\rule[0.5ex]{\linewidth}{1pt}\vspace*{-\baselineskip}\vspace{3.2pt}
				\rule[0.5ex]{\linewidth}{2pt}\\
				\vspace{6.5mm}
				{\large Author}\\
				\vspace{3.5mm}
				{\Large\textsc{Alan Costa dos Santos}}\\
				\vspace{11mm}
				\includegraphics[scale=0.29]{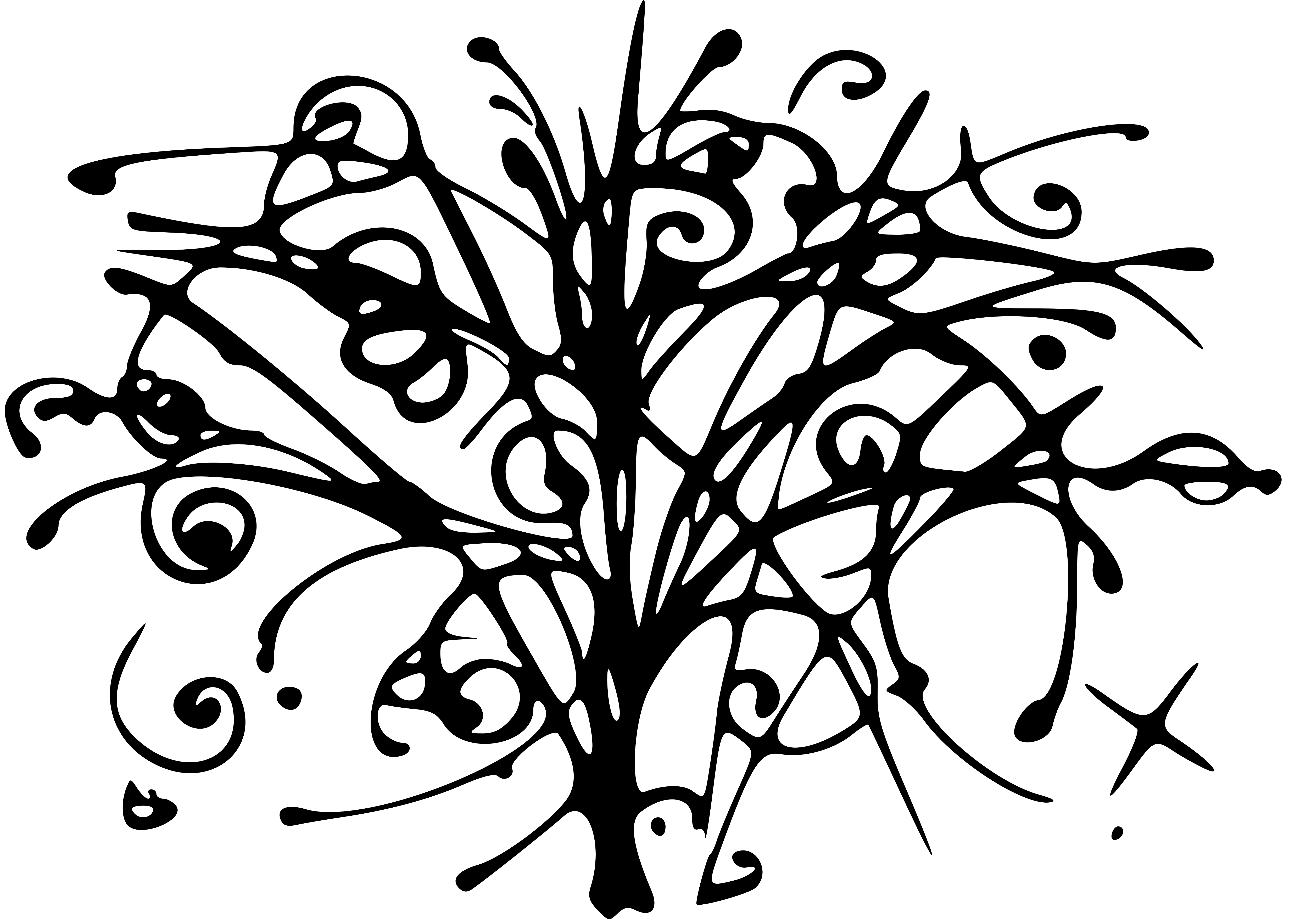}\\
				\vspace{3mm}
				{\large \textsc{Instituto de Física}\\
					\textsc{Universidade Federal Fluminense - UFF}}\\
				\vspace{11mm}
				\begin{minipage}{10cm}
					PhD thesis presented to Fluminense Federal University as a fundamental requirement to obtain the degree of \textsc{Doctor in Physics} by the Institute of Physics.
				\end{minipage}\\
				\vspace{9mm}
				{\large\textsc{February 2020}}
				\vspace{12mm}
			\end{center}
			\begin{flushright}
				{\large Thesis supervisor:}\\
				\vspace{3.5mm}
				{\large\textsc{Prof. Dr. Marcelo S. Sarandy}}
			\end{flushright}
		\end{adjustwidth*}
	\end{SingleSpace}
\end{titlingpage}

\begin{center}
	{\Large\textsc{Alan Costa dos Santos}}\\
	\vspace{50mm}
	{\Large \textsc{Adiabatic Dynamics and Shortcuts to Adiabaticity}}\\[2mm]
	{\large \textsc{Fundamentals and Applications}}\\[4mm]
	\vspace{6.5mm}
	%				{\large \textsc{Instituto de Física}\\
	%					\textsc{Universidade Federal Fluminense - UFF}}\\
\end{center}

\vspace{60mm}
\begin{flushright}
	\begin{minipage}{8cm}
		PhD thesis presented to Fluminense Federal University as a fundamental requirement to obtain the degree of \textsc{Doctor in Physics} by the Institute of Physics.
	\end{minipage}\\
	\vspace{10mm}
	Supervisor: Dr. Marcelo Silva Sarandy
\end{flushright}
\vspace{30mm}
\begin{center}
	{\large Niterói, Rio de Janeiro -- Brazil, February 2020}
\end{center}

\newpage

\includepdf[scale=1]{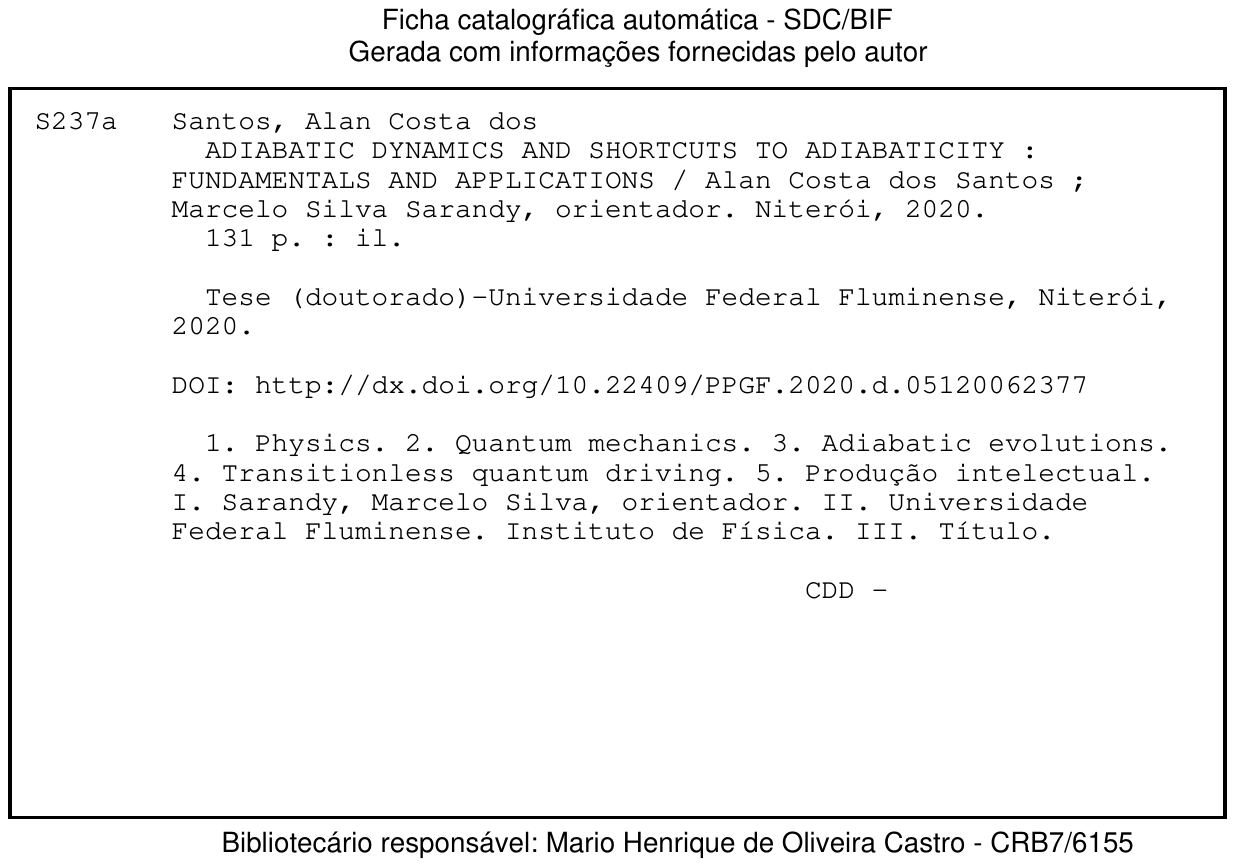}

%\clearemptydoublepage
%
%
% File: abstract.tex
% Author: V?ctor Bre?a-Medina
% Description: Contains the text for thesis abstract
%
% UoB guidelines:
%
% Each copy must include an abstract or summary of the dissertation in not
% more than 300 words, on one side of A4, which should be single-spaced in a
% font size in the range 10 to 12. If the dissertation is in a language other
% than English, an abstract in that language and an abstract in English must
% be included.

\chapter*{Abstract}
\begin{SingleSpace}
	\initial{I}n this thesis, it is presented a set of results in adiabatic dynamics (closed and open system) and transitionless quantum driving that promote some advances in our understanding on quantum control and Hamiltonian inverse engineering. In the context of adiabatic dynamics in closed systems, it is introduced a validation mechanism for the adiabaticity conditions by studing the system dynamics from a non-inertial reference frame. It is shown how the evaluation of relevant adiabatic approximation conditions in different reference frames (inertial and non-inertial) allows to correctly predict the adiabatic behavior of quantum system driven by oscillatory and rotating fields. Such mechanism is applicable to general multi-particle quantum systems, establishing the conditions for the equivalence of the adiabatic behavior as described in inertial or non-inertial frames. As a relevant application to modern quantum devices, the study of adiabaticity yields an adiabatic protocol to ensure a stable charging process of quantum batteries, which allows one to avoid the spontaneous
	discharging regime. 
	%%%%%%%%%%%%%%%%%%%%%%%%%%
	By considering a decohering scenario, validity conditions of the adiabatic approximation are also studied. As a fresh general result with potential applications, it is shown that under decoherence the adiabaticity may still occur
	in the infinite time limit, as it happens for closed systems, for
	a class of initial quantum states. From the viewpoint of basic studies on adiabaticity in open systems, this general approach for adiabatic non-unitary evolution leads to thermodynamics definitions of heat and work in terms of the underlying Liouville superoperator governing the quantum dynamics. Hence, one derives conditions under which an adiabatic open-system quantum dynamics implies in the absence of heat exchange, providing a connection between quantum and thermal adiabaticity.
	%%%%%%%%%%%%%%%%%%%%%%%%%%
	To end, the original contributions of this thesis to the theory of shortcuts to adiabaticity refers to a generalized approach of transitionless quantum driving, where one explores the gauge freedom of the quantal phase factors accompanying adiabatic trajectories. It is shown that this generalized transitionless evolutions can be more robust against decoherence
	than their standard and adiabatic counterparts, even by driving the system through fields that provide a minimal energy demanding
	scheme. The number of theoretical prediction presented in this thesis is experimentally verified through two different experimental setups, namely a qubit encoded in the energy hyperfine structure of a Ytterbium trapped ion, and in nuclear magnetic resonance with a nuclear spin qubit.
	
	\vspace{0.4cm}
	
	\noindent \textbf{Keywords:} Adiabatic dynamics $\cdot$ Transitionless quantum driving $\cdot$ Quantum thermodynamics $\cdot$ Trapped ions $\cdot$ Nuclear magnetic resonance.
	
\end{SingleSpace}

%
% File: abstract.tex
% Author: V?ctor Bre?a-Medina
% Description: Contains the text for thesis abstract
%
% UoB guidelines:
%
% Each copy must include an abstract or summary of the dissertation in not
% more than 300 words, on one side of A4, which should be single-spaced in a
% font size in the range 10 to 12. If the dissertation is in a language other
% than English, an abstract in that language and an abstract in English must
% be included.

\chapter*{Resumo}
\begin{SingleSpace}
	\initial{N}esta tese são apresentados resultados sobre dinâmica adiabática (sistemas fechado e aberto) e evoluções não transicionais que promovem alguns avanços em nosso entendimento sobre controle quântico e engenharia inversa de Hamiltoniano. No contexto de evoluções adiabáticas em sistemas fechados é introduzido um mecanismo de validação para as condições de adiabaticidade estudando a dinâmica do sistema a partir de um referencial não-inercial. É mostrado como o cálculo de relevantes condições de aproximação adiabática em diferentes referenciais (inerciais e não-inerciais) permite prever corretamente o comportamento adiabático de sistemas quânticos conduzidos por campos oscilatórios e rotativos. Este mecanismo é aplicável a sistemas quânticos de muitas partículas em geral, estabelecendo-se as condições de equivalência do comportamento adiabático conforme descrito por referênciais inerciais ou não-inerciais. Como aplicação relevante para dispositivos quânticos modernos, o estudo da adiabaticidade em diferentes referenciais leva a um protocolo adiabático que garante processo de carregamento estável de baterias quânticas, o que permite evitar o regime de descarga espontânea. Considerando um cenário com decoerência, também são estudadas as condições de validade da aproximação adiabática. Como um novo resultado geral com possíveis aplicações é mostrado que sob decoerência a adiabaticidade ainda pode ocorrer no limite de tempo infinito, como ocorre em sistemas fechados, para uma classe de estados iniciais. De um ponto de vista de fundamentos de adiabaticidade em sistemas abertos, essa abordagem geral para evoluções adiabáticas  não-unitárias levam à definições termodinâmicas de calor e trabalho em termos do superoperador de Liouville, que rege a dinâmica quântica. Portanto, deriva-se condições sob as quais uma dinâmica quântica adiabática de sistema aberto implica na ausência de troca de calor, fornecendo uma conexão entre adiabaticidade quântica e térmica. Por fim, as contribuições originais desta tese para a teoria dos atalhos para a adiabaticidade referem-se às abordagens generalizadas de evoluções não transicionais, onde se explora a liberdade de \textit{gauge} das fases quânticas que acompanham as trajetórias adiabáticas. É mostrado que essas evoluções não transicionais generalizadas podem ser mais robustas contra a decoerência do que as versões adiabática e tradicional, mesmo conduzindo o sistema através de campos que fornecem um esquema mínimo de demanda de energia. As previsões teóricas apresentadas nesta tese são verificadas experimentalmente através de dois sistemas experimentais diferentes, a saber, com um qubit codificado nos níveis de energia hiperfino de um íon de Itérbio aprisionado, e em ressonância magnética nuclear com um qubit codificado em um spin nuclear.
	
	\vspace{0.2cm}
	
	\noindent \textbf{Keywords:} Evolução adiabática $\cdot$ Evolução não transicional $\cdot$ Termodinâmica quântica $\cdot$ Íons armadilhados $\cdot$ Ressonância magnética nuclear.

\end{SingleSpace}

%
% File: abstract.tex
% Author: V?ctor Bre?a-Medina
% Description: Contains the text for thesis abstract
%
% UoB guidelines:
%
% Each copy must include an abstract or summary of the dissertation in not
% more than 300 words, on one side of A4, which should be single-spaced in a
% font size in the range 10 to 12. If the dissertation is in a language other
% than English, an abstract in that language and an abstract in English must
% be included.
\newpage

\vspace*{\stretch{4}}

\noindent To my parents Zé Almir and Toinha.\\
To my brothers Ramon, Saymon and Herisson.\\
To my future wife, Dâmaris (my Penny). \\
To my little crazy dog, Nick. \\
\textit{In memoriam}, Francisco José (Chicão).

\chapter*{Acknowledgements}

Four years ago I started a very hard mission. Along these years a number of people have made such mission smoother. Firstly, I would like to thank those people who have held my hands through the most drastic moments of this long path. To my parents, Zé Almir and Toinha, thank you so much for being here all of the moments, for not giving up on me when the things were very hard and for supporting me all the times I needed your love and attention. To my brothers, Ramon, Saymon and Herisson, for being so important in my live and for all very nice moments next to you. You guys are awesome! To end, to Dâmaris for being my best friend, my guide, my boss, my confidant and for being the love of my life. I'm sure you do not know how importantly you are to me, how lovely you are. From now on, I hope to always be by your side to say how much beautiful you are, how much lovely your eyes are, and how much your smile beautifies my life. I see my future in your eyes, your so beautiful eyes. To my parents, brothers and girlfriend, all of you are the reason of everything.

It is very important to highlight here the role of a number of researchers who contributed a lot for my academic life and, some of them, for my personality as researcher. First, for chronological order, I would like to thank to the experimental groups that helped me along my PhD. I would like to thank to the first person who believed in my ideas, Dr. Chang-Kang Hu (University of Science and Technology of China -- USTC). I'm sure that his experimental implementations of my results were very important to this thesis. My friend, many thanks for all your support, patience, attention during our collaborations and for teaching me a bit about the very nice Chinese culture. Moreover, I am grateful to his lab staff for being part of this collaborations, Dr. Yun-Feng Huang (his PhD supervisor), Dr. Jin-Ming Cui, Dr. Chuan-Feng Li, and Dr. Guang-Can Guo. Thank you so much for all your work and support. A second experimental group that had a very important role to this thesis was the group of the Dr. Roberto Sarthour (Brazilian Center for Research in Physics -- CBPF). To Dr. Roberto, Dr. Alexandre Souza and to Dr. Ivan Oliveira, thank you all so much for being so solicitous when requested and for their hospitality in receiving me at CBPF. Thanks to Amanda Nicotina for performing part of the experiment and for your support. There is no words to describe the relevance of the support of these two experimental groups to my thesis. Thank you so much. I also would like to thank to Dr. José Huguenin and to Marcello Passos (Exact Sciences Institute, Federal Fluminense University -- ICEx-UFF). Although our collaborative work does not make part of this thesis, it has great relevance to me. Moreover, thank you so much for your hospitality and support when I was in Volta Redonda. To Marcello Passos, thanks for the coffees, beers and for all our chats.

Again, in a chronological way, I would like to highlight here the people that believed in my theoretical studies on quantum mechanics and for helping me to enhance them. I want to thank to Dr. Frederico Brito and \xcancel{Prof.} Dr. Diogo Soares-Pinto (São Carlos Institute of Physics, University of São Paulo -- IFSC-USP), for their collaborations in two works presented in this thesis. There is no doubt that their collaborations were strongly relevant to achieve good results and discussions. In addition, I have no words to thank all their support and hospitality when I was in São Carlos for 2 weeks. I wish to thank to Dr. Bar\i\c{s} \c{C}akmak (Bah\c{c}e\c{s}ehir University, Istanbul, Turkey), to Dr. Steve Campbell (Trinity College Dublin, Dublin 2, Ireland) and to Dr. Nicolaj Zinner (Aarhus University, DK-8000 Aarhus C, Denmark) for believing in my ideas and for helping me to enhance them. Our work on stability of three-level quantum batteries through adiabatic dynamics is my first independent project, so your collaboration has both academic and personal important meaning to me. Thank you so much. To end, I am grateful to Dr. Andreia Saguia (Federal Fluminense University -- UFF), for your interest in my research and for helping me to provide a generalized study of adiabatic quantum batteries. Thanks for your time, patience and personal/professional advice. 

To end, the man who made this thesis possible, my supervisor Dr. Marcelo Sarandy. Today I can see that I did not have a supervisor, I had a friend. Along these six years (since my master degree), I learned a lot with you. First, I learned how to be student, then how to be a researcher and, the most important, how to have professional ethics. I'm very grateful for all your optimism, even when I was not so excited with my own results, for your patience and wisdom in leading the whole Ph.D. project. I have no words to thank you for all of chats, personal  advice and professional guidance for the last six years.

This thesis was mainly supported by Conselho Nacional de Desenvolvimento Cient\'{\i}fico e Tecnol\'ogico (CNPq-Brazil). The author also acknowledges financial support in part by the Coordena\c{c}\~ao de Aperfei\c{c}oamento de Pessoal de N\'{\i}vel Superior - Brasil (CAPES) (Finance Code 001) and by the Brazilian National Institute for Science and Technology of Quantum Information [CNPq INCT-IQ (465469/2014-0)].

\newpage

\vspace*{\stretch{4}}

%\noindent``Share your knowledge. It is a way to achieve immortality."

%\noindent\textit{Dalai Lama}

%\vspace{1cm}

%\noindent``Sometimes it is the people no one can imagine anything of, who do the things no one can imagine."

%\noindent\textit{Alan Turing}

\begin{center}
	\noindent \textbf{Nordeste Independente} \\ \textit{Braulio Tavares} \& \textit{Ivanildo Vilanova}
\end{center}

\begin{multicols}{2}
	
	\noindent Já que existe no sul esse conceito\\
	Que o nordeste é ruim, seco e ingrato\\
	Já que existe a separação de fato\\
	É preciso torná-la de direito\\
	Quando um dia qualquer isso for feito\\
	Todos dois vão lucrar imensamente\\
	Começando uma vida diferente\\
	De que a gente até hoje tem vivido\\
	Imagina o Brasil ser dividido\\
	E o nordeste ficar independente\\
	
	\noindent Dividindo a partir de Salvador\\
	O nordeste seria outro país\\
	Vigoroso, leal, rico e feliz\\
	Sem dever a ninguém no exterior\\
	Jangadeiro seria o senador\\
	O cassaco de roça era o suplente\\
	Cantador de viola, o presidente\\
	O vaqueiro era o líder do partido\\
	Imagina o Brasil ser dividido\\
	E o nordeste ficar independente\\
	
	\noindent Em Recife, o distrito industrial\\
	O idioma ia ser nordestinense\\
	A bandeira de renda cearense\\
	"Asa Branca" era o hino nacional\\
	O folheto era o símbolo oficial\\
	A moeda, o tostão de antigamente\\
	Conselheiro seria o inconfidente\\
	Lampião, o herói inesquecido\\
	Imagina o Brasil ser dividido\\
	E o nordeste ficar independente
	
	\noindent O Brasil ia ter de importar\\
	Do nordeste algodão, cana, caju\\
	Carnaúba, laranja, babaçu\\
	Abacaxi e o sal de cozinhar\\
	O arroz, o agave do lugar\\
	O petróleo, a cebola, o aguardente\\
	O nordeste é auto-suficiente\\
	O seu lucro seria garantido\\
	Imagina o Brasil ser dividido\\
	E o nordeste ficar independente\\
	
	\noindent Se isso aí se tornar realidade\\
	E alguém do Brasil nos visitar\\
	Nesse nosso país vai encontrar\\
	Confiança, respeito e amizade\\
	Tem o pão repartido na metade\\
	Temo prato na mesa, a cama quente\\
	Brasileiro será irmão da gente\\
	Vai pra lá que será bem recebido\\
	Imagina o Brasil ser dividido\\
	E o nordeste ficar independente\\
	
	\noindent Eu não quero, com isso, que vocês\\
	Imaginem que eu tento ser grosseiro\\
	Pois se lembrem que o povo brasileiro\\
	É amigo do povo português\\
	Se um dia a separação se fez\\
	Todos os dois se respeitam no presente\\
	Se isso aí já deu certo antigamente\\
	Nesse exemplo concreto e conhecido\\
	Imagina o Brasil ser dividido\\
	E o nordeste ficar independente
	
\end{multicols}

\newpage
\renewcommand{\contentsname}{Table of Contents}
\maxtocdepth{subsection}
\tableofcontents*
\addtocontents{toc}{\par\nobreak \mbox{}\hfill{\bf Page}\par\nobreak}
\newpage
\listoftables
\addtocontents{lot}{\par\nobreak\textbf{{\scshape Table} \hfill Page}\par\nobreak}
\newpage
\listoffigures
\addtocontents{lof}{\par\nobreak\textbf{{\scshape Figure} \hfill Page}\par\nobreak}
%\clearemptydoublepage
%
%
% The bulk of the document is delegated to these chapter files in
% subdirectories.
\mainmatter

%%%%%%%%%%%%%%%% Introdution %%%%%%%%%%%%%%%
%\let\textcircled=\pgftextcircled
\chapter{Introduction}
\label{chap:intro}

\initial{A}diabatic dynamics~\cite{Kato:50,Born:28,Messiah:Book} provides a powerful technique widely explored for a number of applications in quantum mechanics that demand our ability of perfectly controlling a quantum system. {By definition, the adiabatic dynamics (also called, adiabatic approximation) of a system driven by a time-dependent Hamiltonian $H_{0}(t)$ is achieved when the system dynamics (approximately) undergoes a trajectory in Hilbert space in which the instantaneous eigenstates of the Hamiltonian do not evolve coupled to each other~\cite{Kato:50,Born:28}. That means, if the system starts with some superposition of stationary states of $H_{0}(t=0)$, then the system is said to undergo an adiabatic dynamics if such superposition is kept (up to phases) for the corresponding stationary states of $H_{0}(t)$ for all later times $t>0$. Such kind of dynamics has been applied in the context of quantum computation~\cite{Farhi:01,Sarandy:05-2,Aharonov:04,Steffen:03}, digitized quantum computing~\cite{Barends:16,Hen:15,Santos:15}, state engineering~\cite{Amniat:12,Masuda:15,Vitanov:99,Unanyan:98,Vitanov:17}, among others~\cite{Tameem:18}. Adiabatic quantum computing has been investigated as a possible candidate to build high technology quantum processors with manufactured artificial spins, such as the D-Wave quantum computer~\cite{Johnson:11,Boixo:14}. In such system the input state is the lowest-energy state of a \textit{tunneling Hamiltonian}, where all qubits of the system are in a superposition of $\ket{0}$ and $\ket{1}$. Then, the system is driven by a time dependent Ising Hamiltonian with a transverse field, where the final state of the system is a classical state associated with the lowest-energy state of a \textit{problem Hamiltonian}. The problem Hamiltonian allows us to adjust the local fields (along the quantization axis) acting on each qubit, as well as the coupling between them. In order to provide a quantum annealing, the driven Hamiltonian interpolate the tunneling and problem Hamiltonians, where the problem Hamiltonian becomes effective meanwhile the influence of the tunneling one is reduced. However, the performance of adiabatic dynamics in achieving some tasks with high fidelity is strongly affected by the competition between the total evolution time (as imposed by the adiabaticity conditions~\cite{Amin:09,Jansen:07,Sarandy:04}) and the decohering time-scale of real physical systems (relaxation, dephasing, etc). When the inevitable environment-system coupling is taken into account, the performance of the adiabatic dynamics depends on the basis in which the decoherence acts, with some enhancement of the robustness of the protocol usually achievable when the decohering effects take place in the instantaneous energy eigenbasis~\cite{Albash:15}. More precisely, in this context the notion of adiabatic dynamics differs from the closed system case due to fundamental issues. Indeed, given the interaction with some environment, the system dynamics is governed by master equations that take into account the non-unitary effects due to such interaction~\cite{Petruccione:Book,Lindblad:76}. Although the adiabatic approximation in open systems takes into account contributions of the Hamiltonian $H_{0}(t)$, it takes into account system-environment interaction as well~\cite{Sarandy:05-1}. Hence, undesired effects will drastically change the system state when the decohering effects do not take place in a preferred basis~\cite{Albash:15}. However, independently of the existence of a preferable decohering basis, we need to consider some alternative approach to deal with these undesired effects. A widely used approach to bypass this problem is the shortcut to adiabaticity based on \textit{transitionless quantum driving} (TQD), as conceived by M. Demirplak and S. Rice~\cite{Demirplak:03,Demirplak:05}. 
	
	Shortcuts to adiabaticity (STA) consist of a set of methods able to significantly speed up an dynamics so that the initial and final state are the same as obtained by an adiabatic dynamics. Such approaches have been proposed in context of inverse Hamiltonian engineering, where we add fields in order to achieve initial and target states of the system, with applications to state engineering~\cite{Huang:17}, tracking of many-body systems across quantum 
	phase transitions~\cite{Saberi:14,Hatomura:17}, quantum computing~\cite{Santos:18-a,Santos:15,Santos:16,Coulamy:16}, quantum thermodynamics~\cite{Deng:18}, among others~\cite{Ibanez:13,Ruschhaupt:12,Mukherjee:16,Chen:16-2,Torrontegui:13,Odelin:19}. TQD is a specific STA, where the path followed by the system from initial to final state is an adiabatic trajectory of some reference time-dependent Hamiltonian $H_{0}(t)$. While a general STA can be obtained for any faster orbit than the adiabatic one, the TQD method drives the system along the exact adiabatic trajectory, for this reason we use the term ``transitionless'', denoting that the system does not make transitions between different instantaneous eigenstates of $H_{0}(t)$, but it is driven by a different Hamiltonian called transitionless Hamiltonian. TQD has been applied to a number of quantum tasks, such as state engineering~\cite{Chen:10,Xia:16,Hu:18}, quantum computing~\cite{Oh:14,Santos:15,Santos:16,Coulamy:16}, thermal engines~\cite{Zheng:16,Lutz:18,Abah:17,Funo:17}, relativistic quantum mechanical dynamics~\cite{Deffner:16}, among others~\cite{Torrontegui:13,Odelin:19}. As proposed by Demirplak and Rice  (also investigated by M. V. Berry~\cite{Berry:09}), the TQD method consists of an dynamics where the system is driven by a sum of fields, namely the adiabatic fields used to implement $H_{0}(t)$ plus auxiliary fields used to implement the \textit{counter-diabatic} Hamiltonian $H_{\text{cd}}(t)$. As its nomenclature suggests, $H_{\text{cd}}(t)$ inhibits \textit{diabatic} transitions that happens when $H_{0}(t)$ drives the system without any constraints imposed by the validity conditions of the adiabatic theorem. Due to the additional fields, the energy cost for implementing TQD are increased and it depends on the how faster the system is driven~\cite{Santos:15}, so that we can establish a trade-off between speed and energy cost of implementing such shortcut~\cite{Campbell-Deffner:17}. Then, any advantage of the usage of the TQD method comes with some additional energy expenditure to implement it.~\cite{Santos:15,Campbell-Deffner:17}. In return, it is important to highlight that speeding up the dynamics implies in a more robust dynamics against a number of decohering environment and systematic errors~. In summary, it would be appreciated if we get a new TQD approach that allows us for achieving the same robustness as the \textit{standard}\footnote{From now on we will use the terminology ``standard'' to refer to the TQD method proposed by Demirplak and Rice.} TQD, but with low-cost energy fields.
	
	In this thesis we focus on advances in adiabatic quantum mechanics theory (closed and open systems) and TQD theory, by exploring a generalized approach for the TQD. The motivation of the original studies and applications presented in this thesis are structured as follows.
	
	In context of adiabatic dynamics in closed systems, besides the significant number of applications of adiabatic dynamics previously mentioned, it remains an open problem the conditions that allow us to guarantee the adiabatic approximation. As we shall see in more details, there are a large number of works in the literature trying to solve this problem by proposing different validity \textit{conditions to adiabaticity} (ACs). In theory, when ACs are satisfied then the adiabatic approximation should be achieved with high fidelity. However, it has been proved that some widely used ACs are neither \textit{necessary} nor \textit{sufficient} in guaranteeing any prediction of the adiabatic behavior~\cite{Suter:08}. By \textit{necessary}, we mean that if the desired behavior is obtained then the condition at hands is \textit{necessarily satisfied}; by \textit{sufficient} we mean that if such condition is satisfied then it should be sufficient to say that the adiabatic approximation is achieved with high fidelity. The search for necessary and sufficient ACs has led to a number of discussions about the existence of necessary and sufficient conditions to adiabaticity~\cite{Teufel:03,Ambainis:04,Tong:05,Jansen:07,Amin:09,Tong:10,Cao:13} as well as new bounds for adiabaticity~\cite{Tong:07,Yu:14,Wang:15}. Then in this thesis we introduce a validation mechanism for the adiabatic approximation by driving
	the quantum system to a non-inertial reference frame~\cite{Hu:19-b}. To this end, first we consider several relevant adiabatic approximation conditions previously derived and show that all of them fail by introducing a suitable oscillating Hamiltonian for a single quantum bit (qubit). Then, by evaluating the adiabatic condition in a rotated non-inertial frame, we show that all of these conditions, including the standard adiabatic condition, can correctly describe the actual dynamics in the original frame, either far from resonance or at a resonant point. We show how such mechanism is applicable to general multi-particle quantum systems, establishing the conditions for the equivalence of the adiabatic behavior as described in inertial or non-inertial frames. 
	%%%
	As a practical application of the study of adiabaticity, we apply our results to charging process of a quantum battery (QB). Among fundamental challenges for useful QBs are both the control of the energy transfer and the stability of the discharge process to an available consumption hub. Indeed, a fully operational loss-free QB requires an inherent control over the energy transfer process, with the ability of keeping the energy retained with no leakage. In this sense, we exploit an adiabatic protocol to ensure a stable charged state of QBs which allows one to avoid the spontaneous discharging regime. Our results are firstly applied to a three-level quantum system in a ladder-type configuration~\cite{Santos:19-a}, then we propose a general design for a QB based on a power observable quantifying the energy transfer rate to a consumption hub~\cite{Santos:20c}. As main applications of the power observable, we introduce a \textit{trapping energy mechanism} for QBs and an asymptotically \textit{stable discharge mechanism}, which is achieved through an adiabatic evolution eventually yielding vanishing power. 
	
	By considering decohering effects on the system dynamics, we study the validity conditions of the adiabatic approximation. First we review some elements introduced in Ref.~\cite{Sarandy:05-1}, then as a general result we show that the adiabatic approximation in open system may still occur in the infinite time limit, as it happens for closed systems, for a class of initial quantum states. Our results can be viewed as a complementary discussion to the results presented in Refs.~\cite{Sarandy:05-1,Sarandy:05-2}, where the authors argued that open system adiabaticity is typically expected to occur at finite time. As a practical application, we implement the adiabatic quantum algorithm for the Deutsch problem, where a distinction is established between the open system adiabatic density operator and the target pure state expected in the computation process. Preferred time windows for obtaining the desired outcomes are then analyzed. On the other hand, from the viewpoint of basic studies on adiabaticity in open systems, 
	here we theoretically study thermodynamic adiabatic processes in open quantum systems, which evolve nonunitarily under decoherence due to its interaction with its surrounding environment. From the general approach for adiabatic non-unitary evolution, we establish heat and work in terms of the underlying Liouville superoperator governing the quantum dynamics. 
	In this scenario, thermodynamic processes with no heat exchange, namely, adiabatic transformations, can be implemented through quantum evolutions in closed systems, even though the notion of a closed system is always an idealization and approximation. As a consequence, we establish a detailed study on adiabatic evolutions (dynamics) and adiabatic processes (thermodynamics), where we derive the conditions such that an adiabatic open-system quantum dynamics implies in the absence of heat exchange, providing a connection between quantum and thermal adiabaticity.
	
	As a last set of results, we present a generalized approach of TQD in closed systems by exploring the gauge freedom of the quantal phases that accompany an adiabatic trajectory. In fact, there is a number of situations for which we need not exactly mimic an adiabatic 
	process, but only assure that the system is kept in an instantaneous eigenstate (independently 
	of its associated quantum phase)~\cite{Santos:15,Santos:16,Coulamy:16,Adolfo:16,Stefanatos:14,Lu:14,Deffner:16,Liang:16,An:16,Xia:16,Zhang:16,Marcela:14,Chen:10}. Therefore, we can suitably choose some dynamics in which the adiabatic phase is replaced for generalized quantal phases $\theta_n(t)$. This provides us a generalized TQD Hamiltonian that depends on a number of free parameters and such parameters can be chosen in order to optimize some physical quantity or process. For example, in this thesis we show situations in which the generalized TQD Hamiltonian is both time-independent and optimal energy cost demanding concerning the adiabatic and standard TQD Hamiltonian. As an application of the results, we realize feasible time-independent TQD 
	Hamiltonians~\cite{Santos:18-b}, which provide the optimal dynamics concerning both their experimental viability 
	and energy cost, as measured by the strength of the counter-diabatic fields required by the implementation~\cite{Hu:18} and by the optimal pulse sequence required to implement some quantum task~\cite{Santos:20b}. The energy-optimal protocol presented here is potentially useful for speeding up digitized adiabatic quantum computing~\cite{Hen:14,Barends:16}. Since digitized adiabatic quantum computing requires the Trotterization of the adiabatic dynamics, our protocol could be useful in reducing the number of steps used in digital adiabatic quantum processes. In addition, such process is independent of the experimental approach used to digitize the adiabatic quantum evolution.
	
	A number of theoretical predictions presented in this thesis are experimentally implemented in two different experimental setups. By using a two-level system encoded in the hyperfine energy structure of a single Ytterbium trapped ion, driven by external microwave fields, we experimentally explore the validation mechanism of the adiabatic theorem in both closed and open systems and the performance of optimal TQD Hamiltonians as provided by the generalized TQD theory. From the experimental novelty side, we report two distinct realizations. In the context of generalized TQD, the experimental realization allowed us to check the fields used to implement an optimal TQD is in fact less intense than the fields used to implement the adiabatic and standard TQD version of the dynamics. In context of the realization of an adiabatic open-system evolution, we use the hyperfine qubit of the Ytterbium trapped ion as our work substance for our application to thermodynamics. Our implementation exhibits high controllability, opening perspectives for analyzing thermal machines (or refrigerators) in open quantum systems under adiabatic evolutions. Due to the requirements of the usual definitions of heat and work, the investigation of thermodynamic quantities in adiabatic dynamics is achieved with time-dependent decoherence effects. To this end, it is introduced an efficient control to a Gaussian noise with time-dependent amplitude, which is then used to simulate a dephasing channel with a time-dependent decoherece rate $\gamma(t)$. We stress here that we focus on superficial details of the experimental setup used in our work, but a detailed description of the system can be found in Chang-kang Hu's PhD thesis~\cite{Hu:Thesis}. The second experimental setup used here is a nuclear spin system driven by nuclear magnetic resonance (NMR) experimental setup. Our experimental investigation shows the performance of a generalized approach for TQD to implement shortcuts to adiabaticity in NMR. We show that the generalized TQD theory requires less energy resources, with the strength of the magnetic field less than that required by standard TQD. As a second discussion, we analyze the experimental implementation of shortcuts to single-qubit adiabatic gates. By adopting generalized TQD, we can provide feasible time-independent driving Hamiltonians, which are optimized in terms of the number of pulses used to implement the quantum dynamics. The robustness of adiabatic and generalized TQD evolutions against typical decoherence processes in NMR is also analyzed.
	
	\newpage
	
\section{Mathematical framework, notation and conventions}
\label{sec:Notation}

In this section we present some important notations and conventions we use throughout this thesis.

\subsection{Matrix operations: Exponential and Logarithm of a matrix}

When we work with numbers (real $\Rmath$ or complex $\Cmath$) and how a number can be connected to another one, the notion of functions is a ``natural" consequence. Thus, we can define a function $f$ as an application $f$ that takes an element $x$ defined on a space $V_x$ and it associates to this element a second one $f(x)$ defined on a space $V_f$. Then, we write $f:x\in V_x \rightsquigarrow f(x) \in V_f$. More specifically, by choosing the domain of each element it is possible to define scalar, vector and matrix functions. Therefore, let $\Mcalb(m)$ be a matrix space whose elements $A$ are diagonalizable $(m\times m)$-matrices with eigenvalues $a_{n} \in \Cmath$ and eigenvectors $\ket{a_{n}} \in \Cmath^m$. Therefore, a matrix function $f$ is defined as $f:A \in \Mcalb(m) \rightsquigarrow f(A) \in \Mcalb(m)$. It is important to highlight that we can express a matrix function through its power series representation
\begin{equation}
	f(A) = \lim_{\lambda \rightarrow 1} f(\lambda A) = \lim_{\lambda \rightarrow 1} \sum_{n=0}^{\infty} \frac{(\lambda A)^n}{n!} \left(\left. \frac{\partial^{n}f(\lambda A)}{\partial \lambda^{n}} \right\vert_{\lambda = 0}\right) \label{EqExpanFunctionF} \text{ , }
\end{equation}
where we assume the existence of the derivatives in the above equation. Now, it is possible to define a number of interesting and useful matrix equations. First, again let us consider a matrix $A$ with eigenvalues $a_{n}$ associated with eigenvectors $\ket{a_{n}}$. Thus, we define the \textit{exponential} of $A$, which is denoted by $e^{A}$ or $\exp(A)$. Then, by using the Eq.~\eqref{EqExpanFunctionF} we can write
\begin{equation}
	e^{A} = \lim_{\lambda \rightarrow 1} e^{\lambda A} = \lim_{\lambda \rightarrow 1} \left[ \sum_{n=1}^{\infty} \frac{(\lambda A)^n}{n!} \right] = \sum_{n=1}^{\infty} \lim_{\lambda \rightarrow 1} \left[\frac{(\lambda A)^n}{n!} \right] \text{ , }
\end{equation}
and therefore we conclude that
\begin{equation}
	e^{A} = \sum_{n=1}^{\infty} \frac{A^n}{n!} \text{ . }
\end{equation}

The notion of exponential of a matrix is a highly useful mathematical tool for many tasks in quantum physics, since the unitary evolution of quantum systems is described by such mathematical elements~\cite{Zettili:Book}. In particular, it is convenient to write the matrix form of $e^{A}$ in the eigenbasis $\{\ket{a_{n}}\}$ of $A$ as
\begin{equation}
	e^{A} = \sum_{n=1}^{N}\ket{a_{n}}\bra{a_{n}} e^{A} \sum_{k=1}^{N}\ket{a_{k}}\bra{a_{k}} = \sum_{n=1}^{N}e^{a_{n}} \ket{a_{n}}\bra{a_{n}} \text{ , }
\end{equation}
where we already used that $\bra{a_{n}} e^{A} \ket{a_{k}} = e^{a_{k}} \delta_{kn}$. On the other hand, we can define the natural logarithm function of a matrix as
\begin{equation}
	e^{\ln A} = A  \text{ , } \label{EqexpLnA}
\end{equation}
which should be valid for any matrix $A$. However, here we can write $\ln A$ in basis $\{\ket{a_{n}}\}$
\begin{equation}
	\ln A = \sum_{n=1}^{N} \ln a_{n} \ket{a_{n}}\bra{a_{n}} \text{ . } \label{EqLogMatrix}
\end{equation}
However, a first remark here is that the above equation is valid up to a condition on the spectrum of $A$. In particular, $\ln 0$ is not well defined, because $\lim_{\lambda \rightarrow 0} \ln \lambda \rightarrow - \infty$. This means that Eqs.~\eqref{EqexpLnA} and~\eqref{EqexpLnA} hold for \textit{nonsingular} matrices $A$

\subsection{The density matrix}

A $d$-dimensional quantum system is fully described through a matrix $\rho$ called \textit{density matrix}. As any matrix, $\rho$ may be written from a predefined matrix basis $\{\rho_{n}\}$ as
\begin{equation}
	\rho = \sum_{n=0}^{D^2-1} c_{n} \rho_{n} \text{ , }
\end{equation}
where $c_{n}$ are parameters to be determined. Thus, assuming that $\rho$ describes a quantum system, it should satisfy the normalization condition $\trnone{\rho}=1$, so that it leads us to a condition on the coefficients $c_{n}$ as
\begin{equation}
	\sum_{n=1}^{D^2} c_{n} \trnone{\rho_{n}} = 1 \text{ . } \label{EqNormCond1}
\end{equation}

In particular, there is a convenient choice of the parameters $c_{n}$ and matrices $\rho_{n}$ so that the above condition is satisfied. Throughout this thesis we will use the basis $\{\rho_{n}\} = \{\1, \{\sigma_{n}\} \}$, where $\{\sigma_{n}\}$ is a set of $D^2-1$ ($D\times D$)-dimensional \textit{traceless} matrices\footnote{We say that a matrix $A$ is said to be traceless matrix when $\trs{A} = 0$.}. Thus, from Eq.~\eqref{EqNormCond1} we can identify $c_{0} = (1/D)$ so that $\rho$ could be rewritten as
\begin{equation}
	\rho = \frac{1}{D} \1 + \sum_{n=1}^{D^2-1} c_{n} \sigma_{n} \text{ . }
\end{equation}

Let us consider a second condition on the matrices $\rho_{n}$ associated with orthogonality the between the elements of $\{\sigma_{n}\}$, namely, let us consider a basis in which
\begin{equation}
	\frac{1}{D}\tr{\sigma_{n}\sigma^{\dagger}_{m}} = \delta_{nm} \text{ , } 
\end{equation}
where $\sigma^{\dagger}_{m}$ is the Hermitian conjugate of $\sigma_{m}$ and $\delta_{nm}$ is the Kronecker's delta\footnote{$\delta_{nm}=1$ when $n=m$ and $\delta_{nm}=0$ otherwise.}. From this consideration, we easily can find each coefficient $c_{n}$ of $\rho$ and $\sigma_{n}$ by computing
\begin{equation}
	\tr{\rho \sigma^{\dagger}_{n} } = \tr{\left(\frac{1}{D} \1 + \sum_{n=1}^{D^2-1} c_{n} \rho_{n}\right) \sigma^{\dagger}_{n}} 
	= \frac{1}{D}\trnone{ \sigma^{\dagger}_{n}}  + \sum_{k=1}^{D^2-1} c_{k} \tr{ \sigma_{k}\sigma^{\dagger}_{n} }
\end{equation}
where we can use that $\trnone{ \sigma^{\dagger}_{n}} = 0$ and $\tr{\sigma_{k}\sigma^{\dagger}_{n}} = D\delta_{kn}$ to get
\begin{equation}
	\tr{\rho \sigma^{\dagger}_{n} } = \varrho_{n} = D c_{n} \text{ . }
\end{equation}

Finally, if we adopt a \textit{traceless} matrix basis $\{\sigma_{1}, \sigma_{2}, \cdots , \sigma_{D^2-1}\}$ that satisfies $\trs{\sigma_{n}\sigma^{\dagger}_{m}} = D\delta_{nm}$, the density matrix can be written as
\begin{equation}
	\rho = \frac{1}{D} \left[ \1 + \sum_{n=1}^{D^2-1} \varrho_{n} \sigma_{n} \right]\text{ . } \label{EqEqRhoCoherence}
\end{equation}
with $\varrho_{n} = \trs{\rho \sigma^{\dagger}_{n} }$. As an example, let us consider the set of the Pauli matrices given by
\begin{equation}
	\sigma_{x} = \begin{bmatrix} 0 & 1 \\ 1 & 0 \end{bmatrix} 
	\text{ \  , \ \ \ \ }
	\sigma_{y} = \begin{bmatrix} 0 & -i \\ i & 0 \end{bmatrix} 
	\text{ \ \ \ \ and \ \ \ \ }
	\sigma_{z} = \begin{bmatrix} 1 & 0 \\ 0 & -1 \end{bmatrix}  \text{ . }
\end{equation}

Evidently we can note that the above matrices constitute a set of Hermitian traceless matrices. In addition, by using the relation ($a,b,c\in\{x,y,z\}$)
\begin{equation}
	\sigma_{a} \sigma_{b} = \delta_{ab}\1 + i \varepsilon_{abc} \sigma_{c} \text{ , }
\end{equation}
where $\varepsilon_{abc}$ is the Levi-Civita's symbol defined as
\begin{equation}
	\varepsilon_{abc} = \left\{\begin{matrix}
		+1 & \text{ if \ } (a,b,c) \in \{ (1,2,3),(3,1,2),(2,3,1) \} \\
		-1 & \text{ if \ } (a,b,c) \in \{ (1,3,2),(2,1,3),(3,2,1) \} \\
		0 & \text{ if $a = b$ or $b = c$ or $c=a$ }
	\end{matrix} \right. \text{ , }
\end{equation}

one finds that
\begin{equation}
	\tr{\sigma_{a} \sigma_{b}} = 2 \delta_{ab} \text{ . }
\end{equation}

Therefore, we conclude that the set of Pauli matrices can be used as a matrix basis for a two-level system and we write
\begin{equation}
	\rho = \frac{1}{2} \left[ \1 + \sigma_{x} \varrho_{x}(t) + \sigma_{y} \varrho_{y}(t) + \sigma_{z} \varrho_{z}(t) \right] = \frac{\1 +  \vec{\varrho} \cdot \vec{\sigma}}{2} \text{ , } \label{EqDensiMatrixDecompGen}
\end{equation}
where we define the matrix Pauli vector $\vec{\sigma} = \sigma_{x} \hat{\text{i}} + \sigma_{y} \hat{\text{j}} + \sigma_{z} \hat{\text{k}}$,
with $\{\hat{\text{i}},\hat{\text{j}},\hat{\text{k}}\}$ being the canonical basis of the $\Rmath^3$ Euclidean space, and the \textit{coherence vector} $\vec{\varrho} = \varrho_{x} \hat{\text{i}} + \varrho_{y} \hat{\text{j}} + \varrho_{z} \hat{\text{k}}$ for a two-level system.

\section{Quantum mechanics in superoperator formalism}\label{SecSuperOpForm}

In the theory of open systems, the dynamics of a quantum system $S$ coupled to an environment is generically described by a master equation that takes into account the system-environment interaction. In particular, here we consider the class of evolution where the system dynamics is described by the time-local master equation~\cite{Alicki:Book07,Petruccione:Book}
\begin{equation}
	\dot{\rho}(t) = \Lcalb_{t}[\rho(t)] \text{ , } \label{EqEqLind}
\end{equation}
where $\Lcalb_{t}[\bullet]$ is the generator of the dynamics and the subscript ``$t$'' makes explicit the possibility of time-dependency of the parameters associated with the environment. As the element $\Lcalb_{t}[\bullet]$ is an operation that leads an operator $A\in \Hcalb_{\text{S}}$ to another $A^{\prime} \in \Hcalb_{\text{S}}$, where $\Hcalb_{\text{S}}$ is the Hilbert space of the system $S$, we call $\Lcalb_{t}[\bullet]$ a \textit{superoperator}. We can deal with the above dynamics in different ways, but here we will consider the case in which we define an extended space where the superoperator becomes an $(D^2 \times D^2)$-dimensional operator $\Lmath(t)$ and the operators (density matrix and observables, for example) become (super-)vectors $\dket{\rho(t)}$ in this new space.

To see how it can be done, we define a set of $D_{-}=D^2 -1$ operators $\Ocal = \{\sigma_{n}\} \in \Hcalb_{\text{S}}$, so that $\trs{\sigma_{n}} = 0$ and $\trs{\sigma_{n}\sigma_{m}^{\dagger}} = D\delta_{nm}$. As already discussed, $\rho(t)$ can be written in form given in Eq.~\eqref{EqEqRhoCoherence}, where the identity is introduced here as a $D$-dimensional operator of this basis in order to satisfy the property $\trs{\rho(t)} = 1$, so that we define $\sigma_{0} = \1$. As an example we consider a two-level system, in which a complete basis $\Ocal$ for this system would be $\Ocal_{\text{tls}} = \{\sigma_{0},\sigma_{n}\} = \{\1,\sigma_{x},\sigma_{y},\sigma_{z}\}$, so that the coherence vector reads as in Eq.~\eqref{EqDensiMatrixDecompGen}. Obviously we could consider other choices for $\Ocal$, depending on the convenience of our choice. Thus, if we substitute $\rho(t)$ in Eq.~\eqref{EqEqLind} by using the expanded form of the Eq.~\eqref{EqEqRhoCoherence}, we find the system of differential equations
\begin{equation}
	\dot{\varrho}_{k} (t) = \frac{1}{D} \sum_{n=0}^{D_{-}} \varrho_{i}(t) \trs{\sigma_{k}^{\dagger} \Lcalb [ \sigma_{i} ]} \text{ , } \label{Eqv1}
\end{equation}
where we assume that $\Lcalb [\bullet]$ is a linear superoperator. Note that if we identify the coefficient $\trs{\sigma_{k}^{\dagger} \Lcalb [ \sigma_{i} ]}$ in the above equation as an element at $k$-th row and $i$-th column of a $(D^2 \times D^2)$-dimensional matrix $\Lmath(t)$, one can write
\begin{equation}
	\dket{\dot{\rho}(t)} = \Lmath(t) \dket{\rho(t)} \text{ , } \label{EqEqSuperLindEq}
\end{equation}
where $\dket{\rho(t)}$ is a $D^2$-dimensional vector with components $\varrho_{n}(t) = \trs{\rho(t)\sigma_{n}^{\dagger}}$, $n=0,1,\cdots D_{-}$. In an analogous way to the coherence vector $\vec{\varrho}(t)$ of the two-level system, we will call $\dket{\rho(t)}$ the \textit{coherence vector}. In addition, in this formalism, the inner product between two density matrices $\xi_{1}$ and $\xi_{2}$ from their coherence vectors $\dket{\xi_{1}}$ and $\dket{\xi_{2}}$ is defined as $\dinterpro{\xi_{1}}{\xi_{2}} = (1/D)\trs{\xi^{\dagger}_{1}\xi_{2}}$, where the conjugate coherence vector $\dbra{\xi_{1}}$ has components given by $\trs{\xi^{\dagger}_{1}\sigma_{n}}$.

Since we have a set of new coherence vectors $\dket{v}$, it is convenient to provide a basis for such vector space. Given the basis of operators $\{\sigma_{n}\}$, one defines the basis of coherence vectors as 
\begin{equation}
	\dket{\sigma_{k}} = \frac{1}{D}\begin{bmatrix} \trs{\sigma^{\dagger}_{1}\sigma_{k}} & \trs{\sigma^{\dagger}_{2}\sigma_{k}} & \cdots & \trs{\sigma^{\dagger}_{D^2-1}\sigma_{k}}
	\end{bmatrix}^{\text{t}} \text{ , }
\end{equation}
so that we can use $\trs{\sigma_{n}\sigma_{m}^{\dagger}} = D\delta_{nm}$ to write
\begin{equation}
	\dket{\sigma_{k}} = \begin{bmatrix} \delta_{1k} & \delta_{2k} & \cdots & \delta_{(D^2-1)k}
	\end{bmatrix}^{\text{t}} \text{ , } \label{EqSuperOpBasis}
\end{equation}
which constitutes the set of $D^2$ linearly independent vectors on the vector space, so that any coherence vector $\dket{v}$ can be written as a linear combination of the vectors in the set $\{\dket{\sigma_{k}}\}$.

An important remark is that, due to the non-Hermiticity of the operator $\Lcalb[\bullet]$, the superoperator $\Lmath(t)$ may not be diagonalizable. Nevertheless, non-Hermitian operators can always be written in the Jordan canonical form, where $\Lmath(t)$ displays a block-diagonal structure $\Lmath_{\text{J}}(t)$ composed by Jordan blocks $J_{n}(t)$ associated with different time-dependent eigenvalues $\lambda_{\alpha}(t)$ of $\Lmath(t)$~\cite{Horn:Book}. Then, firstly we define the following.

\begin{definition}[Jordan block form]\label{DefJordanForm}
	Given a $K\times K$ matrix $L$, the Jordan block form of $L$ reads
	\begin{align}
		L_{J} = \begin{bmatrix}
			J_{k_1}[\lambda_{k_1}] & 0                      & 0        & \cdots & 0 \\
			0                      & J_{k_2}[\lambda_{k_2}] & 0        & \cdots & 0 \\
			\vdots & \ddots & \ddots & \ddots & \vdots \\
			0 & \cdots & 0 & \ddots & 0  \\
			0 & \cdots & \cdots & 0 & J_{k_{N}}[\lambda_{N}] 
		\end{bmatrix}_{K \times K} \text{ , } \label{EqJordanFormMatrixJ}
	\end{align}
	where each block $J_{k_1}[\lambda_{k_1}]$ is given by an upper triangular matrix of the form
	\begin{align}
		J_{k}[\lambda] = \begin{bmatrix}
			\lambda & 1   & 0        & \cdots & 0 \\
			0 &\lambda & 1 & \cdots & 0 \\
			\vdots & \ddots & \ddots & \ddots & \vdots \\
			0 & \cdots & 0 & \lambda & 1  \\
			0 & \cdots & \cdots & 0 & \lambda  
		\end{bmatrix}_{k\times k} \text{ , } \label{EqJordanFormMatrix}
	\end{align}
	with $\lambda$ denoting the eigenvalues of $L$. Alternatively, a $K\times K$ Jordan matrix $L_{J}$ can be denoted as
	\begin{align}
		L_{J} = J_{k_{1}}[\lambda_{1}] \oplus J_{k_2}[\lambda_2] \oplus \cdots \oplus J_{k_N}[\lambda_N] = \bigoplus_{\alpha=1}^{N} J_{k_{\alpha}}[\lambda_{\alpha}] \text{ , \ \ \ } k_1+k_2+\cdots+k_N = K \text{ , }
	\end{align}
	where $N\leq K$ is the number of Jordan block required to write $L_{J}$ in a Jordan block form.
\end{definition}

In some cases, the coefficients $\lambda_n$ may depend on other parameters (time, for example). Then, by assuming that the coefficients $\lambda_n$ depend on a complete set of parameter $\xi = \{\xi_{1},\cdots \xi_{M}\}$, we denote the Jordan matrix as
\begin{align}
	L_{J}(\xi) = J_{k_1}[\lambda_1(\xi)] \oplus J_{k_2}[\lambda_2(\xi)] \oplus \cdots \oplus J_{k_N}[\lambda_N(\xi)]  \text{ . }
\end{align}

The notion of Jordan form is important here because, different from the generator of a unitary dynamics (the Hamiltonian), the generator of the dynamics in an open system $\Lmath(t)$ does not admit a diagonal form in general. However, every square matrix $A$ can be diagonalized by blocks from the \textit{Jordan canonical form theorem}~\cite{Horn:Book}.

\begin{theo}[Jordan canonical form theorem]
	Let be $A \in M_{K}$, where $M_{K}$ is a set of $K \times K$ square matrix. Then, there is a non-singular matrix $S \in M_{K}$, positive integers $N\leq K$ and $k_{1}, \cdots , k_{N}$ with $k_{1} + k_{2} + \cdots k_{N} = K$, and scalars $\lambda_1, \cdots , \lambda_N \in \Cmath$ so that
	\begin{align}
		A(\xi) = S(\xi) J_{A}(\xi) S^{-1}(\xi) \text{ , }
	\end{align}
	where $J_{A}(\xi) = J_{k_1}[\lambda_1(\xi)] \oplus J_{k_2}[\lambda_2(\xi)] \oplus \cdots \oplus J_{k_N}[\lambda_N(\xi)]$ is a Jordan matrix associated with $A(\xi)$.
\end{theo}

Now, by using the above discussion to the superoperator $\Lmath(t)$, we can write its Jordan form through the matrix $S(t)$ which allows us to write
\begin{equation}
	\Lmath_{\text{J}}(t) = S^{-1}(t) \Lmath(t) S(t) = \bigoplus_{\alpha=1}^{N} L_{N_\alpha}[\lambda_{\alpha}(t)] \text{ , } \label{EqEqLindJ}
\end{equation}
where $N$ is the number of distinct eigenvalues $\lambda_{\alpha}(t)$ of $\Lmath(t)$ and each block $L_{N_\alpha}[\lambda_{\alpha}(t)]$ is a $(N_\alpha \times N_\alpha)$-dimensional matrix given as in Eq.~\eqref{EqJordanFormMatrix}. Since the Hilbert space of the system has dimension $D$, one finds $N_1 + N_2 + \cdots + N_N = D^2$. In addition, as an immediate consequence of the structure of $\Lmath_{\text{J}}(t)$, one see that $\Lmath(t)$ does not admit the existence of eigenvectors. Instead eigenvectors, we define \textit{right} $\dket{\Dcalb_{\alpha}^{n_{\alpha}}(t)}$ and \textit{left} quasi-eigenvectors $\dbra{\Ecalb_{\alpha}^{n_{\alpha}}(t)}$ of $\Lmath(t)$ associated with the eigenvalue $\lambda_{\alpha}(t)$, satisfying
\begin{subequations}\label{EqEqEigenStateL}
	\begin{align}
		\Lmath(t)\dket{\Dcalb_{\alpha}^{n_{\alpha}}(t)} &= \dket{\Dcalb_{\alpha}^{(n_{\alpha}-1)}(t)} + \lambda_{\alpha}(t)\dket{\Dcalb_{\alpha}^{n_{\alpha}}(t)} \text{ , } \\
		\dbra{\Ecalb_{\alpha}^{n_{\alpha}}(t)}\Lmath(t) &= \dbra{\Ecalb_{\alpha}^{(n_{\alpha}+1)}(t)} + \dbra{\Ecalb_{\alpha}^{n_{\alpha}}(t)}\lambda_{\alpha}(t) \text{ . }
	\end{align}
\end{subequations}

The set $\{\dket{\Dcalb_{\alpha}^{n_{\alpha}}(t)}\}$ combined with $\{\dbra{\Ecalb_{\alpha}^{n_{\alpha}}(t)}\}$ constitutes a basis for the space associated with the operator $\Lmath(t)$ and satisfies the normalization condition $\dinterpro{\Ecalb_{\beta}^{m_{\beta}}(t)}{\Dcalb_{\alpha}^{n_\alpha}(t)} = \delta_{\beta\alpha} \delta_{m_{\beta}n_{\alpha}}$ and completeness relation
\begin{equation}
	\sum_{\alpha=1}^{N} \sum _{n_{\alpha} = 1}^{N_{\alpha}} \dket{\Dcalb_{\alpha}^{n_{\alpha}}(t)}\dbra{\Ecalb_{\alpha}^{n_{\alpha}}(t)} = \1_{D^2\times D^2} \text{ , }
\end{equation}
where $N$ is the number of Jordan blocks in Eq.~\eqref{EqEqLindJ} and $N_{\alpha}$ is the dimension of the $\alpha$-th Jordan block.

\let\cleardoublepage\clearpage

%%%%%%%%%%%%%%%%%%%%%%%%%%%%%%%%%%%%%%%%%%%%%%%%
%%%%%%%%%%%%%%%%% ADIABATICITY %%%%%%%%%%%%%%%%%
%%%%%%%%%%%%%%%%%%%%%%%%%%%%%%%%%%%%%%%%%%%%%%%%

%%%%%%%%%%%%%%%%% CLOSED SYSTEMS %%%%%%%%%%%%%%%%%
\chapter{Adiabatic dynamics in closed quantum systems}\label{ChapterAdiabaticClosedSystem}

\initial{I}n this chapter we will discuss the adiabatic dynamics in closed quantum systems and the validity of the conditions of the adiabatic theorem. We start this chapter by defining the notion of adiabaticity in closed systems and providing a review on some known adiabaticity conditions (AC's) and their validity in guaranteeing the adiabatic behavior of a quantum system with non-degenerate time-dependent Hamiltonian $H(t)$. It is important to mention that \textit{necessary} and \textit{sufficient} conditions to adiabatic dynamics are not a closed subject. In recent years, it has been proposed theoretical and experimental studies of how to deal with the adiabatic behavior for a particular class of time-dependent Hamiltonian, in which the widely used proposal of AC's do not work. After that, we introduce a new strategy that can be used to recover the ability of AC's in predicting the adiabatic behavior of a number of Hamiltonians. For this reason, we call it a validation mechanism for AC's. As a useful application, we discuss how the study of adiabatic dynamics in the rotating frame allows us for providing a stable charging/discharging process of quantum batteries. Quantum batteries are quantum devices used to store energy to be transferred to some consumption hub whenever required, where their charging/discharging performance (power) explores quantum resources, like coherence and entanglement. The contributions of this thesis presented in this chapter refer to the following works:

$\bullet$ C.-K. Hu, J.-M. Cui, A. C. Santos, Y.-F. Huang, C.-F. Li, G.-C. Guo, F. Brito, and M. S. Sarandy, ``Validation of quantum adiabaticity through non-inertial frames and its trapped-ion realization'', Sci. Rep. \textbf{9}, 10449 (2019).

$\bullet$ A. C. Santos, B. Çakmak, S. Campbell, and N. T. Zinner, ``Stable adiabatic quantum batteries'', Phys. Rev. E \textbf{100}, 032107 (2019).

$\bullet$ A. C. Santos, A. Saguia, and M. S. Sarandy, ``Stable and charge-switchable quantum batteries'', Phys. Rev. E \textbf{101}, 062114 (2020).

\newpage

\section{Review on validity conditions of the adiabatic theorem} \label{SecRevAd}

Let us consider a quantum system $\Scalb$ driven by a time-dependent (non-degenerate) Hamiltonian $H(t)$, its dynamics is dictated by Schrödinger equation
\begin{equation}
	i\hbar\ket{\dot{\psi} (t)} = H(t)\ket{\psi (t)} \text{ , }
\end{equation}
where $\hbar$ is the reduced Planck's constant ($\hbar = h/2\pi$, where $h$ is the Planck's constant). The set of eigenstates and energies of $H(t)$ are denoted as $\ket{E_{n}(t)}$ and $E_{n}(t)$, respectively. In closed quantum systems, the notion of an adiabatic dynamics is well-defined and widely used in several situations and studies, therefore, we shall use it here. An initial definition of adiabaticity is associated with an \textit{independent (uncoupled) dynamics of the eigenstates of the Hamiltonian with energy $E_{n}(t)$}. Thus, if we start the system in a particular eigenstate $\ket{E_{n}(t_{0})}$ with energy $E_{n}(t_{0})$, the definition of adiabaticity establishes that at a later instant $t>t_{0}$ the state of the system will be the updated eigenstate $\ket{E_{n}(t)}$ with energy $E_{n}(t)$. From a more rigorous way, we can define adiabaticity from the viewpoint of the Hilbert space of the system as follows~\cite{Sarandy:04}.

\begin{definition}[Adiabatic dynamics]
	Given a quantum system $\Scalb$ evolving under a time-dependent $H(t)$, the system dynamics is said to be adiabatic if the dynamics of its Hilbert space $\Hcal_{\Scalb}$ can be decomposed into uncoupled Schrödinger eigenspaces with distinct, and non-crossing instantaneous eigenvalues of $H(t)$.
\end{definition}

By \textit{Schrödinger eigenspaces} we mean an $N_{\alpha}$-dimensional complex vector space, whose basis is composed by the $N_{\alpha}$ (orthonormal) eigenvectors of $H(t)$ with the same eigenvalue. Of course, the number of eigenspaces is constrained to the dimension of the Hilbert space $\Hcal_{\Scalb}$. In the case where the Hamiltonian is non-degenerated, Schrödinger eigenspaces become one-dimensional. The adiabatic dynamics is not a general property of quantum systems and it cannot be achieved in any situation. Thus, the main question that arises is: \textit{what is (are) the condition(s) to achieve the adiabatic dynamics in quantum systems?} The answer to this question lead us to the so-called \textit{adiabaticity conditions} (AC's). AC's are a set of conditions (or a single condition) that should be satisfied for any system which follows an adiabatic dynamics.

A first validity condition to adiabatic theorem was initially derived by P. Ehrenfest~\cite{Ehrenfest:916}, after by M. Born and V. Fock~\cite{Born:28} and then by T. Kato~\cite{Kato:50} many years later. These works introduced the \textit{traditional} AC given by adiabatic parameter
\begin{equation}
	C_{\text{trad}} = \max_{t \in [ 0,\tau ]} \left[\max_{m \neq n}\left|\hbar \frac{\bra{E_{m}(t)} \dot{H}(t)\ket{E_{n}(t)}}{\left[ E_{m}(t) - E_{n}(t) \right]^2}\right|\right] \text{ ,} \label{EqAdTradCond}
\end{equation}
where the adiabatic dynamics would be achieved whenever $C_{\text{trad}} \ll 1$~\cite{Ehrenfest:916,Born:28,Kato:50,Sarandy:04}. 

In 2004, Marzlin and Sanders provided a counter-example in which the above condition fails in guaranteeing the adiabatic behavior~\cite{Marzlin:04}. To this end, Marzlin and Sanders investigated the dynamics of a Hamiltonian identical to that of a $\frac{1}{2}$-spin particle in a rotating magnetic field in $xy$-plane. In fact, let us consider the Hamiltonian $H = -\gamma \vec{S}\cdot\vec{B}(t)$ which describes the coupling between the $\frac{1}{2}$-spin particle to a time-dependent magnetic field $\vec{B}(t)$, where $\mu$ is the magnetic dipole momentum and $\vec{S} = S_{x}\hat{\text{i}}+S_{y}\hat{\text{j}}+S_{z}\hat{\text{k}}$ is the $\frac{1}{2}$-spin ``vector'' operator, with $S_{n} = \hbar \sigma_{n}/2$. Therefore, in case where $\vec{B}(t) = B_{0} \hat{\text{k}} + B_{1} \left[\cos(\omega t)\hat{\text{i}} + \sin(\omega t)\hat{\text{j}} \right]$, i.e., a rotating magnetic field in $xy$-plane, we get the Hamiltonian
\begin{equation}
	H(t) = \frac{\hbar \omega_{0}}{2} \sigma_{z} + \frac{\hbar \omega_{1}}{2} \left[ \cos(\omega t)\sigma_{x} + \sin(\omega t)\sigma_{y}\right] \text{ . } \label{EqNMRHamil}
\end{equation}

In case where the transverse field intensity $|B_{1}|$ is much less intense than the longitudinal field $|B_{0}|$, we write $|\omega_{1}| \ll |\omega_{0}|$. The fundamental and excited states of $H(t)$ are, respectively, given by
\begin{subequations}
	\label{EqNMREigenVectors}
	\begin{align}
		\ket{E_{-}(t)} &= -e^{ -i\omega t } \sin \left( \theta / 2 \right) \ket{1} + \cos \left( \theta / 2 \right) \ket{0} \text{ , } \\
		\ket{E_{+}(t)} &=  e^{ -i\omega t } \cos \left( \theta / 2 \right) \ket{1} + \sin \left( \theta / 2 \right) \ket{0} \text{ , }
	\end{align}
\end{subequations}
with energies $E_{\pm}= \pm \hbar \omega_{0}\sec(\theta)/2$, where $\tan(\theta) = \omega_{1}/\omega_{0}$. Since we are considering $|\omega_{1}| \ll |\omega_{0}|$, $\theta$ is a small parameter and, therefore, it means the fundamental state of the system can be written as $\ket{E_{-}(t)} \approx \ket{0}$. Thus, by considering the system is initialized in the fundamental state of $H(0)$, we can write $\ket{\psi(0)} \approx \ket{0}$, then we write $\rho(0) = \ket{\psi(0)}\bra{\psi(0)}= \ket{0}\bra{0}$. The dynamics of the system driven by above Hamiltonian in Eq.~\eqref{EqNMRHamil} is given by (see Appendix~\ref{ApFrameChangeQM-NMR}).
\begin{equation}
	\rho (t) = e^{i\frac{\omega }{2}t\sigma_{z}}e^{-\frac{i}{\hbar} H_{\text{R}}t } \rho(0) e^{\frac{i}{\hbar} H_{\text{R}}t }e^{-i\frac{\omega }{2}t\sigma_{z}} \text{ , }
\end{equation}
where $H_{\text{R}}$ is the Hamiltonian in the rotating frame provided by transformation $R(t) = e^{i\frac{1}{2}\omega t \sigma_{z}}$ and it is given by
\begin{equation}
	H_{\text{R}}(t) = H_{\text{R}} = \hbar \frac{\omega_{0}-\omega}{2} \sigma_{z} + \hbar \frac{\omega_{1}}{2} \sigma_{x} \text{ . } \label{EqNMRHamilR}
\end{equation}

Thus, by defining the projector $\Pcal_{0} = \ket{0}\bra{0}$, the probability of the system to be found at state $\ket{0}$ is given by (here $\tan\theta = \omega_{1}/\omega_{0}$)
\begin{equation}
	p_{0}(t) = \tr{\rho(t)\Pcal_{0}} = \frac{2(1-r)^2 + \left[ 1+ \cos \left( \omega_0 t \sqrt{(1-r)^2 + \tan^2\theta} \right) \right] \tan^2\theta }{2 \left[ (1-r)^2 + \tan^2\theta \right]}  \text{ , } \label{EqProbGroundNMR}
\end{equation}
where we define $r = \omega/\omega_{0}$ as a resonant approximation ratio, that is, when $r \rightarrow 1$ we are close to resonance situation, otherwise we are far-from resonance (or in an intermediate stage). From Eq.~\eqref{EqProbGroundNMR}, it is worth highlighting the maximum and minimum fidelity times as
\begin{equation}
	\tau_{\text{max}} = \frac{2n\pi}{\omega_0 \sqrt{(1-r)^2 + \tan^2\theta}} \quad \text{and} \quad
	\tau_{\text{min}} = \frac{(2n+1)\pi}{\omega_0 \sqrt{(1-r)^2 + \tan^2\theta}} \label{EqtauMaxMinNMR} \text{ , }
\end{equation}
where
\begin{equation}
	p_{0}(\tau_{\text{max}}) = 1 \quad \text{and} \quad  p_{0}(\tau_{\text{min}}) = \frac{(1-r)^2 }{(1-r)^2 + \tan^2\theta} \text{ . } 
\end{equation}

As result, one can see that the minimum fidelity depends on the rf-field configuration. Therefore, now we study three different dynamics associated with near and far from resonance configuration of the rf-field.

\emph{Situation 1: Spin dynamics at resonance ($r \rightarrow 1$)} -- In this limit the above equation reads
\begin{equation}
	p_{0}^{\text{ress}}(t) = \lim_{r \rightarrow 1} p_{0}(t) = \cos^2 \left( \frac{\omega_1 t}{2} \right)  \text{ , } \label{EqProbGroundNMRRess}
\end{equation}
so that for each instant $t = \tau_{n} = (2n+1)\pi/\omega_{1}$, for $n\in \Zmath^{+}$ the probability of approximately finding the system at  the fundamental state (state $\ket{0}$) is zero. Now, note that there is no adiabaticity and it just depends on the limit $r \rightarrow 1$.

\emph{Situation 2: Spin dynamics for very-slowly oscillating rf-field ($r \rightarrow 0$)} -- As a first far-from situation we can consider the case where $r \rightarrow 0$, in which the Eq.~\eqref{EqProbGroundNMR} becomes (the superscript ``ffr'' means far-from resonance)
\begin{equation}
	p_{0}^{\text{ffr}(1)}(t) = \lim_{r \rightarrow 0} p_{0}(t) = \cos^2 \theta + \cos^2 \left( \frac{1}{2} \omega_{1} t \csc \theta\right) \sin^2 \theta  \text{ , } \label{EqProbGroundNMRffr1}
\end{equation}
where the minimum and maximum value of $p_{0}^{\text{ffr}(1)}(t)$ are $\cos^2 \theta$ and $1$, respectively. Therefore, since we are considering that $\omega_{0} \gg \omega_{1}$, so the parameter $\theta \ll 1$ and hence $\cos^2 \theta \approx 1$. It means that for a very-slowly oscillating rf-field we are close to the fundamental state and adiabaticity is approximately obtained.

\emph{Situation 3: Spin dynamics for highly oscillating rf-field ($r \gg 1$)} -- Now, in the limit of high oscillating rf-field, we can write
\begin{equation}
	p_{0}^{\text{ffr}(2)}(t) = \lim_{r \gg 1} p_{0}(t) \rightarrow 1 - \Ocal\left( r^{-2} \right)
\end{equation}
in which we are neglecting terms smaller than $r^{-2}$. It is a bit counter-intuitive because we have a highly oscillating field driven the system through an adiabatic path (approximately).

Now, given that we know the exact solution for the dynamics of the system (consequently the transition probabilities), let us compute each AC in order to verify their validation. First, the traditional adiabaticity condition in Eq.~\eqref{EqAdTradCond} to the above dynamics reads
\begin{equation}
	C_{\text{trad}}^{\text{nmr}} = \max_{t \in [ 0,\tau ]} \left| \hbar \frac{|\bra{E_{-}(t)} \dot{H}(t)\ket{E_{+}(t)}|}{\left[ E_{-}(t) - E_{+}(t) \right]^2}\right| = \frac{1}{2} |r\sin(\theta) \cos(\theta)| \text{ . } \label{EqTACnmr}
\end{equation}

The above equation presents a strictly increasing linear behavior of the quantity $C_{\text{trad}}^{\text{nmr}}$ as function of $r$. Regarding the Situations 1 and 3 above, the prediction of the adiabaticity as provided by $C_{\text{trad}}^{\text{nmr}}$ is not valid. In particular, in this case, the traditional condition is neither \textit{sufficient} and nor \textit{necessary} condition to predict adiabaticity. In fact, the Situation 1 shows that we do not have adiabaticity even when the parameter $C_{\text{trad}}^{\text{nmr}}$ is small (concerning its value for other choices of $r$). Therefore, it means the traditional condition is not \textit{sufficient} to guarantee adiabaticity. On the other hand, Situation 3 shows that we have adiabaticity when $C_{\text{trad}}^{\text{nmr}}$ is not small (since $r \gg 1$, then it is possible to find $r$ so that $C_{\text{trad}}^{\text{nmr}} \gg 1$). Here, one concludes that the traditional condition is not \textit{necessary} to get adiabaticity.

After Marzlin and Sanders results, other works discussed about the viability of using quantitative conditions to guarantee the validity of the adiabatic theorem in quantum systems. D. M. Tong \textit{et al}~\cite{Tong:05} showed that quantitative conditions do not guarantee the validity of the adiabatic theorem, where they provide a revision of the results presented by Marzlin and Sanders. In order to solve that problem, D. M. Tong \textit{et al}~\cite{Tong:07} provided a new condition (actually, a set of conditions) which should guarantee the adiabatic approximation when satisfied. Without loss of generality we can write the set of conditions as
\begin{subequations}
	\label{EqAdTongCondAll}
	\begin{align}
		C_{\text{Tong}}^{(\text{a})} &= \max_{t \in [ 0,\tau ]} \left[\max_{m \neq n}\left|\hbar \frac{|\bra{E_{m}(t)} \dot{H}(t)\ket{E_{n}(t)}|}{\left[ E_{m}(t) - E_{n}(t) \right]^2}\right|\right] \text{ , } \label{EqAdTongCondA}\\
		C_{\text{Tong}}^{(\text{b})} &= \max_{m \neq n} \left|\hbar
		\frac{d}{dt} 
		\left[\frac{\bra{E_{m}(t)} \dot{H}(t)\ket{E_{n}(t)}}{\left[ E_{m}(t) - E_{n}(t) \right]^2}\right]
		\right|_{M} \tau \text{ , } \label{EqAdTongCondB}\\
		C_{\text{Tong}}^{(\text{c})} &= \max_{m \neq n,l \neq m} \left|\hbar
		\frac{\bra{E_{m}(t)} \dot{H}(t)\ket{E_{n}(t)}}{\left[ E_{m}(t) - E_{n}(t) \right]^2}
		\right|_{M} \left|\interpro{E_{m}(t)}{\dot{E}_{l}(t)}\right|_{M}\tau \text{ , } \label{EqAdTongCondC}
	\end{align}
\end{subequations}
where~\eqref{EqAdTongCondA} is exactly the traditional condition~\eqref{EqAdTradCond}, but $C_{\text{Tong}}^{(\text{b})}$ and $C_{\text{Tong}}^{(\text{c})}$ are two additional conditions introduced in Ref.~\cite{Tong:07}. Here the index ``$M$'' in $|f(t)|_{M}$ denotes the mean value of $f(t)$ in the interval $t\in [0,\tau]$. Here it is worth mentioning that some years after Ref.~\cite{Tong:07}, D. M. Tong authored a paper~\cite{Tong:10} where he argued how quantitative conditions could be necessary to guarantee the validity of the adiabatic approximation, but no discussion about the sufficiency is considered in. Since adiabaticity is achieved whenever all the above conditions satisfy $C_{\text{Tong}}^{(\text{x})}\ll 1$, then here we can define the Tong's coefficient from set of coefficient in Eq.~\eqref{EqAdTongCondAll} as
\begin{equation}
	C_{\text{Tong}} = \max \left\{ C_{\text{Tong}}^{(\text{a})}, C_{\text{Tong}}^{(\text{b})}, C_{\text{Tong}}^{(\text{c})} \right\} \text{ , } \label{EqAdTongCond}
\end{equation}
so that $C_{\text{Tong}}\ll 1$ would be enough in guaranteeing adiabaticity. Without more calculations, we can conclude that the Tong's conditions are not satisfied to the Marzlin-Sanders Hamiltonian in Eq.~\eqref{EqNMRHamil}, since the condition $C_{\text{Tong}}^{(\text{a})}$ will provide the same result as in Eq.~\eqref{EqTACnmr}.

A third interesting AC to be considered here is the AC as provided by Jian-da Wu \textit{et al}~\cite{Jianda-Wu:08}, where the authors introduced the notion of \textit{quantum geometric potential} associated with an adiabatic path, also called \textit{adiabatic orbit} in Ref.~\cite{Jianda-Wu:08}. The Wu's condition establishes that an $N$-level quantum system can follow an adiabatic orbit with fidelity (probability) $\Fcalb(\delta) = (1-\delta)^2$ if
\begin{equation}
	\max_{m,k,n \neq m}\frac{|\gamma_{km}(t)|}{|E_{n}(t) - E_{m}(t) + \Delta_{mn}(t)|} \leq \frac{\delta}{\sqrt{N-1}} \text{ , }
\end{equation}
with $\gamma_{nm}(t) = i\interpro{E_{n}(t)}{\dot{E}_{m}(t)}$ being the geometrical Berry's phase~\cite{Berry:84} and $\Delta_{mn}(t)$ the the previously mentioned \textit{quantum geometric potential}~\cite{Jianda-Wu:08} given by
\begin{equation}
	\Delta_{mn}(t) = \gamma_{mm}(t) - \gamma_{nn}(t) + \frac{d}{dt} \arg \left[ \gamma_{nm}(t) \right] \text{ . }
\end{equation}

Since the above equation allows us to determinate the probability of finding the system following an adiabatic orbit as function of $\delta$, it is reasonable to adopt $\delta \ll 1$ for guaranteeing that we will get an adiabatic dynamics with high probability. By assuming that, here we consider the Wu's AC as
\begin{equation}
	C_{\text{Wu}} = \max_{t \in [ 0,\tau ]} \left[\max_{m,k,n \neq m}\frac{\sqrt{N-1}|\gamma_{km}(t)|}{|E_{n}(t) - E_{m}(t) + \Delta_{mn}(t)|}\right] \text{ , } \label{EqAdCondWu}
\end{equation}
which should satisfy $C_{\text{Wu}}\ll 1$ for an adiabatic behavior. Basically, the Wu's condition requires we need to take the traditional AC and replace the term $E_{n}(t) - E_{m}(t)$ for $E_{n}(t) - E_{m}(t) + \Delta_{mn}(t)$ instead. Therefore, due the new term $\Delta_{mn}(t)$, it is convenient to study whether $C_{\text{Wu}}$ allows us to describe the Marzlin-Sanders Hamiltonian dynamics. Thus, we can write
\begin{equation}
	C_{\text{Wu}}^{\text{nmr}} = \max_{t \in [ 0,\tau ]} \frac{|\gamma_{10}(t)|}{|E_{1}(t) - E_{0}(t) + \Delta_{10}(t)|} = \frac{|r\sin(\theta) \cos(\theta)|}{\sqrt{ 4 + r^2 + r^2 \cos(2\theta) \left[ 2 + \cos(2\theta)\right] }}\text{ , }
\end{equation}
where we already used the Eqs~\eqref{EqNMREigenVectors}. Due the non-trivial behavior of $C_{\text{Wu}}^{\text{nmr}}$ concerning the resonance ratio $r$, we will analyze its validity to describe the adiabatic dynamics after we present the last AC to be considered in this thesis.

To end, we discuss now the AC derived by Andris Ambainis and Oded Regev~\cite{Ambainis:04}. An important element in Ambainis-Regev AC is associated with a parameter $\lambda(t)$ which quantifies the spacing between instantaneous energy $E_{n}(t)$, associated with the target eigenstate of $H(t)$, and the its adjacent energy levels $E_{n+1}(t)$ and $E_{n-1}(t)$. For example, a possible choice of the parameter $\lambda(t)$ could be $\lambda(t) = \left[E_{n+1}(t) +E_{n-1}(t)\right]/2$. Therefore, given our choice on $\lambda(t)$, the Ambainis-Regev condition establishes that, given the quantity if we get
\begin{equation}
	\frac{10^5}{\delta^2 \tau} \max \left\{ \frac{||\tau\dot{H}(t)||^3}{\lambda^4} , \frac{||\tau\dot{H}(t)||\cdot ||\tau^2\ddot{H}(t)||}{\lambda^3} \right\}  \leq 1 \text{ , }
\end{equation}
so the instantaneous evolved state $\ket{\psi(t)}$ will be close to eigenstate $\ket{E_{n}(t)}$, where the distance between each other is $\delta$. In order to write the Ambainis-Regev AC in terms of an adiabaticity coefficient $C_{\text{AR}}$, we consider
\begin{equation}
	C_{\text{AR}} = \max_{t\in [0,\tau]} \left[\max \left\{ \frac{\hbar||\dot{H}(t)||^3}{\lambda^4} , \frac{||\dot{H}(t)||\cdot ||\ddot{H}(t)||}{\lambda^3} \right\}\right] \tau^2\text{ , } \label{EqAdCondAR}
\end{equation}
where we neglected the term $10^5/\delta^2$ because we are interested here in qualitative characteristic of the AC's. That is, we expect that near to resonance the AC's provide bigger values concerning that one obtained far-from resonance. For this reason, the term $10^5/\delta^2$ does not develop a important role in our study. Thus, by applying the above condition to the Marzlin-Sanders Hamiltonian, we obtain
\begin{equation}
	C_{\text{AR}}^{\text{nmr}} = \max \left\{ \frac{r^3 \cos\theta \sin^{3}\theta}{2 \sqrt{2}} , \frac{r^3 \cos\theta \sin^{2}\theta}{2} \right\}\omega_0^{2}\tau^2 \text{ , }
\end{equation}
since we have a two-level system we already wrote $\lambda(t) = E_{1}(t) - E_{0}(t)$.

\begin{figure}
	%	\input{Figs/FigAdCondSutter.plt}
	%	\vspace{5.2cm}
	\centering
	\includegraphics[scale=0.55]{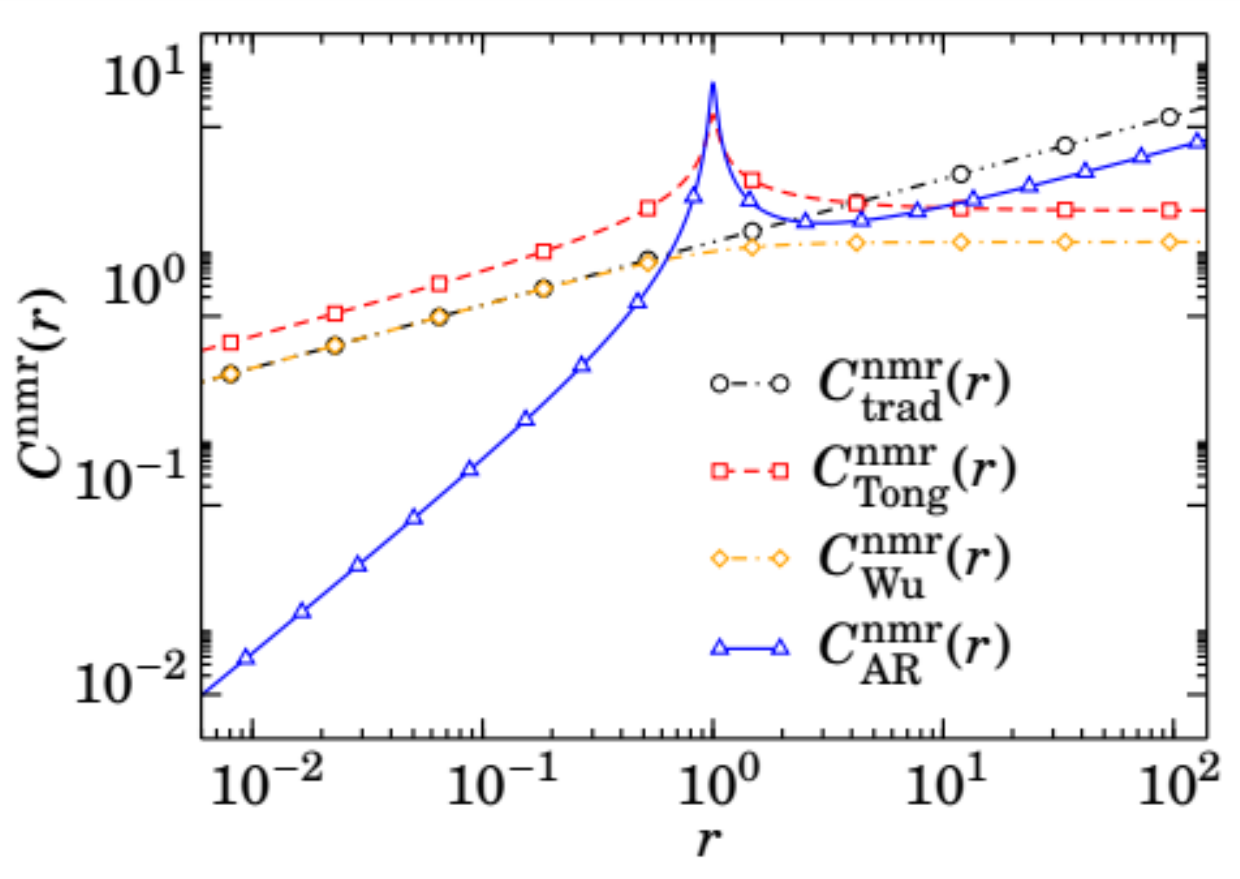}
	\caption{Adiabatic parameters $C_{n}$ as function of the resonance ratio $r$ for the Hamiltonian in Eq.~\eqref{EqNMRHamil}. Here we consider the total evolution time $\tau = \tau_{\text{min}}$ given in Eq.~\eqref{EqtauMaxMinNMR}, with $n=0$ and $\theta = 0.03$, so that $\omega_{1} \approx 3 \cdot 10^{-2} \omega_{0}$.}
	\label{FigAdCondSutter}
\end{figure}

In order to make a comparison between the AC considered here, in Fig.~\ref{FigAdCondSutter} we present each quantity computed to the Hamiltonian in Eq.~\eqref{EqNMRHamil}. From Eq.~\eqref{EqProbGroundNMR}, the adiabatic behavior is obtained when the parameter $r$ is large enough so that the oscillatory contribution of the co-sine function Eq.~\eqref{EqProbGroundNMR} can be neglected, otherwise the oscillations will provide small values of the probability and it means we do not have adiabaticity. Now, we can study this phenomenon under viewpoint of the ACs shown in Fig.~\ref{FigAdCondSutter}, where the same parameters are also considered in Ref.\cite{Suter:08}. 

Firstly, let us focus on the situation in which we have large values for $r$. For this range of values for $r$, adiabaticity is observed and the AC should describe such phenomena. However, the graphs in Fig.~\ref{FigAdCondSutter} shows that the conditions $C_{\text{AR}}^{\text{nmr}}$ and $C_{\text{trad}}^{\text{nmr}}$ suggest we should not get the adiabatic behavior, while the conditions $C_{\text{Tong}}^{\text{nmr}}$ and $C_{\text{Wu}}^{\text{nmr}}$ presents a consistent (expected) behavior. This analysis suggests that the conditions $C_{\text{AR}}^{\text{nmr}}$ and $C_{\text{trad}}^{\text{nmr}}$ \textit{are not necessary conditions} to guarantee adiabaticity, since the system evolves through an adiabatic  path even when such conditions do not indicate such phenomenon. On the other side, in case where we have small values for $r$, the conditions $C_{\text{Tong}}^{\text{nmr}}$ and $C_{\text{Wu}}^{\text{nmr}}$ present inconsistency with the Eq.~\eqref{EqProbGroundNMR}, at the same time that the conditions $C_{\text{AR}}^{\text{nmr}}$ and $C_{\text{trad}}^{\text{nmr}}$ shows an expected behavior. As conclusion, it is possible to see that the conditions $C_{\text{trad}}^{\text{nmr}}$ and $C_{\text{Wu}}^{\text{nmr}}$ \textit{are not sufficient} in guaranteeing adiabaticity. %In summary, the Table~\ref{TableAdCondi} shows the conclusion obtained from this dynamics. 

%\begin{table}[t!]
%	\centering
%	\begin{tabular}{c|cc}
%		Condition & Necessary      & Sufficient     \\ \hline
%		$C_{\text{trad}}$      & No             & No             \\
%		$C_{\text{Tong}}$      & Apparently yes & Apparently yes \\
%		$C_{\text{Wu}}$        & Apparently yes & No             \\
%		$C_{\text{AR}}$        & No             & Apparently yes
%	\end{tabular}
%\caption{Table with a summary on the validity conditions considered here, where we discuss on their \textit{necessity} and %\textit{sufficiency} in guaranteeing adiabaticity. It is important to mention that the above table is based on main results discussed in %Ref.~\cite{Suter:08}.}
%\label{TableAdCondi}
%\end{table}

While the traditional adiabatic parameter is neither necessary nor sufficient to guarantee adiabaticity, \textit{apparently} the Tong's conditions seems to be both necessary and sufficient to predict adiabaticity. Here we are using the term ``apparently'' because we are not able to confirm such result in general. By the way, we shall see that all of the above conditions are neither necessary nor sufficient to predict adiabaticity. 

\section{Frame-dependence of the adiabatic parameters} \label{SecAdFrameDependence}

In general, the quantities presented in Eqs.~\eqref{EqAdTradCond},~\eqref{EqAdTongCond},~\eqref{EqAdCondWu} and~\eqref{EqAdCondAR} are computed taking into account that no frame transformation into Schrödinger's equation is computed. We mean, either the system Hamiltonian is written taking into account the laboratory frame or no frame is explicitly considered. We just compute them. For example, when one writes the NMR Hamiltonian given in Eq.~\eqref{EqNMRHamil}, so it is intrinsic that we are computing the system energy for the Laboratory frame. However, when we consider the time-dependent transformation $R(t) = e^{i\omega t \sigma_{z}/2}$ as done to obtain the new time-independent Hamiltonian $H_{\text{R}}$ in Eq.~\eqref{EqNMRHamilR}, in an inconsistently way we are implementing a change of frame in Schrödinger's equation. Thus, it makes sense to consider the notion of reference frame changes in quantum mechanics and it was done in 1997 by W. H. Klink~\cite{Klink:97}. Here, we will present how the adiabatic parameters discussed in previous sections depend on the frame representation used to describe adiabaticity.

To study the description of the quantum mechanics in different reference frame, let us rewrite the Schrödinger's equation from the Von-Neumann form as
\begin{equation}
	\dot{\rho}(t) = \frac{1}{i\hbar} [H(t),\rho(t)] \text{ , } \label{EqVonNeumannLab}
\end{equation}
where $\rho(t)$ is the density matrix and $H(t)$ is the time-dependent Hamiltonian which drives the system under laboratory frame. Now, let us consider a second observer, say Bob, that describes the system dynamics from a different frame. For such observer, the system state reads
\begin{equation}
	\rho_{\text{Bob}}(t) = B(t) \rho(t) B^{\dagger}(t) \text{ , }
\end{equation}
for some unitary operator $B(t)$. Particularly, $B(t)$ needs to be unitary because $\rho_{\text{Bob}}(t)$ should be a density matrix, so that it has the same properties as $\rho(t)$. Therefore, by putting $\rho_{\text{Bob}}(t)$ into Eq.~\eqref{EqVonNeumannLab} we get the dynamical equation in Bob's frame as
\begin{equation}
	\dot{\rho}_{\text{Bob}}(t) = \frac{1}{i\hbar} [H_{\text{Bob}}(t),\rho_{\text{Bob}}(t)] \text{ , } \label{EqVonNeumannBob}
\end{equation}
where $H_{\text{Bob}}(t)$ is the Hamiltonian in this new frame and reads
\begin{equation}
	H_{\text{Bob}}(t) = B(t) H(t) B^{\dagger}(t) + i \hbar \dot{B}(t)B^{\dagger}(t) \text{ . }
\end{equation}

This means that, under viewpoint of Bob, the system is driven by different fields. In general, notice that the term $i \hbar \dot{B}(t)B^{\dagger}(t)$ arises when we consider a time-dependent transformation on the original equation. W. H. Klink identified the quantity $i \hbar \dot{B}(t)B^{\dagger}(t)$ as a ``fictitious potential'' which arises due to the frame change as provided by $B(t)$~\cite{Klink:97}. The term ``fictitious potential'' is used here because the above equation has an analog in classical mechanics, where we rewrite the Newton's law of motion in different reference frame.

\subsection{Adiabatic dynamics of a two-level system in an oscillating field}

In particular, the results previously discussed are valid for the case considered in Ref.~\cite{Suter:08}; namely, a spin-$\frac{1}{2}$ in a rotating magnetic field. No discussion of how this result can be applied to other systems (different Hamiltonians) is discussed. Now, we discuss how these results are not valid for a system of a spin-$\frac{1}{2}$ in the presence of an \textit{oscillating} magnetic field. The following results will be used to motivate the study of the adiabatic parameters from different reference frames, where a mechanism to recover the validity of the AC's arise as an alternative approach~\cite{Hu:19-b}.

Thus, consider the following Hamiltonian
\begin{equation}
	H(t) = \frac{\hbar}{2}\left[ \omega_{0} \sigma_{z} + \omega_{1} \sin(\omega t)\sigma_{x} \right] \text{ , } \label{EqOscHamil}
\end{equation}
that describes the dynamics of a two-level system interacting with an external \textit{oscillating} field. Here we are interested in a situation where $|\omega_{0}|\gg |\omega_{1}|$, since our main interest here is studying the resonance effects on the adiabaticity behavior of the system. The Hamiltonian in Eq.~\eqref{EqOscHamil} exhibits a resonance situation when we set the external field in a configuration so that $\omega \approx \omega_{0}$ (see Appendix~\ref{ApFrameChangeQM-TLSOscila}). In terms of the dimensionless parameter $r$ ($\omega = r\omega_{0}$) the resonance condition becomes $r \approx 1$.

\subsection{Experimental verification in a trapped ion system and adiabaticity conditions} \label{SecExpAdCondCS}

In order to study the resonance phenomena of a two level system driven by the oscillating field which characterizes the Hamiltonian in Eq.~\eqref{EqOscHamil}, it was used a trapped ion setup where the hyperfine energy level of the Ytterbium ($^{171}$Yb$^+$) trapped ion in a six-needles Paul trap was used as a two-level system~\cite{Olmschenk:07}. Since the focus of this thesis is not to provide an experimental description of the system used in our experiment, the reader who wishes to get more details on the experimental setup could read the Appendix~\ref{ApTrappeIon} and/or Ref.~\cite{Olmschenk:07}. 

As discussed in Appendix~\ref{ApTrappeIonQubit}, our qubit is encoded in a hyperfine energy levels of the ground state $^{2}S_{1/2}$ as shown in Fig.~\ref{FigAdiabExpTrappedIon}{\color{blue}a}. The ground state $^{2}S_{1/2}$ is constituted of two orbitals with different values of the total momentum angular $F$ denoted by $^{2}S_{1/2}\ket{F=0}$ and $^{2}S_{1/2}\ket{F=1}$. In particular, the state $^{2}S_{1/2}\ket{F=1}$ can be decomposed in three degenerate states (at zero external field) $^{2}S_{1/2}\ket{F=1,m_{F} = -1}$, $^{2}S_{1/2}\ket{F=1,m_{F} = 0}$ and $^{2}S_{1/2}\ket{F=1,m_{F} = 1}$. When we introduce an external magnetic field, due the Zeeman effect, these three states become non-degenerate (degeneracy break) with a frequency shift $\delta_{\text{Z}}$ (first order Zeeman shift), as shown in Fig.~\ref{FigAdiabExpTrappedIon}{\color{blue}b}. Thus, the qubit can be encoded on these states as $\ket{0} \equiv \text{} ^{2}S_{1/2}\ket{F=0,m_{F} = 0}$ and $\ket{1} \equiv \text{} ^{2}S_{1/2}\ket{F=1,m_{F} = 0}$, whose transition frequency is given by $\omega_{\text{s}}$. Transitions between these states is achieved by introducing a high controllable resonant microwave with adjustable Rabi frequency $\Omega_{\text{r}}(t)$ and oscillation frequency $\omega_{\text{f}}$.

\begin{figure}[t!]
	%\input{Figs/FigAdiabExpTrappedIon.plt}
	%\vspace{5.55cm}
	\centering
	\includegraphics[scale=0.7]{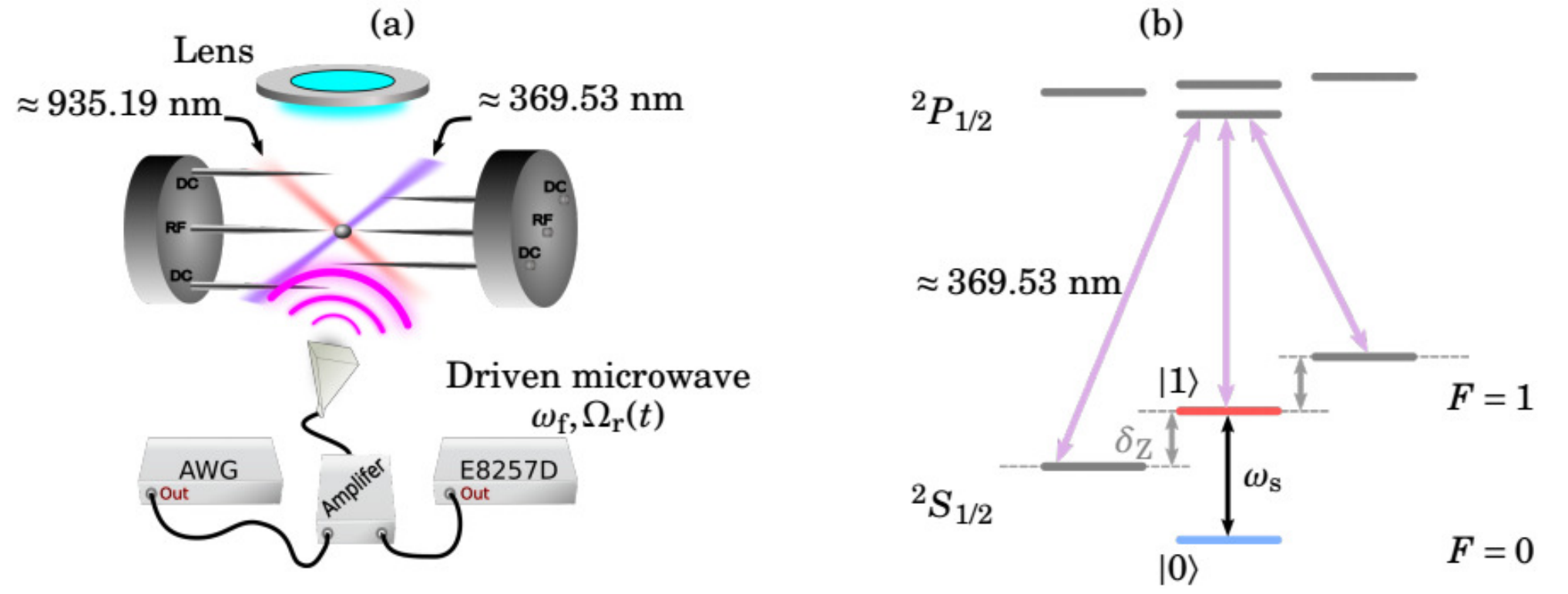}
	\caption{({\color{blue}a}) Schematic representation of the experimental setup used to manipulate the trapped ion Ytterbium qubit.  A $369.53$~nm readout laser is used to state tomography, meanwhile the qubit is encoded in the hyperfne energy levels and	coherently driven by a microwave controlled with a programmable arbitrary waveform generator (AWG). The light at $935.2$~nm is used to initialization and cooling~(see more detail in Appendix~\ref{ApTrappeIon}). ({\color{blue}a}) Relevant energy levels of the Ytterbium $^{171}$Yb$^+$ in the presence of a static magnetic field $\vec{B}$ to promote some Zeeman shift $\delta_{\text{Z}} = 1.4$~MHz/G in degenerate states with $F=1$. The states used as our two-level system are highlighted above.}
	\label{FigAdiabExpTrappedIon}
\end{figure}

As shown in Appendix~\ref{ApManipulationQubit}, the parameters of the Hamiltonian in Eq.~\eqref{EqOscHamil} can be efficiently implemented using a single microwave with suitable choices for oscillation and Rabi frequencies. In particular, the component $\sigma_{z}$ of the Hamiltonian can be implemented by a detuning given by $\omega_{0} = \omega_{\text{s}} - \omega_{\text{f}}$, while the second component $\sigma_{x}$ is controlled by a time-dependent Rabi frequency given by $\Omega_{\text{r}}(t) = \omega_{1} \sin(\omega t)$, which can be suitably adjusted from a field with time-dependent intensity. By starting from the initial state $\ket{\psi(0)} = \ket{0}$, in Fig.~\ref{FigAdFidelityOscHamil} we show the fidelity for the ground state for different choices of the parameter $\omega$. The fidelity here is defined as
\begin{equation}
	\Fcalb\left(\rho,\rho_{\text{ad}}\right) = \tr{ \sqrt{\sqrt{\rho_{\text{ad}}(t)} \rho(t) \sqrt{\rho_{\text{ad}}(t)} }}  \text{ , }
\end{equation}
where $\rho(t)$ is the solution of the Schrödinger equation and $\rho_{\text{ad}}(t)$ is the adiabatic solution. To compute the experimental values for fidelity we do state tomography as discussed in Appendix~\ref{ApIonTomography}. As theoretically predict in Appendix~\ref{ApFrameChangeQM-TLSOscila}, when $r \approx 1$ we can see transitions between the energy levels of the system, where the minimum fidelity becomes $\Fcalb\left(\rho,\rho_{\text{ad}}\right) \approx 0$ for $\tau_{\text{min}} \approx 50~\mu$s, then adiabaticity is broken when $r \approx 1$. Otherwise, for sufficiently small or large values of $r$, adiabaticity is observed for all of evolution time.

Now, to compute the AC's for the Hamiltonian in Eq.~\eqref{EqOscHamil}, we write the set of eigenvectors of $H(t)$ is given by
\begin{equation}
	\ket{E_{\pm}(t)} = \frac{1}{\Ncal_{\pm}(t)}\left[\alpha_{\pm}(t)\ket{1} + \ket{0}\right] \label{EqStatesHOscInert}
\end{equation}
where $\Ncal_{\pm}^{2}(t) = 1+ \alpha_{\pm}^2(t)$, with
$\alpha_{\pm}(t) = \csc \theta \csc (\omega t)\left[ 2 \cos\theta \pm  \Sigma  \right] /2$ and
$\Sigma^2 = 3+\cos(2\theta)-2\cos(2\omega t)\sin^2 (\theta)$. The energies are
\begin{equation}
	E_{\pm}(t) = \pm \frac{1}{4} \hbar \omega_{0} \sec (\theta)\Sigma \text{ . }
\end{equation}

Thus, we compute the AC's as provided by Eqs.~\eqref{EqAdTradCond},~\eqref{EqAdTongCond},~\eqref{EqAdCondWu} and~\eqref{EqAdCondAR} and the result is presented in Fig.~\ref{FigAdCondOscHamilInertial}. From all of curves in Fig.~\ref{FigAdCondOscHamilInertial}, we set the value of the total evolution time given by experimental one shown in Fig.~\ref{FigAdFidelityOscHamil} and the maximization in each equation is obtained in interval $t \in [0,\tau]$, so it does not matter the multiplicative value of $\tau$ for the conditions $C_{\text{Tong}}$ and $C_{\text{AR}}$ because we are computing the maximum value for all $t \in [0,\tau]$. In fact, any value of $C_{\text{Tong}}$ and $C_{\text{AR}}$ in some sub-interval of $t \in [0,\tau]$ will provides a same qualitative behavior shown in Fig.~\ref{FigAdCondOscHamilInertial}. As result, the AC's tell us that adiabaticity is observed for values of $r \ll 1$, while such adiabatic behavior is less probable as $r$ becomes sufficiently big. In particular, it is important to note the case $r \approx 1$ in comparative with the case $r \gg 1$. In conclusion, based on Figs.~\ref{FigAdCondOscHamilInertial} and~\ref{FigAdFidelityOscHamil}, as we can see, all of the conditions considered here do not allow us to get a correct description of the adiabatic behavior of the system. In summary, we can say that it allows us to conclude that all of AC's are \textit{neither sufficient}, because the validity is violated when we compare $r \approx 1$ with the results for $r \gg 1$, \textit{nor necessary} when we consider the case $r \ll 1$ with that one for $r \gg 1$. %The summary shown in Table~\ref{TableAdCondi} is wrong, where now we establish the new correct summary as shown in  Table~\ref{TableAdCondiOsc}.

\begin{figure}
	%\input{Figs/FigAdFidelityOscHamil.plt}
	%\vspace{5.25cm}
	\centering
	\includegraphics[scale=0.55]{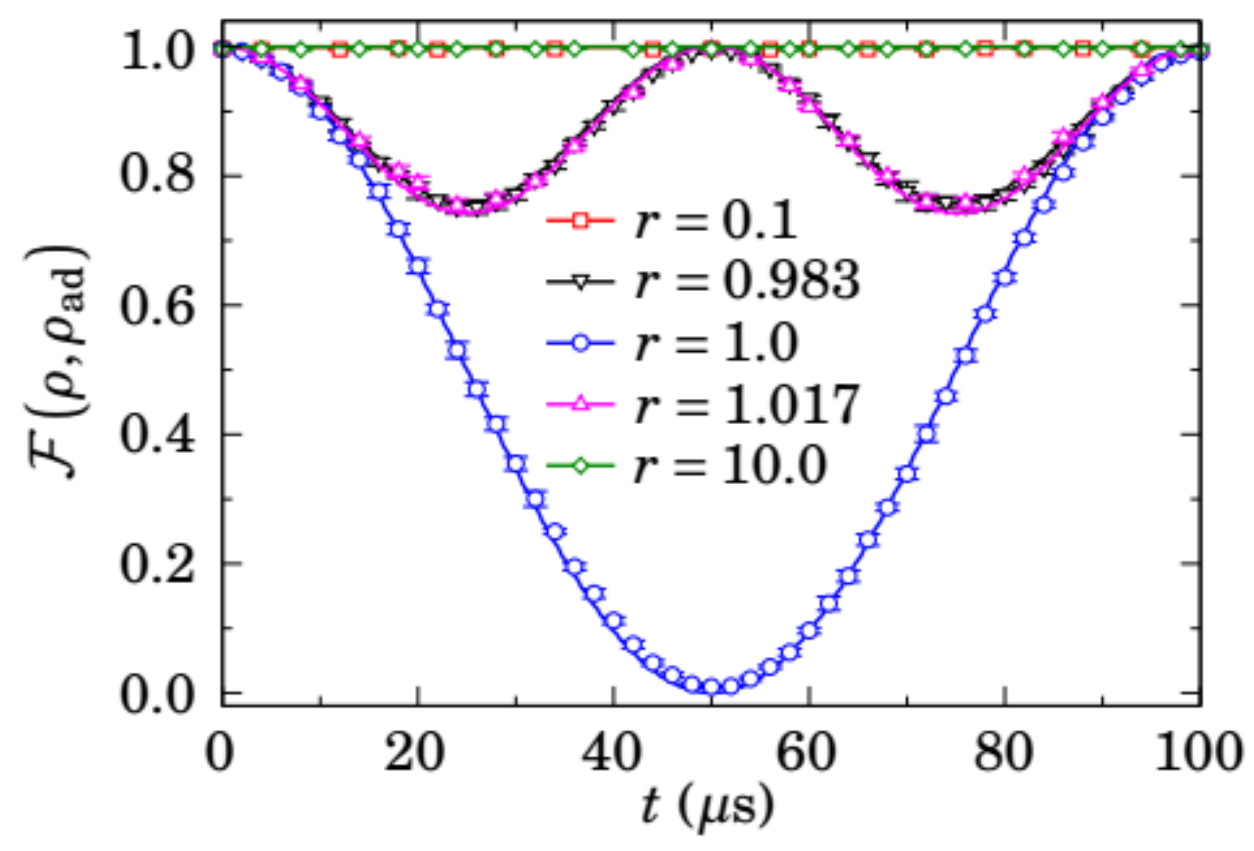}
	\caption{Theoretical and experimental fidelities for the quantum dynamics for different values of the resonance parameter $r$. The symbols and lines represent experimental data and theoretical results, respectively. The error bars are obtained from 60,000 binary-valued measurements for each data point and are not larger than $1.6\%$. We set $\omega_{0} = 2\pi \times 1.0$~MHz, $\omega_{\text{T}} = 2\pi \times 20.0$~KHz.}
	\label{FigAdFidelityOscHamil}
\end{figure}

%\begin{table}[b!]
%	\centering
%	\begin{tabular}{c|cc}
%		Condition & Necessary      & Sufficient     \\ \hline
%		$C_{\text{trad}}$      & No             & No             \\
%		$C_{\text{Tong}}$      & No 			& No \\
%		$C_{\text{Wu}}$        & No 			& No             \\
%		$C_{\text{AR}}$        & No             & No
%	\end{tabular}
%	\caption{Table~\ref{TableAdCondi} updated.}
%	\label{TableAdCondiOsc}
%\end{table}

%%%%%%%%%%%%%%%%%%%%%%%%%%%%%%%%%%%%%%%%%%%%%%%%%%%%%%%%%%%%%%%%%%%%%%%%%%%%%%
\subsection{Adiabatic parameters in non-inertial frames}

A widely used procedure to deal with inconsistencies in the adiabatic parameters is providing a new parameter able to solve the problem previously obtained. We discussed a bit about that in Sec.~\ref{SecRevAd}, where we showed four different adiabatic parameters. However, here we will follow a different approach to deal with the result shown in previous section.

It is evident the role of resonance in adiabaticity breaking of quantum system. By taking into account that the resonance is a non-linear phenomena (because small fields can drastically affect the system dynamics) and that it is convenient to consider such phenomena through a non-inertial frame change, here we will consider the adiabaticity conditions as provided from the non-inertial viewpoint. In other words, given the dynamics in the non-inertial frame we want to compute the adiabatic parameters from this new frame. To introduce the non-inertial frame, we use the unitary time-dependent operator $\Ocalb(t) = e^{i \omega t \sigma_{z}/2}$ and from Eq.~\eqref{EqVonNeumannBob} the dynamics in the non-inertial frame is given by
\begin{equation}
	\dot{\rho}_{\Ocalb}(t) = \frac{1}{i\hbar} [H_{\Ocalb}(t) ,\rho_{\Ocalb}(t)] \text{ , } \label{EqDynRotFrame}
\end{equation}
where $H_{\Ocalb}(t) = \Ocalb(t)H(t)\Ocalb^{\dagger}(t)+ i\hbar \dot{\Ocalb}(t)\Ocalb^{\dagger}(t)$ and $\rho_{\Ocalb}(t) = \Ocalb(t)\rho(t)\Ocalb^{\dagger}(t)$. After some calculation, it is possible to show that
\begin{equation}
	H_{\Ocalb}(t) = \hbar\frac{\omega_{0} - \omega}{2} \sigma_{z} + \hbar \frac{\sin ( \omega t) \tan \theta}{2} \vec{\omega}_{xy}(t)\cdot\vec{\sigma}_{xy} \text{ , } \label{EqOscHamilNonInertial}
\end{equation}
with $\vec{\omega}_{xy}(t) = \omega_{0}[\cos(\omega t) \hat{x} - \sin(\omega t) \hat{y}]$ and $\vec{\sigma}_{xy} = \sigma_{x} \hat{x} + \sigma_{y} \hat{y}$. The spectrum of the Hamiltonian in the rotating frame is therefore
\begin{equation}
	\ket{E_{\pm}^{\Ocalb}(t)} = \frac{1}{\Ncal^{\Ocalb}_{\pm}} \left[ e^{i \omega t} \alpha^{\Ocalb}_{\pm} (t) \ket{1} + \ket{0} \right] \text{ , } \label{EqStatesHOscNonInert}
\end{equation}
where $\Ncal^{\Ocalb}_{\pm} = \sqrt{ |\alpha^{\Ocalb}_{\pm} (t)|^2 + 1 }$, with
\begin{equation}
	\alpha^{\Ocalb}_{\pm} (t) = \frac{\csc (\omega t)}{\omega_{0}\sqrt{2}} \left[ \sqrt{2}(\omega_{0} - \omega) \cot \theta \pm \sqrt{-2\omega^2 + 4\omega \omega_{0} -\omega_{0}^2 - \omega_{0}^2\cos (2\omega t) + 2 (\omega - \omega_{0})^2 \csc^2\theta} \right] \text{ , }
\end{equation}
and the set of eigenvalues are
\begin{equation}
	E_{\pm}^{\Ocalb}(t) = \pm\frac{\hbar\sec\theta}{4} \sqrt{ 2\omega^2 - 4\omega\omega_{0} + 3\omega_{0}^2 + \left(  2\omega^2 - 4\omega\omega_{0} + \omega_{0}^2 \right)\cos (2\theta) - 2\omega_{0}^2 \cos(2\omega t) \sin^2\theta } \text{ . } \label{EqEnergyNonInertial}
\end{equation}

\begin{figure}[t!]
	%\input{Figs/FigAdCondOscHamilInertial.plt}
	%\vspace{5.2cm}
	\centering
	\includegraphics[scale=0.55]{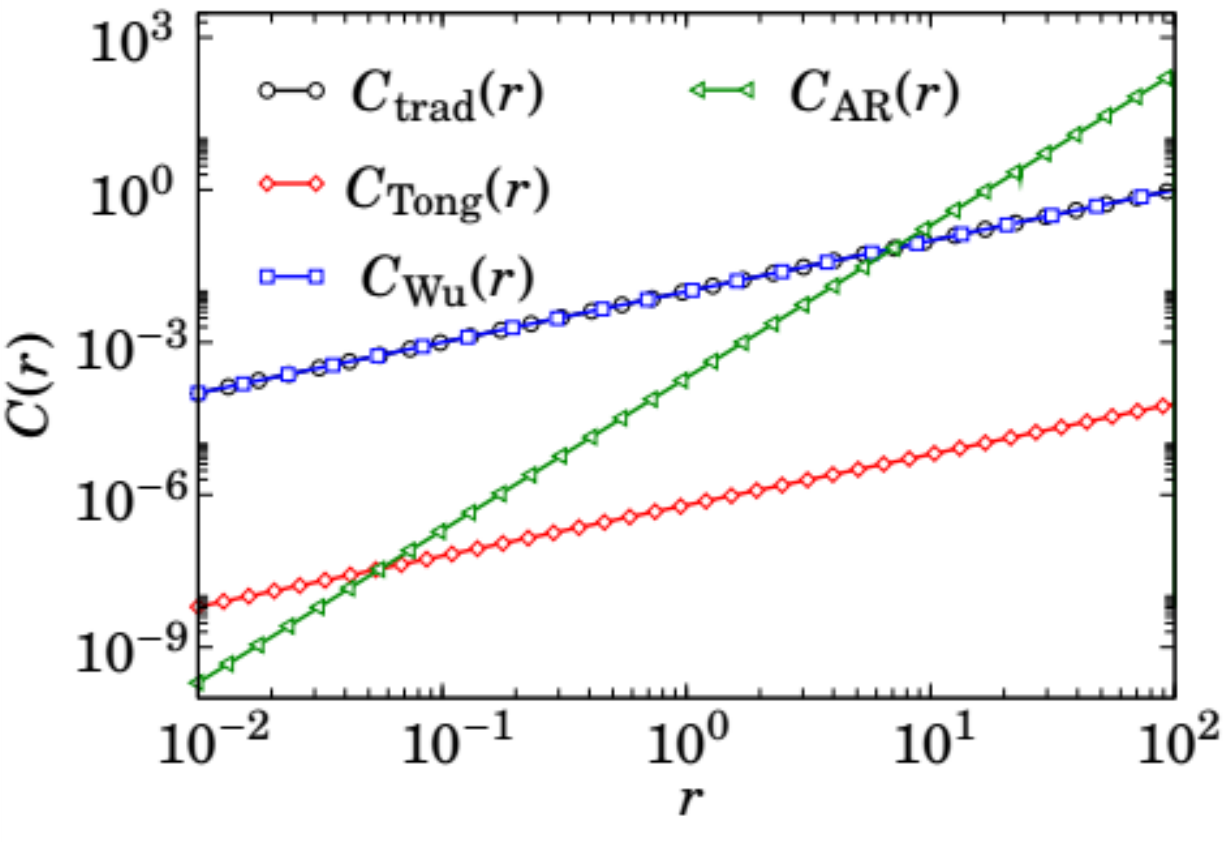}
	\caption{AC's computed for the Hamiltonian written in Eq.~\eqref{EqOscHamil}. The parameters used here are $\omega_{0} = 2\pi \times 1.0$~MHz, $\omega_{1} = 2\pi \times 20.0$~KHz, and $\tau=100~\mu$s.}
	\label{FigAdCondOscHamilInertial}
\end{figure}

A first point to be highlighted here is the influence of $\omega$ in the spectrum of the non-inertial Hamiltonian. Therefore, by computing the adiabatic parameters in this new frame the result is given as in Fig.~\ref{FigAdCondOscHamilNonInertial}.

Thus, it is possible to see that AC's when computed from the non-inertial frame of the system allows ut to get a description consistent with experimental the results in Fig.~\ref{FigAdFidelityOscHamil}. In order to give some physical meaning to this result, let us consider the energy gap between ground and excited state given by
\begin{equation}
	g^{\Ocalb}(t) = E_{+}^{\Ocalb}(t) - E_{-}^{\Ocalb}(t) \text{ , }
\end{equation}
so from Eq.~\eqref{EqEnergyNonInertial} we can write the minimum gap as
\begin{equation}
	g_{\text{min}}^{\Ocalb} = \min_{t\in[0,\tau]} g^{\Ocalb}(t) = \hbar \omega_{0} \left \vert 1-r \right \vert \text{ , }
\end{equation}
where we already used $\omega = r\omega_{0}$. Now, we can see that the peak in Fig.~\ref{FigAdCondOscHamilNonInertial} is due to the minimum gap going to zero at resonance $r \rightarrow 1$, at the same time that the gap becomes sufficiently large when $r \gg 1$ and it is of order of $\hbar \omega_{0}$ when $r \rightarrow 0$. 

\begin{figure}[t!]
	\centering
	\includegraphics[scale=0.55]{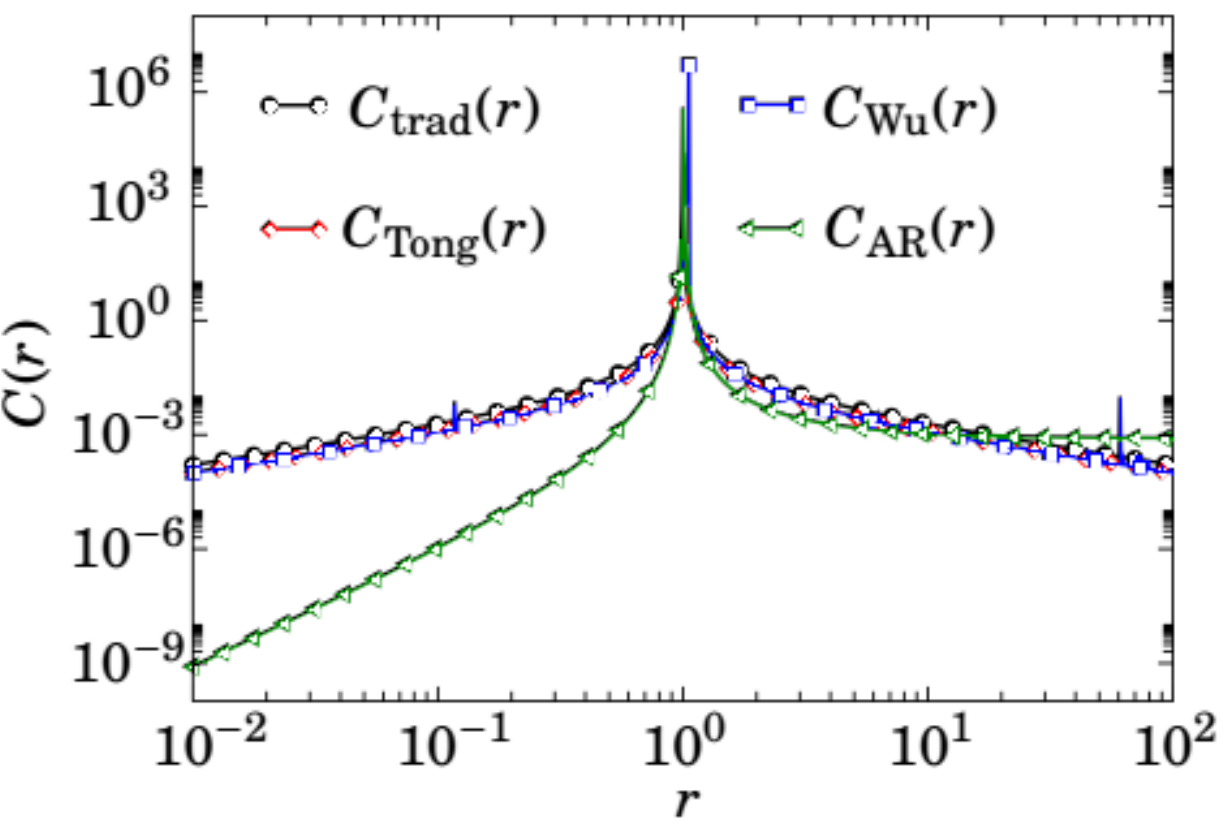}
	%\input{Figs/FigAdCondOscHamilNonInertial.plt}
	%\vspace{5.2cm}
	\caption{AC's computed for the Hamiltonian in Eq.~\eqref{EqOscHamilNonInertial}. The parameters are $\omega_{0} = 2\pi \times 1.0$~MHz, $\omega_{1} = 2\pi \times 20.0$~KHz, and $\tau=100~\mu$s.}
	\label{FigAdCondOscHamilNonInertial}
\end{figure}

As a result, it is possible to see that the AC's allow us to perfectly describe the adiabatic behavior (or its breaking) of the system considered here when computed in a different reference frame. This suggests a frame-dependence of the adiabatic validity conditions in guaranteeing the adiabatic approximation, as we pictorially describe in Fig.~\ref{FigSchemeFrameChange}. It is important to highlight that even the traditional adiabatic condition becomes sufficient and necessary in this new frame. Thus, the immediate question to be addressed in next section is: \textit{Is this frame change applicable to any quantum system?}

\section{Validation of quantum adiabaticity through non-inertial frames}

As previously mentioned, here we are interested in providing a different approach to deal with inconsistencies of the AC's, so we are not providing a new AC. The results present in the two previous sections are valid for the particular Hamiltonian in Eq.~\eqref{EqOscHamil}, so the question we have is: \textit{Is the frame change a universal strategy to validate the AC's?}

\begin{figure}[t]
	\centering
	\includegraphics[scale=0.7]{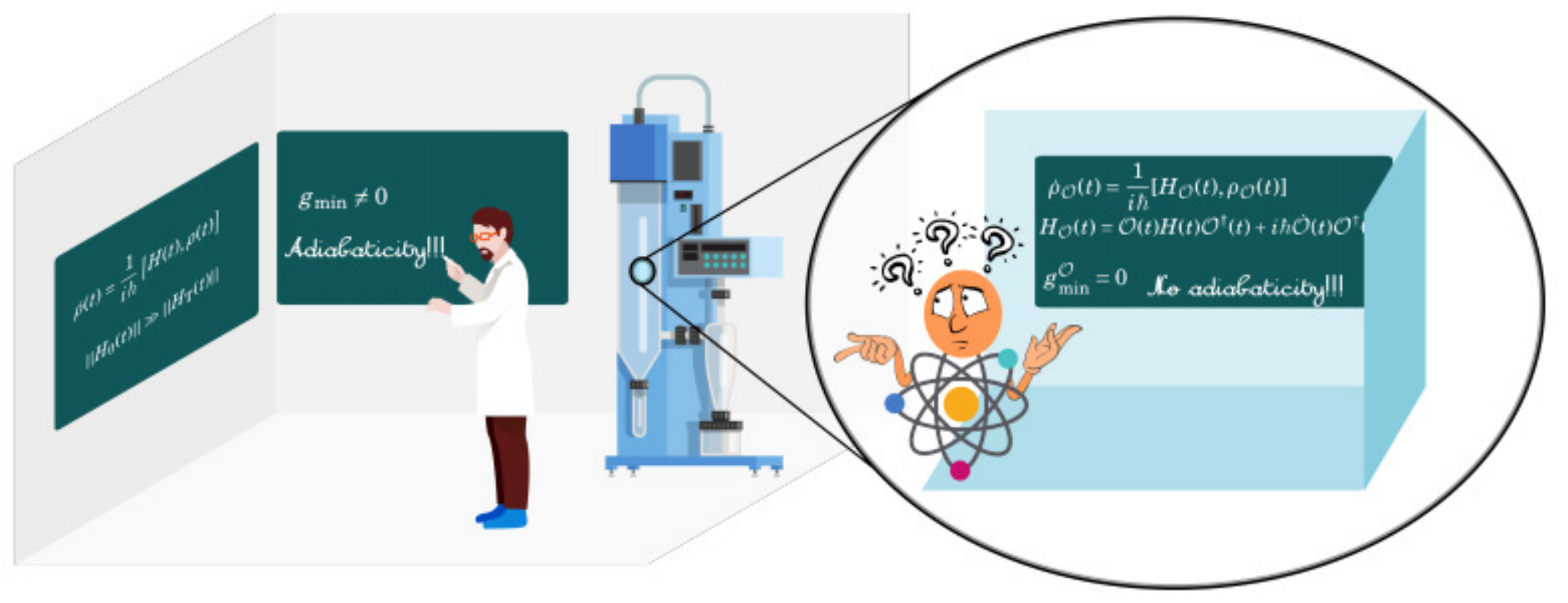}
	\caption{Pictorial representation of how the working substance deals with the adiabatic paths in its natural non-inertial frame. While the experimentalist in the figure computes adiabaticity in inertial (laboratory) frame, the working substance follows a non-adiabatic behavior. The nature does not care about what we are computing!}
	\label{FigSchemeFrameChange}
\end{figure}

In order to extend the previous study to a large (or infinite) class of Hamiltonians, let us consider the following Hamiltonian
\begin{equation}
	H(\omega,t) = \hbar \omega_{0} H_{0} + \hbar \omega_{\text{T}} H_{\text{T}}(\omega,t) \text{ , } \label{EqGenHAdOrig}
\end{equation}
where the contributions $\hbar \omega_{0} H_{0}$ and $\hbar \omega_{\text{T}} H_{\text{T}}(\omega,t)$ depend on the longitudinal and transverse fields $\vec{B}_{0}$ and $\vec{B}_{\text{T}}(t)$, respectively, this means we can write $[H_{\text{T}}(\omega,t),H_{0}] \neq 0$. Here we are denoting the fields by $\vec{B}$, but it does not necessarily denote a magnetic field. We are considering that $\omega_{0}$ is the natural frequency of the Hamiltonian $H_{0}$, so note that the frequency transition between energy levels of the system is a multiple of $\omega_{0}$ and it depends on the set of eigenvalues for $H_{0}$. Here we are interested in a non-linear phenomenon that arises from the resonance configuration of the system, so a second assumption is related with the Hamiltonian $\hbar \omega_{\text{T}} H_{\text{T}}(\omega,t)$ and it reads as $||\omega_{\text{T}} H_{\text{T}}(\omega,t)||\ll||\omega_{0} H_{0}||$, physically it means we have very small transverse fields acting on the system, the ideal scenario to study these non-linear phenomena. Under the number of above considerations, it is important to identify the dynamics given by
\begin{equation}
	\dot{\rho} (t) = \frac{1}{i\hbar} \left[H(\omega,t), \rho (t) \right] \text{ , }
\end{equation}
as a dynamics obtained from a inertial frame (for example, from laboratory frame). However, it is possible to define a non-inertial frame associated with the operator $\Ocalb(t) = \exp \left( i \omega H_{0} t \right)$. In this new frame, the dynamics is governed by the Hamiltonian
\begin{equation}
	H_{\Ocalb}(\omega,t) = \Ocalb(t)H(\omega,t)\Ocalb^{\dagger}(t)+ i\hbar \dot{\Ocalb}(t)\Ocalb^{\dagger}(t) \text{ , } 
\end{equation}
and the equation for the system dynamics is given by the Eq.~\eqref{EqDynRotFrame}. Note that in above Hamiltonian we have a contribution due to the frame change, in same way that fictitious forces arise in classical mechanics when we do a frame change from inertial to non-inertial ones. Thus, the above equation can be rewritten as
\begin{equation}
	H_{\Ocalb}(\omega,t) = \hbar \left( \omega_{0} - \omega \right) H_{0} + \hbar \omega_{\text{T}} H_{\Ocalb , \text{T}}(\omega,t) \text{ , } 
\end{equation}
where $H_{\Ocalb,\text{T}}(\omega,t) = \Ocalb(t)H_{\text{T}}(t)\Ocalb^{\dagger}(t)$. Therefore, due to the quantity $ \omega_{0} - \omega $ in the first term of $H_{\Ocalb}{(\omega,t)}$, which comes from contribution of the term $i\hbar \dot{\Ocalb}(t)\Ocalb^{\dagger}(t)$ to the term $\hbar \omega_{0} H_{0}$, in this new frame it becomes evident that the contribution of $H_{\text{T}}(\omega,t)$ cannot be ignored in this new frame. From this new frame is possible to see that the term $H_{\text{T}}(\omega,t)$ promotes transitions in energy levels of the system when $||H_{\Ocalb,\text{T}}(\omega,t)|| \sim || \hbar \left( \omega_{0} - \omega \right) H_{0} ||$, so the dynamics as provided by Eq.~\eqref{EqDynRotFrame} represents a dynamics under system reference. In this way, we can clearly see that the adiabatic behavior in inertial frame becomes impossible because we will get transitions between energy levels of the Hamiltonian. In fact, once we are in a regime where $||\omega_{\text{T}} H_{\text{T}}(\omega,t)||\ll||\omega_{0} H_{0}||$, the instantaneous \textit{time-dependent} eigenvectors of $H(\omega, t)$ are approximately equal to the eigenvectors of $\omega_{0} H_{0}$. Moreover, note that we did not anything about adiabaticity in non-inertial frame, so we could have some adiabaticity in this new frame. However, as previously shown, there are situations in which the ACs computed in non-inertial frame allow us to get some general result on adiabaticity in inertial frame. This suggests a connection between adiabaticity in different frames. To deal with this question, we have the following theorem (see proof in Appendix~\ref{AppProofTheoremsAdiab}).

\begin{theo}\label{TheoAdiab}
	Consider a Hamiltonian $H(t)$ and its non-inertial counterpart $H_{\Ocalb}(t)$,
	with $\Ocalb(t)$ an arbitrary unitary transformation. The eigenstates of $H(t)$ and $H_{\Ocalb}(t)$ are denoted by $\ket{E_{k}(t)}$ and $\ket{E^{\Ocalb}_{m}(t)}$,
	respectively. Then, if a quantum system $\Scalb$ is prepared at time $t=t_0$ in a particular eigenstate $\ket{E_{k}(t_{0})}$ of $H(t_0)$, then the adiabatic evolution
	of $\Scalb$ in the inertial frame, governed by $H(t)$, implies in the adiabatic evolution of $\Scalb$ in the non-inertial frame, governed by $H_{\Ocalb}(t)$,
	if and only if
	\begin{equation}
		|\bra{E^{\Ocalb}_{m}(t)}\Ocalb(t)\ket{E_{k}(t)}| = |\bra{E^{\Ocalb}_{m}(t_{0})}\Ocalb(t_{0})\ket{E_{k}(t_{0})}| \text{ , }\label{EqCondTheoAdiab}
	\end{equation}
	$\forall$ $t,m$, where $t \in [t_{0},\tau]$, with $\tau$ denoting the total time of evolution.
	Conversely, if the adiabatic dynamics in the non-inertial frame starts in $\ket{E^{\Ocalb}_{m}(t_{0})}$, then the dynamics in the
	inertial frame is also adiabatic if and only if Eq.~\eqref{EqCondTheoAdiab} is satisfied.
\end{theo}

If the hypothesis of the above theorem is satisfied, then adiabaticity in inertial frame can be used as witness of an adiabatic behavior in non-inertial frame, so that no calculation of the AC's in the changed frame is required. Furthermore, we can use the Theorem~\ref{TheoAdiab} to predict the adiabaticity broken in inertial frame from the non-adiabatic behavior in non-inertial frame. In fact, from the Morgan's law, if $A$ and $B$ then $C$ ($A \wedge B \rightarrow C$), therefore if not $C$ then not $A$ or not $B$ ($\neg C \rightarrow \neg A \vee \neg B$)~\cite{DeMorgan:Collection,Hurley:Book}. In this case, the instance $A$ is the Theorem~\ref{TheoAdiab}, $B$ is adiabaticity in inertial frame and $C$ the adiabatic behavior in non-inertial frame. Thus, not $C$ means non-adiabatic behavior in non-inertial frame, since the instance $A$ remains valid (the hypothesis of the theorem is satisfied) then we get not $B$, that is, non-adiabatic dynamics in inertial frame. The opposite is also true, if Eq.~\eqref{EqCondTheoAdiab} is satisfied we can associated adiabaticity in non-inertial frame with the adiabatic dynamics in inertial frame and, by using the Morgan's law, non-adiabatic behavior in inertial frame implies non-adiabaticity in non-inertial frame.

From this general study, the dynamics presented in Fig.~\ref{FigAdFidelityOscHamil} can be correctly explained from curves in Fig.~\ref{FigAdCondOscHamilNonInertial} by using the Theorem~\ref{TheoAdiab}, because for situations where we are at resonance and far-from resonance the Theorem~\ref{TheoAdiab} is satisfied. For a detailed discussion, see Appendix~\ref{AppAppliTheorem1}.

%%%%%%%%%%%%%%%%%%%%%%%%%%%%%%%%%%%%%%%%%%%%%%%%%%%%%%%%%%%%%%%%%%%%%%%%%%%%%%
\subsection{Revising the adiabatic dynamics of a spin}

Now let us revisit the problem of the adiabatic dynamics of a spin-$\frac{1}{2}$ in presence of a rotating magnetic field, as discussed in Sec.~\ref{SecRevAd}. As a direct consequence of the time-independent Hamiltonian obtained in non-inertial frame, as provided by the  Eq.~\eqref{EqNMRHamilR}, the dynamics is always (trivially) adiabatic. However, the Theorem~\ref{TheoAdiab} is not applicable to this system because the hypothesis of such theorem is not obeyed here, since near to resonance the initial state in this case is not an individual eigenstate of $H_{\text{R}}$. In fact, the initial state in non-inertial frame is prepared in an eigenstate of $\sigma_z$, but at resonance we have $H_{\text{R}} \propto \sigma_x$. For this reason, we need to derive a new strategy to deal with this case. Thus, by considering the situations in which the Hamiltonian in inertial frame is time-independent, we derive the following theorem (the proof is shown in Appendix~\ref{AppProofTheoremsAdiab}).

\begin{theo} \label{TheoAdiabTI}
	Consider a Hamiltonian $H(t)$ and its non-inertial counterpart $H_{\Ocalb}(t)=H_{\Ocalb}$,
	with $\Ocalb(t)$ an arbitrary unitary transformation and $H_{\Ocalb}$ a constant Hamiltonian.
	The eigenstates of $H(t)$ and $H_{\Ocalb}$ are denoted by $\ket{E_{k}(t)}$ and $\ket{E^{\Ocalb}_{m}}$,
	respectively. Then, if a quantum system $\Scalb$ is prepared at time $t=t_0$ in a particular eigenstate $\ket{E_{n}(t_{0})}$ of $H(t_0)$,
	then the adiabatic evolution of $\Scalb$ in the inertial frame, governed by $H(t)$, occurs if and only if
	\begin{equation}
		| \bra{E_{k}(t)}U_{\Ocalb}(t,t_{0})\ket{E_{n}(t_{0})} | = |\interpro{E_{k}(t_{0})}{E_{n}(t_{0})}|  \,\,\,\,\, \forall t,k  \,\,\, \text{ , } \label{ApU2}
	\end{equation}
	where $t \in [t_{0},\tau]$, with $\tau$ denoting the total time of evolution, and
	$U_{\Ocalb}(t,t_{0}) = \Ocalb^{\dagger}(t)e^{-\frac{i}{\hbar} H_{\Ocalb}(t-t_{0})}\Ocalb(t_{0})$.
\end{theo}

The experimental results in Ref.~\cite{Suter:08} can be validated by Theorem ~\ref{TheoAdiabTI}, since
the Hamiltonian in Eqs.~\eqref{EqNMRHamil} and~\eqref{EqNMRHamilR} satisfy the Eq.~(\ref{ApU2}) in a far-from resonance situation and
violates it at resonance. In fact, the initial state $|\psi(0)\rangle$ can be approximately written as $|\psi(0)\rangle = \ket{E_{n}(0)} \approx |n\rangle$ [with $\sigma_z |n\rangle =(-1)^{(n+1)} |n\rangle$]. Thus, Eq.~\eqref{ApU2} provides the condition $| \bra{k}e^{-\frac{i}{\hbar} H_{\Ocalb}t}\ket{n} | = |\interpro{k}{n}| = \delta_{kn}$, $\forall k$ and $\forall t\in [0,\tau]$.
In a far-from-resonance situation, we have
$H_{\Ocalb}^{\text{nmr}} \approx (\omega _{0}-\omega )\sigma_{z}/2$, and we conclude that $| \bra{k}e^{-iH_{\Ocalb}t/\hbar}\ket{n} | \approx \delta_{kn}$. This shows that the dynamics in the inertial frame is (approximately) adiabatic far from resonance. On the other hand, near to resonance, we get $H_{\Ocalb}^{\text{nmr}} \approx (\omega _{\text{rf}}/2)\sigma_{x}$, where we can immediately conclude that
$| \bra{k}e^{-iH_{\Ocalb}t/\hbar}\ket{n} | \not\approx \delta_{kn}$ is not valid for any $t\in [0,\tau]$.

\section{Stable quantum batteries through adiabatic dynamics}

Batteries are devices able to store energy to be transferred to a consumption hub at later times. In this sense, if we use quantum mechanical effects to provide more efficient batteries concerning the power during charging/discharging processes and available/extractable energy, it emerges the possibility to design quantum devices known as \textit{quantum batteries} (QBs)~\cite{Alicki:13,PRL2013Huber,PRL2017Binder,Ferraro:18,PRL_Andolina}. QBs are potentially relevant as fuel for other quantum devices and, more generally, for boosting a proper development of quantum networks. In particular, they have been proposed in a number of distinct experimental architectures, such as spin systems~\cite{Le:18}, quantum cavities~\cite{Binder:15,Fusco:16,Zhang:18,Ferraro:18}, superconducting transmon qubits~\cite{Santos:19-a}, and quantum oscillators~\cite{Andolina:18,Andolina:19}. Among fundamental challenges for useful QBs are both the control of the energy transfer and the stability of the discharge process to an available consumption hub (CH) (see, e.g., Refs.~\cite{Santos:19-a,Gherardini:19,Kamin:20-1} for recent discussion about the stability topic). In this section we will consider the application of adiabatic dynamics to design stable quantum batteries~\cite{Santos:20c}. In addition, it is presented a power switchable scheme of quantum batteries based on a \textit{power operator} that describes the instantaneous energy transfer rate from battery to a consumption hub.

\subsection{Stable charging process in three-level quantum batteries} \label{SubSecStableThreeLevelAdQB}

Given a quantum system and its internal Hamiltonian $H_{0}$, which sets the system energy scale, we define the energy of the system as $E = \tr{H_{0}\rho}$, where $\rho$ is the matrix density of the system. Here we define the energy system in terms of the population of $\rho$ concern the eigenstates of $H_{0}$, where any control Hamiltonian $H_{\text{c}}(t)$ does not contribute to the quantity $E$. Then, we can define the system energy when we set the reference Hamiltonian $H_{0}$. Thus, if we drive the system using the Hamiltonian $H(t)$, the energy of the system changes in according with
\begin{equation}
	E(t) = \tr{H_{0}\rho(t)} = \tr{H_{0}U(t)\rho(0)U^{\dagger}(t)} \text{ , }
\end{equation}
where $U(t)$ is the evolution operator. Despite the energy to be computed from above equation, in context of quantum batteries it does not mean the available energy in system is given by $E(t)$; The available energy is obtained from notion of \textit{ergotropy} $\Ecalb$~\cite{Allahverdyan:04}. The ergotropy is the maximum amount of work that can be extracted from a quantum battery and reads as
\begin{equation}
	\Ecalb_{\rho} = \tr{H_{0}\rho} - \tr{H_{0}\sigma_{\rho}} \text{ , }
\end{equation}
where $\sigma_{\rho}$ is the \textit{passive state}, defined as a state in which no energy can be extracted from the system by unitary transformations $V$ defined on the set $\Vcalb$ of all unitary operations on the Hilbert space $\Hcalb$ of the system. Mathematically we write
\begin{equation}
	\tr{H_{0}\sigma_{\rho}} - \tr{H_{0}V\sigma_{\rho}V^{\dagger}} \leq 0 \text{ , for any } V \in \Vcalb \text{ . }
\end{equation}

In this thesis, since we are interested in unitary dynamics that leads the system from empty charge battery state to the full charge state, the passive state is identified as the ground state $\ket{E_{0}}$ of $H_{0}$. Then, to define our three-level quantum battery, let us consider a non-degenerate quantum battery as a $d$-level quantum system described by the Hamiltonian
\begin{equation}
	H_{0} = \sum_{n}^d \varepsilon_{n} \ket{\varepsilon_{n}}\bra{\varepsilon_{n}} \text{ , } \label{H0}
\end{equation}
with $\varepsilon_{1} < \varepsilon_{2} < \dots < \varepsilon_{d}$. The battery is said to be in a {\it passive} state if no energy can be extracted from it through some cyclic process, e.g. the turning on of some potential $V(t)$, for a time $\tau$, that satisfies the boundary conditions $V(0) = V(T) = 0$. Interestingly, an array of such batteries may not be passive with some work being available, but requiring some collective processing of the batteries~\cite{Binder:15, PRL2017Binder,Le:18,Ferraro:18}. A notable exception to this is the thermal states which are said to be {\it completely passive} as even with such collective processing one cannot extract any work~\cite{Pusz:78}. {\it Active} states are those that allowed for work to be extracted through some cyclic process, with maximal amount of extractable work is called the ergotropy~\cite{Allahverdyan:04}. We remark that $d = 2$ is a special case where all passive states are completely passive as any diagonal state in the energy eigenbasis is necessarily a thermal state at some temperature. In what follows we will be concerned with stable charging of a single three-level quantum battery, $d=3$.

We will be interested in examining the energy stored in a quantum battery through some, possibly non-unitary, process. The system energy at time $t$ is simply $E(t) = \trs{H_{0} \rho(t)}$ and if we assume our battery begins in the ground state, $\ket{\varepsilon_1}$, the ergotropy is the difference in energy between the final and initial battery states after the charging process is over,
\begin{equation}
	\Ecalb (t) = E(t) - E_{\text{gs}} = \trs{H_{0} \rho(t)} - \varepsilon_{1} \text{ , }
\end{equation}
with the maximum stored energy $\Ecalb_{\text{max}} = \hbar (\omega_{3}-\omega_{1})$ achievable via a process which transfers all the population from the initial ground state to the maximally excited state. In this particular study, we consider a charging process done through Stimulated Raman Adiabatic Passage (STIRAP)~\cite{Vitanov:17}, where we can ensure a stable adiabatic quantum battery (as we shall see). It is possible to generalize the proposed setting to systems with $d > 3$ using chain STIRAP processes that connect the lowest lying energy eigenstate to the highest one~\cite{Vitanov:17}. Therefore, we assume that system to be driven by the time-dependent Hamiltonian~\cite{Vitanov:17,Marangos:98,Fleischhauer:05,Bergmann:98}
\begin{equation}
	H_{\text{c}}(t) = \hbar\Omega_{12}(t)e^{-i\omega_{12}t} \ket{\varepsilon_{1}}\bra{\varepsilon_{2}} + \hbar\Omega_{23}(t)e^{-i\omega_{23}t} \ket{\varepsilon_{2}}\bra{\varepsilon_{3}} + \text{h.c}  \text{ . } \label{Ht}
\end{equation}
In this case, the complete Hamiltonian which describes the dynamics of the system can be written as $H(t) = H_{0} + H_{\text{c}}(t)$. The Hamiltonian $H_{\text{c}}(t)$ develops the role of a \textit{quantum charger} for our three-level quantum battery as we need to couple our system to the external fields described by $H_{\text{c}}(t)$ in order to charge the battery. While the bare Hamiltonian is important for dictating the amount of energy stored in the battery, the dynamics is driven by the interaction Hamiltonian $H_{\text{c}}(t)$. In fact, by considering the dynamics of the system in a general time-dependent interaction picture, the new Hamiltonian can be written as~\cite{Whaley:84}
\begin{equation}
	\dot{\rho}_{\text{int}}(t) = \frac{1}{i\hbar} [H_{\text{int}}(t),\rho_{\text{int}}(t)] \text{ , }
\end{equation}
where $\rho_{\text{int}}(t) = e^{iH_{0}t} \rho(t) e^{-iH_{0}t}$ and
\begin{equation}
	H_{\text{int}}(t) = \hbar \Omega_{12}(t) \ket{\varepsilon_{1}}\bra{\varepsilon_{2}} + \hbar \Omega_{23}(t) \ket{\varepsilon_{2}}\bra{\varepsilon_{3}} + \text{h.c} \text{ , }
\end{equation}
where we already assumed that both fields in Eq.~\eqref{Ht} are on resonance with the energy levels of the battery. Thus, it is possible to show that, in this new representation, we can get the population in each energy level from $P_{n} = \trs{\hat{P}_{n}\rho_{\text{int}}(t)} = \trs{\hat{P}_{n}\rho(t)}$, where $\hat{P}_{n}$ is the projector $\hat{P}_{n} = \ket{\varepsilon_{n}}\bra{\varepsilon_{n}}$. In addition, $\trs{H_{0}\rho_{\text{int}}(t)}  =  \trs{H_{0}\rho(t)}$, so that $\Ecalb(t) = \trs{H_{0}\rho_{\text{int}}(t)} -E_{\text{gs}}$ can be obtained from the dynamics in the rotating frame. As already discussed in Sec.~\ref{SecAdFrameDependence}, it is important to highlight that the adiabatic behavior in this frame does not mean adiabaticity in the inertial frame, for this reason we can deal with an adiabatic charging process of quantum batteries. Therefore, the system dynamics as in the above equation can be considered in our study without loss of generality, so that the study of the charging procedure of our battery through an adiabatic dynamics will be done in this new frame. By computing the set of eigenvectors of the new Hamiltonian $H_{\text{int}}(t)$ we find
\begin{subequations}
	\label{EigenStates}
	\begin{align}
		\ket{E_{-}(t)} &= \frac{1}{\sqrt{2}} \left[
		\frac{\Omega_{12}(t)}{\Delta(t)} \ket{\varepsilon_{1}} -
		\ket{\varepsilon_{2}} + 
		\frac{\Omega_{23}(t)}{\Delta(t)} \ket{\varepsilon_{3}}
		\right] \\
		\ket{E_{0}(t)} &= 
		\frac{\Omega_{23}(t)}{\Delta(t)} \ket{\varepsilon_{1}} -
		\frac{\Omega_{12}(t)}{\Delta(t)} \ket{\varepsilon_{3}} \label{E0}\\
		\ket{E_{+}(t)} &= \frac{1}{\sqrt{2}} \left[
		\frac{\Omega_{12}(t)}{\Delta(t)} \ket{\varepsilon_{1}} +
		\ket{\varepsilon_{2}} + 
		\frac{\Omega_{23}(t)}{\Delta(t)} \ket{\varepsilon_{3}} 
		\right]\text{ , }
	\end{align}
\end{subequations}
associated with eigenvalues $E_{\pm}(t) = \pm \hbar \Delta(t)$ and $E_{0}(t) =0$, where $\Delta^2(t) = \Omega_{12}^2(t)+\Omega_{23}^2(t)$. Then, as mentioned, we will assume the process starts with the battery state $\ket{\psi(0)} = \ket{\varepsilon_{1}}$. This state can be written as a combination of different elements of Eqs.~\eqref{EigenStates}, depending on the initial values of the parameters $\Omega_{12}(0)$ and $\Omega_{23}(0)$. Therefore we can choose different charging protocols associated with distinct choices of the parameters $\Omega_{12}(t)$ and $\Omega_{23}(t)$, by adjusting how the external fields act on the system at the start of the evolution. We will show that while some protocols will lead to an unstable charged state, and therefore would require a carefully timed decoupling of the battery from the charging fields, by exploiting the STIRAP technique, we can achieve a stable and robust charged state.

\emph{Unstable charging --} From Eqs.~\eqref{EigenStates}, it is possible to show that the initial state $\ket{\psi(0)}$ can be written as a combination of the states $\ket{E_{-}(0)}$ and $\ket{E_{+}(0)}$ if we set $\Omega_{12}(0) \neq 0$ and $\Omega_{23}(0) = 0$. In fact, by considering this initial values for $\Omega_{12}(t)$ and $\Omega_{23}(t)$, one can show that
\begin{equation}
	\ket{\psi(0)} = \frac{\ket{E_{+}(0)} + \ket{E_{-}(0)}}{\sqrt{2}} = \ket{\varepsilon_{1}} \text{ . }
\end{equation}

In case where we admit that the system undergoes an adiabatic dynamics, we find the evolved state~\cite{Sarandy:04,Kato:50,Messiah:Book,Amin:09}
\begin{equation}
	\ket{\psi^{\text{ad}}(t)} = \frac{1}{\sqrt{2}}\left[ e^{-\frac{i}{\hbar}\int_{0}^{t} E_{+}(t^{\prime})dt^{\prime}}\ket{E_{+}(t)} + e^{-\frac{i}{\hbar}\int_{0}^{t} E_{-}(t^{\prime})dt^{\prime}}\ket{E_{-}(t)} \right] \text{ , } \label{AdCuns}
\end{equation}
where we already used the parallel transport condition $\interpro{E_{n}(t)}{\dot{E}_{n}(t)}=0$, for all $n$. Thus, we write
\begin{equation}
	\ket{\psi^{\text{ad}}(t)} = \frac{\cos\Phi(t)}{\Delta(t)}\Big(\Omega_{12}(t) \ket{\varepsilon_{1}} + \Omega_{23}(t) \ket{\varepsilon_{3}} \Big) - i \sin\Phi(t) \ket{\varepsilon_{2}},
\end{equation}
where $\Phi(t)= \int_{0}^{t} \Delta(t)d\xi$. Therefore, one finds the ergotropy $\Ecalb(t)  =  \bra{\psi^{\text{ad}}(t)}H_{0}\ket{\psi^{\text{ad}}(t)} - \bra{\varepsilon_{1}}H_{0}\ket{\varepsilon_{1}}$ as
\begin{equation}
	\Ecalb(t) = \hbar \frac{\cos^2\Phi(t)}{\Delta^2(t)}\left[ \omega_{1} \Omega^2_{12}(t) + \omega_{3}\Omega^2_{23}(t)\right] + \hbar \omega_{2} \sin^2\Phi(t) - \hbar\omega_{1} \label{Cuns} \text{ . }
\end{equation}

From the above equation, one makes evident that to achieve maximal ergotropy firstly we must fix the final values for the parameters $\Omega_{12}(t)$ and $\Omega_{32}(t)$ at some cutoff time (or charging time) $\tau_c$ in order to get $\Omega_{12}(\tau_{\text{c}})=0$ and $\Omega_{32}(\tau_{\text{c}})\neq0$. This involves particular initial and final conditions on the parameters $\Omega_{12}(t)$ and $\Omega_{32}(t)$ to fully charge the battery. Secondly, the instant in which the system achieves the full charge is when $\cos\Phi(\tau_{\text{c}}) = 1$. Under these constraints, we achieve maximum ergotropy, $\Ecalb(\tau_{\text{c}})= \Ecalb_{\text{max}}$. However, from Eq.~\eqref{Cuns} one can see that for $t  >  \tau_{\text{c}}$ the battery charge cannot be kept at its maximum value, and rather it will continue to oscillate between fully charged and fully dissipated states due to the action of the fields. We describe a protocol which leads to this situation as an unstable battery charging process. In addition, the function $\Phi(t)$ depends on the integration from $0$ to some instant $t > \tau_{\text{c}}$, the sine and cosine functions could become highly oscillating, such that that after $t > \tau_{\text{c}}$ we can have many maximum and minimum values for the ergotropy. We understand this as follows: in an adiabatic regime of the charging process, there is an intrinsic discharging process due to the relative quantal phases in Eq.~\eqref{AdCuns}. The adiabatic phase associated with different adiabatic paths (eigenstates), promote destructive and constructive superpositions of the components $\ket{\varepsilon_{3}}$ of the states $\ket{E_{+}(t)}$ and $\ket{E_{-}(t)}$. Consequently, we observe the natural discharging as a process due to destructive interference from $\ket{\varepsilon_{3}}$. Thus, a charging strategy that begins the adiabatic evolution with $\Omega_{12}(0) \neq  0$ and $\Omega_{32}(0) = 0$ does not lead to a stable and robust quantum battery. We remark that the above result is not a particular feature of adiabatic charging process. Actually, this \textit{spontaneous discharging} is an intrinsic characteristic of different systems where the oscillatory behavior of the quantal phases promotes some (partial) destructive interference as obtained in Eq.~\eqref{Cuns}.

\emph{Stable charging via STIRAP --} An alternative strategy for our quantum battery is through the eigenstate $\ket{E_{0}(t)}$, the so-called \textit{dark state}~\cite{Fleischhauer:05}. Again, the system initial state needs to be $\ket{\varepsilon_{1}}$, then from Eqs.~\eqref{EigenStates} we can see that to follow the dark state path we need to set the initial values of the parameters $\Omega_{12}(0) = 0$ and $\Omega_{32}(0) \neq  0$. Thus,
\begin{equation}
	\ket{\psi(0)} = \ket{E_{0}(0)} = \ket{\varepsilon_{1}} \text{ , } \label{AdCstb}
\end{equation}
By letting the system undergo adiabatic dynamics, the evolved state becomes
\begin{equation}
	\ket{\psi^{\text{ad}}(t)} = \ket{E_{0}(t)} = \frac{1}{\sqrt{2}} \left[
	\frac{\Omega_{23}(t)}{\Delta(t)} \ket{\varepsilon_{1}} -
	\frac{\Omega_{12}(t)}{\Delta(t)} \ket{\varepsilon_{3}}
	\right]  \text{ , } \label{EqAdBatteryStable}
\end{equation}
with no quantal phase accompanying the evolution, because the adiabatic phase is null, once we have $E_{0}(t)=0$ and $\interpro{E_{0}(t)}{\dot{E}_{0}(t)} = 0$. The ergotropy is then
\begin{eqnarray}
	\Ecalb(t) = \hbar \frac{\omega_{3}\Omega_{12}^2(t) + \omega_{1}\Omega_{23}^2(t)}{\Delta^2(t)} - \hbar\omega_{1} \label{Cstb} \text{ , }
\end{eqnarray}
which achieves its maximum value when $\Omega_{12}(\tau_{\text{c}})\neq 0$ and $\Omega_{23}(\tau_{\text{c}}) = 0$, without any assumption about the value of $\tau_{\text{c}}$, in stark contrast to the unstable charging process. Clearly to get a fully charged battery both initial and final conditions on the parameters $\Omega_{12}(t)$ and $\Omega_{32}(t)$ are required. However, by exploiting the STIRAP protocol we can avoid the oscillatory behavior otherwise present due to accumulated quantal phases.

A second important physical lesson of these results is associated with the intrinsic characteristics of dark states. Unlike the other eigenstates of the Hamiltonian driving the system, the dark state does not allow population inversion even when we put the fields on resonance with the system. In fact, from the evolved state given in Eq.~\eqref{EqAdBatteryStable} we can see that any change in $\ket{\psi^{\text{ad}}(t)}$ just depends on the changes in parameters $\Omega_{12}(t)$ and $\Omega_{32}(t)$ associated with intensity of the fields used to drive the system. This property allows us to design a robust battery that does not suffer from spontaneous discharging if the control fields are not switched off after the charging process. Thus, the emergence of the dark state further highlights the relevance of three-level (or $N$-level) systems over the more commonly considered two-level qubits in designing stable quantum batteries~\cite{CampbellBatteries}.

\subsubsection*{Three-level superconducting transmon quantum batteries}

In previous section we showed as the stability problem of quantum batteries can be suitably solved with a three-level quantum system, where a spontaneous discharging (dashed red curve in Fig.~\ref{FigSchemeQuantumBattery}{\color{blue}a}) arises due to coherent oscillations of the system~\cite{Andolina:18,Ferraro:18,Le:18,Andolina:19,Zhang:18}. In fact, in classical batteries we have an stable charging process, in sense that after a ``charge time'' $\tau_{\text{c}}$ the charge in battery becomes constant in time (blue continuum curve in Fig.~\ref{FigSchemeQuantumBattery}{\color{blue}a}). Therefore, such a behavior is highly wished for developing robust (stable) quantum batteries. In particular, the three-level battery introduced here can be implemented in several physical systems in which we can encode a ladder three-level systems, as sketched in Fig.~\ref{FigSchemeQuantumBattery}{\color{blue}b} shown like trapped ion systems and superconducting circuit QED system~\cite{devoret2012,wendin2017,gu2017}, for example. Here we propose that superconducting transmon qubits are particularly suitable candidates~\cite{Koch:07,you2007,Schreier08,barends2013}, the ladder-type three-level system is schematically presented in Fig.~\ref{FigSchemeQuantumBattery}{\color{blue}c}. These qubits are fabricated (typically planar) chips and consist of two Josephson junctions, with capacitance $C_{\text{J}}$ and energy $E_{\text{J}}$, that are shunted by a large capacitor with capacitance $C_{\text{B}}$. While the quantized circuit corresponding to a standard $LC$ circuit (capacitance-inductance) will result in a harmonic oscillator, the Josephson junction functions as a non-linear inductor and distorts the spectrum of the oscillator away from the equally-spaced one. The great success of the transmon qubit derives from its large ratio of Josephson to capacitative energy $E_{\text{J}}/E_{\text{C}}$, where $E_{\text{C}} = e^2 /2C$ with $C = C_{\text{J}} +C_{\text{B}} + C_{\text{g}}$. As discussed by Koch {\it et al.}~\cite{Koch:07}, a large $E_{\text{J}}/E_{\text{C}}$ renders the system very insensitive to charge noise, hence enhancing the lifetime. There is a catch however: the anisotropy in the spectrum also scales with $E_{\text{J}}/E_{\text{C}}$ and goes down with increasing ratio. Recall that the anisotropy needs to be significant to ensure that we are away from the equally-spaced case and can address individual levels to produce well-defined qubits. Fortunately, the anisotropy scales as a power law, while the noise sensitivity depends exponentially on this ratio. Hence, one may find a ``sweet spot'' with good anisotropy and long lifetimes. In practice, one typically aims for $E_{\text{J}}/E_{\text{C}} \sim  80-100$~\cite{Koch:07}. Typical energy level splittings in transmon qubits are of the order of 10~GHz, while the anisotropies are of the order of 100~MHz, and while this is much smaller than the splitting, modern microwave techniques are more than adequate to address such levels~\cite{gu2017}.

\begin{figure}[t!]
	%\input{Figs/SchemeQuantumBattery.plt}
	%\vspace{3.85cm}
	\centering
	\includegraphics[scale=0.55]{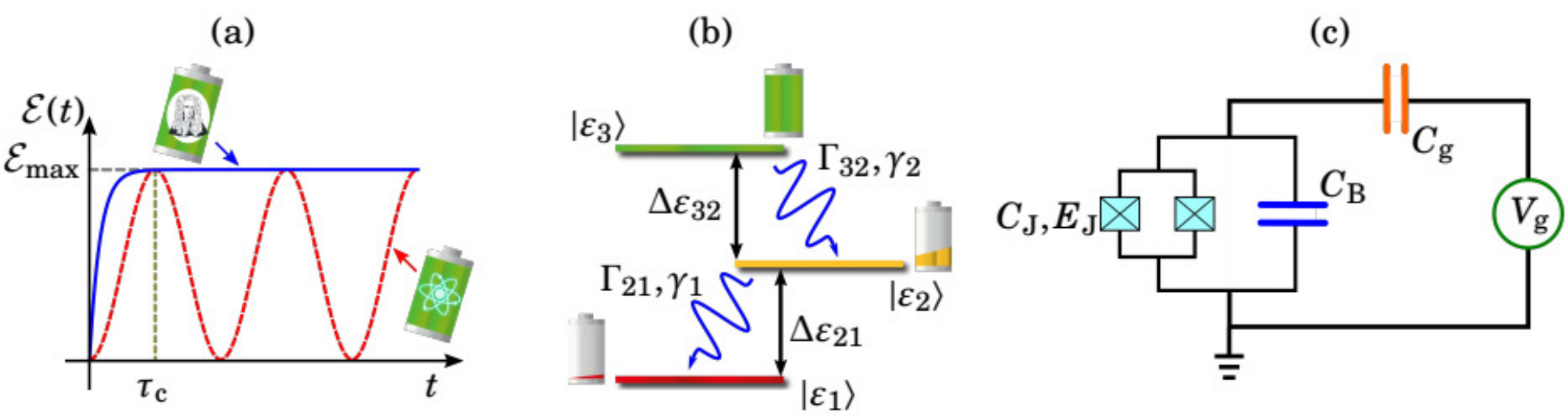}
	\caption{({\color{blue}a}) Schematic representation of the stored charge in classical and quantum batteries. While classical batteries present a charge stability after the charging time $\tau_{\text{c}}$, quantum batteries typically have oscillatory behavior. ({\color{blue}b}) Energy level structure of the ladder-type three-level system used here, where for a transmon qutrit the gap between energy levels read as $\Delta\varepsilon_{21} \sim 2\sqrt{2 E_{\text{J}}E_{\text{C}}}$ and $\Delta\varepsilon_{32} \sim \Delta\varepsilon_{21} - E_{\text{C}}$. Therefore, the maximum stored energy is given by means of the  energies $E_{\text{J}}$ and $E_{\text{C}}$ as $\Ecalb_{\text{max}} = \varepsilon_{3} - \varepsilon_{1} = \Delta\varepsilon_{32} + \Delta\varepsilon_{21} = 4\sqrt{2 E_{\text{J}}E_{\text{C}}} - E_{\text{C}}$. ({\color{blue}c}) The sketch of a superconducting transmon qubit circuit, where a Josephson junction of capacitance $C_{\text{J}}$ is shunted by a large capacitance $C_{\text{B}}$.}
	\label{FigSchemeQuantumBattery}
\end{figure}

The previous results have focused on an idealized setting where our quantum battery does not suffer any environmentally induced spoiling effects. However, because we want to propose a realistic implementation of our device, we need to deal with the performance and stability of our quantum battery when the most relevant environmental effects are taken into consideration (see Refs.~\cite{PRA_Ivanov,Vitanov:17} for other studies exploring decoherence effects on STIRAP protocols). In particular we will consider a dynamics governed by a Lindblad master equation~\cite{Lindblad:76} which takes into account both relaxation and dephasing phenomena, corresponding to the most natural non-unitary effects in superconducting circuits~\cite{Li:11,Martinis:03,JPCS_Li}, which we propose as a natural platform to realize our battery. The dynamics of the system is then given by the open-system dynamics equation
\begin{equation}
	\dot{\rho}_{\text{int}}(t) = \frac{1}{i\hbar} [H_{\text{int}}(t),\rho_{\text{int}}(t)] + \Lcal_{\text{rel}}[\rho_{\text{int}}(t)]+ \Lcal_{\text{dep}}[\rho_{\text{int}}(t)] \text{ , } \label{LindEq}
\end{equation}
where the superoperators $\Lcal_{\text{rel}}[\bullet]$ and $\Lcal_{\text{dep}}[\bullet]$ describe the relaxation and dephasing phenomena, respectively, and can be written as
\begin{subequations}
	\label{RelTerm}
	\begin{align}
		\Lcal_{\text{rel}}[\bullet] &= \sum_{k\neq j}\Gamma_{kj} \left[\sigma_{kj}\bullet\sigma_{jk} - \frac{1}{2}\{\sigma_{kk},\bullet\} \right] \text{ , } \\
		\Lcal_{\text{dep}}[\bullet] &= \sum_{j=2,3}\gamma_{j} \left[\sigma_{jj}\bullet\sigma_{jj} - \frac{1}{2}\{\sigma_{jj},\bullet\} \right] \text{ , }
	\end{align}
\end{subequations}
where $\sigma_{kj}=\ket{\varepsilon_{k}}\bra{\varepsilon_{j}}$ and $\Gamma_{kj}=\Gamma_{jk}$. Building on the general definitions we have introduced in Eqs.~\eqref{RelTerm}, we would like to clarify two important points on the characteristics of noise we consider in the rest of this work. First, the relaxation processes we consider are only the sequential decays, meaning, $\ket{\varepsilon_{3}}  \rightarrow  \ket{\varepsilon_{2}}$ and $\ket{\varepsilon_{2}}  \rightarrow  \ket{\varepsilon_{1}}$ characterized by the rates $\Gamma_{32}$ and $\Gamma_{21}$, respectively. We do not take into account the nonsequential decay mechanism which is responsible from inducing transitions like $\ket{\varepsilon_{3}}  \rightarrow  \ket{\varepsilon_{1}}$, since the rate associated with such a process, $\Gamma_{31}$, is an order of magnitude smaller for transmon qubits~\cite{Peterer:15} and, as we shall discuss later, one can take it as a negligible parameter in our discussion. Second, $\gamma_2$ and $\gamma_3$ determine the rates at which the superpositions between $\ket{\varepsilon_{1}}$ and $\ket{\varepsilon_{2}}$, and $\ket{\varepsilon_{1}}$ and $\ket{\varepsilon_{3}}$ are suppressed, respectively. Together, they also contribute to the dephasing of superpositions between $\ket{\varepsilon_{2}}$ and $\ket{\varepsilon_{3}}$. However, due to the nature of the STIRAP protocol with the dark state, the only dephasing rate that has an impact on the charging protocol is $\gamma_3$, since the state $\ket{\varepsilon_{2}}$ is never populated during the process. Clearly, there are a number of timescales and relevant noise parameters to fix in order to quantitatively assess the effectiveness of our protocol. In what follows, we will focus on those parameter ranges most relevant for transmon qubits, which provide a promising candidate architecture. Nevertheless, we expect the qualitative behavior discussed to hold in other relevant settings.

\emph{Stable charging under dissipation and decoherence --} As a first important question, we examine the effect that environmental spoiling mechanisms have on the charging process itself. To this end we consider $\Omega_{12}(t) = \Omega_{0} f(t)$ and $\Omega_{23}(t) = \Omega_{0} [1-f(t)]$, where $f(t)$ is a function which satisfies $f(0) = 0$ and $f(\tau) = 1$, such that the boundary conditions on $\Omega_{12}(t)$ and $\Omega_{23}(t)$ are satisfied and we realise the stable charging via STIRAP. We can readily examine the behavior of the ergotropy as a function of the dimensionless parameter $\Omega_{0}\tau$. In addition to the ergotropy, equally important is assessing the charging power of quantum batteries~\cite{Campaioli:18}, which we define as 
\begin{equation}
	\Pcalb(\tau) = \frac{\Ecalb(\tau)}{\tau} \text{ , } \label{CP}
\end{equation}
where $\Ecalb(\tau)$ needs to be understood as the amount of energy transferred to the battery from external fields during the time interval $\tau$. In order to make a meaningful comparison, we rescale $\Pcalb(\tau)$ with the maximal attainable power $\Pcalb_\text{max}$. As argued by Binder {\it et al}, it is physically reasonable to bound the amount of energy available for a given charging protocol~\cite{Binder:15}. The most efficient charging process therefore corresponds to one which needs only enough energy to fully charge the battery, in our case $\hbar(\omega_3 - \omega_1)$. We can then exploit the quantum speed limit~\cite{DeffnerReview} to determine the minimum time, $\tau_\text{QSL}$, needed for some time-independent process to charge the battery and thus corresponds to the most powerful charging obtainable, under this energy constraint~\cite{Binder:15}. Thus, $\Pcalb_\text{max}=\pi/(2\hbar(\omega_3 - \omega_1))$. We fix the functional form of $f(t)$ to be a simple linear ramp, $f(t) = t / \tau$. Naturally, one could consider any other ramp that satisfies the boundary conditions, however, as STIRAP is an adiabatic protocol, the means by which one manipulates the system is of little consequence. While from one ramp to another some qualitative differences may emerge in the behavior as one approaches the adiabatic regime, the quantitative features outlined in what follows persist.

\begin{figure}[t!]
	%\input{Figs/FigErgoPowerQB.plt}
	%\vspace{5.4cm}
	\centering
	\includegraphics[scale=0.55]{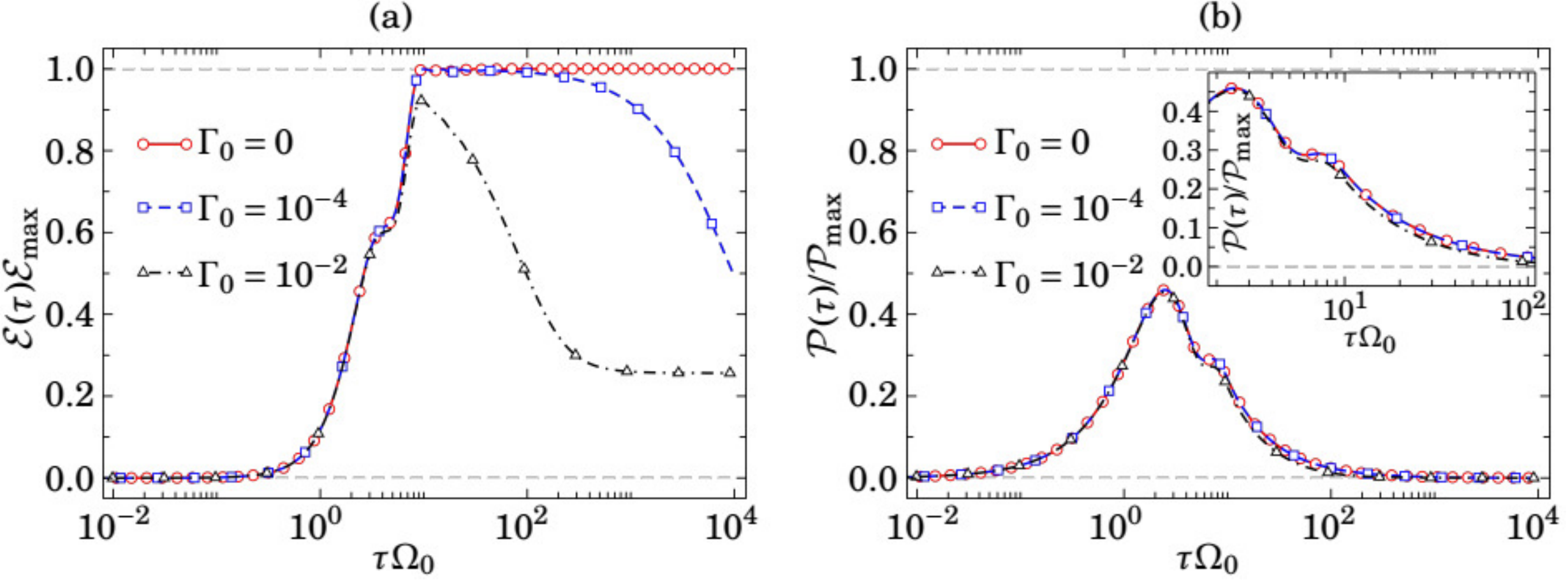}
	\caption{({\color{blue}a}) Ergotropy and ({\color{blue}b}) power for a linear ramp including the effects of both relaxation and dephasing. We chose the energy spectrum of our three-level system as $\omega_1 = 0$, $\omega_2 = \omega_{0}$ and $\omega_3 = 1.95\omega_{0}$ in order to account for the anharmonicity of the energy gaps in transmon qubits, resulting in $\Ecalb_{\text{max}} = 1.95\hbar\omega_{0}$. We set the rates characterizing the noise as $\Gamma_{32} = 2\Gamma_{21}$, $\gamma_2 = \Gamma_{21}$, and $\gamma_3 = \Gamma_{32} = 2\Gamma_{21}$, and $\Gamma_{21}=\Gamma_{0}\Omega_{0}$, again to match the state-of-the-art parameters measured for transmon qubits.}
	\label{FigErgoPowerQB}
\end{figure}

Therefore, by letting the system evolves under Eq.~\eqref{LindEq}, the instantaneous charge and charging power are shown in Fig.~\ref{FigErgoPowerQB}. The Fig.~\ref{FigErgoPowerQB}{\color{blue}a} shows the instantaneous charge (ergotropy) as funtion of the parameter $\Omega_{0}\tau$ for different values of the decohering parameters which we specify in details in the caption of the figure. Such parameters are considered in line with the transmon implementation we propose here. The asymptotic curve close to $1$ corresponds to the free-decohering situation. As already discussed in previous section, the stabilizing performance of three-level batteries adiabatically driven provide a very stable charging process. One can see that fast evolution gives a vanishingly small ergotropy as these timescales are far from the adiabatic limit, therefore the STIRAP protocol is ineffective and no population inversion can be observed. As we increase $\tau$, in line with the adiabatic theorem~\cite{Sarandy:05-1,Amin:09}, the maximum ergotropy grows and we achieve a fully charged state when the STIRAP protocol is faithfully implemented. We clearly see that in the case of no decoherence the charged state is perfectly stable for $\tau \gtrsim 10/\Omega_{0}$. Conversely, the ergotropy is affected when the decoherence effects become more significant. For small values of decoherence (blue, dashed curve) the STIRAP protocol is quite robust and only becomes significantly adversely affected when the times scales are an order of magnitude slower than strictly necessary. As the environmental effects are increased we find that achieving a fully charged battery is not possible, however, we can identify a range of values for $\tau$ for which we get the optimal stored charge, as we can see in the black, dotted curve in Fig.~\ref{FigErgoPowerQB}{\color{blue}a}. Thus, a given $\tau$ sets the speed of the adiabatic evolution and we can see that an optimality criterion between total evolution time and decoherence effects appears.

The power is shown in Fig.~\ref{FigErgoPowerQB}{\color{blue}b}, where we rescaled it using the maximum power value $\Pcalb_{\text{max}}$, which is only weakly affected for reasonable environmental parameters. Naturally, for fast protocols where the battery fails to charge the resulting power is negligible. As $\tau$ increases the charging power also increases until it reaches a maximum of $ \sim 0.5\Pcalb_\text{max}$. We are considering the optimal power $\Pcalb_{\text{max}}$ as provided by the quantum speed limit, then it is natural to see the discrepancy between $\Pcalb_\text{max}$ and the maximum power for an unitary dynamics. However, it is interesting to note that the maximum power does not correspond to when the battery is fully charged. By comparing Figs.~\ref{FigErgoPowerQB}{\color{blue}a} and~\ref{FigErgoPowerQB}{\color{blue}b} we see that, for all the considered noise values, the maximum ergotropy is achieved for $\tau \sim 10/\Omega_{0}$, which corresponds to $\Pcalb \sim 0.25\Pcalb_\text{max}$. Thus we find that there is a trade-off between the maximum achievable ergotropy and the power when stably charging a quantum battery via STIRAP. A promising method to boost the power of our protocol would be to employ so-called shortcuts-to-adiabaticity~\cite{Torrontegui:13,Odelin:19}. However, these techniques invariably come at the cost of some additional resources which will affect the resulting efficiency and power, but nevertheless may prove useful to ensure both fast and stable quantum batteries.

As already mentioned, it is important to stress that the results shown in Fig.~\ref{FigErgoPowerQB} do not take into account decay transitions between $\ket{\varepsilon_{3}}  \rightarrow  \ket{\varepsilon_{1}}$. This assumption is justified since, in case where no noise mechanism acts on our system, the adiabatic behavior is achieved for $\tau\Omega_{0} \sim 10$. Notice that the highest decay rate we consider is $\Gamma_{21}/\Omega_{0} \leq 10^{-2}$, which allows us to write $(\Gamma_{21}\tau/\tau\Omega_{0})\leq 10^{-2}$ or $\Gamma_{21}\tau \leq 10^{-2}\tau\Omega_{0}$. Replacing the parameters with their values in the adiabatic limit we find $\Gamma_{21}\tau  \lesssim  10^{-2}10 = 10^{-1}$. Thus, $\Gamma_{21}$ is indeed relevant to our discussion. As experimentally shown~\cite{Peterer:15}, the timescale $\tau_{31}$ for the process $\ket{\varepsilon_{3}}  \rightarrow  \ket{\varepsilon_{1}}$ is $\tau_{31}   \sim   10^2 \tau_{21}$, where $\tau_{21}$ is the timescale of the process $\ket{\varepsilon_{2}}  \rightarrow  \ket{\varepsilon_{1}}$. Therefore we can write $\Gamma_{31} \sim 10^{-2} \Gamma_{21}$. From this we have that $\Gamma_{31}\tau \sim 10^{-2} \Gamma_{21}\tau \lesssim 10^{-3}$, which allows us to conclude that the relaxation due to non-sequential rates, i.e. $\Gamma_{31}$, are negligible for the adiabatic time scales considered in this work.

\subsection{Stable and power switchable two-level quantum batteries}

In previous section we show how the stability of quantum batteries can be obtained by introducing a three-level quantum battery driven by adiabatic fields. There is not doubt that the additional energy level is very helpful for the results, since we drive the system through a dark-state, but here we will investigate the role of the adiabatic dynamics in a general context of quantum batteries. To this end, we introduce here the notion of power operator, a new quantum observable that develops a important role in theory of quantum batteries, as we shall see.

\subsubsection*{The power operator}

The derivation of the power operator takes into account the following physical model. Consider a composite quantum system described by a Hilbert space $\Bcal \otimes \Acal$, where $\Bcal$ is associated with a QB and $\Acal$ refers to an auxiliary system, which will play the role of a CH. Our choice is just a convenience, our results are suitably applicable to cases where the Hilbert space $\Bcal$ is replaced by a quantum charger and $\Acal$ becomes a fully discharged quantum battery. Therefore, our approach is useful for both charge and discharge process of quantum batteries. However, here we will consider the discharging process. The dynamics for the composite system is driven by the Hamiltonian $H(t) = H_{0} + H_{\text{int}}(t)$, where $H_{0}$ is the individual energy contribution for both QB and CH subsystems, while $H_{\text{int}}(t)$ is the corresponding interaction Hamiltonian. 
The Hamiltonian $H_{0}$ can be written as $H_{0} = H^{\Bcal}_{0} + H^{\Acal}_{0}$, with $H^{\Bcal}_{0}$ being the battery inner Hamiltonian and $H^{\Acal}_{0}$ the auxiliary free Hamiltonian. In general, the inner parts of the QB may interact with each other, with this inner interaction included in Hamiltonian $H^{\Bcal}_{0}$.

As already discussed, the success of the energy transfer from the QB to the CH can be measured by the {\textit{ergotropy}, i.e. the maximum amount of work which can be further extracted from the CH. Then, here we identify $\sigma_{\rho} = \rho^{\Acal}_{\text{gs}}$, where $\rho^{\Acal}_{\text{gs}}$ is the ground state of $H^{\Acal}_{0}$, because $\Acal$ is our system of interest. The contribution $\trs{H\sigma_{\rho}}$ to the ergotropy is the empty energy state, which is here given by $E^{\Acal}_{\text{emp}} = \trs{H^{\Acal}_{0}\rho^{\Acal}_{\text{gs}}}$. Hence, the ergotropy of the CH reads $\Ecalb^{\Acal}_{0} = \tr{\rho_{0}^{\Acal}H^{\Acal}_{0}} - E^{\Acal}_{\text{emp}}$, where $\rho_{0}^{\Acal}$ is the CH initial state. Notice that, by starting the evolution in the ground state of the CH, we have $\Ecal^{\Acal}_{0}=0$. However, by coupling the QB to the CH, an energy transfer process will take place, which will be governed by Schr\"odinger equation. It will be useful here to adopt the interaction picture, where the density operator $\rho_{\text{int}}(t)$ is driven by $i\hbar\dot{\rho}_{\text{int}}(t) = [H_{\text{int}}(t),\rho_{\text{int}}(t)]$, with the dot symbol denoting time derivative throughout the paper. The amount of energy $C(t)$ available at instant $t$ is obtained from the instantaneous ergotropy, reading 
	\begin{equation}
		C(t) = \Ecalb^{\Acal}_{0}(t) = \tr{H_{0}^{\Acal}\rho_{\text{int}}(t)} - E^{\Acal}_{\text{emp}} \text{ . }
		\label{Ct}
	\end{equation}
	Eq.~(\ref{Ct}) provides the amount of energy transferred from the QB to the CH, which can be further used as a source of work. Throughout our discussion we are assuming that the CH final state is a pure state, otherwise the above equation cannot be considered as ergotropy~\cite{Santos:20c}. Now, the starting point for the QB proposal is to consider the instantaneous power, which is defined as $\Pcalb(t) = \dot{C}(t)$. Notice that $\Pcalb(t)$ quantifies the transfer rate at which energy is available at the CH. Then, from definition of the power operator, we get
	\begin{equation}
		\Pcalb(t) = \dot{C}(t) = \tr{H_{0}^{\Acal}\dot{\rho}(t)} = \frac{1}{i\hbar}\tr{H_{0}^{\Acal}[H(t),\rho(t)]} \text{ , }
	\end{equation}
	where we used Schr\"odinger equation for $\rho(t)$. Thus, we find
	\begin{align}
		\Pcalb(t) &= \frac{1}{i\hbar}\tr{H_{0}^{\Acal}\left[H(t)\rho(t)-\rho(t)H(t)\right]} = \frac{1}{i\hbar}\tr{H_{0}^{\Acal}H(t)\rho(t)-H_{0}^{\Acal}\rho(t)H(t)} \nonumber \\
		&= \frac{1}{i\hbar}\tr{H_{0}^{\Acal}H(t)\rho(t)-H(t)H_{0}^{\Acal}\rho(t)} = \frac{1}{i\hbar}\tr{ [H_{0}^{\Acal},H_{\text{C}}(t)]\rho(t)} \text{ , }
	\end{align}
	where we used that $[H_{0}^{\Acal},H(t)] = [H_{0}^{\Acal},H_{0} + H(t)] = [H_{0}^{\Acal},H_{\text{C}}(t)]$. Then, by defining
	\begin{equation}
		\hat{\Pcalb}(t) = \frac{1}{i\hbar}[H_{0}^{\Acal},H_{\text{C}}(t)] \text{ , } \label{EqPowerOperator}
	\end{equation}
	we conclude that
	\begin{equation}
		\Pcalb(t) = \tr{ \hat{\Pcalb}(t)\rho(t)} \text{ . }
	\end{equation}
	
	For convenience, we need to consider the system dynamics in interaction picture so that the internal Hamiltonian $H_{0}$ does not contribute to the dynamics and for the charging process (see e.g. the section~\ref{SubSecStableThreeLevelAdQB}). Therefore, by considering Schr\"odinger equation in the interaction picture, we have
	\begin{equation}
		\dot{\rho}_{\text{int}}(t) = \frac{1}{i\hbar} [H_{\text{int}}(t),\rho_{\text{int}}(t)] \text{ , } \label{VonNeuInt}
	\end{equation}
	where the dot symbol denotes time derivative, $H_{\text{int}}(t) = \Zcal^{\dagger}(t) H_{\text{C}}(t) \Zcal(t)$, and 
	$\rho_{\text{int}}(t) = \Zcal^{\dagger}(t)\rho(t)\Zcal(t)$, with $\Zcal^{\dagger}(t) = e^{iH_{0}t/\hbar}$. 
	It is important to mention that the relevant quantities to be computed here are independent of the frame 
	used to study the dynamics. In fact, by computing the ergotropy 
	in the interaction picture, we get
	\begin{align}
		C_{\text{int}}(t) &= \tr{H_{0}^{\Acal}\rho_{\text{int}}(t)} - E^{\Acal}_{\text{emp}} = \tr{\Zcal^{\dagger}(t)H_{0}^{\Acal}\rho(t)\Zcal(t)} - E^{\Acal}_{\text{emp}} \nonumber \\
		&= \tr{H_{0}^{\Acal}\rho(t)\Zcal(t)\Zcal^{\dagger}(t)} - E^{\Acal}_{\text{emp}} 
		= \tr{H_{0}^{\Acal}\rho(t)} - E^{\Acal}_{\text{emp}}  = C(t) \text{ . }
	\end{align}
	
	By following the same calculation, in the interaction picture, we can use the Eq.~\eqref{VonNeuInt} to find that 
	\begin{equation}
		\Pcalb_{\text{int}}(t) = \tr{ \hat{\Pcalb}_{\text{int}}(t)\rho_{\text{int}}(t)} = \Pcalb(t) \text{ , }
	\end{equation}
	where $\Pcalb_{\text{int}}(t)$ is the power in interaction picture, which reads $\hat{\Pcalb}_{\text{int}}(t) = [H_{0}^{\Acal},H_{\text{int}}(t)]/(i\hbar)$. It shows that no information is lost if we consider our analysis from the interaction picture, then whenever it is convenient we will consider our discussion from interaction picture.
	
	In conclusion, it is possible to introduce a Hermitian operator able to compute the expected value of the power, with the above equation showing that the eigenvalues of $\hat{\Pcalb}(t)$ provide the expected values of a measurement of such physical quantity. As a first application of the observable $\hat{\Pcalb}(t)$, we can show that the \textit{adiabatic dynamics} allows for a {\it stable energy transfer to the CH}. In fact, by defining that the system QB-CH undergoes an adiabatic energy transfer process when the composite state $\ket{\psi(t)}$ evolves adiabatically under the Hamiltonian $H(t)$ that drives the system. Let the system be initialized in the state $\ket{\psi(0)} = \sum_{n} c_{n} \ket{E_{n}(0)}$, where $\{\ket{E_{n}(0)}\}$ is the set of instantaneous eigenstates of $H(0)$. 
	Then, in the adiabatic regime, we have $\ket{\psi_{\text{ad}}(t)} = \sum_{n} c_{n} e^{i\theta^{\text{ad}}_{n}(t)}\ket{E_{n}(t)}$, where $\theta^{\text{ad}}_{n}(t)$ are the adiabatic phases accompanying the adiabatic dynamics~\cite{Berry:84}. 
	Therefore
	\begin{equation}
		\Pcalb_{\text{ad}}(t) = \frac{1}{i\hbar} \sum\nolimits_{n,m} c_{n} c^{\ast}_{m} e^{i\Delta^{\text{ad}}_{nm}(t)} \bra{E_{n}(t)} [H_{0}^{\Acal},H(t)] \ket{E_{m}(t)} \text{ , }
	\end{equation}
	with $\Delta^{\text{ad}}_{nm}(t) = \theta^{\text{ad}}_{n}(t) - \theta^{\text{ad}}_{m}(t)$ and where we have used $[H_{0}^{\Acal},H(t)] = [H_{0}^{\Acal}, H_{\text{C}}(t)]$. Thus, we get
	\begin{equation}
		\Pcalb_{\text{ad}}(t) = \frac{1}{i\hbar} \sum\nolimits_{n,m} c_{n} c^{\ast}_{m} e^{i\Delta^{\text{ad}}_{nm}(t)} \left[E_{m}(t)-E_{n}(t)\right]\bra{E_{n}(t)}H_{0}^{\Acal}\ket{E_{m}(t)} 
		\text{ . } \label{EqAdPowerGen}
	\end{equation}
	Now, let us consider the case where our system starts in a single eigenstate of the Hamiltonian $H(0)$, for example the $k$-th eigenstate of $H(0)$. By doing that, we have $c_{n} = \delta_{nk}$ and the above equation becomes
	\begin{equation}
		\Pcalb_{\text{ad}}(t) = \frac{1}{i\hbar} \sum\nolimits_{n,m} \delta_{nk} \delta_{mk} e^{i\Delta^{\text{ad}}_{nm}(t)} \left[E_{m}(t)-E_{n}(t)\right]\bra{E_{n}(t)}H_{0}^{\Acal}\ket{E_{m}(t)} = 0\text{ . } 
	\end{equation}
	
	Therefore, since the stability condition for an interval $[t_{1},t_{2}]$ can be mathematically written as $\Pcalb(t) = 0$ for any $t \in [t_{1},t_{2}]$, this result allows us to conclude that adiabatic dynamics is a strategy to get a stable charging process in quantum batteries because $\Pcalb_{\text{ad}}(t\geq \tau_{\text{ad}}) = 0$, where $\tau_{\text{ad}}$ is the required time to achieve the adiabatic regime. The same calculation as be done for the situation where we have a degenerate eigenspace of the Hamiltonian, where get the same result.
	
	As a second application of the observable $\hat{\Pcalb}(t)$, it is introduced a general {\it energy trapping mechanism}~\cite{Santos:20c}. To begin with, assume that $H_{\text{int}}(t)$ is a constant Hamiltonian, 
	namely, $H_{\text{int}}(t)=H_{\text{int}}$, which implies that $\hat{\Pcalb}(t)=\hat{\Pcalb}$ is also a constant operator. Then, notice that, if $\hat{\Pcalb}$ and $H_{\text{int}}$ commute, 
	we have that  $\hat{\Pcalb}$ is a constant of motion. Assume now that there is a common eigenstate $\ket{p_{0}}$ of $\hat{\Pcalb}$ and $H_{\text{int}}$ with power eigenvalue $p_{0} = 0$. 
	Therefore, by initially preparing the quantum system at the quantum state $\ket{\psi(0)} = \ket{p_{0}}$, no amount of energy can be extracted from the QB. In fact, let us consider the 
	initial state of the system as $\ket{\psi(0)} = \ket{p_{0}}$, where we have $\hat{\Pcalb} \ket{p_{0}} = 0$ and $H_{\text{int}}\ket{p_{0}} = E_{p0}\ket{p_{0}}$. The evolved state reads 
	$\ket{\psi(t)} = e^{-iE_{p0}t/\hbar}\ket{p_{0}}$. Therefore, the instantaneous power is $\Pcalb(t) = \bra{\psi(t)} \hat{\Pcalb} \ket{\psi(t)} = \bra{p_{0}} \hat{\Pcalb} \ket{p_{0}} = 0$. 
	Hence, no amount of energy can be introduced or extracted from the system driven by $H_{\text{int}}$. Remarkably, this conclusion holds even if $\hat{\Pcalb}$ and $H_{\text{int}}$ 
	does {\it not} commute, as long as they share at least a single eigenstate with power eigenvalue $p_{0} = 0$. This less restrictive situation is indeed even more interesting and is it explored in Ref.~\cite{Santos:20c}. Indeed, the the trapping mechanism opens perspectives for a class of QBs in which the energy is not transferred even when the battery is connected 
	to the CH. 
	
	As we shall see, we can also develop an activation mechanism, where we can control the exact time to discharge the energy to the CH. In this thesis we will focus on adiabatic performance of such devices and the trapping mechanism will be considered in this framework. However in Ref.~\cite{Santos:20c} we study the trapping mechanism in terms of time-independent fields, where a Bell quantum battery is proposed because we use Bell states to achieve optimal performance of the battery. 
	
	\subsubsection*{Adiabatic model for stable QBs}
	
	As already introduced, phenomenon known as \textit{spontaneous discharge} is not desired, but in our case we can circumvents this problem if we eliminate the energy oscillates between the QB and the CH after the QB starts its energy distribution. In turn, there will be energy revivals in the QB for specific times, which prevents a stable battery discharge. Here, inspired by Eq.~\eqref{EqAdPowerGen}, we can provide a general approach to stabilizing the charging processes of QBs. The control of the battery discharge is again provided by a time-dependent Hamiltonian $H(t)$. It is assumed that $H(t)$ can be turned on and off with high control. As we turn it on, it is also able to ensure stability in the discharge process. As in the case of the Bell QB discussed in Ref.~\cite{Santos:20c}, we begin by taking the initial state of the whole system as $\ket{\psi(0)}=\ket{\beta_{11}}_{\Bcal} \ket{0}_{\Acal}$. Then, let us consider a suitable three-qubit time-dependent Hamiltonian, which reads (in interaction picture)
	\begin{eqnarray}
		H(t) = [1-f(t)] H_{\text{ini}} + [1-f(t)]f(t) H_{\text{mid}} + f(t) H_{\text{fin}} \text{ . } \label{EqHAdTwoCelQB}
	\end{eqnarray}
	
	Each part of the above Hamiltonian is important to the protocol. The initial Hamiltonian $H_{\text{ini}} =  \hbar J (\sigma_{x}^{\Bcal_1}\sigma_{x}^{\Bcal_2} + \sigma_{y}^{\Bcal_1}\sigma_{y}^{\Bcal_2})$ sets the initial battery state $\ket{\psi(0)}$ as its ground state, so that the desired state $\ket{00}_{\Bcal} \ket{1}_{\Acal}$ will be the system at the end of the evolution if we choose the final Hamiltonian $H_{\text{fin}}$ as $H_{\text{fin}} =  \hbar J ( \sigma_{z}^{\Bcal_1}\sigma_{z}^{\Acal} + \sigma_{z}^{\Bcal_2}\sigma_{z}^{\Acal} )$. To end, we need to guarantee that the system will flow an adiabatic dynamics in interaction picture we introduce the middle Hamiltonian $H_{\text{mid}} =  \hbar J ( \sigma_{x}^{\Bcal_1}\sigma_{x}^{\Bcal_2} + \sigma_{y}^{\Bcal_1}\sigma_{y}^{\Bcal_2} + \sigma_{x}^{\Bcal_2}\sigma_{x}^{\Acal} + \sigma_{y}^{\Bcal_2}\sigma_{y}^{\Acal} )$, so that there are no level crossings for $H(t)$, which is assured by the middle Hamiltonian $H_{\text{mid}}$. This allows for the quantum evolution towards the target state through the adiabatic dynamics.
	
	\begin{figure}
		%\input{Figs/FigAdChargeQBCH.plt}
		%\vspace{5.25cm}
		\centering
		\includegraphics[scale=0.6]{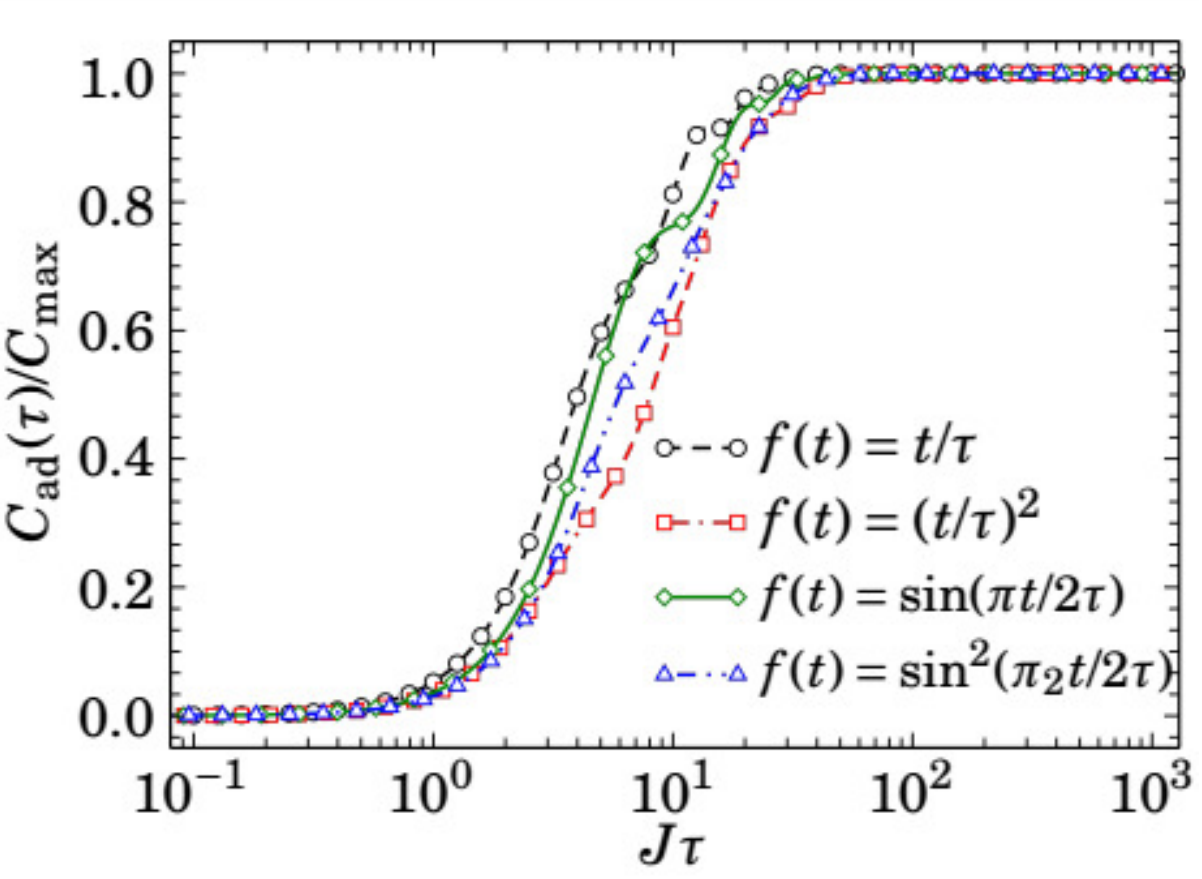}
		\caption{Amount of energy, as multiple of $C_{\text{max}} = 2\hbar \omega_{0}$, transferred from battery to consumption center through an adiabatic dynamics as function of the parameter $J\tau$. Here we considered different function interpolations $f(t)$.}\label{FigAdiabatic}
	\end{figure}
	
	Before all, it is important to deal with the problem of the double degenerate spectrum of the ground state of $H_{\text{fin}}$. Since both states $\ket{00}_{\Bcal} \ket{1}_{\Acal}$ (desired state) and $\ket{11}_{\Bcal} \ket{0}_{\Acal}$ are eigenvectors of $H_{\text{fin}}$ with same eigenvalue, then it remains to prove that the final state $\ket{00}_{\Bcal} \ket{1}_{\Acal}$ is achieved due to the symmetries of $H(t)$. Then, by defining the parity operator $\Pi_{z} = \sigma_{z}^{\Bcal_1}\sigma_{z}^{\Bcal_2}\sigma_{z}^{\Acal}$, we can verify that $[H(t),\Pi_{z}] = 0$ for all $t$. Therefore, $H(t)$ preserves the parity of the state throughout the evolution. Since we start the evolution at the state $\ket{\psi(0)}$, which has parity eigenvalue $-1$, the system will evolve to subsequent states with the same parity. Hence, transitions to $\ket{11}_{\Bcal} \ket{0}_{\Acal}$ at the end of the evolution are forbidden, so that all the energy of the QB can be adiabatically transferred to the CH, achieving the final state $\ket{00}_{\Bcal} \ket{1}_{\Acal}$. Then, by considering the CH system as a two-level system with internal Hamiltonian $H^{\Acal}_{0} = \hbar \omega_{0} \sigma_{z}^{\Acal}$, we illustrate in Fig.~\ref{FigAdiabatic} the instantaneous charge transferred to the CH. One shows that this transfer is indeed stable for several different choices of the interpolation function $f(t)$. In agreement with Eq.~(\ref{EqAdPowerGen}), the QB discharge process is stable and robust against variations of the interpolation scheme. 
	
	\newpage

\section{Conclusions of this chapter}

In this chapter we discussed the main results of this thesis concerning adiabatic dynamics in closed system. As a first result~\cite{Hu:19-b}, we studied a validation mechanism for the ACs already introduced in the literature. From a particular example of a two-level system driven by a time-dependent oscillating Hamiltonian, we showed that all the conditions widely used in literature are neither sufficient nor necessary in predicting the adiabatic dynamics of such evolution. Then, we have introduced a framework to validate such ACs for generic discrete multi-particle Hamitonians, which is
rather convenient to analyze quantum systems at resonance. Our strategy is based on the analysis of ACs in a suitably
designed non-inertial reference frame. By using this new strategy, we have both theoretically and experimentally shown that
the relevant ACs considered here, which include the traditional AC, are sufficient and necessary to
describe the adiabatic behavior of the system. In this case, sufficiency and
necessity are fundamentally obtained through the non-inertial frame map, with all the conditions failing to
point out the adiabatic behavior in the original reference frame. The experimental realization has been performed through a single trapped Ytterbium ion, with excellent agreement with the theoretical results. More
generally, the validation of ACs has been expanded to arbitrary Hamiltonians through the Theorems~\ref{TheoAdiab} and~\ref{TheoAdiabTI}, with
detailed conditions provided for a large class of Hamiltonians in the form of Eq.~\eqref{EqGenHAdOrig}. Therefore, instead of
looking for new approaches for defining ACs, we have introduced a mechanism based on “fictitious potentials”
(associated with non-inertial frames) to reveal a correct indication of ACs, both at resonance and of-resonant
situations. In addition, as a further example, we discussed how the validation mechanism through non-inertial
frames can be useful to describe the results presented in Ref.~\cite{Suter:08}, where the adiabatic dynamics of a single spin-$1/2$ in NMR had been previously investigated.

As an application of these general results, in a second moment we investigated the performance of adiabatic quantum batteries, where we have shown how to adiabatically charge/discharge quantum batteries in two different contexts. Firstly, we introduced a three-level quantum battery able to be fully charged through a very stable (adiabatic) trajectory. To this end, we employ STIRAP that allows one to bypass the undesired spontaneous discharging due to imprecise control on the fields that occur if the charging process couples directly only two levels of the battery, e.g. the ground and maximally excited states. While (effective) qubit batteries require careful manipulation of the charging fields, our three-dimensional quantum battery is able to exploit the STIRAP protocol to ensure a robust and stable charge. We explicitly consider the effects of the most relevant sources of noise and have shown that even for moderate values of decoherence and dissipation, our adiabatic quantum battery is quite robust. For more severe environmental effects we have shown that an optimal time emerges that dictates the maximal achievable ergotropy. We finally proposed that superconducting transmon qubits provide a promising implementation for adiabatic quantum batteries. Our results show that clear advantages can be gained by exploiting higher-dimensional quantum systems. As such we expect that extending our analysis to consider arrays of high-dimensional quantum batteries, and the role of entanglement in the collective charging process, will be of significant interest~\cite{Alicki:13, Binder:15, PRL2013Huber}.

After, we discussed on general results of adiabatic charging/discharging process of quantum batteries. Since two level system are more easy to control in experimental setups, it is quite important to circumvent the issue of stable charging of two-level systems. Then, we have proposed a system constituted by a QB (charged) and a CH that will be receive the amount of energy initially stored in battery. Of course, our system can be easily adapted to the situation of a initially discharged QB coupled to a charger. As main element of our study, we have introduced a new observable associated with the power, where the power operator is the source of two main applications: (i) a trapping energy mechanism based on a common eigenstate between the power operator and the interaction Hamiltonian, in which the battery can indefinitely retain its energy even if it is coupled to the CH; (ii) an asymptotically stable discharge mechanism, which is achieved through an adiabatic evolution leading to vanishing power. Then, we can solve in general the stability of the energy transfer process, with no backflow of energy from the CH to the QB. More specifically, we introduced a piecewise time-dependent Hamiltonian that stabilizes the energy discharging process through the adiabatic dynamics. The scaling of this model of QB can be achieved by adding Bell pairs to the QB, with each Bell pair providing a controllable mechanism for full and stable charge of a qubit in the CH. Experimental implementations may inherit hardware designs originally proposed for adiabatic quantum computing, being available with current technology.  

In both protocols of QB introduced here, we showed that they allow for the design of batteries that are robust to intrinsic errors in real physical scenarios concerning unknown delays in turning off the charging fields. 

\let\cleardoublepage\clearpage

\chapter{Adiabatic dynamics in open quantum systems}

\initial{W}hen we consider the adiabatic theorem as derived in the previous chapter, we take into account the absence of non-unitary effects on our system, where the system follows a unitary dynamics from some initial state to the final one. In cases we need to deal with decohering effects, due to the coupling of our system with some environment, the dynamical equation is drastically changed and the notion of adiabaticity previously established is breakdown. For this reason, we need to derive an alternative approach in that case. In this section we review the main results about adiabaticity in open systems, particularly we will following the most general approach as derived in Ref.~\cite{Sarandy:05-1}. It is possible to find alternative studies on adiabatic behavior in open systems~\cite{Venuti:16,Yi:07}, but all of the them can be viewed as a particular application of the most general study in Ref.~\cite{Sarandy:05-1}. The results presented in this chapter refer to the following works

$\bullet$ C.-K. Hu, A. C. Santos, J.-M. Cui, Y.-F. Huang, M. S. Sarandy, C.-F. Li, and G.-C. Guo,
“Adiabatic quantum dynamics under decoherence in a controllable trapped-ion setup”,
Phys. Rev. A \textbf{99}, 062320 (2019).

$\bullet$ C.-K. Hu, A. C. Santos, J.-M. Cui, Y.-F. Huang, D. O. Soares-Pinto, M. S. Sarandy, C.-F. Li,
and G.-C. Guo, “Quantum thermodynamics in adiabatic open systems and its trapped-ion
experimental realization”, arXiv e-prints , arXiv:1902.01145 (2019), arXiv:1902.01145
[quant-ph].

\section{Adiabatic approximation in open systems} \label{SecAdApprox}

As discussed in section~\ref{SecSuperOpForm}, the non-diagonalizable characteristic of the superoperator $\Lcalb(t)$ needs to be taken into account before defining adiabaticity in a non-unitary evolution. In this thesis we will consider the general definition of adiabaticity as established in Ref.~\cite{Sarandy:05-1}.

\begin{definition}[Adiabaticity in open systems]\label{DefAdiabOS}
	An open quantum system is said to undergo an adiabatic dynamics if the evolution of its Hilbert-Schmidt space can be decomposed into uncoupled Lindblad-Jordan eigenspaces associated with distinct, time-dependent, non-crossing eigenvalues of $\Lcalb[\bullet]$.
\end{definition}

By Lindblad-Jordan eigenspaces we mean a set of vectors composed by the eigenvectors of a single Jordan block of the Lindblad superoperator. Therefore, the independent dynamics of a Jordan block of $\Lcalb[\bullet]$ is analogous to the independent dynamics of a degenerated eigenspace of the Hamiltonian through an adiabatic dynamics in closed system. However, it is important to highlight this discussion as an analogy. In cases where $\Lcalb[\bullet]$ admits uni-dimensional Jordan blocks, which means that $\Lcalb[\bullet]$ is diagonalizable, the adiabatic dynamics in open system is not the same as adiabatic dynamics in closed system, since the set of eigenvectors of $\Lcalb[\bullet]$ does not characterize the spectrum of the Hamiltonian that drives the system. In addition, as we shall see, some eigenvectors of $\Lcalb[\bullet]$ does not represent realistic physical states of a system, unlike the set of eigenvectors of the system Hamiltonian.

Thus, the immediate question is: \textit{under what conditions can we achieve the adiabatic behavior in open systems?} To answer this question, our starting point is writing the evolved state $\dket{\rho(t)}$ expanded in the basis $\dket{\Dcalb_{\alpha}^{n_{\alpha}}(t)}$ as
\begin{equation}
	\dket{\rho(t)} = \sum_{\alpha=1}^{N} \sum _{n_{\alpha} = 1}^{N_{\alpha}} r_{\alpha}^{n_{\alpha}}(t) \dket{\Dcalb_{\alpha}^{n_{\alpha}}(t)} \text{ , } \label{EqRhoExpD}
\end{equation}
where $r_{\alpha}^{n_{\alpha}}(t)$ are coefficients to be determined. By putting the above equation in Eq.~\eqref{EqEqSuperLindEq}, we use the eigenvalue equations in Eqs.~\eqref{EqEqEigenStateL} to write
\begin{align}
	\dot{r}_{\beta}^{k}(t) %& = \lambda_{\beta}(t)r_{\beta}^{k}(t) + r_{\beta}^{k+1}(t) - 
	%\sum_{\alpha=1}^{N} 
	%\sum _{n_{\alpha} = 1}^{N_{\alpha}} r_{\alpha}^{n_{\alpha}}(t) \dinterpro{\Ecalb_{\beta}^{k}(t)}{\dot{\Dcalb}_{\alpha}^{n_{\alpha}}(t)} \nonumber \\
	& = \lambda_{\beta}(t)r_{\beta}^{k}(t) - r_{\beta}^{k}(t) \dinterpro{\Ecalb_{\beta}^{k}(t)}{\dot{\Dcalb}_{\beta}^{k}(t)} + r_{\beta}^{k+1}(t) 
	- \sum _{n_{\beta} \neq k}^{N_{\beta}} r_{\beta}^{n_{\beta}}(t) \dinterpro{\Ecalb_{\beta}^{k}(t)}{\dot{\Dcalb}_{\alpha}^{n_{\alpha}}(t)} \nonumber \\ 
	&- \sum_{\alpha\neq \beta}^{N} \sum _{n_{\alpha} = 1}^{N_{\alpha}} r_{\alpha}^{n_{\alpha}}(t) \dinterpro{\Ecalb_{\beta}^{k}(t)}{\dot{\Dcalb}_{\alpha}^{n_{\alpha}}(t)} \text{ . } \label{Eqrdot}
\end{align}

The two first terms in right-hand side of the above equation is associated with the perfect uncoupled evolution, where we have decoupling between different Jordan blocks and inside of each block. The third and fourth terms tell us about the coupling between the $k$-th and all of other eigenvectors inside a single block. Therefore, in according with the Definition~\ref{DefAdiabOS}, both cases could be considered as genuine adiabatic dynamics in context of open systems. 

\subsection{One-dimensional Jordan blocks} \label{SecAdApprox-1DJB}

Before considering the most generalized case, let us restrict our first analysis to the case in which the $\Lmath(t)$ admits an Jordan decomposition with one-dimensional Jordan blocks in Eq.~\eqref{EqEqLindJ}. Under this assumption, the quasi-eigenstate relations in Eqs.~\eqref{EqEqEigenStateL} become genuine eigenstate equations given by
\begin{subequations}\label{EqEqEigenOneL}
	\begin{align}
		\Lmath(t)\dket{\Dcalb_{\alpha}(t)} &= \lambda_{\alpha}(t)\dket{\Dcalb_{\alpha}(t)} \text{ , } \\
		\dbra{\Dcalb_{\alpha}(t)}\Lmath(t) &= \dbra{\Ecalb_{\alpha}(t)}\lambda_{\alpha}(t) \text{ . }
	\end{align}
\end{subequations}

Hence, the Eq.~\eqref{Eqrdot} can be reduced to
\begin{equation}
	\dot{r}_{\beta}(t) = \lambda_{\beta}(t)r_{\beta}(t) - r_{\beta}(t) \dinterpro{\Ecalb_{\beta}(t)}{\dot{\Dcalb}_{\beta}(t)} - \sum_{\alpha\neq \beta}^{N} r_{\alpha}(t) \dinterpro{\Ecalb_{\beta}(t)}{\dot{\Dcalb}_{\alpha}(t)} \text{ . } \label{EqrdotOne}
\end{equation}

The equation above allows us to get a first notion of how we can make an adiabatic approximation, since the last term couples eigenvectors with distinct eigenvalues. Then, we define 
\begin{equation}
	r_{\beta}(t) = p_{\beta}(t) e^{\int_{t_{0}}^{t} \lambda_{\beta}(\xi)d\xi} \text{ , }
\end{equation}
and one finds a equation for $p_{\beta}(t)$ given by
\begin{equation}
	\dot{p}_{\beta}(t) = - p_{\beta}(t) \dinterpro{\Ecalb_{\beta}(t)}{\dot{\Dcalb}_{\beta}(t)} - \sum_{\alpha\neq \beta}^{N} p_{\alpha}(t) e^{\int_{t_{0}}^{t} \lambda_{\alpha}(\xi) - \lambda_{\beta}(\xi) d\xi}\dinterpro{\Ecalb_{\beta}(t)}{\dot{\Dcalb}_{\alpha}(t)} \text{ , } \label{EqpdotOne}
\end{equation}
with the second term in right-hand side being the responsible for coupling distinct Jordan-Lindblad eigenspaces during the evolution. If we are able to use some strategy to minimize the effects of such term in the above equation, we can approximate the dynamics to
\begin{equation}
	\dot{p}_{\beta}(t) \approx - p_{\beta}(t) \dinterpro{\Ecalb_{\beta}(t)}{\dot{\Dcalb}_{\beta}(t)} \text{ , }
\end{equation}
so that the adiabatic solution of the dynamics for $r_{\beta}(t)$ can be written as
\begin{equation}
	r_{\beta}(t) = r_{\beta}(t_{0}) e^{\int_{t_{0}}^{t} \lambda_{\beta}(\xi)d\xi}
	e^{- \int_{t_{0}}^{t} \dinterpro{\Ecalb_{\beta}(\xi)}{\dot{\Dcalb}_{\beta}(\xi)}d\xi} \text{ . } \label{Eqr1BJdec}
\end{equation}
where we already used $p_{\beta}(t_{0}) = r_{\beta}(t_{0})$. In conclusion, if the system undergoes an adiabatic dynamics along a non-unitary process, the evolved state can be written as
\begin{equation}
	\dket{\rho^{1\text{D}}_{\text{ad}}(t)} = \sum _{\alpha=1}^{N} r_{\alpha}(t) \dket{\Dcalb_{\alpha}(t)} = \sum _{\alpha=1}^{N} r_{\alpha}(t_{0}) e^{\int_{0}^{t} \Lambda_{\alpha}(t^{\prime})dt^{\prime}}\dket{\Dcalb_{\alpha}(t)} \label{EqAdEvol1D} \text{ , }
\end{equation}
with $\Lambda_{\alpha}(t)=\lambda_{\alpha}(t)-\dinterpro{\Ecalb_{\alpha}(t)}{\dot{\Dcalb}_{\alpha}(t)}$ being the generalized adiabatic phase which accompanying the dynamics of the $n$-th eigenvector. Throughout this thesis, the superscript ``1D'' in some physical parameter, quantity or state means that such element is valid if the adiabatic Lindbladian admits one-dimensional Jordan block decomposition.

The adiabatic evolved state in Eq.~\eqref{EqAdEvol1D} is not achieved in any situation, but we can find conditions in which such dynamics is expected. To derive adiabaticity conditions for open systems, let us rewrite the Eq.~\eqref{EqpdotOne} as
\begin{equation}
	e^{\int_{t_{0}}^{t} \dinterpro{\Ecalb_{\beta}(\xi)}{\dot{\Dcalb}_{\beta}(\xi)} d\xi }\frac{d}{dt} \left[p_{\beta}(t) e^{- \int_{t_{0}}^{t} \dinterpro{\Ecalb_{\beta}(\xi)}{\dot{\Dcalb}_{\beta}(\xi)} d\xi }\right] = - \sum_{\alpha\neq \beta}^{N} p_{\alpha}(t) e^{\int_{t_{0}}^{t} \lambda_{\alpha}(\xi) - \lambda_{\beta}(\xi) d\xi}\dinterpro{\Ecalb_{\beta}(t)}{\dot{\Dcalb}_{\alpha}(t)} \text{ , }
\end{equation}
so that one finds
\begin{equation}
	p_{\beta}(t) e^{- \int_{t_{0}}^{t} \dinterpro{\Ecalb_{\beta}(\xi)}{\dot{\Dcalb}_{\beta}(\xi)} d\xi } - p_{\beta}(t_{0}) = - \sum_{\alpha\neq \beta}^{N} \int_{t_{0}}^{t} F_{\alpha\beta}(\xi)  e^{\int_{t_{0}}^{\xi} \lambda_{\alpha}(\xi^{\prime}) - \lambda_{\beta}(\xi^{\prime}) d\xi^{\prime}} d\xi\text{ , }
\end{equation}
where we defined
\begin{equation}
	F_{\alpha\beta}(t) = e^{-\int_{t_{0}}^{t} \dinterpro{\Ecalb_{\beta}(\xi)}{\dot{\Dcalb}_{\beta}(\xi)} d\xi } p_{\alpha}(t) \dinterpro{\Ecalb_{\beta}(t)}{\dot{\Dcalb}_{\alpha}(t)} \text{ . } \label{EqFalphaBeta}
\end{equation}

Therefore, in order to achieve the adiabatic approximation we need to satisfy
\begin{equation}
	G_{\alpha\beta}(t) = \left \vert \int_{t_{0}}^{t} F_{\alpha\beta}(\xi)  e^{\int_{t_{0}}^{\xi} \lambda_{\alpha}(\xi^{\prime}) - \lambda_{\beta}(\xi^{\prime}) d\xi^{\prime}} d\xi \right \vert \ll 1 \text{ , }
\end{equation}
for all $\alpha$ and $\beta$. Now, by using the normalized time $s = t/\tau$ in above equation, we get
\begin{equation}
	G_{\alpha\beta}(s) = \left \vert \int_{s_{0}}^{s} \tilde{F}_{\alpha\beta}(s^{\prime})  e^{\tau\int_{s_{0}}^{s^{\prime}} \lambda_{\alpha}(s^{\prime\prime}) - \lambda_{\beta}(s^{\prime\prime}) ds^{\prime\prime}} ds^{\prime} \right \vert \text{ , }
\end{equation}
with (here we adopt $d_{s}f(s) = df(s)/ds$)
\begin{equation}
	\tilde{F}_{\alpha\beta}(s) = e^{-\int_{s_{0}}^{s} \dinterpro{\Ecalb_{\beta}(s^{\prime})}{d_{s}\Dcalb_{\beta}(s^{\prime})} ds^{\prime} } p_{\alpha}(s) \dinterpro{\Ecalb_{\beta}(s)}{d_{s}\Dcalb_{\alpha}(s)} \text{ . }
\end{equation}

Now, we can use that (by defining $\Gcalb_{\alpha\beta}(s) = \lambda_{\alpha}(s) - \lambda_{\beta}(s)$)
\begin{equation}
	\frac{d}{ds} \left[ \frac{\tilde{F}_{\alpha\beta}(s)  e^{\tau\int_{s_{0}}^{s} \Gcalb_{\alpha\beta}(s^{\prime})ds^{\prime}}}{\Gcalb_{\alpha\beta}(s)} \right] = \frac{d}{ds} \left[ \frac{\tilde{F}_{\alpha\beta}(s)}{\Gcalb_{\alpha\beta}(s)} \right]e^{\tau\int_{s_{0}}^{s} \Gcalb_{\alpha\beta}(s^{\prime}) ds^{\prime}} + \tau \tilde{F}_{\alpha\beta}(s)  e^{\tau\int_{s_{0}}^{s} \Gcalb_{\alpha\beta}(s^{\prime}) ds^{\prime}} \text{ , }
\end{equation}
to rewrite $G_{\alpha\beta}(s)$ as
\begin{align}
	G_{\alpha\beta}(s) &= \left \vert \frac{1}{\tau} \int_{s_{0}}^{s} \frac{d}{ds^{\prime}} \left[ \frac{\tilde{F}_{\alpha\beta}(s^{\prime})  e^{\tau\int_{s_{0}}^{s^{\prime}} \Gcalb_{\alpha\beta}(s^{\prime\prime})ds^{\prime\prime}}}{\Gcalb_{\alpha\beta}(s^{\prime})} \right] -
	\frac{d}{ds^{\prime}} \left[ \frac{\tilde{F}_{\alpha\beta}(s^{\prime})}{\Gcalb_{\alpha\beta}(s^{\prime})} \right]e^{\tau\int_{s_{0}}^{s^{\prime}} \Gcalb_{\alpha\beta}(s^{\prime\prime}) ds^{\prime\prime}}
	ds^{\prime} \right \vert \text{ , } \nonumber \\
	&= \left \vert \frac{\tilde{F}_{\alpha\beta}(s)  e^{\tau\int_{s_{0}}^{s} \Gcalb_{\alpha\beta}(s^{\prime})ds^{\prime}}}{\tau\Gcalb_{\alpha\beta}(s)} - \frac{\tilde{F}_{\alpha\beta}(s_0) }{\tau\Gcalb_{\alpha\beta}(s_0)} -
	\frac{1}{\tau}\int_{s_{0}}^{s} \frac{d}{ds^{\prime}} \left[ \frac{\tilde{F}_{\alpha\beta}(s^{\prime})}{\Gcalb_{\alpha\beta}(s^{\prime})} \right]e^{\tau\int_{s_{0}}^{s^{\prime}} \Gcalb_{\alpha\beta}(s^{\prime\prime}) ds^{\prime\prime}}
	ds^{\prime} \right \vert \text{ . }
\end{align}

Therefore, from above equation we have two different conditions to be satisfied. In fact, if we want to get $G_{\alpha\beta}(s)$ as smaller as possible, an \textit{sufficient criterion} to adiabaticity reads
\begin{subequations}\label{EqAdCondOpenSystem}
	\begin{align}
		\text{(C1)} \quad \quad & \left \vert \frac{\tilde{F}_{\alpha\beta}(s)  e^{\tau\int_{s_{0}}^{s} \Gcalb_{\alpha\beta}(s^{\prime})ds^{\prime}}}{\tau\Gcalb_{\alpha\beta}(s)} \right \vert \ll 1 \text{ , } \label{EqACOpenC1} \\
		\text{(C2)} \quad \quad &\left \vert \frac{1}{\tau}\int_{s_{0}}^{s} \frac{d}{ds^{\prime}} \left[ \frac{\tilde{F}_{\alpha\beta}(s^{\prime})}{\Gcalb_{\alpha\beta}(s^{\prime})} \right]e^{\tau\int_{s_{0}}^{s^{\prime}} \Gcalb_{\alpha\beta}(s^{\prime\prime}) ds^{\prime\prime}}
		ds^{\prime} \right \vert \ll 1 \text{ , } \label{EqACOpenC2}
	\end{align}
\end{subequations}
that allows us achieving the adiabatic behavior when satisfied for every $s_0\leq s \leq 1$. If the above equation is satisfied for all $\alpha$, the $\beta$-th eigenvector evolves uncoupled from the other eigenvectors, but in case the above equation is satisfied for all $\alpha$ and $\beta$, all eigenvectors of the spectrum of $\Lmath(t)$ evolve uncoupled from each other. Moreover, it is worth highlighting here that the above conditions are very similar to the Tong's conditions $C_{\text{Tong}}^{(\text{a})}$ and $C_{\text{Tong}}^{(\text{b})}$ given in Eq.~\eqref{EqAdTongCondAll}. In fact, by setting $s_{0} = 0$, the second condition can be easily rewritten as
\begin{equation}
	\text{(C2-Tong)} \quad \quad \left \vert \frac{1}{\tau} \frac{d}{ds^{\prime}} \left[ \frac{\tilde{F}_{\alpha\beta}(s^{\prime})}{\Gcalb_{\alpha\beta}(s^{\prime})} \right]e^{\tau\int_{s_{0}}^{s^{\prime}} \Gcalb_{\alpha\beta}(s^{\prime\prime}) ds^{\prime\prime}}
	\right \vert_M  \ll 1 \text{ , } \label{EqACOpenC2Max}
\end{equation}
where we have used that $\left \vert \int_{x_{0}}^{x_{1}}f(x)dx \right \vert \leq \int_{x_{0}}^{x_{1}} \left \vert f(x) \right \vert dx$ and the subscript $M$ denotes the maximal value of the modulus, so that the validity of the condition (C2-Tong) implies into validity of (C2). Notice that the above equation can be viewed as a generalization of the condition $C_{\text{Tong}}^{(\text{b})}$ for open systems. In summary, we get the adiabaticity coefficients in open system given by
\begin{subequations}
	\label{EqAdCoeffOS}
	\begin{align}
		\Xi_{\alpha \beta}^{(1)}(s) &= \left \vert \frac{\tilde{F}_{\alpha\beta}(s)  e^{\tau\int_{s_{0}}^{s} \Gcalb_{\alpha\beta}(s^{\prime})ds^{\prime}}}{\tau\Gcalb_{\alpha\beta}(s)} \right \vert \text{ , } \label{EqAdCoeffOS1} \\
		\Xi_{\alpha \beta}^{(2)}(s) &= \left \vert \frac{1}{\tau}\frac{d}{ds} \left[ \frac{\tilde{F}_{\alpha\beta}(s)}{\Gcalb_{\alpha\beta}(s)} \right]e^{\tau\int_{s_{0}}^{s} \Gcalb_{\alpha\beta}(s^{\prime}) ds^{\prime}}
		\right \vert \text{ , } \label{EqAdCoeffOS2}
	\end{align}
\end{subequations}
so that the conditions of the Eq.~\eqref{EqAdCondOpenSystem} can be rewritten as
\begin{equation}
	\max \left\{ \max_{s\in[0,1]} \Xi_{\alpha \beta}^{(1)}(s) , \max_{s\in[0,1]} \Xi_{\alpha \beta}^{(2)}(s) \right\}\ll 1 \text{ . } \label{EqAdCondOpenSystem-2}
\end{equation}

As it was observed by Sarandy and Lidar~\cite{Sarandy:05-1,Sarandy:05-2}, the nature of the function $\Gcalb_{\alpha\beta}(t)$ is an important element that needs to be addressed in details. In fact, since $\Gcalb_{\alpha\beta}(t): t \in \Rmath \leadsto \Gcalb_{\alpha\beta} \in \Cmath$, the argument in the exponential of the Eqs.~\eqref{EqAdCondOpenSystem} could admit both real and imaginary parts. On the one hand, the imaginary part of $\Gcalb_{\alpha\beta}(t)$ can be neglected due to module in Eqs.~\eqref{EqACOpenC2Max} and~\eqref{EqACOpenC1} (or due the Riemann-Lebesgue lemma~\cite{Sarandy:04}). On the other hand, the real part of $\Gcalb_{\alpha\beta}(t)$ could promote some divergence of the quantities in Eq.~\eqref{EqAdCondOpenSystem} for long evolution times (due the integral in argument of the exponential). Therefore, this means that adiabaticity is not achieved in the regime where $\tau \rightarrow \infty$, but it would be possible to find a range on the total evolution time $\tau$ so that the adiabatic approximation could be obtained~\cite{Sarandy:05-2}. However, as we shall see in Sec.~\ref{SecApplicAdiabDA}, it is possible to find situations in which the adiabatic approximation is achieved in the regime $\tau \rightarrow \infty$.

\subsection{General Jordan blocks}

In a general case, where we have multi-dimensional Jordan blocks in Eq.~\eqref{EqEqLindJ}, one needs to start from the coupled set of equations given in Eq.~\eqref{Eqrdot}. Before, without loss of generality we write
\begin{equation}
	r_{\beta}^{k}(t) = p_{\beta}^{k}(t) e^{\int_{t_{0}}^{t} \lambda_{\beta}(\xi)d\xi} \text{ , }
\end{equation}
in which the argument $\int_{t_{0}}^{t} \lambda_{\beta}(\xi)d\xi$ could be associated with a dynamical phase (in analogy with closed system case), and the Eq.~\eqref{Eqrdot} becomes
\begin{align}
	\dot{p}_{\beta}^{k}(t) & = - p_{\beta}^{k}(t) \dinterpro{\Ecalb_{\beta}^{k}(t)}{\dot{\Dcalb}_{\alpha}^{k}(t)} + p_{\beta}^{k+1}(t) 
	- \sum _{n_{\beta} \neq k}^{N_{\beta}} p_{\beta}^{n_{\beta}}(t) \dinterpro{\Ecalb_{\beta}^{k}(t)}{\dot{\Dcalb}_{\beta}^{k}(t)} \nonumber \\ 
	&- \sum_{\alpha\neq \beta}^{N} \sum _{n_{\alpha} = 1}^{N_{\alpha}} p_{\alpha}^{n_{\alpha}}(t) e^{\int_{t_{0}}^{t} \lambda_{\alpha}(\xi) - \lambda_{\beta}(\xi) d\xi}  \dinterpro{\Ecalb_{\beta}^{k}(t)}{\dot{\Dcalb}_{\alpha}^{n_{\alpha}}(t)} \text{ . } \label{Eqpdot}
\end{align}

As already mentioned, the last term in above equation is the `diabatic' contribution to the dynamics, so the Jordan blocks become uncoupled when such term is small enough. Therefore, by following the same procedure as before we can show that the adiabatic dynamics is approximately achieved when
\begin{subequations}\label{EqAdCondOpenSystemGenBlock}
	\begin{align}
		\text{(C1)} \quad \quad & \left \vert \frac{\tilde{F}^{\text{gen}}_{\alpha\beta}(s)  e^{\tau\int_{s_{0}}^{s} \Gcalb_{\alpha\beta}(s^{\prime})ds^{\prime}}}{\tau\Gcalb_{\alpha\beta}(s)} \right \vert \ll 1 \text{ , } \label{EqACOpenGenBlockC1} \\
		\text{(C2)} \quad \quad &\left \vert \frac{1}{\tau}\int_{s_{0}}^{s} \frac{d}{ds^{\prime}} \left[ \frac{\tilde{F}^{\text{gen}}_{\alpha\beta}(s^{\prime})}{\Gcalb_{\alpha\beta}(s^{\prime})} \right]e^{\tau\int_{s_{0}}^{s^{\prime}} \Gcalb_{\alpha\beta}(s^{\prime\prime}) ds^{\prime\prime}}
		ds^{\prime} \right \vert \ll 1 \text{ , } \label{EqACOpenGenBlockC2}
	\end{align}
\end{subequations}
where $\tilde{F}^{\text{gen}}_{\alpha\beta}(s^{\prime})$ defines the generalized version of the Eq.~\eqref{EqFalphaBeta} as
\begin{equation}
	\tilde{F}^{\text{gen}}_{\alpha\beta}(s) = \sum _{n_{\alpha} = 1}^{N_{\alpha}} p_{\alpha}^{n_{\alpha}}(s) e^{-\int_{s_{0}}^{s} \dinterpro{\Ecalb_{\beta}^{k}(s^{\prime})}{d_{s^{\prime}}\Dcalb_{\beta}^{k}(s^{\prime})} ds^{\prime} } \dinterpro{\Ecalb_{\beta}^{k}(s)}{d_{s}\Dcalb_{\alpha}^{n_{\alpha}}(s)} \text{ , }
\end{equation}
so that $\tilde{F}^{\text{gen}}_{\alpha\beta}(s) = \tilde{F}_{\alpha\beta}(s)$ for one-dimensional Jordan blocks. Therefore, if the conditions in Eq.~\eqref{EqAdCondOpenSystemGenBlock} are satisfied, the Eq.~\eqref{Eqpdot} can be (approximately) written as
\begin{equation}
	\dot{p}_{\beta}^{k}(t)  = - p_{\beta}^{k}(t) \dinterpro{\Ecalb_{\beta}^{k}(t)}{\dot{\Dcalb}_{\alpha}^{n_{\alpha}}(t)} + p_{\beta}^{k+1}(t) 
	- \sum _{n_{\beta} \neq k}^{N_{\beta}} p_{\beta}^{n_{\beta}}(t) \dinterpro{\Ecalb_{\beta}^{k}(t)}{\dot{\Dcalb}_{\beta}^{n_{\alpha}}(t)} \text{ , }
\end{equation}
and the solution to $p_{\beta}^{k}(t)$ from above equation is the \textit{open system adiabatic solution}. The above equation can be rewritten in a compact way as
\begin{equation}
	\dot{\vec{p}}_{\beta}(t) = \left[ \tilde{\1}_{\text{u-shift}} - \Gcal_{\beta}(t) \right] \vec{p}_{\beta}(t) \text{ , } \label{Eqpvec}
\end{equation}
where we already used $p_{\beta}^{N_{\alpha}+1}(t) = 0$. Here $\Gcal_{\beta}(t)$ is a $(N_{\beta}\times N_{\beta})$-dimensional matrix whose elements are $\Gcal_{kn_{\alpha}}(t) = \dinterpro{\Ecalb_{\beta}^{k}(t)}{\dot{\Dcalb}_{\beta}^{n_{\alpha}}(t)}$, $\vec{p}_{\beta}(t)$ is a vector with $N_{\beta}$ components $p_{\beta}^{k}(t)$ and $\tilde{\1}_{\text{u-shift}}$ an upper shift matrix which reads
\begin{equation}
	\tilde{\1}_{\text{u-shift}} = \begin{bmatrix}
		0 & 1 & 0 & \cdots & 0 \\
		0 & 0 & 1 & \cdots & 0 \\
		\vdots & \ddots & \ddots & \ddots & \vdots \\
		0 & \cdots & \cdots & 0 & 1 \\
		0 & \cdots & \cdots & 0 & 0
	\end{bmatrix} \text{ . }
\end{equation}

Thus, it is possible to see that the uncoupled evolution within a single Jordan block is not necessarily obtained even when $\Gcal(t) \approx 0$. In this case, the most convenient way to write the adiabatic evolution operator $\Ucalb (t,t_{0})$ is from the definition of individual inner evolution operators $\Ucalb_{\alpha} (t,t_{0})$ for each bloch given by
\begin{equation}
	\Ucalb_{\alpha} (t,t_{0}) = e^{\int_{t_{0}}^{t} \lambda_{\alpha}(\xi)d\xi} \sum _{n_{\alpha} = 1}^{N_{\alpha}} \sum _{m_{\alpha} = 1}^{N_{\alpha}} u_{n_{\alpha}m_{\alpha}}(t)\dket{\Dcalb_{\alpha}^{n_{\alpha}}(t)}\dbra{\Ecalb_{\alpha}^{m_{\alpha}}(t_{0})} \label{EqUalpha}\text{ , }
\end{equation}
where the elements $u_{n_{\alpha}m_{\alpha}}(t)$ make explicit the case in which there is inner transition in a single Jordan block and $e^{\int_{t_{0}}^{t} \lambda_{\alpha}(\xi)d\xi}$ is a phase independent on the inner structure of the block. The functions $u_{n_{\alpha}m_{\alpha}}(t)$ are determined from solution of the Eq.~\eqref{Eqpvec}, therefore it strongly depends on the elements of matrix of $\Gcal(t)$. From this definition, the complete evolution operator $\Ucalb (t,t_{0})$ is conveniently written as
\begin{equation}
	\Ucalb (t,t_{0}) = \sum_{\alpha = 1}^{N} \Ucalb_{\alpha} (t,t_{0}) \text{ . } \label{EqUAdOS}
\end{equation}

Notice that $\Ucalb (t,t_{0})$ does not admit transitions between two eigenvectors from different blocks, therefore it is our adiabatic evolution operator. An immediate result from above operator is that, depending on the matrix $\Gcal(t)$, we may be unable to find the analytical solution of the adiabatic dynamics. It happens whenever the Jordan block of interest does not admit solution to Eq.~\eqref{Eqpvec}, so if our dynamics does not start from a vector in a Jordan block with unknown solution for the Eq.~\eqref{Eqpvec}, we can analytically get the adiabatic solution. 

Before following our discussion, it is important to mention that $\Ucalb (t,t_{0})$ admits an operator $\Ucalb^{-1} (t,t_{0})$ so that $\Ucalb (t,t_{0})\Ucalb^{-1} (t,t_{0}) = \1$, where $\Ucalb^{-1} (t,t_{0})$ is given by
\begin{equation}
	\Ucalb^{-1} (t,t_{0}) = \sum_{\alpha = 1}^{N} \Ucalb_{\alpha}^{-1} (t,t_{0}) \text{ , }
\end{equation}
with each operator $\Ucalb_{\alpha}^{-1} (t,t_{0})$ obtained from
\begin{equation}
	\Ucalb_{\alpha}^{-1} (t,t_{0}) = e^{\int_{t_{0}}^{t} \lambda_{\alpha}(\xi)d\xi} \sum _{n_{\alpha} = 1}^{N_{\alpha}} \sum _{m_{\alpha} = 1}^{N_{\alpha}} \tilde{u}_{n_{\alpha}m_{\alpha}}(t)\dket{\Dcalb_{\alpha}^{n_{\alpha}}(t_{0})}\dbra{\Ecalb_{\alpha}^{m_{\alpha}}(t)} \text{ , }
\end{equation}
where the coefficients should satisfies
\begin{equation}
	\sum _{j_{\nu} = 1}^{N_{\nu}} 
	u_{n_{\nu}j_{\nu}}\tilde{u}_{j_{\nu}m_{\nu}} = \delta_{n_{\nu}m_{\nu}} \text{ . } \label{EqUcoeff}
\end{equation}

In addition, the operator $\Ucalb (t,t_{0})$ can be designed in order to be identified as the operator which ``block-diagonalizes'' $\Lmath(t)$, that is, if we impose four additional conditions given by
\begin{subequations}
	\begin{align}
		\sum _{n_{\nu} = 1}^{N_{\nu}} \tilde{u}_{l_{\nu}(n_{\nu}-1)} u_{n_{\nu}l_{\nu}} = 0 &\quad \text{ , } \quad \sum _{n_{\nu} = 1}^{N_{\nu}} \tilde{u}_{l_{\nu}n_{\nu}} u_{n_{\nu}l_{\nu}} = 1 \label{Eqmu01} \text{ , } \\
		\sum _{n_{\nu} = 1}^{N_{\nu}} \tilde{u}_{l_{\nu}n_{\nu}} u_{n_{\nu}(l_{\nu}+1)} = 0 \quad &\text{ and } \quad \sum _{n_{\nu} = 1}^{N_{\nu}} \tilde{u}_{l_{\nu}(n_{\nu}-1)} u_{n_{\nu}(l_{\nu}+1)} = 1 \text{ , } \label{EqmuND01}
	\end{align}
	\label{Eqmu}
\end{subequations}
so one can show that
\begin{equation}
	\Lmath_{\text{J}}(t) = \Ucalb^{-1} (t,t_{0}) \Lmath(t) \Ucalb (t,t_{0}) \text{ . } \label{EqLDiagU}
\end{equation}

A detailed proof of Eqs.~\eqref{EqUcoeff} and~\eqref{Eqmu} is provided in Appendix~\ref{ApUAd}. In particular, from Eqs.~\eqref{EqUalpha} and \eqref{EqAdEvol1D} we can redo the above analysis in order to show that the adiabatic evolution operator for one-dimensional Jordan blocks reads
\begin{equation}
	\Ucalb^{1\text{D}} (t,t_{0}) = \sum _{\alpha=1}^{N} e^{\int_{t_{0}}^{t} \Lambda_{\alpha}(\xi)d\xi} \dket{\Dcalb_{\alpha}(t)}\dbra{\Ecalb_{\alpha}(t_{0})} \label{EqUalphaOneJB}\text{ . }
\end{equation}

In fact, by starting the system in an arbitrary initial state $\dket{\rho(t_{0})} = r_{\alpha}(t_{0})\dket{\Dcalb_{\alpha}(t_{0})}$, one shows that the solution $\dket{\rho^{1\text{D}}_{\text{ad}}(t)}$ given in Eq.~\eqref{EqAdEvol1D} can be obtained from above equation as $\dket{\rho^{1\text{D}}_{\text{ad}}(t)} = \Ucalb^{1\text{D}} (t,t_{0})\dket{\rho(t_{0})}$. Moreover, as we shall see, the knowledge of the evolution operator for adiabatic dynamics is the main element required to study transitionless quantum driving in open systems. For this reason, the evolution operators derived in this section are important results of this thesis.

\section{Application to adiabatic Deutsch algorithm} \label{SecApplicAdiabDA}

For the applications of adiabatic approximation in open system considered in this thesis, the system dynamics is given by
\begin{equation}
	\dot{\rho}(t) = \Hcalb[\rho(t)] + \Rcalb[\rho(t)] \text{ , }
\end{equation}
with $\Hcalb[\bullet]$ is the unitary part of the dynamics given by
\begin{equation}
	\Hcalb[\bullet] = \frac{1}{i\hbar} [H(t),\bullet] \text{ , }
\end{equation}
and $\Rcalb[\bullet]$ is the non-unitary contribution to the dynamics. In particular, here we will be interested in a number of evolution given by Lindblad master equation~\cite{Lindblad:76}
\begin{equation}
	\Rcalb[\rho(t)] = \frac{1}{2}\sum_{n} \left[ 2\Gamma_{n}(t)\rho(t)\Gamma^{\dagger}_{n}(t) - \left\{\Gamma^{\dagger}_{n}(t)\Gamma_{n}(t),\rho(t)\right\}\right]
	\text{ , } \label{EqEqLindPartic}
\end{equation}
where $\Gamma_{n}(t)$ are the Lindblad operators.

\subsection{Adiabatic Deutsch algorithm under dephasing}

%\begin{figure}[t!]
%	\centering
%	\subtop[Classical circuit]{\includegraphics[scale=0.17]{Figs/DeutschAlgoCirc.png}\label{Fig-DeutschAlgoClassCirc}}
%	\subtop[Quantum circuit]{\includegraphics[scale=0.17]{Figs/spetrum.pdf}\label{Fig-DeutschAlgoQuanCirc}}
%	\caption{{\color{blue}(a)} Spectrum of the noise source, where the dots are the measured data and the solid line is a Gaussian fit. {\color{blue}(b)} Linear fit of the experimental results for relationship between $A_{0}$ and $\gamma_{\text{d}}$. Without driving noise (noise amplitude is zero), the dephasing rate of the qubit is  fitted as 5.25 Hz, which is caused by magnetic fluctuation in our lab. {\color{blue}(c)} Rabi oscillations between the $\ket{0}$ and $\ket{1}$ states under different noise intensities.}
%	\label{Fig-DeutschAlgo}
%\end{figure}

%\begin{figure}[t]
%	\centering
%	\input{Figs/DeutschAlgoFig.plt}
%	\setlength{\unitlength}{1.0cm}
%	\vspace{0.7cm}
%	\caption{bla bla bla}
%\end{figure}

At beginning of the development of quantum computation the main question was finding some concrete example (some algorithm) in which quantum mechanically driven system could provide computation tasks with high performance concerning its classical counterpart. In this sense, in 1985, David Deutsch derived a quantum algorithm in which we can get a enhanced performance from its classical one~\cite{Deutsch:85}, called \textit{Deutsch algorithm}. Although the Deutsch's algorithm is not a important (powerful) algorithm in quantum computation, since it does not provide a exponential gain concerning the its classical version, it is the first main example of how powerful a quantum algorithm can be. The problem addressed in Deutsch's algorithm is how to determinate whether a dichotomic real function $f: x\in\{0,1\} \rightarrow f(x)\in\{0,1\}$ is \textit{constant} (the output result $f(x)$ is the same regardless the input value $x$) or \textit{balanced} (the output result $f(x)$ assumes different values in according with the input value $x$). The adiabatic version of such algorithm is obtained through the operator $\Ocalb_{f}$ which computes $f$ (oracle), given by
\begin{equation}
	\Ocalb_{f} = (-1)^{f(0)}\ket{0}\bra{0} + (-1)^{f(1)}\ket{1}\bra{1} \text{ , }
\end{equation}
thus one can write the adiabatic Hamiltonian which implements the algorithm as
\begin{equation}
	H^{\text{DA}}(t) = U_{f}(t) H_{0} U_{f}^{\dagger}(t) \text{ , }
\end{equation}
where $H_{0} = -\hbar \omega \sigma_{x}/2$ and $U_{f}(t) = \exp(i\frac{\pi}{2}\frac{t}{\tau}\Ocalb_{f})$. At $t=0$ we have $H^{\text{DA}}(0) =H_{0}$, so that the initial input state is written as $\ket{\psi_{\text{inp}}} = \ket{+} = (1/\sqrt{2})(\ket{0}+\ket{1})$. One can show that, when the evolution is slow enough, the adiabatic dynamics of the system driven by the Hamiltonian $H^{\text{DA}}(t)$ will be given by
\begin{equation}
	\rho_{\text{cs}}^{\text{DA}}(t) = \frac{1}{2} \left[ \1 + \cos \left( \frac{\pi F}{2}\frac{t}{\tau} \right) \sigma_{x} - \sin \left( \frac{\pi F}{2}\frac{t}{\tau} \right) \sigma_{y} \right] \text{ , } \label{EqDAClosedSol}
\end{equation}
where $F = 1-(-1)^{f(0)+f(1)}$ and the subscript ``cs'' denotes that $\rho_{\text{cs}}^{\text{DA}}(t)$ is obtained from adiabatic solution for closed systems. By considering the system has its dynamics governed by an environment so that the main effects on the system is phase damping with rate $\gamma(t)$, whose contribution reads
\begin{equation}
	\Rcalb[\bullet] = \gamma(t) \left[ \sigma_{z} \bullet \sigma_{z} - \bullet \right] \label{EqRcalbDephasing}
\end{equation}
we describe the system evolution as
\begin{equation}
	\dot{\rho}^{\text{DA}}(t) = - \frac{i}{\hbar} [H^{\text{DA}}(t),\rho^{\text{DA}}(t)] + \gamma(t) \left[ \sigma_{z} \rho^{\text{DA}}(t) \sigma_{z} - \rho^{\text{DA}}(t) \right] \text{ . } \label{EqLindDephDA}
\end{equation}

In order to study the adiabatic dynamics of the system in according with above equation, let us to rewrite it in superoperator formalism as
\begin{equation}
	\dket{\dot{\rho}^{\text{DA}}(t)} = \Lmath^{\text{DA}}(t) \dket{\rho^{\text{DA}}(t)} \text{ , }
\end{equation}
where
\begin{equation}
	\Lmath^{\text{DA}}(t) = \begin{bmatrix}
		0 & 0 & 0 & 0 \\ 0 & -2 \gamma & 0 & \omega \sin \left( \frac{\pi F}{2}\frac{t}{\tau} \right) \\ 0 & 0 & -2 \gamma & \omega \cos \left( \frac{\pi F}{2}\frac{t}{\tau} \right) \\ 0 & - \omega \sin \left( \frac{\pi F}{2}\frac{t}{\tau} \right) & - \omega \cos \left( \frac{\pi F}{2}\frac{t}{\tau} \right) & 0
	\end{bmatrix} \text{ , }
\end{equation}
whose the eigenvectors are (the superscript ``t'' denotes transpose)
\begin{subequations}
	\label{EqDARightEigenVec}
	\begin{align}
		\dket{\Dcalb^{\text{DA}}_{0}(t)} &= 
		\left[ \frac{}{} 1 \quad 0 \quad 0 \quad 0\text{ } \right]^{\text{t}} \text{ , } \\
		\dket{\Dcalb^{\text{DA}}_{1}(t)} &= 
		\left[ \text{ } 0 \quad -\cos \left( \frac{\pi F}{2}\frac{t}{\tau} \right) \quad \sin \left( \frac{\pi F}{2}\frac{t}{\tau} \right) \quad 0\text{ } \right]^{\text{t}} \text{ , } \\
		\dket{\Dcalb^{\text{DA}}_{2}(t)} &= 
		\left[ \text{ } 0 \quad \Delta_{+}(t) \sin \left( \frac{\pi F}{2}\frac{t}{\tau}\right) \quad \Delta_{+}(t) \cos \left( \frac{\pi F}{2}\frac{t}{\tau}\right) \quad 1\text{ } \right]^{\text{t}} \text{ , } \\
		\dket{\Dcalb^{\text{DA}}_{3}(t)} &= 
		\left[ \text{ } 0 \quad \Delta^{-1}_{+}(t) \sin \left( \frac{\pi F}{2}\frac{t}{\tau}\right) \quad \Delta^{-1}_{+}(t) \cos \left( \frac{\pi F}{2}\frac{t}{\tau}\right) \quad 1\text{ } \right]^{\text{t}} \text{ , }
	\end{align}
\end{subequations}
and left eigenvectors
\begin{subequations}
	\label{EqDALeftEigenVec}
	\begin{align}
		\dbra{\Ecalb^{\text{DA}}_{0}(t)} &= \left[ \frac{}{} 1 \quad 0 \quad 0 \quad 0\text{ } \right] \text{ , } \\
		\dbra{\Ecalb^{\text{DA}}_{1}(t)} &= \left[ \text{ } 0 \quad -\cos \left( \frac{\pi F}{2}\frac{t}{\tau} \right) \quad \sin \left( \frac{\pi F}{2}\frac{t}{\tau} \right) \quad 0\text{ } \right] \text{ , } \\
		\dbra{\Ecalb^{\text{DA}}_{2}(t)} &= \frac{1}{2} \left[ \text{ } 0 \quad \omega \sin \left( \frac{\pi F}{2}\frac{t}{\tau}\right) \quad \omega \cos \left( \frac{\pi F}{2}\frac{t}{\tau}\right) \quad - \frac{\Delta_{-}(t)}{\sqrt{\gamma^2(t)-\omega^2}}\text{ } \right] \text{ , } \\
		\dbra{\Ecalb^{\text{DA}}_{3}(t)} &= \frac{1}{2}\left[ \text{ } 0 \quad - \omega \sin \left( \frac{\pi F}{2}\frac{t}{\tau}\right) \quad -\omega \cos \left( \frac{\pi F}{2}\frac{t}{\tau}\right) \quad \frac{\Delta_{+}(t)}{\sqrt{\gamma^2(t)-\omega^2}}\text{ } \right] \text{ , }
	\end{align}
\end{subequations}
with eigenvalues $\lambda_{0}(t) = 0$, $\lambda_{1}(t) = -2 \gamma(t) $, $\lambda_{2}(t) = - \Delta_{+}(t) $ and $\lambda_{3}(t) = -\Delta_{-}(t)$, where we defined $\Delta_{\pm}(t) = \gamma(t) \pm \sqrt{\gamma^2(t) - \omega^2}$. The non-degenerate spectrum of $\Lmath^{\text{DA}}(t)$ shows that it can be written in one-dimensional Jordan block form, so that the adiabatic behavior of the system can be obtained from adiabatic solution discussed in Sec.~\ref{SecAdApprox-1DJB}. By writing the density matrix of the initial state $\rho^{\text{DA}}(0) = \ket{\psi_{\text{inp}}}\bra{\psi_{\text{inp}}} = \ket{+}\bra{+} = (1/2)(\1+\sigma_{x})$ in superoperator formalism, one shows that the initial state can be written as a linear combination of the vectors $\dket{\Dcalb^{\text{DA}}_{0}(0)}$ and $\dket{\Dcalb^{\text{DA}}_{1}(0)}$ as
\begin{equation}
	\dket{\rho^{\text{DA}}(0)} = \begin{bmatrix} \text{ } 1 & 1 & 0 & 0\text{ } \end{bmatrix}^{\text{t}} = \dket{\Dcalb^{\text{DA}}_{0}(0)} - \dket{\Dcalb^{\text{DA}}_{1}(0)} \text{ . } \label{EqIniStaDA}
\end{equation}
Hence, if we let the system evolves adiabatically, the evolved state is obtained from Eq.~\eqref{EqAdEvol1D} as
\begin{equation}
	\dket{\rho^{\text{DA}}(t)} = \dket{\Dcalb^{\text{DA}}_{0}(t)} - e^{-2\int_{t_{0}}^{t} \gamma(\xi)d\xi}\dket{\Dcalb^{\text{DA}}_{1}(t)} \label{EqDAAdEvolution} \text{ , }
\end{equation}
where we already used that $\Lambda_{1}(t) = \lambda_{1}(t)=-2 \gamma(t)$ and $\Lambda_{0}(t) = \lambda_{0}(t)=0$, since from Eqs.~\eqref{EqDARightEigenVec} and~\eqref{EqDALeftEigenVec} one can see that $\dinterpro{\Ecalb_{0}(t)}{\dot{\Dcalb}_{0}(t)}=\dinterpro{\Ecalb_{1}(t)}{\dot{\Dcalb}_{1}(t)}=0$. Now, by rewrite the above equation in matrix notation we get
\begin{equation}
	\dket{\rho^{\text{DA}}(t)} = \left[ \text{ } 1 \quad - e^{-2\int_{t_{0}}^{t} \gamma(\xi)d\xi}\cos \left( \frac{\pi F}{2}\frac{t}{\tau} \right) \quad - e^{\int_{t_{0}}^{t} \lambda_{1}(\xi)d\xi}\sin \left( \frac{\pi F}{2}\frac{t}{\tau} \right) \quad 0\text{ } \right]^{\text{t}}  \text{ , }
\end{equation}
where we can easily to determinate the components of the coherence vector associated with density matrix $\rho^{\text{DA}}(t)$ and write
\begin{equation}
	\rho^{\text{DA}}(t) = \frac{1}{2} \left[ \1 + e^{-2\int_{t_{0}}^{t} \gamma(\xi)d\xi}\cos \left( \frac{\pi F}{2}\frac{t}{\tau} \right) \sigma_{x} - e^{-2\int_{t_{0}}^{t} \gamma(\xi)d\xi}\sin \left( \frac{\pi F}{2}\frac{t}{\tau} \right)\sigma_{y} \right]  \text{ . }\label{EqDAAdSolOS}
\end{equation}

In unitary dynamics limit $\gamma(t) \rightarrow 0$, it is straightforward to see that we recover the density matrix of the unitary dynamics given in Eq.~\eqref{EqDAClosedSol}. However, the above solution is just valid when the adiabatic regime is achieved, in other words, when the adiabaticity conditions discussed in Sec.~\ref{SecAdApprox}. In this case more precisely, the relevant conditions are related with one-dimensional Jordan blocks as discussed in Sec.~\ref{SecAdApprox-1DJB}. Thus, by defining $\Gcal_{\beta \alpha}(s) = \dinterpro{\Ecalb_{\alpha}(s)}{d_{s}\Dcalb_{\beta}(s)}$, where we are using $s = t/\tau$, we find the non-null elements given by
\begin{subequations}\label{EqCoefficDA}
	\begin{align}
		\Gcal_{12}(s) &= -\frac{\pi F \Delta_{+}(s)}{2\omega} \text{ \, \, , \, \, }
		\Gcal_{13}(s) = -\frac{\pi F \omega}{2 \Delta_{+}(s)} \text{ , } \\
		\Gcal_{21}(s) &= - \Gcal_{31}(s) = \frac{\pi F \omega}{4 \sqrt{\gamma^2(s) - \omega^2}} \text{ , } \\
		\Gcal_{22}(s) &= - \Gcal_{32}(s) = \frac{ \Delta_{+}(s) d_{s}\gamma(s)}{2 \left(\gamma^2(s) - \omega^2\right)} \text{ , } \\
		\Gcal_{23}(s) &= - \Gcal_{33}(s) = - \frac{ \omega^2 d_{s}\gamma(s)}{2 \left(\gamma^2(s) - \omega^2\right) \Delta_{+}(s) } \text{ , }
	\end{align}
\end{subequations}
with $\Delta_{\pm}(s) = \gamma(s) \pm \sqrt{\gamma^2(s) - \omega^2}$. From above equation, it is possible to see that we can simplify our discussion if we consider $\gamma(t) = \gamma(s) = \gamma_{0}$, so that $\Gcal_{22}(s)= \Gcal_{32}(s)=\Gcal_{23}(s) =\Gcal_{33}(s)=0$, moreover, the rest of the above quantities are time-independent by using this assumption. Thus, from Eqs.~\eqref{EqAdCoeffOS} we have the adiabaticity coefficients given by
\begin{subequations}\label{EqCoef1DA}
	\begin{align}
		\Xi_{12}^{(1)}(s) &= \frac{e^{-s\tau\gamma_{0}} F\pi \left \vert \omega p_{1}(s) \right \vert  }{4\tau \left \vert \omega^{2} - \gamma_{0} \left( \gamma_{0} - i\sqrt{ \omega^2 - \gamma_{0}^2 }\right)\right \vert } \text{ \, \, , \, \, }
		\Xi_{13}^{(1)}(s) = \frac{e^{-s\tau\gamma_{0}} F\pi \left \vert \omega p_{1}(s) \right \vert  }{4\tau \left \vert \omega^{2} - \gamma_{0} \left( \gamma_{0} + i\sqrt{ \omega^2 - \gamma_{0}^2 }\right)\right \vert } \text{ , } \\
		\Xi_{21}^{(1)}(s) &= \frac{e^{s\gamma_{0}\tau} F\pi \left \vert \omega p_{3}(s) \right \vert }{2\tau\left \vert \gamma_{0} + i\sqrt{ \omega^2 - \gamma_{0}^2 }\right \vert} \text{ \, \, , \, \, } 
		\Xi_{31}^{(1)}(s) = \frac{e^{s\gamma_{0}\tau} F \pi |p_{2}(s)|}{2\omega^3\tau} \left \vert \omega^{2} - 2\gamma_{0} \left( \gamma_{0} + i\sqrt{ \omega^2 - \gamma_{0}^2 } \right) \right \vert \text{ , }
	\end{align}
\end{subequations}
where we already used that $\omega \geq \gamma$, so that $\left \vert\text{exp}\sqrt{\gamma_{0}^2 - \omega^2}\right \vert=1$. To end, the second adiabatic coefficients are
\begin{subequations}\label{EqCoef2DA}
	\begin{align}
		\Xi_{12}^{(2)}(s) &= \frac{e^{-s\tau\gamma_{0}} F\pi \left \vert \omega d_{s}p_{1}(s) \right \vert  }{4\tau \left \vert \omega^{2} - \gamma_{0} \left( \gamma_{0} - i\sqrt{ \omega^2 - \gamma_{0}^2 }\right)\right \vert } \text{ \, \, , \, \, }
		\Xi_{13}^{(2)}(s) = \frac{e^{-s\tau\gamma_{0}} F\pi \left \vert \omega d_{s}p_{1}(s) \right \vert  }{4\tau \left \vert \omega^{2} - \gamma_{0} \left( \gamma_{0} + i\sqrt{ \omega^2 - \gamma_{0}^2 }\right)\right \vert } \text{ , } \\
		\Xi_{21}^{(2)}(s) &= \frac{e^{s\gamma_{0}\tau} F\pi \left \vert \omega d_{s}p_{3}(s) \right \vert }{2\tau\left \vert \gamma_{0} + i\sqrt{ \omega^2 - \gamma_{0}^2 }\right \vert} \text{ \, \, , \, \, } 
		\Xi_{31}^{(2)}(s) = \frac{e^{s\gamma_{0}\tau} F \pi |d_{s}p_{2}(s)|}{2\omega^3\tau} \left \vert \omega^{2} - 2\gamma_{0} \left( \gamma_{0} + i\sqrt{ \omega^2 - \gamma_{0}^2 } \right) \right \vert \text{ . }
	\end{align}
\end{subequations}

The other coefficients are zero and this means that some Lindblad-Jordan eigenspaces are uncoupled between each other for any $\tau > 0$. In fact, we found that $\Xi_{\alpha0}^{(1)}(s) = \Xi_{\alpha0}^{(2)}(s) = 0$ so that the eigenvector $\dket{\Dcalb_{0}(s)}$ evolves independent on the $\alpha$-th eigenvector of $\Lmath(s)$. On the other hand, $\Xi_{0\alpha}^{(1)}(s) = \Xi_{0\alpha}^{(2)}(s) = 0$ means that the $\alpha$-th eigenvector of $\Lmath(s)$ evolves independent on the eigenvector $\dket{\Dcalb_{0}(s)}$.

In Eqs.~\eqref{EqCoef1DA}, due to the positive value of the exponential in above equations, the value of $\Xi_{21}^{(1)}(s)$ and $\Xi_{31}^{(1)}(s)$ should increase as $\gamma_{0}\tau \rightarrow \infty$, however in some cases there is a critical values $\Xi_{21}^{(1)}(s)$ and $\Xi_{31}^{(1)}(s)$ for some values of $\gamma_{0}\tau$. Therefore, as discussed in Ref.~\cite{Sarandy:05-1}, for a range of values for $\gamma_{0}\tau$ the adiabatic behavior is approximately obtained, but as $\gamma_{0}\tau$ grows the coefficients $\Xi_{21}^{(1)}(s)$ and $\Xi_{31}^{(1)}(s)$ become big enough so that we achieve a regime in which an \textit{adiabaticity breakdown} happens. On the other hand, let us consider a more robust analysis by studying the behavior of the other coefficients $\Xi_{12}^{(1)}(s)$ and $\Xi_{13}^{(1)}(s)$. The negative exponential in $\Xi_{12}^{(1)}(s)$ and $\Xi_{13}^{(1)}(s)$ shows a different behavior because as $\tau$ increases the coefficient $e^{-s\gamma_{0}\tau}/\gamma_{0}\tau$ decreases. Moreover, in the limit $\gamma_{0} \tau \rightarrow \infty$, one gets $\Xi_{12}^{(1)}(s) \rightarrow 0$ and $\Xi_{13}^{(1)}(s) \rightarrow 0$. The same result is obtained for the coefficients in Eqs.~\eqref{EqCoef2DA}. From this analysis, we have a scenario where in the limit of slow dynamics ($\gamma_{0} \tau \gg 1$) the eigenvector $\dket{\Dcalb_{1}(s)}$ does not evolve uncoupled from $\dket{\Dcalb_{2}(s)}$ and $\dket{\Dcalb_{3}(s)}$ (no adiabaticity), but the eigenvectors $\dket{\Dcalb_{2}(s)}$ and $\dket{\Dcalb_{3}(s)}$ are decoupled from $\dket{\Dcalb_{1}(s)}$. However, interestingly, here we will discuss how it is possible to proof that adiabaticity can be obtained in the regime $\gamma_{0} \tau \gg 1$.

By using the Eqs.~\eqref{EqCoefficDA}, from set of equations given in Eq.~\eqref{EqrdotOne}, one finds $\dot{r}_{0}(t) = 0$ and
\begin{subequations}
	\begin{align}
		d_{s}r_{1}(s) & = \tau\lambda_{1}(s)r_{1}(s)  - r_{2}(s) \Gcal_{12}(s) - r_{3}(s) \Gcal_{13}(s) \text{ , } \label{Eqr1dotDA} \\
		d_{s}r_{2}(s) & = \tau\lambda_{2}(s)r_{2}(s)  - r_{1}(s) \Gcal_{21}(s) \text{ , } \label{Eqr2dotDA} \\
		d_{s}r_{3}(s) & = \tau\lambda_{3}(s)r_{3}(s) - r_{1}(s) \Gcal_{31}(s) \text{ . } \label{Eqr3dotDA} 
	\end{align}
\end{subequations}

Now, we use the discussion above in order to rewrite the Eqs.~\eqref{Eqr2dotDA} and~\eqref{Eqr3dotDA} in adiabatic regime as ($j = 2,3$)
\begin{equation}
	d_{s}r_{j}(s)  \approx \tau\lambda_{j}(s)r_{j}(s) \text{ \, \, \, } \Rightarrow \text{ \, \, \, } r_{j}(s) = r_{j}(0) e^{\tau \int_{0}^{s}\lambda_{j}(s^{\prime})ds^{\prime}}  \text{ . } \label{EqSolr2andr3}
\end{equation}

As it was already discussed here, the coefficients $r_{2}(s)$ and $r_{3}(s)$ do not depend on the other coefficients. However, by rewriting the Eq.~\eqref{Eqr1dotDA} in this regime we get
\begin{equation}
	d_{s}r_{1}(s)  = \tau\lambda_{1}(s)r_{1}(s) - r_{2}(0) e^{\tau \int_{0}^{s}\lambda_{2}(s^{\prime})ds^{\prime}} \Gcal_{12}(s) - r_{3}(0) e^{\tau \int_{0}^{s}\lambda_{3}(s^{\prime})ds^{\prime}} \Gcal_{13}(s) \text{ , }
\end{equation}
thus, note that if we consider the initial state such that $r_{2}(0) = r_{3}(0) = 0$, then we find
\begin{equation}
	d_{s}r_{1}(s)|_{r_{j}(0)=0}  = \tau\lambda_{1}(s)r_{1}(s) \text{ \, \, \, } \Rightarrow \text{ \, \, \, } r_{1}(s)|_{r_{j}(0)=0} = r_{1}(0) e^{\int_{0}^{s}\tau\lambda_{1}(s^{\prime})ds^{\prime}} \text{ , }
\end{equation}
where the adiabatic dynamics is obtained. From this example we conclude that the adiabatic approximation in open system requires a more detailed study, we mean, the validity conditions of the adiabatic dynamics alone are not enough to guarantee the adiabatic behavior in open systems. In conclusion, this discussion allows us to guarantee that the Deutsch algorithm considered here achieves its adiabatic dynamics when very-slowly driven fields are used. In fact, from Eq.~\eqref{EqIniStaDA} it is possible to see that the initial condition $r_{2}(0) = r_{3}(0) = 0$ is satisfied and the dynamics is dictated by the Eq.~\eqref{EqDAAdEvolution}.

\subsection{Experimental verification and asymptotic adiabatic dynamics for open systems}

In order to verify the theoretical results obtained in last section, where the adiabatic dynamics in open system is obtained in asymptotic long evolution total time, in this section we will present the experimental verification in a trapped ion Ytterbium qubit~\cite{Hu:19-a}. However, different from Sec.~\ref{SecExpAdCondCS} where the dynamics is unitary, here we need to consider the non-unitary effects on the dynamics because of the long time coherence of the system~\cite{Leibfried:03}. In fact, as shown in Appendix~\ref{ApTrappeIonDeco}, while the Rabi frequency achieved in our experiment is of the order of KHz, the natural dephasing rate $\gamma_{\text{nd}}$ of the system used in our experiment was computed as $\gamma_{\text{nd}} \approx 3.03$~Hz. For this reason, a simulation of non-unitary effects on the two-level system is required and it is done with the experimental apparatus shown in Fig.~\ref{FigAdiabExpComTrappedIon}. For more details about the feature used in the simulation, see Appendix~\ref{SecExpAdCondCS}.

\begin{figure}[t!]
	%\input{Figs/FigAdiabExpTrappedIon.plt}
	%\vspace{4.1cm}
	\centering
	\includegraphics[scale=0.65]{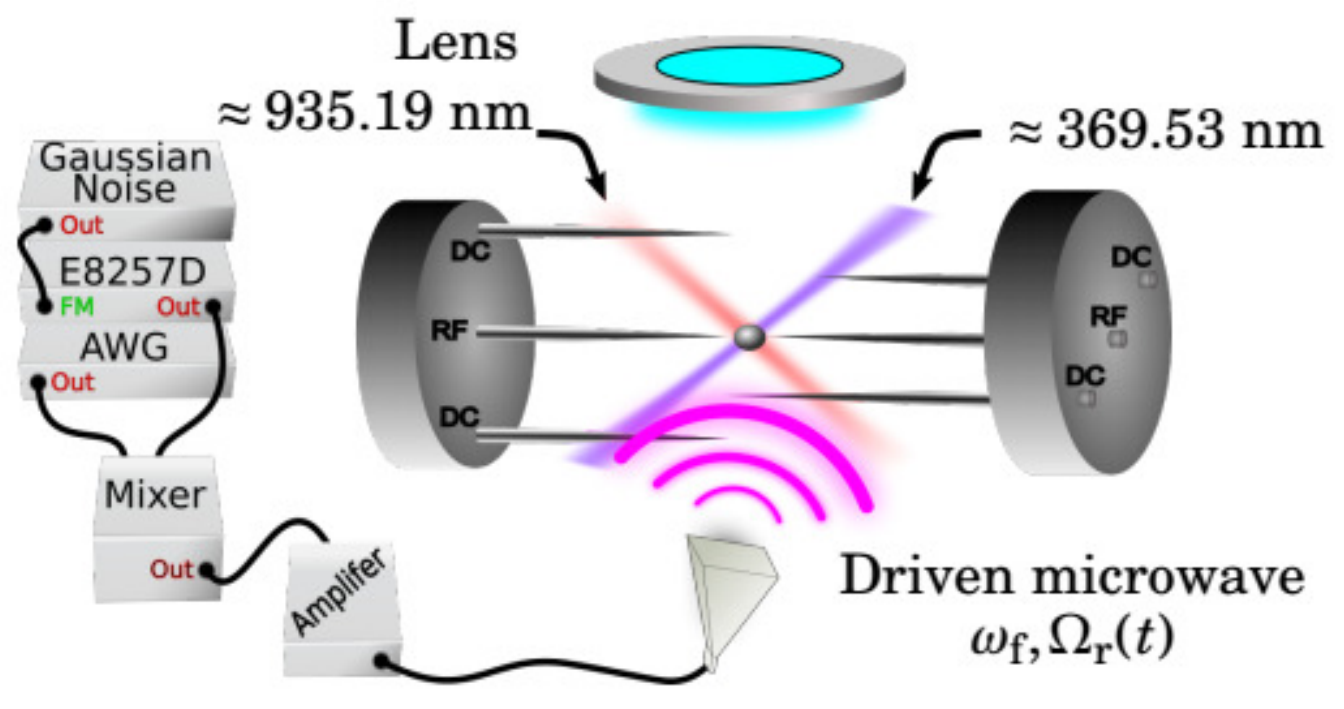}
	\caption{A modified version of the experimental setup presented in Fig.~\ref{FigAdiabExpTrappedIon}{\color{blue}a}, where here we have a new electronic component used to introduce a Gaussian noise in driving microwave.}
	\label{FigAdiabExpComTrappedIon}
\end{figure}

As a quantifier of the adiabatic dynamics we define the fidelity between two density matrices $\rho_{1}$ and $\rho_{2}$ given by
\begin{equation}
	\Fcalb(\rho_{1},\rho_{2}) = \text{Tr}\sqrt{ \sqrt{\rho_{1}} \rho_{2} \sqrt{\rho_{1}}} \text{ , } \label{EqFidelityGeneral}
\end{equation}
from what we can set here two important fidelities. By putting $\rho_{1}$ as the desired dynamics and $\rho_{2}$ as the numerical solution of the dynamics, we define the theoretical fidelity $\Fcalb_{\text{the}}(t)$. Otherwise, when $\rho_{2}$ assume the values of the experimental density matrix (obtained from tomography process) we have the experimental fidelity $\Fcalb_{\text{exp}}(t)$. Said that, the theoretical/experimental fidelity of the system follows an adiabatic path in open system as function of the total evolution time $\tau$ for the Deutsch algorithm reads
\begin{equation}
	\Fcalb_{\text{the(exp)}}(\tau) = \text{Tr}\sqrt{ \sqrt{\rho^{\text{DA}}(\tau)} \rho_{\text{the(exp)}}(\tau) \sqrt{\rho^{\text{DA}}(\tau)}} \text{ , }
\end{equation}
with $\rho_{\text{the(exp)}}(\tau)$ being obtained through numerical solution (experimental realization) and $\rho^{\text{DA}}(\tau)$ is given by Eq.~\eqref{EqDAAdSolOS}. In case where $\Fcalb_{\text{the(exp)}}(\tau) \approx 1$, we have the adiabatic solution with high fidelity, otherwise the dynamics is not adiabatic. For the experiment of the Deutsch algorithm we take a balanced function as example and we set the parameter $\omega = 2\pi \times 10$~KHz, so that in Fig.~\ref{FigGraphsDA} we present the results for the fidelity $\Fcalb(\tau)$. In all of the experiments realized in this work, the error bars are obtained from the standard deviation associated with 60 000 repetitions of the experiment. For every fidelity, we perform state tomography by measuring the qubit in the three Pauli bases ($\sigma_{x}$, $\sigma_{y}$, and $\sigma_{z}$)~\cite{James:01}, with every basis measured $2000$ times and repeated 10 times.

\begin{figure}[t!]
	%\input{Figs/FigFidelAdiabDA.plt}
	%\vspace{5.4cm}
	\centering
	\includegraphics[scale=0.6]{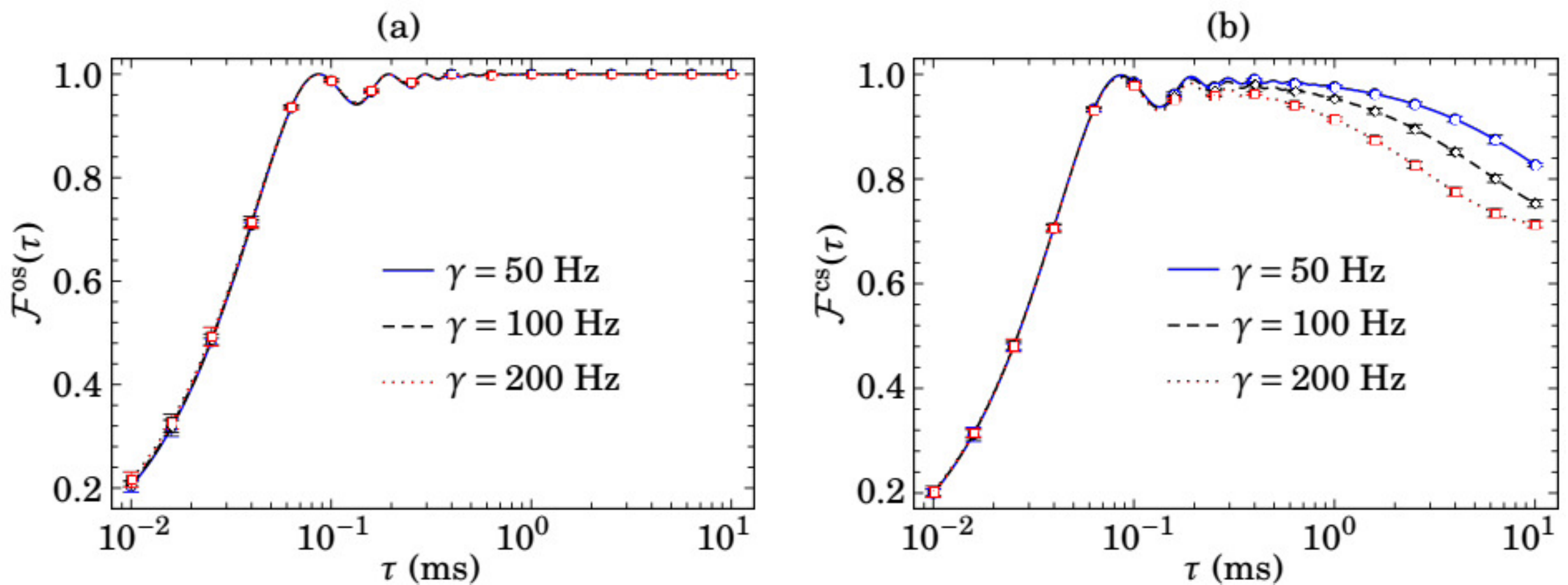}
	\caption{({\color{blue}a}) Fidelity for the adiabatic open system behavior of the Deutsch algorithm. It is shown $\Fcalb^{\text{os}}_{\text{the}}(\tau)$ (continuum lines) and its experimental counterpart $\Fcalb^{\text{os}}_{\text{exp}}(\tau)$ (symbols) as functions of $\tau$ for several values of the dephasing rate $\gamma$. ({\color{blue}b}) Fidelity for the pure target state in the Deutsch algorithm for several values of the dephasing rate $\gamma$. Lines represent $\Fcalb^{\text{cs}}_{\text{the}}(\tau)$ and the dots its experimental counterpart $\Fcalb^{\text{cs}}_{\text{exp}}(\tau)$. In both graphs we consider a balanced function with $\omega = 2\pi \times 10$~KHz. Due the high controllability and measurement fidelity, some experimental points data are superposed.}
	\label{FigGraphsDA}
\end{figure}

It is argued in Ref.~\cite{Sarandy:05-1} that there should be an optimal time for the adiabatic approximation. While this is true in general, here fidelity increases to its maximum as $\tau \rightarrow \infty$ due to the initial conditions of the system state. Actually, as discussed in previous section, the relation between adiabatic and the decoherence time scales allows us to get an adiabatic dynamics for very-slowly driven Lindbladian. It is also important to stress here that the high fidelity for the adiabatic behavior as $\tau \rightarrow \infty$ does not imply in a high fidelity for the target state of the algorithm. Indeed, the output for a balanced function $f$ is the pure state $\ket{-}$, which can be represented by the density operator $\rho_{\text{cs}}^{\text{DA}}(\tau)$ given in Eq.~\eqref{EqDAClosedSol}. In order to quantify the success of the algorithm under decoherence, we define the fidelity of the real dynamics with respect to the target output density operator given by
\begin{equation}
	\Fcalb^{\text{cs}}_{\text{exp}}(\tau) = \text{Tr}\sqrt{ \sqrt{ \rho_{\text{cs}}^{\text{DA}}(\tau) } \rho_{\text{the(exp)}}(t) \sqrt{\rho_{\text{cs}}^{\text{DA}}(\tau)}} \text{ . }
\end{equation}

Observe that, since we are evolving under decoherence, the target state $\rho_{\text{cs}}^{\text{DA}}(\tau)$ is distinct of the adiabatic solution given in Eq.~\eqref{EqDAAdSolOS}, with the adiabatic solution reducing to the target state in the limit of vanishing decoherence. The fidelity with respect to the pure target state in Fig.~\ref{FigGraphsDA}{\color{blue}b} decays for long times as the decoherence rate increases. This is in contrast with the fidelity with respect to the mixed-state adiabatic evolution for open system in Fig.~\ref{FigGraphsDA}{\color{blue}a}, which goes to one as $t \rightarrow \infty$. This is because adiabaticity in open systems is related to the decoupling of Jordan blocks in a nonunitary evolution, which is only approximately equivalent (in the weak coupling regime) to achieving a pure target state after a computation process. Therefore, by focusing on pure target states, it is possible to see a preferred time-window exhibiting maximum fidelity. In Fig.~\ref{FigGraphsDA}{\color{blue}a}, we show two windows with high fidelities, which are highlighted in light green color. Notice that these windows also correspond to high fidelities in Fig.~\ref{FigGraphsDA}{\color{blue}b}, which means that adiabaticity in open systems is indeed able to provide the target state of the Deutsch algorithm with high fidelity for convenient measurement times. Therefore, this allows us to understand a bit more how to obtain a favorable trade-off between the necessary time to achieve the adiabatic regime and the (restrictive) time scale of decohering effects.

Now, let us to generalize the above result for a more general situation. By using the master equation, it follows that the superoperator $\Lcalb_{t}[\bullet]$ satisfies $\trs{\Lcalb_{t}[\rho(t)]} = 0$, since $\trs{\dot{\rho}(t)} = 0$.
Thus, if we consider the basis $\{\sigma_{n}\}$, with $\sigma_{0}=\1$, the first row of the matrix representation of $\Lmath(t)$ is vanishing.
In fact, by adopting that the matrix elements of $\Lmath(t)$ are written as $\Lmath_{mn}(t) = (1/D)\trs{\sigma_{m}^{\dagger}\Lcalb_{t}[\sigma^{n}]}$, we have
\begin{equation}
	\Lmath_{0n}(t) = \frac{1}{D}\trs{\Lcalb_{t}[\sigma^{n}]} = 0 \text{ . }
\end{equation}
Then $\Lmath(t)$ has at least one eigenvalue zero with eigenvector constant. We further assume that $\Lmath(t)$ admits one-dimensional Jordan decomposition and
that there are no eigenvalue crossings in the spectrum of $\Lmath(t)$.
If we suitably order the basis so that the first eigenvalue is $\lambda_{0} = 0$, we get the associated eigenvector as
\begin{equation}
	\dket{\Dcalb_{0}} = \begin{bmatrix}
		1 & 0 & 0 & \cdots & 0 & 0
	\end{bmatrix}^t \text{ . }
\end{equation}
As an immediate result, it follows that $\dbra{\Ecalb_{0}} = \dket{\Dcalb_{0}}^{t}$. In addition, such vector is associated with the maximally mixed state $(1/D)\1$. In fact, the elements $\varrho_{n}(t)$ of $\dket{\rho(t)}$ are given by $\varrho_{n}(t) = \trs{\sigma_{n}^{\dagger}\rho(t)}$, so that in the basis $\{\1,\sigma_{n}\}$ we have
\begin{equation}
	\varrho_{n}(t) = \frac{1}{D}\trs{\sigma_{k}^{\dagger} \1} = \delta_{0n} \text{ . }
\end{equation}
where we use that $\trs{\sigma_{k}^{\dagger}\sigma_{n}} = D \delta_{kn}$ and $\trs{\sigma_{n}} = D\delta_{n0}$. Therefore, by writing the density matrix as in Eq.~\eqref{EqEqRhoCoherence}, it is possible to conclude that any physical initial state $\dket{\rho(0)}$ \textit{must} be written as a combination of $\dket{\Dcalb_{0}}$ and other vectors $\dket{\Dcalb_{\beta\neq 0}(t)}$. Hence, the initial state can be generally written as
\begin{equation}
	\dket{\rho(0)} = c_0(0) \dket{\Dcalb_{0}} +\sum\nolimits_{\beta \neq 0} c_{\beta}(0) \dket{\Dcalb_{\beta}(0)} \text{ , } \label{ApInitState}
\end{equation}
where $c_0(0)=1$ and $c_{\beta}(0)$ are general complex coefficients.

The dynamics of the vanishing eigenvalue subspace can be studied from equation
\begin{equation}
	\dot{c}_{0}(t) = \sum_{\beta \neq0} c_{\beta}(t) \dinterpro{\Ecalb_{0}}{\dot{\Dcalb}_{\beta}(t)}e^{\int_{0}^{t} \lambda_{\beta}(t^\prime)dt^\prime} \text{ . }\label{Apc0dot}
\end{equation}
where we already used $\lambda_{0}(t) = 0$ and $\dinterpro{\Ecalb_{\beta}(t)}{\dot{\Dcalb}_{0}(t)} = 0$. Now, using that the supervectors $\dket{\dot{\Dcalb}_{\beta}(t)}$ and $\dbra{\Ecalb_{\beta}(t)}$ satisfy $\dinterpro{\Ecalb_{\alpha}(t)}{{\Dcal}_{\beta}(t)} = \delta_{\beta \alpha}$ for any $\beta$ and $\alpha$, we have
\begin{equation}
	\frac{d}{dt}\dinterpro{\Ecalb_{\alpha}(t)}{\Dcalb_{\beta}(t)} = \dinterpro{\dot{\Ecalb}_{\alpha}(t)}{\Dcalb_{\beta}(t)}+ \dinterpro{\Ecalb_{\alpha}(t)}{\dot{\Dcalb}_{\beta}(t)} = 0 \text{ , } \nonumber
\end{equation}
and consequently we conclude that
\begin{equation}
	\dinterpro{\dot{\Ecalb}_{\alpha}(t)}{\Dcalb_{\beta}(t)} = - \dinterpro{\Ecalb_{\alpha}(t)}{\dot{\Dcalb}_{\beta}(t)} \text{ . }
\end{equation}
By using this result in Eq.~\eqref{Apc0dot}, we get $\dot{c}_{0}(t) = 0$, which then implies in $c_0(t)=1, \, \forall t$.

Concerning the dynamics of the remaining eigenstates of $\Lmath(t)$, let us assume that
the initial state is such that a single Jordan block is populated in addition to the block associated with $\lambda_0(t)$,
i.e., $c_\eta(0) \ne 0$ for a single $\eta \in \left[1,D^2-1\right]$.
We then start by looking at the adiabatic condition in Eqs.~\eqref{EqAdCondOpenSystemGenBlock}, they tell us whether or not
$\beta$ can evolve decoupled from $\alpha$, but it does not provide any information whether $\alpha$ can evolve decoupled from $\beta$.
Assume that the adiabaticity conditions hold asymptotically in time ($t \rightarrow \infty$),
due to the fact that $\Re\left[\lambda_{\alpha}(t)-\lambda_{\beta}(t)\right] < 0, \, \forall \alpha,t$.
Provided the absence of level crossings as a function of time,
this condition selects the largest eigenvalue $\lambda_{\beta}(t)$.
Then, we write $\dot{c}_{\beta}(t) =  b(t) c_{\beta}(t)$, with $b(t)$ denoting a complex coefficient. The solution reads $c_{\beta}(t) = c_{\beta}(0)e^{\int_{0}^{t}b(t^\prime)dt^\prime}$.
On the other hand, once the parameter $\xi_{\alpha\beta}(t)$ does not necessarily satisfy the adiabaticity conditions, we generically write
$\dot{c}_{\alpha}(t) = a_\beta(t) c_{\beta}(t) + \sum_{\alpha^\prime \neq \beta} a_{\alpha^\prime} (t) c_{\alpha^\prime}(t)$, with $a_{\alpha^\prime}(t)$ and $a_{\beta}(t)$
complex coefficients. However, by imposing $c_{\beta}(0) = 0$,
it yields ${c}_{\beta}(t) = 0$ and consequently
$\dot{c}_{\alpha}(t) = \sum_{\alpha^\prime \neq \beta} a_{\alpha^\prime} (t) c_{\alpha^\prime}(t)$. By iteratively applying this argument after decoupling $\lambda_\beta(t)$,
we finally obtain that $\dot{c}_{\eta}(t) = a_{\eta}(t) c_{\eta}(t)$, stopping as $\lambda_\eta(t)$ becomes the largest eigenvalue of the remaining set.
Thus, by assuming a single $c_\eta(0) \ne 0$ in the initial state, which is equivalent to assuming an initial superposition of only two Jordan blocks, an adiabatic dynamics
for $t\rightarrow \infty$ is always achieved, reading
\begin{equation}
	\dket{\rho(t)} = \dket{\Dcalb_{0}} + c_{\eta}(t) \dket{\Dcalb_{\eta}(t)} \text{ . }
\end{equation}

\section{Thermodynamics of adiabatic dynamics in open system}

Thermodynamics is the knowledge field that encompasses physical processes of systems composed by a large number of particles subject to interaction with its environment. In particular, processes where some amount of energy exchanges that may arise from the difference of temperature between two systems (heat) or due to changes in force fields acting on some system (work). In this sense, a new knowledge field emerges when we use notions of thermodynamics to quantum mechanical systems: \textit{quantum thermodynamics}. Quantum thermodynamics deals with thermodynamics quantities (heat and work) of systems that present physical properties inherent to quantum systems (entanglement and coherence, for example). For example, quantum thermodynamics may be used to understand how coherence allows us to design efficient thermal engines at quantum level~\cite{Scovil:59,Klatzow:19}, and entanglement is used to enhance charging power in quantum batteries~\cite{Ferraro:18,PRL_Andolina,PRL2017Binder}. Here we will consider a particular question related with the thermodynamics of adiabatic dynamics. To this end, we present a short discussion on how to define heat and work for quantum mechanical systems driven by master equations as in Eq.~\eqref{EqEqLind}. Because it is not the scope of this thesis to carefully present general results and features of quantum thermodynamics, for further details on quantum thermodynamic and its applications we recommend some relevant references~\cite{Deffner-Campbel:Book,Alicki:18,CP_Vinjanampathy,Anders:13,Gemmer:Book}.

\subsection{Work and heat in quantum thermodynamics}

In classical thermodynamics, the first thermodynamics law establishes that the internal energy variation of a system is given by
\begin{equation}
	d U = \dbar W + \dbar Q \text{ , } \label{EqdU}
\end{equation}
where $\dbar Q$ is some amount of heat exchange between the system of interest and its environment, and $\dbar W$ is work realized by/on the system. Here $\dbar$ denotes that the heat and word depend on the thermodynamics path followed by the system. In above equation, heat and work are very well defined and it can be obtained from observation macroscopic properties of the system. However, physical quantities in quantum mechanics are obtained from its postulates and, consequently, from Hermitean observables. In general, it has been shown that it is not possible to associate an observable for the thermodynamic definition of heat and of work~\cite{Talkner:07}. Then, the starting point widely used to define such physical quantities in quantum systems is from the definition of internal energy~\cite{Alicki:79,Kieu:04}. If we consider a quantum system driven by a time-dependent Hamiltonian $H(t)$, we can can compute the instantaneous internal energy of the system as
\begin{equation}
	U (t) = \tr{\rho (t) H(t)} \text{ , }
\end{equation}
given the instantaneous system state at instant $t$. Therefore, the energy variation of the system can be directly obtained as
\begin{equation}
	d U (t) = \frac{d}{dt} \tr{\rho (t) H(t)} dt = \tr{\dot{\rho} (t) H(t)} dt + \tr{\rho (t) \dot{H}(t)} dt \text{ . } \label{EqdUQS}
\end{equation}

Now, in order to get a relation between the above equation with Eq.~\eqref{EqdU}, we need to consider two important observations. First, let be the system evolving free from any external interaction with other systems, that is, we have a closed system and it follows a unitary dynamics. In this case, we get
\begin{equation}
	d U_{\text{cs}} (t) = \tr{\rho (t) \dot{H}(t)} dt \text{ , } \label{EqdUQSUniDy}
\end{equation}
because we can use the von Neumann equation for the density operator to write
\begin{equation}
	\tr{\dot{\rho} (t) H(t)} = \frac{1}{i\hbar}\tr{[H(t),\rho (t)] H(t)} = \frac{1}{i\hbar}\tr{H(t)\rho (t) H(t) - \rho (t) H(t) H(t)} = 0 \text{ , }
\end{equation}
valid whenever the system dynamics is unitary (closed system). That means any variation in the internal energy of the system, through a unitary dynamics, is due to time-dependence of the Hamiltonian. Consequently, some external field control is required when we want to change the system energy. For this reason, it is convenient to associate the energy variation $\Delta U_{\text{cs}} (t_{1},t_{0})$ in interval $t \in [t_{1},t_{0}]$ with work realized by/on the system $W(t_{1},t_{0})$ as
\begin{equation}
	W(t_{1},t_{0}) = \Delta U_{\text{cs}} (t_{1},t_{0}) = \int_{t_{0}}^{t_{1}} \underbrace{\tr{\rho (t) \dot{H}(t)} dt}_{\dbar W} \text{ . } \label{EqWorkQS}
\end{equation}

Because we have an addition in Eq.~\eqref{EqdU}, when $W(t_{1},t_{0}) > 0$ we say that work is done \textit{on} the system, otherwise work is done \textit{by} the system. Here we need to be careful because the above equation is valid for unitary dynamics, so when the system is coupled to an external system, we cannot identity the above equation as work in general~\cite{Deffner-Campbel:Book,Anders:13}. On the other hand, let us consider the case where the system is driven with a \textit{time-independent} Hamiltonian $H(t) = H$, then from Eq.~\eqref{EqdUQS} we get
\begin{equation}
	d U_{\text{os}} (t) = \tr{\dot{\rho} (t) H } dt = \tr{\Lcalb_{t}[\rho(t)]H}dt \text{ , } \label{EqdUQSNonUniDy}
\end{equation}
where we used the Eq.~\eqref{EqEqLind} in last equality. The meaning of the subscript ``os'' in above equation is because $d U_{\text{os}} (t) \neq 0$ if the system is coupled to some external system (reservoir), otherwise we get $d U_{\text{os}} (t) = 0$. In fact, this result is immediate from Eq.~\eqref{EqdUQSUniDy} with $H(t) = H$. Then, when the system is driven by time-independent fields, the energy of our system just change if we couple it to some external system which will transfer energy to our system of interest. For this reason, we identify the above quantity with the amount of \textit{heat} exchanged between the system and its environment as
\begin{equation}
	Q(t_{1},t_{0}) = \Delta U_{\text{os}} (t_{1},t_{0}) = \int_{t_{0}}^{t_{1}} \underbrace{\tr{\Lcalb_{t}[\rho(t)]H} dt}_{\dbar Q} \text{ . } \label{EqHeatQS}
\end{equation}

An important remark here is that in \textit{steady state} regime we get $\dbar Q = 0$. In fact, because we have that the steady state $\rho_{\text{ss}}$ satisfies $\Lcalb_{t}[\rho_{\text{ss}}] = 0$, in the limit where $\Lcalb_{t}[\rho(t)] \rightarrow \Lcalb_{t}[\rho_{\text{ss}}] = 0$, we get $\dbar Q = 0$. From this result, if the dynamics admits a steady state, we can establish a maximum amount of heat exchanged between system and reservoir as
\begin{equation}
	Q_{\text{ss}} = \int_{t_{0}}^{\tau_{\text{ss}}} \tr{\Lcalb_{t}[\rho(t)]H} dt = \tr{H\rho_{\text{ss}}} - \tr{H\rho(t_{0})} \text{ . } \label{EqHeatQS-SS}
\end{equation}
where $\tau_{\text{ss}}$ is the time required to achieve the steady state $\rho_{\text{ss}}$.

\subsection{Adiabatic process \textit{versus} adiabatic dynamics}

In thermodynamics, by definition adiabaticity is associated to a process with no heat exchange between the system and its reservoir. On the other hand, as discussed in Sec.~\ref{SecAdApprox}, adiabaticity in dynamics of quantum systems is well defined in a different way. Therefore, the question that arises here is: \textit{How are thermodynamics adiabatic processes related with an adiabatic dynamics?} As a preliminary result (but insufficient yet), we can restrict our analysis to the closed system case. In general, as already discussed here, unitary processes do not admit energy transference as heat, since no system is coupled to our main system. Because this result is independent on the kind of dynamics (adiabatic or not), we can say that \textit{adiabatic processes in closed systems are not related with an adiabatic dynamics}. Therefore, in order to give a more complete study on this issue, it is convenient to consider an open system dynamics because some heat exchange is admissible. To this end, we need to introduce heat and work in a general way from superoperator formalism, where adiabatic dynamics in open system is well defined.

Therefore, let us consider the heat exchange as given in Eq.~\eqref{EqdUQSNonUniDy}. To derive the corresponding expression in the superoperator formalism we first define the basis of operators given by $\{\sigma_{i}\}$, $i=0,\cdots,D_{-}$ (again we are using $D_{-} = D^2-1$), where $\trs{\sigma^{\dagger}_{i}\sigma_{j}} = D\,\delta_{ij}$. In this basis, we can write $\rho (t)$ and $H(t)$ generically as
\begin{equation}
	H(t) = \frac{1}{D}\sum_{n=0}^{D_{-}} h_{n}(t)\sigma ^{\dagger}_{n} \text{ \, and \, } \rho(t) = \frac{1}{D}\sum_{n=0}^{D_{-}} \varrho_{n}(t)\sigma_{n} \text{ , } \label{EqHroSuper}
\end{equation}
where we have $h_{n}(t) = \trs{H(t)\sigma_{n}}$ and $\varrho_{n}(t) = \trs{\rho(t)\sigma^{\dagger}_{n}}$. Then, we get
\begin{equation}
	\dbar Q = \frac{1}{D^2}\left( \sum _{n,m=0}^{D_{-}}  \trs{\Lcalb[\varrho_{n}(t)\sigma_{n}] h_{m}(t)\sigma ^{\dagger}_{m}} \right) dt = \frac{1}{D^2}\left( \sum _{n,m=0}^{D_{-}} \varrho_{n}(t) h_{m}(t) \trs{\Lcalb[\sigma_{n}] \sigma ^{\dagger}_{m}} \right) dt \label{EqdHeatSuperoForm} \text{ .}
\end{equation}
Now, we use the definition of the matrix elements of the superoperator $\Lmath(t)$, associated with $\Lcalb[\bullet]$, which reads $\Lmath_{mn} = (1/D)\trs{\sigma ^{\dagger}_{m}\Lcalb[\sigma_{n}] }$, so that we write
\begin{equation}
	\dbar Q = \frac{1}{D}\left( \sum _{n,m=0}^{D_{-}} h_{m}(t) \Lmath_{mn} \varrho_{n}(t) \right) dt \label{dQ3open} \text{ .}
\end{equation}
In conclusion, by defining the vector elements
\begin{align}
	\dbra{h(t)} &= \begin{bmatrix}
		h_{0}(t) & h_{1}(t) & \cdots & h_{D_{-}}(t)
	\end{bmatrix}^{\text{t}} \text{ , } \label{EqBraH}\\
	\dket{\rho(t)} &= \begin{bmatrix}
		\varrho_{0}(t) & \varrho_{1}(t) & \cdots & \varrho_{D_{-}}(t)
	\end{bmatrix}\text{ , } \label{EqKetRho}
\end{align}
we can rewrite Eq.~\eqref{dQ3open}, yielding
\begin{equation}
	\dbar Q = \frac{1}{D}\dinterpro{h(t)}{\Lmath (t)|\rho(t)}dt \text{ . } \label{EqQsuperOp}
\end{equation}
Equivalently, from Eq.~\eqref{EqWorkQS} we can use the Eq.~\eqref{EqHroSuper} to write $\dot{H}(t) = (1/D)\sum_{n=0}^{D_{-}} \dot{h}_{n}(t)\sigma ^{\dagger}_{n}$ and, consequently
\begin{equation}
	\dbar W =  \frac{1}{D}\sum_{n=0}^{D_{-}}\dot{h}_{n}(t)\trs{\rho (t) \sigma ^{\dagger}_{n}} dt =  \frac{1}{D}\sum_{n=0}^{D_{-}}\dot{h}_{n}(t)\varrho_{n}(t) dt \text{ , }
\end{equation}
where we used the definition of the coefficients $\varrho_{n}(t)$ in last equality. Now, by using Eqs.~\eqref{EqBraH} and~\eqref{EqKetRho} into above equation, we conclude that
\begin{equation}
	\dbar W = \frac{1}{D}\dinterpro{\dot{h}(t)}{\rho(t)} dt \text{ . } \label{EqWorkSuperoForm}
\end{equation}

To end, since the system is evolving under a non-unitary dynamics, some entropy variation is computed along the evolution. Here we use the entropy variation as provided by the variation in Von Neumann entropy given by
\begin{equation}
	\Delta S(t,t_{0}) = S(t) - S(t_{0}) = - \tr{\rho(t) \ln [\rho(t)]} + \tr{\rho(t_{0}) \ln [\rho(t_{0})]} \text{ , }
\end{equation}
that can be written in a differential way as
\begin{equation}
	d S(t) = - \left[\frac{d}{dt} \tr{\rho(t) \ln [\rho(t)]}\right] dt = - \trs{ \dot{\rho} (t)\log \rho(t) } - \trs{ \dot{\rho} (t)} \text{ , }
\end{equation}
and by using that $\trs{ \dot{\rho} (t)} = 0$, because $\trs{ \rho (t)} = 1$, and $ \dot{\rho} (t) = \Lcalb_{t} [\rho(t)]$ we get
\begin{equation}
	d S(t) = - \trs{ \dot{\rho} (t)\log \rho(t) }dt = - \trs{ \Lcalb_{t} [\rho(t)] \log \rho(t) }dt \text{ . }
\end{equation}

Now, let us to write
\begin{equation}
	\log \rho(t) = \frac{1}{D}\sum_{n=0}^{D_{-}} \varrho_{n}^{\log}(t)\sigma_{n}^{\dagger} \text{ , } 
\end{equation}
so that we can define the vectors $\dbra{\rho_{\log}(t)}$ associated to $\log \rho(t)$ with components $\varrho_{n}^{\log}(t)$ obtained as $\varrho_{n}^{\log}(t) = \trs{\sigma_{n}\log \rho(t)}$. Thus, we get
\begin{equation}
	dS (t) = -\frac{1}{D^2}\sum_{m=0}^{D_{-}}\sum_{n=0}^{D_{-}}\varrho_{m}(t) \varrho_{n}^{\log}(t)\trs{ \Lcalb_{t} [\sigma_{m}]\sigma_{n}^{\dagger} } dt = - \frac{1}{D}\dbra{\rho_{\log}(t)}\Lmath (t) \dket{\rho(t)} dt \text{ . }  \label{EqVarEntropySuperop}
\end{equation}

Thus, from Eqs.~\eqref{EqdHeatSuperoForm},~\eqref{EqWorkSuperoForm} and~\eqref{EqVarEntropySuperop} we can study the thermodynamics of adiabatic dynamics in a general framework. However, since work of an adiabatic dynamics has been addressed in a number of applications, here we are interested in study the implications of the adiabatic dynamics in the expression of heat derived in Eqs.~\eqref{EqdHeatSuperoForm} and~\eqref{EqVarEntropySuperop}. To this end, let us consider the most general case in which the system starts from a superposition $\dket{\rho (0)} = \sum \nolimits _{\alpha,k_{\alpha}} r^{k_{\alpha}}_{\alpha}\dket{\Dcalb^{k_{\alpha}}_{\alpha}(0)}$, associated with the initial arbitrary matrix density $\rho(0)$. Then, if we take into account the situations where the system evolves adiabatically, we use the evolution operator in Eqs.~\eqref{EqUAdOS} and~\eqref{EqUalpha} to write the evolved as 
\begin{align}
	\dket{\rho_{\text{ad}} (t)} &= \Ucalb (t) \dket{\rho (0)} = \sum _{\beta}^{N} \sum _{k_{\beta}}^{N_{\beta}} \sum_{\alpha = 1}^{N} \sum _{n_{\alpha},m_{\alpha}}^{N_{\alpha}} r^{k_{\beta}}_{\beta}\eta_{n_{\alpha}m_{\alpha}}(t)\dket{\Dcalb_{\alpha}^{n_{\alpha}}(t)}\dinterpro{\Ecalb_{\alpha}^{m_{\alpha}}(0)}{\Dcalb^{k_{\beta}}_{\beta}(0)} \nonumber \\
	&= \sum _{\beta}^{N} \sum _{n_{\beta},k_{\beta}}^{N_{\beta}} r^{k_{\beta}}_{\beta}\eta_{n_{\beta}k_{\beta}}(t)\dket{\Dcalb_{\beta}^{n_{\beta}}(t)}
\end{align}
where we defined $\eta_{n_{\alpha}m_{\alpha}}(t) = e^{\int_{0}^{t} \lambda_{\alpha}(\xi)d\xi}u_{n_{\alpha}m_{\alpha}}(t)$. By using the above result, from Eq.~\eqref{EqQsuperOp} we write
\begin{equation}
	\dbar Q_{\text{ad}} = \frac{1}{D}\dinterpro{h(t)}{\Lmath (t)|\rho_{\text{ad}}(t)}dt =
	\frac{1}{D}\sum _{\beta}^{N} \sum _{n_{\beta},k_{\beta}}^{N_{\beta}} r^{k_{\beta}}_{\beta}\eta_{n_{\beta}k_{\beta}}(t)
	\dinterpro{h(t)}{\Lmath (t)|\Dcalb_{\beta}^{n_{\beta}}(t)}dt \text{ , }
\end{equation}
as the amount of heat exchanged between the system and its environment along an adiabatic dynamics in open systems. A more useful expression can be found using the quasi-eigenvector equation for $\Lmath (t)$ given in Eq.~\eqref{EqEqEigenStateL} and reads as
\begin{equation}
	\dbar Q_{\text{ad}}(t) = \frac{1}{D}\sum _{\beta}^{N} \sum _{n_{\beta},k_{\beta}}^{N_{\beta}} r^{k_{\beta}}_{\beta}\eta_{n_{\beta}k_{\beta}}(t) \left[\dinterpro{h(t)}{\Dcalb_{\beta}^{n_{\beta}-1}(t)} + \lambda_{\beta}(t)\dinterpro{h(t)}{\Dcalb_{\beta}^{n_{\beta}}(t)}\right] dt 
	\text{ . } \label{EqHeatAdOpen}
\end{equation}

Moreover, let us consider the case where we have an one-dimensional Jordan block decomposition of $\Lmath (t)$. In this case, we use the adiabatic solution in Eq.~\eqref{EqQsuperOp} to write
\begin{equation}
	\dbar Q_{\text{ad}}^{\text{1D}}(t) = \frac{1}{D}\dinterpro{h(t)}{\Lmath (t)|\rho^{1\text{D}}_{\text{ad}}(t)}dt = \frac{1}{D}\sum _{\beta}^{N} r_{\beta} e^{\int_{0}^{t} \Lambda_{\beta}(t^{\prime})dt^{\prime}} \lambda_{\beta}(t) \dinterpro{h(t)}{\Dcalb_{\beta}^{n_{\beta}}(t)} dt 
	\text{ . } \label{EqHeatAdOpen1D}
\end{equation}

The Eqs.~\eqref{EqHeatAdOpen} and~\eqref{EqHeatAdOpen1D} constitute the main result of this section, where it is evident the possibility of heat exchange during an adiabatic dynamics in open system. In addition, the Eq.~\eqref{EqHeatAdOpen1D} allows us to find a condition for achieving an adiabatic process through an adiabatic dynamics. In fact, note that if we start the dynamics from a superposition of eigenvectors $\{\Dcalb_{\beta}^{n_{\beta}}(0)\}$ of $\Lmath (0)$ with $\lambda_{\beta}(t) = 0$ for all $t$, we get $\dbar Q_{\text{ad}}^{\text{1D}}(t) = 0$ for all $t$. In conclusion, we can establish that \textit{an adiabatic dynamics in quantum mechanics is not in general associated with an adiabatic process in quantum thermodynamics, with a sufficient condition for thermal adiabaticity being the evolution within an eigenstate set with vanishing eigenvalue of $\Lmath (t)$}.

To end, from Eq.~\eqref{EqVarEntropySuperop}, we can write the von Neumann entropy variation as
\begin{equation}
	dS (t) = - \frac{1}{D}\dbra{\rho^{\text{ad}}_{\log}(t)}\Lmath (t) \dket{\rho^{\text{ad}}(t)}dt = - \frac{1}{D}\sum \nolimits _{i,k_{i}} c^{(k_{i})}_{i} e^{\int_{0}^{t} \tilde{\lambda}_{i,k_{i}}(t^{\prime})dt^{\prime}}\dbra{\rho^{\text{ad}}_{\log}(t)}\Lmath (t) \dket{\Dcalb^{(k_{i})}_{i}(t)}dt \nonumber \text{ , }
\end{equation}
so that we can use the Eq.~\eqref{EqEqEigenStateL} to write
\begin{equation}
	dS (t) =\frac{1}{D}\sum \nolimits _{i,k_{i}} c^{(k_{i})}_{i} e^{\int_{0}^{t} \tilde{\lambda}_{i,k_{i}}(t^{\prime})dt^{\prime}}
	\Gamma_{i,k_{i}}(t)dt \text{ , } \label{EntropyProAp}
\end{equation}
where $\Gamma_{i,k_{i}}(t) = \dinterpro{\rho^{\text{ad}}_{\log}(t)}{\Dcalb^{(k_{i}-1)}_{i}(t)} + \lambda_{i}(t) \dinterpro{\rho^{\text{ad}}_{\log}(t)}{\Dcalb^{(k_{i})}_{i}(t)}$, with $\dbra{\rho^{\text{ad}}_{\log}(t)}$ standing for the adiabatic evolved state associated with $\dbra{\rho_{\log}(t)}$.

\emph{Thermal adiabaticity for a qubit adiabatic dynamics --} As an illustration of how adiabatic dynamics is an adiabatic process, let us begin by considering a two-level system initialized in a thermal equilibrium
state $\rho_{\text{th}}(0)$ for the Hamiltonian $H(0)$ at inverse temperature $\beta\!=\!1/k_{\text{B}}T$, where $k_{\text{B}}$ and $T$ are the
Boltzmann's constant and the absolute temperature, respectively. Now, let the system be governed by a Lindblad
equation, where the environment acts as a \textit{dephasing} channel in the energy eigenstate basis
$\{\ket{E_{n}(t)}\}$ of $H(t)$. Thus, we describe the coupling between the system and its reservoir through
\begin{equation}
	\dot{\rho} (t) = \Lcalb [\rho(t)] = \frac{1}{i\hbar} [H(t),\rho(t)] + \gamma (t) \left[ \Gamma^{\text{dp}}(t) \rho(t) \Gamma^{\text{dp}}(t) - \rho(t) \right] \text{ , }
\end{equation}
where $\Gamma^{\text{dp}}(t) = \ket{E_{1}(t)}\bra{E_{1}(t)} - \ket{E_{0}(t)}\bra{E_{0}(t)}$, with $\ket{E_{n}(t)}$ being the instantaneous eigenstates of $H(t)$, here we are assuming $E_{1}(t) > E_{0}(t)$ for all $t$. In this case, the set of eigenvectors of $\Lcalb [\bullet]$
can be obtained from set of operators $P_{nm}(t) = \ket{E_{n}(t)}\bra{E_{m}(t)}$, because we have
\begin{align}
	\Lcalb [P_{nm}(t)] &= \frac{1}{i\hbar} [H(t),P_{nm}(t)] + \gamma (t) \left[ \Gamma^{\text{dp}}(t) P_{nm}(t) \Gamma^{\text{dp}}(t) - P_{nm}(t) \right] \nonumber \\
	&= \frac{1}{i\hbar} \left[ E_{n}(t) - E_{m}(t) \right]P_{nm}(t) + \gamma (t) \left[ \Gamma^{\text{dp}}(t) P_{nm}(t) \Gamma^{\text{dp}}(t) - P_{nm}(t) \right]P_{nm}(t)
\end{align}

Now, by computing the second term in above equation we get
\begin{align}
	\Gamma^{\text{dp}}(t) P_{nm}(t) \Gamma^{\text{dp}}(t) &= \left[\ket{E_{1}(t)}\bra{E_{1}(t)} - \ket{E_{0}(t)}\bra{E_{0}(t)}\right] \ket{E_{n}(t)}\bra{E_{m}(t)}\left[\ket{E_{1}(t)}\bra{E_{1}(t)} - \ket{E_{0}(t)}\bra{E_{0}(t)}\right] \nonumber \\
	&= \left[\delta_{1n} \ket{E_{1}(t)} - \delta_{0n} \ket{E_{0}(t)} \right] \left[ \delta_{1m}\bra{E_{1}(t)} - \delta_{0m}\bra{E_{0}(t)} \right]
	\nonumber \\
	&= \sum_{k=0,1} \delta_{kn}\delta_{km} \ket{E_{k}(t)}\bra{E_{k}(t)} - \sum_{k=0,1} \delta_{kn}\delta_{\bar{k}m} \ket{E_{k}(t)} \bra{E_{\bar{k}}(t)} \text{ , }
\end{align}
so that we can write
\begin{equation}
	\Gamma^{\text{dp}}(t) P_{nm}(t) \Gamma^{\text{dp}}(t) = \left\{ \begin{matrix}
		P_{nn}(t) & \text{for } n = m \\
		- P_{nm}(t) & \text{for } n \neq m 
	\end{matrix} \right. \text{ , }
\end{equation}
where we conclude that $P_{nm}(t)$ satisfies the eigenvalue equation $\Lcalb [P_{nm}(t)] = \lambda_{nm}(t)P_{nm}(t)$ with
\begin{equation}
	\lambda_{nm}(t) = \left\{ \begin{matrix}
		0 & \text{for } n = m \\
		\left[ E_{n}(t)-E_{m}(t) \right]/i\hbar - 2\gamma(t) & \text{for } n \neq m 
	\end{matrix} \right. \text{ . }
\end{equation}

Therefore, by writing $P_{nm}(t)$ in superoperator formalism the of $\dket{\Dcalb_{nm}(t)}$, with components $\Dcalb^{(i)}_{nm}(t) = \tr{P_{nm}(t)\sigma_{i}}$, the eigenvalue equation for $\Lmath (t)$ can be written as 
\begin{equation}
	\Lmath (t) \dket{\Dcalb_{nm}(t)} = \lambda_{nm}(t) \dket{\Dcalb_{nm}(t)} \text{ , }
\end{equation}
where $\lambda_{nm}(t) = E_{n}(t)-E_{m}(t) - 2(1-\delta_{nm})\gamma(t)$. Since we are considering the initial state as a thermal state, in the superoperator formalism the initial state $\rho_{\text{th}}(0)$, we can write
\begin{equation}
	\rho_{\text{th}}(0) = \Zcalb^{-1}(0) \sum _{n} e^{-\beta E_{n}(0)}P_{nn}(0) \text{ \, \, } \Rightarrow \text{ \, \, } \dket{\rho_{\text{th}}(0)} = \Zcalb^{-1}(0) \sum _{n} e^{-\beta E_{n}(0)}\dket{\Dcalb_{nn}(0)} \text{ , }
\end{equation}
where $\Zcal(0)=\tr{e^{-\beta H(0)}}$ is the
partition function of the system. Note that, since $\dket{\rho_{\text{th}}(0)}$ is given by a superposition of eigenvectors of $\Lmath (t)$ with eigenvalue $\lambda_{nn}(t) = 0$ for all $t$, we obtain from Eq.~\eqref{EqHeatAdOpen1D} that $\dbar Q_{\text{ad}} = 0$ for all $t$ whenever the evolution is adiabatic. Therefore,
thermal adiabaticity is achieved for an arbitrary open-system adiabatic dynamics subject to dephasing in the energy eigenbasis.
Hence, any internal energy variation for this situation should be identified as \textit{work}.

\subsection{Experimental realization of heat exchange in adiabatic dynamics}

Now, we will consider a case where the open system dynamics is adiabatic and some amount of heat is exchanged between system and reservoir. The system used here is the same as before, a two level system evolving under dephasing. However, here we conveniently change the decohering basis and now the decohering acts in computational basis $\ket{0}$ and $\ket{1}$, instead the eigenstate basis of $H(t)$. Then the dynamics of the system is given by
\begin{equation}
	\dot{\rho}(t) = \Lcalb[\rho(t)] = \frac{1}{i\hbar} [H_{x} , \rho(t)] + \gamma(t) \left[ \sigma_{z} \rho(t) \sigma_{z} - \rho(t) \right] \text{ , } \label{EqThermDynamicsDeph}
\end{equation}
with $H_{x} = \hbar \omega \sigma_{x}$. Since we are interested in studying the heat transfer during the adiabatic dynamics, we drive the system with a time-independent Hamiltonian, so that $\dbar W = 0$ for all $t$ and therefore any variation in internal energy of the system is due to heat flux. Moreover, because we want to drive the system through an adiabatic dynamics, the parameter $\gamma(t)$ needs to be time-dependent.

The system is initialized in the thermal of $H_{\text{x}}$ at inverse temperature $\beta$ written as
\begin{equation}
	\rho(0) = \frac{1}{2} \left( \1 + \tanh[\beta \hbar \omega] \sigma_{x} \right) \text{ , }
\end{equation}
where rewrite the above state in superoperator formalism as the state $\dket{\rho(0)}$ in the basis $\Ocal_{\text{tls}}$ as
\begin{equation}
	\dket{\rho(0)} = \dket{\sigma_{0}} - \tanh[\beta \hbar \omega]\dket{\sigma_x} \text{ , }
\end{equation}
with the basis $\dket{\sigma_k}$ ($k = \{0,x,y,z\}$) defined in Eq.~\eqref{EqSuperOpBasis}. In this dynamics, the superoperator $\Lmath (t)$ associated with the generator $\Lcalb[\bullet]$ reads
\begin{equation}
	\Lmath (t) = \begin{bmatrix}
		0 & 0            & 0            & 0 \\
		0 & -2 \gamma(t) & 0            & 0 \\
		0 & 0            & -2 \gamma(t) & -2 \omega \\
		0 & 0            & 2 \omega & 0 
	\end{bmatrix} \text{ . }
\end{equation}

It is possible to show that the set $\{\dket{\sigma_{0}},\dket{\sigma_{x}}\}$ satisfies the eigenvalue equation for $\Lmath (t)$ as
\begin{equation}
	\Lmath (t) \dket{\sigma_{0}} = 0 \text{ \ \ , \ \ } \Lmath (t)\dket{\sigma_{x}} = -2\gamma(t)\dket{\sigma_{x}} \text{ , } \label{EigenAp}
\end{equation}
so, if the dynamics is adiabatic, we can write the evolved state as $\dket{\rho_{\text{ad}}(t)} = r_{1}(t)\dket{\sigma_{0}} + r_{x}(t)\dket{\sigma_{x}}$, where $r_{y}(t) =r_{y}(0) = 0$ and $r_{z}(t)=r_{z}(0)=0$ because the coefficients evolve independently form each other. Thus, from the adiabatic solution in open quantum system given in Eq.~\eqref{EqAdEvol1D},
we obtain $r_{1}(t) = 1$ and $r_{x}(0)=-\tanh[\beta \hbar \omega]$, so that we can use $\tilde{\lambda}_{1}=0$ and $\tilde{\lambda}_{x}=-2\gamma(t)$ to obtain
\begin{equation}
	\dket{\rho_{\text{ad}}(t)} = \dket{\sigma_{0}} - e^{-2 \int_{0}^{t}\gamma(\xi)d\xi}\tanh[\beta \hbar \omega] \dket{\sigma_{x}} \text{ . } \label{EqAprhox}
\end{equation}

Notice that by rewriting the above equation in the standard operator formalism, one gets
\begin{equation}
	\rho_{\text{ad}}(t) = \frac{1}{2} \left[\1 - e^{-2 \int_{0}^{t}\gamma(\xi)d\xi}\tanh(\beta \hbar \omega) \sigma_{x}\right] \text{ , } \label{EqrhoxStand}
\end{equation}
moreover, by using this formalism, it is also possible to show that the dephasing channel can be used as a thermalization process if we suitably choose the parameter $\gamma(t)$ and the total evolution time $\tau_{\text{dec}}$. In fact, we can define a new inverse temperature $\beta_{\text{deph}}$ so that Eq.~(\ref{EqAprhox}) behaves as thermal state, namely,
\begin{equation}
	\dket{\rho(t)} = \dket{\sigma_{0}} - \tanh[\beta_{\text{deph}} \hbar \omega] \dket{\sigma_{x}} \text{ , } \label{ThermalAp}
\end{equation}
where we immediately identify
\begin{equation}
	\beta_{\text{deph}} = \frac{1}{\hbar \omega} \text{arctanh}\left[ e^{-2 \int_{0}^{t}\gamma(\xi)d\xi}\tanh(\beta \hbar \omega) \right] \text{ . } \label{EqBetaDeph}
\end{equation}

Now, we use the Eq.~\eqref{EqHeatAdOpen} to compute the exchanged heat during an infinitesimal time interval $dt$ as
\begin{equation}
	\dbar Q_{\text{ad}}(t) = 2 \hbar \tanh(\beta \hbar \omega) \omega \gamma(t) e^{-2 \int_{0}^{t}\gamma(\xi)d\xi} dt \text{ , }
\end{equation}
and therefore we get the total exchanged heat as a function of the total evolution time $\tau_{\text{dec}}$ in which the system remains coupled to reservoir as
\begin{equation}
	Q_{\text{ad}}(\tau_{\text{dec}}) = \int_{0}^{\tau_{\text{dec}}} \left(\frac{\dbar Q_{\text{ad}}(t)}{dt}\right) dt = \hbar \omega \tanh(\beta \hbar \omega) \left( 1 - e^{-2 \bar{\gamma} \tau_{\text{dec}}} \right) \text{ , } \label{EqQad}
\end{equation}
where $\bar{\gamma} = (1/\tau_{\text{dec}})\int_0^{\tau_{\text{dec}}} \gamma(\xi) d\xi$ is the average dephasing rate during $\tau_{\text{dec}}$. A first remark here is that $Q_{\text{ad}}(\tau_{\text{dec}})>0$ for any value of $\bar{\gamma}>0$ and $\tau$, so this means that the reservoir here works as an \textit{artificial} thermal reservoir at inverse temperature $\tilde{\beta}=\beta_{\text{deph}}<\beta$. We also considered a similar analysis in an experimental simulation of a thermal engine in all optical experimental setup~\cite{Passos:19}. In addition, due to the negative argument in the exponential, the higher the mean-value of $\gamma(t)$ the faster the heat exchange ends. Then, we can further compute the maximum exchanged heat from Eq.~\eqref{EqQad}, in the limit $\bar{\gamma} \tau_{\text{dec}} \gg 1$, as a quantity independent of the environment parameters and given by $Q_{\text{max}} = \hbar \omega \tanh(\beta \hbar \omega)$. This result can be also obtained from Eqs.~\eqref{EqHeatQS-SS} and~\eqref{EqrhoxStand} by using the asymptotic steady state $\rho_{\text{ss}} = \rho_{\text{ad}}(t \rightarrow \infty) \propto \1$. It would be worth highlighting that, for quantum thermal machines weakly coupled to thermal reservoirs at different temperatures \cite{Alicki:79}, the maximum heat $Q_{\text{max}}$ is obtained with high-temperature hot reservoirs~\cite{Geva:92,Henrich:07,Lutz:16-Science}.

The experimental realization of the above results was done through the trapped ion system already used previously, but in this case we need to consider an additional challenger due to the parameter $\gamma(t)$ to be time-dependent. To do that, it is used the same procedure as before with an additional element able to provide a time-dependent decohering quantum channel with high controllability. For further details, please see Appendix~\ref{ApTrappeIonDeco}. In Fig.~\ref{FigAdHeat} we show the experimental results for the heat exchange $Q_{\text{ad}}(\tau_{\text{dec}})$ as a function of $\tau_{\text{dec}}$, where we have chosen $\gamma(t)\!=\!\gamma_{0}(1+t/\tau_{\text{dec}})$, where here the parameter $\tau_{\text{dec}}$ is experimentally controlled through the time interval associated to the action of our decohering quantum channel. The solid curves in Fig.~\ref{FigAdHeat} are computed from Eq.~\eqref{EqQad}, while the experimental points are directly computed through the variation of internal energy as $Q(\tau_{\text{dec}})\!=\!U_{\text{fin}} - U_{\text{ini}}$, where $U_{\text{fin(ini)}}\!=\!\trs{\rho_{\text{fin(ini)}}H_{x}}$. The computation of $U_{\text{fin(ini)}}$ is directly obtained from quantum state tomography of $\rho_{\text{fin(ini)}}$ for each value of $\tau_{\text{dec}}$. Although the maximum exchanged heat is independent of $\gamma_0$, the initial dephasing rate $\gamma_0$ affects the \textit{power} for which the system exchanges heat with the reservoir for a given evolution time $\tau_{\text{dec}}$. 

\begin{figure}[t!]
	%\input{Figs/FigHeatTherm.plt}
	%\vspace{5.4cm}
	\centering
	\includegraphics[scale=0.6]{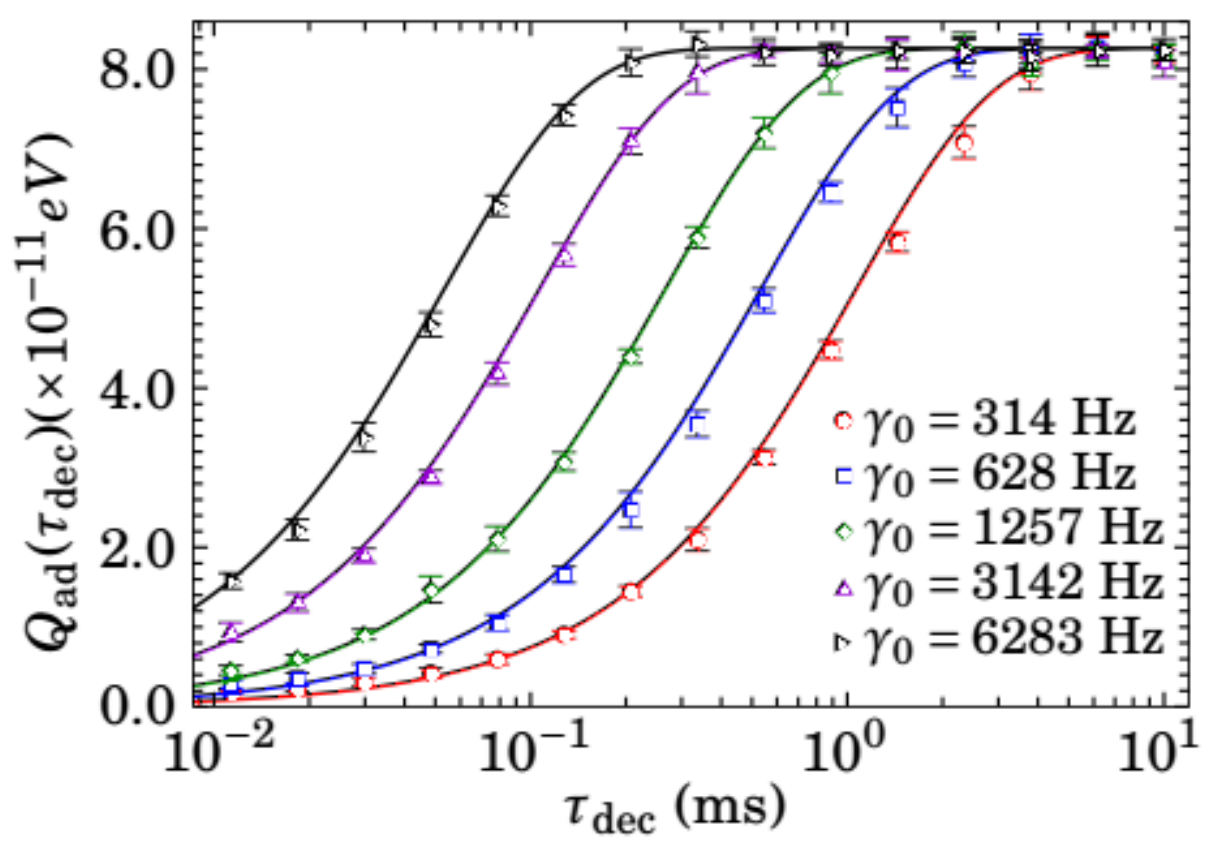}
	\caption{Heat $Q_{\text{ad}}(\tau_{\text{dec}})$ as a function of the total evolution time $\tau_\text{dec}$ for several values of the parameter $\gamma_{0}$. We use $\hbar\omega = 82.662$~peV and $\beta^{-1} = 17.238$~peV, with the physical constants $\hbar \approx 6.578 \cdot 10^{-16}$~eV$\cdot$s and $k_{\text{B}} \approx 8.619 \cdot 10^{-5}$~eV/K.}
	\label{FigAdHeat}
\end{figure}

By completeness and to provide a more detailed view of this heat exchange, we analyze the von Neumann entropy during the evolution. From Eq.~\eqref{EqVarEntropySuperop}, for the particular case of an one-dimensional Jordan block decomposition, we write the entropy variation in our dynamics as
\begin{equation}
	dS (t) =\frac{1}{2}\sum \nolimits _{i = {0}}^{1} c_{i} e^{\int_{0}^{t} \tilde{\lambda}_{i}(t^{\prime})dt^{\prime}}
	\lambda_{i}(t) \dinterpro{\rho^{\text{ad}}_{\log}(t)}{\Dcalb_{i}(t)}dt \text{ , }
\end{equation}
where we used that $\Gamma_{i}(t) = \lambda_{i}(t) \dinterpro{\rho^{\text{ad}}_{\log}(t)}{\Dcalb_{i}(t)}$. By computing $\dbra{\rho^{\text{Ad}}_{\log}(t)}$, where we find 
\begin{equation}
	\dbra{\rho^{\text{Ad}}_{\log}(t)} = \log \left(\frac{1 - g^2(t)}{4} \right) \dbra{1} -2 \text{arctanh} [g(t)] \dbra{x} \text{ , }
\end{equation}
with $g(t)\!=\!e^{-2 \int_{0}^{t}\gamma(\xi)d\xi}\tanh(\beta \hbar \omega)$. Then, from the set of adopted values for our parameters and the spectrum of the Lindbladian, we conclude that 
\begin{equation}
	dS (t) = 4 g(t) \gamma(t) \text{arctanh} [g(t)] dt \text{ . }
\end{equation}

Notice that the relation between heat and entropy can be obtained by rewriting the exchanged heat $d Q$ in the interval $dt$ as 
$\dbar Q_{\text{ad}}(t) = 2 \hbar \omega \gamma(t) g(t)  dt$. In conclusion, the energy variation can indeed be identified as heat exchanged along the adiabatic dynamics. Indeed, by computing the thermodynamic relation between $dS(t)$ and $\dbar Q_{\text{ad}}(t)$ we get $dS(t) = \beta_{\text{deph}} \dbar Q_{\text{ad}}(t)$, where $\beta_{\text{deph}}$ is the  inverse temperature of the simulated thermal bath shown in Eq.~\eqref{EqBetaDeph}.

To guarantee that the amount of heat shown in Fig.~\ref{FigAdHeat} are associated with an exchange heat in an adiabatic dynamics, we computed the fidelity $\Fcalb (\tau_{\text{dec}})$ between the evolved state (solution of the Eq.~\eqref{EqThermDynamicsDeph}) and the adiabatic solution given in Eq.~\eqref{EqrhoxStand}. In Table~\ref{TableAdFidelThermo} we show the minimum value of the fidelity along the evolution, that is, we compute the quantity $\Fcalb_{\text{min}} = \min_{\tau_{\text{dec}}} \Fcalb (\tau_{\text{dec}})$. From Table~\ref{TableAdFidelThermo} it is possible to see that all of heat exchange behavior shown in Fig.~\ref{FigAdHeat} are associated with an adiabatic dynamics in open system with high fidelity.

\begin{table}[t!]
	\centering
	\caption{Minimum value of the experimental fidelity $\Fcalb_{\text{min}} = \min_{\tau_{\text{dec}}} \Fcalb (\tau_{\text{dec}})$ for each choice of $\gamma_{0}$. The maximum experimental error $\Delta\Fcalb_{\text{min}}$ for $\Fcalb_{\text{min}}$ is about $\Delta\Fcalb_{\text{min}} = 0.13\%$ of $\Fcalb_{\text{min}}$.}
	\label{TableAdFidelThermo}
	\begin{tabular}{cccccc}
		\multicolumn{1}{l|}{$\gamma_{0}$}         & 314 Hz    & 628 Hz    & 1257 Hz   & 3142 Hz   & 6283 Hz \\ \hline
		\multicolumn{1}{l|}{$\Fcalb_{\text{min}}$} & 0.9971(3) & 0.9965(4) & 0.9980(7) & 0.9952(8) & 0.9942(9)
	\end{tabular}
\end{table}

Despite we have provided a specific open-system adiabatic evolution, we can determine infinite classes
of system-environment interactions exhibiting the same amount of heat exchange $\dbar Q$. In particular, there are infinite engineered environments 
that are able to extract a maximum heat amount $Q_{\text{max}}$. In fact, it is possible to determine infinite classes of systems exhibiting the same amount of heat exchange $\dbar Q$. This is provided in Theorem 1 below.

%%%%%%%%%%%%%%%%%%%%%
\begin{theo}
	Let ${\cal S}$ be an open quantum system governed by a time-local master equation in the form
	$\dot{\rho}(t) = \Hcalb [\rho (t)] + \Rcalb_{t}  [\rho (t)]$, where $\Hcalb[\bullet] = (1/i\hbar) [H,\bullet]$ and
	$\Rcalb_{t}[\bullet]=\sum \nolimits _{n} \gamma_{n} (t) [ \Gamma_{n}(t) \bullet \Gamma^{\dagger}_{n}(t) - (1/2)\{  \Gamma^{\dagger}_{n}(t) \Gamma_{n}(t),\bullet\} ]$.
	The Hamiltonian $H$ is taken as a constant operator so that no work is realized by/on the system.
	Assume that the heat exchange between ${\cal S}$ and its reservoir during the quantum evolution is given by $\dbar Q$.
	Then, any unitarily related adiabatic dynamics driven by $\dot{\rho}^{\prime}(t) = \Hcalb^{\prime} [\rho^{\prime} (t)] + \Rcalb^{\prime}_{t}  [\rho^{\prime} (t)]$,
	where $\dot{\rho}^\prime (t) = U \dot{\rho}(t) U^\dagger$, $\Hcalb^{\prime}[\bullet] = U\Hcalb[\bullet]U^{\dagger}$ and $\Rcalb^{\prime}_{t}  [ \bullet] = U \Rcalb_{t} [ \bullet]U^{\dagger}$, for some constant unitary $U$,
	implies in an equivalent heat exchange $\dbar Q^\prime = \dbar Q$.
\end{theo}

%%%%%%%%%%%%%%%%%%%% Proof
\begin{proof}
	Let us consider that $\rho(t)$ is solution of
	\begin{equation}
		\dot{\rho}(t) = \Hcalb [\rho (t)] + \Rcalb_{t}  [\rho (t)] \text{ , }
	\end{equation}
	so, by multiplying both sides of the above equation by $U$ (on the left-hand-side) and $U^{\dagger}$ (on the right-hand-side), we get
	\begin{align}
		U\dot{\rho}(t)U^{\dagger} &= U\Hcalb [\rho (t)]U^{\dagger} + U\Rcalb_{t}  [\rho (t)]U^{\dagger} \nonumber \\
		&= \frac{1}{i\hbar}U[H,\rho (t)]U^{\dagger} + \sum \nolimits _{n} \gamma_{n} (t) U\Gamma_{n}(t)\rho (t) \Gamma^{\dagger}_{n}(t)U^{\dagger} - \frac{1}{2}\sum \nolimits _{n} \gamma_{n} (t)U\{  \Gamma^{\dagger}_{n}(t) \Gamma_{n}(t),\rho (t)\} ]U^{\dagger} \text{ , }
	\end{align}
	thus, by using the relations $[UAU^{\dagger},UBU^{\dagger}] = U[A,B]U^{\dagger}$ and $\{UAU^{\dagger},UBU^{\dagger}\} = U\{A,B\}U^{\dagger}$, we find
	\begin{equation}
		\dot{\rho}^{\prime}(t) = \frac{1}{i\hbar}[UHU^{\dagger},\rho^{\prime} (t)] + \sum \nolimits _{n} \gamma_{n} (t) \Gamma^{\prime}_{n}(t)\rho^{\prime} (t) \Gamma^{\prime\dagger}_{n}(t) - \frac{1}{2}\sum \nolimits _{n} \gamma_{n} (t)\{  \Gamma^{\prime\dagger}_{n}(t) \Gamma^{\prime}_{n}(t),\rho^{\prime} (t)\}  \text{ , }
	\end{equation}
	where $\Gamma^{\prime} (t) = U\Gamma_{n}(t)U^{\dagger}$. In conclusion, we get that $\rho^{\prime}(t) = U\rho (t)U^{\dagger}$ is a solution of
	\begin{equation}
		\dot{\rho}^{\prime}(t) = \Hcalb^{\prime} [\rho^{\prime} (t)] + \Rcalb^{\prime}_{t}  [\rho^{\prime} (t)] \text{ , }
	\end{equation}
	where
	\begin{align}
		\Hcalb^{\prime}[\bullet] &= \frac{1}{i\hbar}[UHU^{\dagger},\bullet] = U\Hcalb[\bullet]U^{\dagger}, \\
		\Rcalb_{t}^{\prime}[\bullet] &= \sum \nolimits _{n} \gamma_{n} (t) [\Gamma^{\prime}_{n}(t)\rho^{\prime} (t) \Gamma^{\prime\dagger}_{n}(t) - \frac{1}{2}\{\Gamma^{\prime\dagger}_{n}(t) \Gamma^{\prime}_{n}(t),\rho^{\prime} (t)\} ] = U\Rcalb_{t}[\bullet]U^{\dagger} \text{ . }
	\end{align}
	Now, by taking into account that the Hamiltonian $H$ is a constant operator, we have that no work is realized by/on the system.
	Then, by computing the amount of heat extracted from the system in the prime dynamics during an interval $t \in [0,\tau]$, we obtain
	\begin{equation}
		Q^{\prime} = \trs{H^{\prime}\rho^{\prime}(\tau)} - \trs{H^{\prime}\rho^{\prime}(0)} \text{ , }
	\end{equation}
	where, by definition, we can use $\rho^{\prime}(t) = U\rho(t)U^{\dagger}$, $ \forall t \in [0,\tau]$. Hence
	\begin{equation}
		Q^{\prime} = \trs{H^{\prime}U\rho(\tau)U^{\dagger}} - \trs{H^{\prime}U\rho(0)U^{\dagger}} = \trs{U^{\dagger}H^{\prime}U\rho(\tau)} - \trs{U^{\dagger}H^{\prime}U\rho(0)} = Q  \text{ , }
	\end{equation}
	where we have used the cyclical property of the trace and that $Q = \trs{H\rho(\tau)} - \trs{H\rho(0)}$.
\end{proof}

As an example of application of the above theorem, let us consider a system-reservoir interaction governed by $\Rcalb^{\text{x}}_{t} [\bullet] = \gamma (t) \left[ \sigma_{x} \bullet \sigma_{x} - \bullet \right]$ (bit-flip channel). We can then show that the results previously obtained for dephasing can be reproduced if the quantum system is initially prepared in thermal state of $H_{\text{y}}^{0}=\omega\sigma_{y}$. Such a result is clear if we choose $U = R_{\text{x}}(\pi/2)R_{\text{z}}(\pi/2)$. Then, it follows that $\Rcalb^{\text{x}}_{t}  [ \bullet] = U\Rcalb^{\text{z}}_{t}[ \bullet]U^{\dagger}$ and $\Hcalb^{\prime}[\bullet] = U\Hcalb[\bullet]U^{\dagger}$, where $R_{\text{z(x)}}(\theta)$ are rotation matrices with angle $\theta$ around $z(x)$-axes for the case of a single qubit. Thus, the above theorem assures that the maximum exchanged heat will be $Q_{\text{max}} = \hbar \omega \tanh[\beta \hbar \omega]$.

\section{Conclusions of this chapter}

In this chapter we studied foundations on adiabatic dynamics of open system, where we could to apply our results in context of quantum thermodynamics. First, we have shown that, in completeness with the results present in literature~\cite{Sarandy:05-1,Sarandy:05-2}, we have shown a class of quantum systems where the adiabatic approximation for open systems holds in the asymptotic time limit $t\rightarrow \infty$. However, such result is an arbitrarily valid consequence of the one-dimensional Jordan decomposition of the Lindblad superoperator that drives the system, the absence of level crossings, and the initialization of the system as a superposition of only two eigenstates of such operator. As an application, we studied both theoretical and experimental implementation of an adiabatic version of the Deutsch algorithm, analyzing its run-time optimality in open systems. In contrast with the general picture previously derived in the literature~\cite{Sarandy:05-1,Sarandy:05-2}, we have shown that the adiabatic approximation for open systems holds in the asymptotic time limit $t\rightarrow \infty$ in the Deutsch algorithm. However, it is worth mentioning that, if we are interested in maximizing the probability of measuring the target (pure) state as the outcome of the algorithm, instead of the open system adiabatic mixed-state density operator, then there is an \textit{optimal time window} for the system measurement (for previously related discussions on this topic, 
see also Refs.~\cite{Steffen:03,Sarandy:05-2,Albash:15,Keck:17}). Concerning the experimental results, we have reported, to the best of our knowledge, the first experimental investigation of the adiabatic approximation in a fully controllable open system, where the total evolution time and the decoherence rates can be freely set to verify the quantitative validity conditions for the adiabatic behavior. By using a hyperfine qubit encoded in ground-state energy levels of a trapped Ytterbium ion, all of the theoretical predictions highlighted above have been successfully realized. 

As an application, we studied the thermodynamics of adiabatic evolutions in context of open systems. From the general approach for adiabaticity in open quantum systems discussed here, we provided a relationship between adiabaticity in quantum mechanics and 
in quantum thermodynamics. In particular, we derived a sufficient condition for which the adiabatic dynamics in open quantum systems leads to adiabatic processes in thermodynamics. By using a particular example of a single qubit undergoing an open-system adiabatic evolution path, we have illustrated the existence of both adiabatic and diabatic regimes in quantum thermodynamics, computing the associated heat fluxes in the processes. As a further result, we also proved the existence of an infinite family of decohering systems exhibiting the same maximum heat exchange. From the experimental side, we have realized adiabatic open-system evolutions using an Ytterbium trapped ion, with its hyperfine energy level encoding a qubit (work substance). In particular, heat exchange and entropy production can be optimized along the adiabatic path as a function of the total evolution time. Our implementation exhibits high controllability, opening perspectives for analyzing thermal machines (or refrigerators) in open quantum systems under adiabatic evolutions. The associated effects of the engineered reservoirs on the thermal efficiencies are left for future research.

Then, the results presented in this chapter provide a framework to exploit the experimental realization of adiabaticity in quantum computing under decoherence, tackling features such as optimal run-time of an algorithm, asymptotic time behavior, competition between adiabatic and relaxation time-scales, outcome fidelities, thermodynamic processes via adiabatic dynamics, among others.

%%%%%%%%%%%%%%%%%%%%%%%%%%%%%%%%%%%%%%%%%%%%%%%%
%%%%%%%%%%%%%%%% TQD EVOLUTIONS %%%%%%%%%%%%%%%%
%%%%%%%%%%%%%%%%%%%%%%%%%%%%%%%%%%%%%%%%%%%%%%%%
\chapter{Transitionless Quantum Driving in closed systems} \label{Chapter-TQDCS}

\initial{A}diabatic dynamics is a useful approach to achieve a number of quantum tasks. However, as already previously highlighted, the constraints on the speed evolution is a problem when we consider the undesired effects. As a strategy to bypass these undesired effects in adiabatic dynamics, Mustafa Demirplak and Stuart A. Rice proposed a special dynamics that allows us to mimic an adiabatic dynamics without the requirement of slowly driving the system. The problem solved by Demirplak and Rice in 2003 is how to speed up an adiabatic dynamics~\cite{Demirplak:03}. To this end, Demirplak and Rice showed how to use an auxiliary field to inhibit natural transitions between energy level of the system when we implement fast evolutions. This additional field constitutes a \textit{counter diabatic field paradigm}, as mentioned in their revised paper, published two years later~\cite{Demirplak:05}. In this chapter we will review the method developed by Demirplak and Rice, revisited some years later by Michael Victor Berry~\cite{Berry:09}. As a contribution of this thesis to this field, we show how to generalize such method and we highlight the benefits we can get from this generalized approach. In recent years, the term ``counter diabatic'' has been replaced by the terminology \textit{transitionless quantum driving} (TQD) used by Berry~\cite{Berry:09}. Here we will use the Berry's terminology and the term ``counter diabatic'' will be used whenever it is convenient for us. The results presented in this chapter are associated with the following works:

$\bullet$ A. C. Santos and M. S. Sarandy, ``Generalized shortcuts to adiabaticity and enhanced robustness against decoherence'', J. Phys. A: Math. Theor. \textbf{51}, 025301 (2018).

$\bullet$ C.-K. Hu, J.-M. Cui, A. C. Santos, Y.-F. Huang, M. S. Sarandy, C.-F. Li, and G.-C. Guo, ``Experimental implementation of generalized transitionless quantum driving'', Opt. Lett. \textbf{43}, 3136 (2018).

$\bullet$ A. C. Santos, ``Quantum gates by inverse engineering of a Hamiltonian'', J. Phys. B: At. Mol. Opt. Phys. \textbf{51}, 015501 (2018).

$\bullet$ A. C. Santos, A. Nicotina, A. M. Souza, R. S. Sarthour, I. S. Oliveira, and M. S. Sarandy,
``Optimizing NMR quantum information processing via generalized transitionless quantum
driving'', EPL (Europhysics Letters) \textbf{129}, 30008 (2020).

\section{Standard TQD in closed quantum systems} \label{Sec-T-TQDCS}

Consider $H_{0}(t)$ the driven reference (adiabatic) Hamiltonian of a quantum system with dimension $D$, $\ket{E_{n}(t)}$ the set of instantaneous eigenstates of $H(t)$ with eigenvalues $E_{n}(t)$. The solution for the dynamics of the system that evolves through an adiabatic path (trajectory) is given by
\begin{equation}
	\ket{\psi_{\text{ad}} (t)} = \sum_{n=1}^{D} c_{n} \exp \left[-\frac{i}{\hbar} \int_{0}^{t} E_n(t) + i \hbar \interpro{\dot{E}_{n}(t)}{E_{n}(t)} dt \right] \ket{E_{n}(t)} \text{ , }
\end{equation}
with the coefficients $c_{i}$ defined by the initial state of the system $\ket{\psi(0)} = \sum\nolimits_{n}^{D} c_{n} \ket{E_{n}(0)}$. By convenience in the notation, from now on we will define the set of \textit{adiabatic phase} $\theta^{\text{ad}}_{n}(t)$ as quantal phase given by $\theta^{\text{ad}}_{n}(t)=\theta^{\text{din}}_{n}(t)+\theta^{\text{geo}}_{n}(t)$, where $\theta^{\text{din}}_{n}(t)$ and $\theta^{\text{geo}}_{n}(t)$ are the dynamical and geometrical phases defined as~\cite{Berry:84}
\begin{equation}
	\theta^{\text{din}}_{n}(t) = -\frac{1}{\hbar} E_{n}(t)   \text{ \ \ \ and \ \ \ }
	\theta^{\text{geo}}_{n}(t) =-i \interpro{\dot{E}_{n}(t)}{E_{n}(t)} \text{ . } \label{EqThetaDinGeo}
\end{equation}

In this way, it is straightforward to define the evolution operator $U_{\text{ad}}(t)$ of an adiabatic dynamics as
\begin{equation}
	U_{\text{ad}}(t) = \sum_{n=1}^{D} e^{i \int_{0}^{t} \theta^{\text{ad}}_{n}(t)dt} \ket{E_{n}(t)}\bra{E_{n}(0)} \text{ , }
\end{equation}
because we can check that $\ket{\psi_{\text{ad}} (t)} = U_{\text{ad}}(t)\ket{\psi(0)}$. From the above operator $U_{\text{ad}}(t)$, Demirplak and Rice raised the question whether there is another Hamiltonian (not the adiabatic) that shares of the same evolution operator. Since we know the evolution operator of a given dynamics, we can use an inverse engineering approach to determinate the Hamiltonian that implements such dynamics. The Hamiltonian that mimics the adiabatic dynamics can be obtained from inverse engineering using that
\begin{align}
	H_{\text{tqd}}^{\text{std}}(t) &= -i\hbar U_{\text{ad}}\left( t\right) \dot{U}_{\text{ad}}^{\dag }\left(t\right) 
	= \sum_{n=1}^{D} \sum_{m=1}^{D}e^{i \int_{0}^{t}\theta^{\text{ad}}_{n}(\xi)d\xi} \ket{E_{n}(t)}\interpro{E_{n}(0)}{E_{m}(0)}
	\frac{d}{dt}\left[ e^{-i \int_{0}^{t}\theta^{\text{ad}}_{n}(\xi)d\xi} \bra{E_{m}(t)} \right] \nonumber \\
	&= H_{0}(t) + i \hbar \sum_{n}^{D} \ket{\dot{E}_{n}(t)} \bra{E_{n}(t)} + \interpro{\dot{E}_{n}(t)}{E_{n}(t)} 
	\ket{E_{n}(t)} \bra{E_{n}(t)} \text{ , } \label{EqHTQDstd}
\end{align}
where we already used the Eqs.~\eqref{EqThetaDinGeo} and $H_{0}(t) = \sum_{n} E_{n}(t) \ket{E_{n}(t)} \bra{E_{n}(t)}$, the spectral decomposition of $H_{0}(t)$. A first point to be highlighted here is the absence of any condition on the total evolution time, so a system driven by the above Hamiltonian follows an adiabatic path framed by the reference Hamiltonian $H_{0}(t)$ for any instant $t$. More rigorously, we showed that naturally the speed of the evolution is constrained by the energy cost of the implementation, with faster 
evolutions being more energy demanding~\cite{Santos:15}. Therefore, note that this approach requires our ability to include external fields to the original field used to implement $H_{0}(t)$. These additional fields work as inhibitors of diabatic transitions between energy levels of the system. For this reason, these fields has been referred as \textit{counter diabatic} fields and the additional term in $H_{\text{tqd}}^{\text{std}}(t)$ as \textit{counter diabatic Hamiltonian}. Thus, we write $H_{\text{tqd}}^{\text{std}}(t) = H_{0}(t) + H_{\text{cd}}(t)$, with
\begin{equation}
	H_{\text{cd}}(t) = i \hbar \sum_{n}^{D} \ket{\dot{E}_{n}(t)} \bra{E_{n}(t)} + \interpro{\dot{E}_{n}(t)}{E_{n}(t)} 
	\ket{E_{n}(t)} \bra{E_{n}(t)} \text{ . } \label{EqHcdStd}
\end{equation}

Here we are using the label ``std'' in $H_{\text{tqd}}^{\text{std}}(t)$ to denote the Hamiltonian that implements a TQD as provided by the ``standard" model developed by Demirplak and Rice. It is convenient because we shall see how to generalize this approach, in which this approach is recovered as particular case of the generalized approach.

The standard TQD method has been investigated in a number of different applications, where different discussions on the energetic cost of implementing such evolution has been considered. For example, studies on energetic cost of standard TQD were considered through the Frobenius norm
of the Hamiltonian~\cite{Santos:15,Coulamy:16,Santos:16}, quantum speed limit inequalities~\cite{Campbell-Deffner:17} and average thermodynamics work~\cite{Adolfo:14}. In all of these studies, standard TQD is more costing than its adiabatic counterpart and we can say that speed up adiabatic processes through this method demands additional energy resources. However, in this section we will present how to provide more efficient (energetically) TQD Hamiltonian by exploring the Gauge freedom of the phases that accompanying the adiabatic and, consequently, the TQD evolution.

\section{Generalized TQD in closed quantum systems} \label{Sec-G-TQDCS}

In some protocols of inverse engineering of a Hamiltonian~\cite{Kang:16,Santos:18-a}, the starting point is the equation
\begin{equation}
	H(t) = -i \hbar U(t) \dot{U}^{\dagger}(t) \text{ , }
\end{equation}
with $U(t)$ being the operator that describes the desired evolution. Now, observe that there are a number of adiabatic processes in which the quantal phases collected along the evolution can be negligible, situations where we are interested in the adiabatic trajectory associated to a single eigenstate of $H_{0}(t)$. For this reason, we can use the quantal phases that accompanying a TQD as a free parameter and define the generalized TQD evolution operator given by
\begin{equation}
	U_{\text{tqd}}(t) = \sum_{n=1}^{D} e^{i \int_{0}^{t}\theta _{n}(\xi)d\xi} \ket{E_{n}(t)}\bra{E_{n}(0)} \text{ , }
\end{equation}
with $\theta _{n}(t)$ being arbitrary parameters, that can be adjusted to minimize some physical quantity, for example.
The TQD Hamiltonian considered in previous section is a particular Hamiltonian that can be found from above equation. In fact, by adopting the choice $\theta _{n}(t) = \theta^{\text{ad}} _{n}(t)$, we get $U_{\text{tqd}}(t) = U_{\text{ad}}(t)$. The operator $U_{\text{tqd}}(t)$ implements an generalized TQD, because it allows for a TQD with arbitrary quantal phases, consequently we can find the generalized TQD Hamiltonian as
\begin{align}
	H_{\text{tqd}}(t) &= -i \hbar U_{\text{tqd}}(t) \dot{U}_{\text{tqd}}^{\dagger}(t) = \sum_{n=1}^{D} \sum_{m=1}^{D} 
	e^{i \int_{0}^{t}\theta _{n}(\xi)d\xi} \ket{E_{n}(t)}\interpro{E_{n}(0)}{E_{m}(0)}
	\frac{d}{dt} \left[ e^{i \int_{0}^{t}\theta _{m}(\xi)d\xi} \bra{E_{m}(t)}\right] \nonumber \\
	&= i \hbar \sum_{n=1}^{D} \ket{\dot{E}_{n}(t)}\bra{E_{n}(t)} + i \theta _{n} (t) \ket{E_{n}(t)}\bra{E_{n}(t)} \text{ . }
\end{align}

Now, if we identify the functions $\theta _{n}\left( t\right)$ as it has originally been identified with the adiabatic 
phase $\theta_n(t) = \theta^{\text{ad}}_n(t)$, which 
exactly mimics an adiabatic evolution, then we can write 
$H_{\text{tqd}}(t)$ as $H_0(t) + H_{\text{CD}}(t) = H_{\text{tqd}}^{\text{std}}(t)$. Now, it is relevant to highlight here the absence of the original reference Hamiltonian $H_0(t)$ in generalized TQD Hamiltonian $H_{\text{tqd}}(t)$. This result makes clear that we do not need to know how to implement $H_0(t)$ if we want to drive the system through an adiabatic trajectory of $H_0(t)$ via a TQD method. Moreover, if we define the energy cost of TQD using the Frobenius norm, as done in Refs.~\cite{Santos:15,Santos:16,Coulamy:16}, one notes that $||H_{\text{tqd}}(t)|| \neq ||H_{\text{tqd}}^{\text{std}}(t)||$. Therefore, a question is raised here: \textit{Is the generalized TQD method energetically more efficient than the standard one?}

\subsection{Energetically optimal transitionless quantum driving}

As already mentioned, there is a number of situations for which an adiabatic 
process do not need to be mimicked exactly, but only assure that the system is kept in an instantaneous eigenstate of the original Hamiltonian $H_0(t)$ (independently 
of its associated quantum phase)~\cite{Santos:15,Santos:16,Coulamy:16,Adolfo:16,Stefanatos:14,Lu:14,Deffner:16,Liang:16,An:16,Xia:16,Zhang:16,Marcela:14,Chen:10}. 
This generalized dynamics in terms of arbitrary phases $\theta_n(t)$ will be denoted 
as a {\it transitionless} evolution. Now, we will show that $\theta_n(t)$ can be nontrivially 
optimized in transitionless evolutions both in terms of energy cost and robustness against 
decoherence effects. In this direction, we adopt as a measure of energy cost the average 
Hilbert-Schmidt norm of the Hamiltonian throughout the 
evolution, which is 
given by~\cite{Zheng:16,Santos:15,Campbell-Deffner:17,Nathan:14}
\begin{equation}
	\Sigma  \left( \tau \right) = \frac{1}{\tau }\int_{0}^{\tau }||H(t)|| \text{ }dt = \frac{1}{\tau }\int_{0}^{\tau }\sqrt{\text{Tr}%
		\left[ {H}^{2}\left( t\right) \right] }\text{ }dt \text{ \ ,} \label{EqCostGeneric}
\end{equation}
where $\tau$ denotes the total evolution time.
The energy cost $\Sigma  \left( \tau \right)$ aims at identifying changes in the energy coupling constants and gap structure, 
which typically account for the effort of speeding up adiabatic processes. It is a well-defined measure for 
finite-dimensional Hamiltonians exhibiting non-degeneracies in their energy spectra, so it is applicable, e.g., for  
generic systems composed of a finite number of quantum bits (qubits) under magnetic or electric fields 
(see Refs.~\cite{Nathan:14,Zheng:16} for similar cost measures). More specially, the above energy quantifier can be adequately applied to the Hamiltonians that describe the quantum systems within the scope of this thesis. Note that Eq.~(\ref{EqCostGeneric}) is non-invariant with respect to a change of the zero energy offset, but we can adopt a fixed reference frame where it can be used to quantify the energy cost involved in attempts of accelerating the adiabatic path, either via an increasing of the energy gap in the adiabatic approach or via a reduction of $\tau$ by adjusting the relevant energy couplings in the counter-diabatic theory. In addition, as $\tau$ can be set by the quantum speed limit~\cite{Deffner:13}, Eq.~(\ref{EqCostGeneric}) allows us to establish a trade-off between speed and energy cost for an arbitrary dynamics~\cite{Santos:15,Campbell-Deffner:17}. For instance, in NMR experimental setups, the  quantity $||H(t)||$ represents how intense a magnetic field $\vec{B}(t)$ is expected to be in order to control the speed of such a dynamics. In conclusion, the quantifier defined in Eq.~\eqref{EqCostGeneric} will be used here to measure the energy cost expanded by a generalized TQD with regards to its standard and adiabatic counterparts.

As a first result, we discuss how, for a fixed time $\tau$, the generalized approach of TQD allows us to minimize the energy cost in a TQD by a suitable choice of the arbitrary parameters $\theta_{n}\left( t\right)$. This optimization can be analytically 
derived, which is established by Theorem~\ref{TheoOptmEner} below (its derivation is provided in Appendix~\ref{ApProofTheoOptTQD}).

\begin{theo} \label{TheoOptmEner}
	Consider a closed quantum system under adiabatic evolution governed by a Hamiltonian 
	$H_0\left( t\right) $. The energy cost to implement its generalized transitionless counterpart,  
	driven by the Hamiltonian $H_{\text{tqd}}(t)$, can be minimized by setting  
	\begin{equation}
		\theta _{n}\left( t\right) = \theta ^{\min}_{n}\left( t\right) = -i \langle \dot{E}_{n}(t)|E_{n}(t)\rangle \text{ .}
		\label{OptTeta}
	\end{equation}
\end{theo}

In particular, for any evolution that satisfies the quantum parallel-transport condition, given by $\langle \dot{E}_{n}(t)|E_{n}(t)\rangle = 0$~\cite{Berry:84}, the energy cost to implement a transitionless evolution is 
always optimized by choosing $\theta_{n}^{\text{min}}\left( t\right)=0$. 
This approach is useful for providing both realistic and energetically optimal Hamiltonians in several physical scenarios. 
For example, by considering nuclear spins driven by a magnetic field $\vec{B}(t)$ in a nuclear magnetic resonance setup, the energy 
cost can be optimized by adjusting $\theta_{n}^{\text{min}}\left( t\right)$ such that the magnitude $B(t)$ of the magnetic field is reduced, 
since $||H(t)|| \propto B(t) $.

As a second result, it is possible to proof that  
the generalized TQD approach can be used as a tool to yield time-independent 
Hamiltonians for transitionless evolutions. In general, the Hamiltonian $H_{\text{tqd}}(t)$ has its form 
constrained both by the choice of the phases $\theta_{n} \left( t\right) $ and by eigenstates 
of the adiabatic Hamiltonian $H_0(t)$. Thus, now we delineate under what conditions we can choose 
the set $\{\theta_{n} \left( t\right)\}$ in order to obtain a \textit{time-independent} Hamiltonian for a transitionless evolution.
To answer this question, we impose $\dot{H}_{ \text{tqd} }\left( t\right)=0$ considering 
arbitrary phases $\theta_{n}\left( t\right)$. This leads to the Theorem~\ref{TheoTimeIndep} below 
(its derivation is provided in Appendix~\ref{ApProofTheoOptTQD}).

\begin{theo} \label{TheoTimeIndep}
	Let $H_0\left( t\right) $ be a discrete quantum Hamiltonian, with $\{ \ket{E_{m}(t)} \}$ denoting its 
	set of instantaneous eigenstates. If $\{\ket{E_{m}(t)}\}$ satisfies $\interpro{E_{k}(t)}{\dot{E}_{m}(t)} = c_{km}$, 
	with $c_{km}$ complex constants $\forall$ $k,m$, then a family of time-independent 
	Hamiltonians $H^{\{\theta\}}$ for generalized transitionless evolutions can be defined by setting 
	$\theta_{m}\left( t\right) = \theta$, with $\theta$ a single arbitrary real constant $\forall m$.
\end{theo}

In case where the hypothesis of the Theorem~\ref{TheoTimeIndep} are satisfied and the parallel-transport condition is also satisfied for all eigenstates of $H_{0}(t)$, we can get an energetically optimal time-independent TQD Hamiltonian by setting $\theta_{m}\left( t\right) = \theta_{m}^{\text{min}}\left( t\right) = 0$, as we shall see soon. Although we use the Frobenius norm of the Hamiltonian as measurement of energy cost, we can define different quantifiers and physical quantities to be optimized using these arbitrary quantal phases. Here we will present three applications where different approaches to energy cost optimization are considered.

%For example, in optimal control theory, it is defined an operator $\Ocalb$ to be optimized through the function $J[\psi]$ given by~
%\begin{align}
%J[\psi] = \bra{\psi(t)} \Ocalb \ket{\psi(t)} \text{ . }
%\end{align}

%However, there is a different approach of optimal control related with total evolution time of an evolution. In this case, the above equation is not useful because we need to define the operator $\Ocalb$, but there are studies on a \textit{minimum time function}~. Applications of the 

\subsection{Experimental implementation of Landau-Zener transitionless dynamics}

As a first application, let us consider the dynamics of a 
two-level quantum system, i.e., a qubit, evolving under the Landau-Zener Hamiltonian (here we use the normalized time $s = t/\tau$)
\begin{align}
	H_0^{\text{LZ}}(s) &= \hbar \Delta \sigma_{z} + \hbar \Omega_{0} (s) \sigma_{x} \label{LZAdHam}
\end{align}
where $\Delta$ is a detuning, $\Omega (s)$ is the Rabi frequency and the time $s$ is defined for the interval~$[0,1]$. The adiabatic dynamic is obtained from a very-slow continuum evolution for the Hamiltonian $H_{\text{LZ}}(t)$. The above equation can be rewritten as
\begin{align}
	H_0^{\text{LZ}}\left( s\right) = \hbar \Delta \left[ \sigma _{z}+\tan \vartheta \left( s\right) \sigma _{x}\right] \text{ , }
	\label{H-LZ}
\end{align}
with $\tan \vartheta \left( s\right)$ a dimensionless time-dependent parameter associated with the Rabi frequency, with $\vartheta \left( s\right) = \arctan[\Omega_{0} (s)/\Delta]$. This Hamiltonian describes transitions in two-level systems exhibiting anti-crossings 
in its energy spectrum~\cite{Zener:32}. In particular, it can be applied, e.g., to perform adiabatic population transfer in a two-level system driven by a chirped field~\cite{Demirplak:03} and to investigate molecular collision processes \cite{Lee:79}.
The instantaneous ground $|E_{-}\left( s\right) \rangle$ 
and first excited $|E_{+}\left( s\right) \rangle$ states of 
$H_0^{\text{LZ}}\left( s\right)$ are
\begin{subequations}\label{EqEigenStatesHLZ}
	\begin{align}
		|E_{-}\left( s\right) \rangle  &=-\sin \left[ \frac{\vartheta \left(
			s\right) }{2}\right] |1\rangle +\cos \left[ \frac{\vartheta \left( s\right) 
		}{2}\right] |0\rangle \text{ ,}  \label{FundLZ} \\
		|E_{+}\left( s\right) \rangle  &= \cos \left[ \frac{\vartheta \left(
			s\right) }{2}\right] |1\rangle +\sin \left[ \frac{\vartheta \left( s\right) 
		}{2}\right] |0\rangle\text{ . }  \label{EqExcLZ}
	\end{align}
\end{subequations}

The system is initialized in the ground state $|E_{-}\left( 0\right) \rangle =|0\rangle $ 
of $H_0^{\text{LZ}}\left( 0\right)$. By considering a unitary dynamics and a sufficiently 
large total evolution time (adiabatic time), the qubit evolves to the 
instantaneous ground state $|E_{-}\left( s\right) \rangle$ of $H_0^{\text{LZ}}\left( s\right)$. By applying the standard TQD approach, it is possible to get the TQD Hamiltonian $H_{\text{std}}^{\text{LZ}}(s)$ given by $H_{\text{std}}^{\text{LZ}}(s) = H_0^{\text{LZ}}\left( s\right) + H_{\text{cd}}^{\text{LZ}}\left( s\right)$, where
\begin{align}
	H_{\text{cd}}^{\text{LZ}}\left( s\right) = \frac{1}{\tau}i \hbar \sum_{k=\pm }|d_{s}E_{k}\left( s\right)
	\rangle \langle E_{k}\left( s\right) |=\frac{\hbar d_{s}\vartheta \left( s\right) 
	}{2\tau }\sigma _{y} \text{ , }
	\label{EqExcLZ-std}
\end{align}
because the eigenstates of the Hamiltonian $H_0^{\text{LZ}}\left( s\right)$ the parallel-transport condition. 

Now, let us discuss the generalized transitionless dynamics theory for the Landau-Zener model. 
For optimal energy cost, Eq.~(\ref{OptTeta}) establishes 
$\theta _{n}\left( t\right) =\langle d_{s}E_{\pm }\left( s\right)
|E_{\pm }\left( s\right) \rangle =0$ for the states in Eqs.~\eqref{EqEigenStatesHLZ}. 
Therefore, the optimal Hamiltonian ${H}_{\text{opt}}^{\text{LZ}}\left( s\right) $ is given by 
\begin{align}
	H_{\text{opt}}^{\text{LZ}}\left( s\right) =H_{\text{cd}}^{\text{LZ}}\left( s\right) \text{ . }
	\label{EqExcLZ-opt}
\end{align}

From Eq.~(\ref{EqExcLZ-std}), we can see that $H_0^{\text{LZ}}\left( s\right)$ satisfies the hypotheses of  
Theorem \ref{TheoTimeIndep} if, and only if, we choose the linear interpolation 
$\vartheta (s) = \vartheta _{0}s$. Thus, we adopt this choice 
for simplicity and, consequently, we 
have $H_{\text{opt}}^{\text{LZ}}\left( s\right) = (\vartheta _{0}/2\tau ) \sigma _{y}$.
We observe that a complete (avoided) level crossing picture for the Landau-Zener model 
is described by varying the parameter $\tan \vartheta \left( s\right)$ from $-\infty$ to $+\infty$. 
Here, we are taking a narrower range for $\tan \vartheta \left( s\right)$, which simplifies the 
description of the transitionless dynamics for the model.

Again, the experimental realization was done using the Ytterbium trapped ion setup shown in Fig.~\eqref{FigAdiabExpComTrappedIon}. We use it because we aim to investigate both energy cost and robustness against decoherence of the optimal TQD scheme. In our experiment we used both resonant and off-resonant microwaves fields to implement adiabatic and TQD dynamics.
The adiabatic Hamiltonian dynamics was performed by using a non-resonant microwave,
where the time-dependent effective Rabi frequency is given by $\Omega_{\text{eff}}(s) = [\Omega_{0}^2(s) + \Delta^2]^{1/2}$.
This is also known as ``generalized'' Rabi frequency, i.e, Rabi frequency with detuning~\cite{Knight:80,Mogilevtsev:08,Brunel:98}. To drive the system by using the standard TQD,
we use an independent microwave to simulate the counter-diabatic term $H_{\text{cd}}(s)$.
This additional field is a resonant microwave ($\Delta = 0$),
whose Rabi frequency can be obtained from $H_{\text{cd}}(s)$ in Eq.~\eqref{EqExcLZ-std} as $\Omega_{0\text{cd}}(s) = d_{s}\vartheta \left( s\right) / 2\tau$.
On the other hand, different from standard TQD Hamiltonian,
the optimal dynamics driven by $H_{\text{opt}}(s)$ could be implemented using a single resonant microwave with Rabi frequency $\Omega_{0\text{cd}}(s)$,
where we have turned-off the non-resonant field used for simulate $H_{0}(s)$. 

In order to quantify the energy resources employed in the quantum evolution, we study the intensity of the fields used to perform the adiabatic dynamics as well as standard and optimal TQD.
The field intensity is associated with the Rabi frequency through the relation $I(s) = \Gamma \Omega_{0}^2(s)$, where $\Gamma$ is a constant that depends on the microwave amplifier.
Considering the whole evolution time, we can define the average intensity field as
\begin{align}
	\bar{I} (\tau) = \frac{1}{\tau} \int_{0}^{\tau} I(t) dt = \Gamma \int_{0}^{1} |\Omega (s)|^2 ds \text{ , }
\end{align}
where we use the parametrization $s = t/\tau$. Thus, from Eqs. \eqref{LZAdHam}, \eqref{EqExcLZ-std} and \eqref{EqExcLZ-opt}, we get the field intensity for the adiabatic, standard and optimal TQD Hamiltonian, respectively, as
\begin{subequations}\label{EqIntensityFieldsTrappedIon}
	\begin{align}
		\bar{I}_{0} (\tau) &= \Gamma \int \nolimits _{0}^{1}|\Omega_{\text{0}}(s)|^{2}ds = \Delta^2 \Gamma \int \nolimits _{0}^{1}\tan^2 \vartheta \left( s\right) ds \text{ , } \\
		\bar{I}_{\text{std}} (\tau) &= \bar{I}_{0} (\tau) + \bar{I}_{\text{opt}} (\tau) = \Delta^2 \Gamma \int \nolimits _{0}^{1}\tan^2 \vartheta \left( s\right) ds + \frac{\Gamma}{4\tau^2} \int \nolimits _{0}^{1} | d_{s}\vartheta \left( s\right) |^2 ds \text{ , } \\
		\bar{I}_{\text{opt}} (\tau) &= \Gamma \int \nolimits _{0}^{1}|\Omega_{\text{CD}}(s)|^{2}ds = \frac{\Gamma}{4\tau^2} \int \nolimits _{0}^{1} | d_{s}\vartheta \left( s\right) |^2 ds  \text{ . }
	\end{align}
\end{subequations}

Notice that for the standard TQD we find $\bar{I}_{\text{std}} (\tau) = \bar{I}_{0} (\tau) + \bar{I}_{\text{opt}} (\tau)$, since
the standard TQD field is composed by both the adiabatic and the optimal TQD contributions. The parameter $\Gamma$ is a characteristic parameter that depends on the device used to amplify and modulate the driving microwave acting on the system. Then, in order to give a description of the energy cost independent of any auxiliary device, we take the relative field intensities expressed in unities of the adiabatic intensity $\bar{I}_{0}$. By doing this, it is possible to disregard the constant $\Gamma$.
In fact, we define $\Icalb_{\text{std}} (\tau) = \bar{I}_{\text{std}} (\tau) / \bar{I}_{0}$ and $\Icalb_{\text{opt}} (\tau) = \bar{I}_{\text{opt}} (\tau) / \bar{I}_{0}$, and adopt the normalization $\Icalb_{0} (\tau) = 1$. 

From Eqs.~\eqref{EqIntensityFieldsTrappedIon}, it is possible to identify a value of $\tau_{\text{B}}$, for the total evolution time, in which the intensity fields for implementing optimal TQD becomes less intense than the adiabatic intensity. In fact, firstly note that $\bar{I}_{\text{std}} (\tau) \geq \bar{I}_{\text{opt}} (\tau)$, but on the other side, we can see that the value of $\tau$ implies into $\bar{I}_{\text{opt}} (\tau) \geq \bar{I}_{0} (\tau)$ or $\bar{I}_{\text{opt}} (\tau) \leq \bar{I}_{0} (\tau)$. So, here we impose the desired situation in which $\bar{I}_{\text{opt}} (\tau) \leq \bar{I}_{0} (\tau)$ to obtain the relation
\begin{eqnarray}
	\frac{\Gamma}{4\tau^2} \int \nolimits _{0}^{1} | d_{s}\vartheta \left( s\right) |^2 ds \leq \Delta^2 \Gamma \int \nolimits _{0}^{1}\tan^2 \vartheta \left( s\right) ds \text{ , }
\end{eqnarray}
therefore we conclude that, in this situation, $\tau$ should satisfy
\begin{eqnarray}
	\tau \geq \frac{1}{|\Delta|} \sqrt{\frac{\int \nolimits _{0}^{1} | d_{s}\vartheta \left( s\right) |^2 ds}{\int \nolimits _{0}^{1}\tan^2 \vartheta \left( s\right) ds}} \text{ . }
\end{eqnarray}

This result shows the existence of a boundary total evolution time $\tau_{\text{B}}$ given by
\begin{eqnarray}
	\tau_{\text{B}} = \frac{1}{|\Delta|} \sqrt{\frac{\int \nolimits _{0}^{1} | \vartheta^{\prime} \left( s\right) |^2 ds}{\int \nolimits _{0}^{1}\tan^2 \vartheta \left( s\right) ds}} \text{ , } \label{EqTauBoundary}
\end{eqnarray}
so that if $\tau > \tau_{\text{B}}$ we get $\bar{I}_{\text{Opt}} (\tau) < \bar{I}_{\text{Ad}} (\tau)$, otherwise we find $\bar{I}_{\text{Opt}} (\tau) > \bar{I}_{\text{Ad}} (\tau)$. In case where $\tau = \tau_{\text{B}}$ we have identical field intensities. It would be worth highlighting that $\tau_{\text{B}}$ does not depend on the $\Gamma$. Notice that, after $\tau_{\text{B}}$, the shortcut to adiabaticity defined by the optimal TQD can be implemented by spending less energy resources, as measured by the field intensity, than the adiabatic approach.

\begin{figure}[t!]
	%\input{Figs/FigFidelityTrappedIon-CapIII.plt}
	%\vspace{5.4cm}
	\centering
	\includegraphics[scale=0.6]{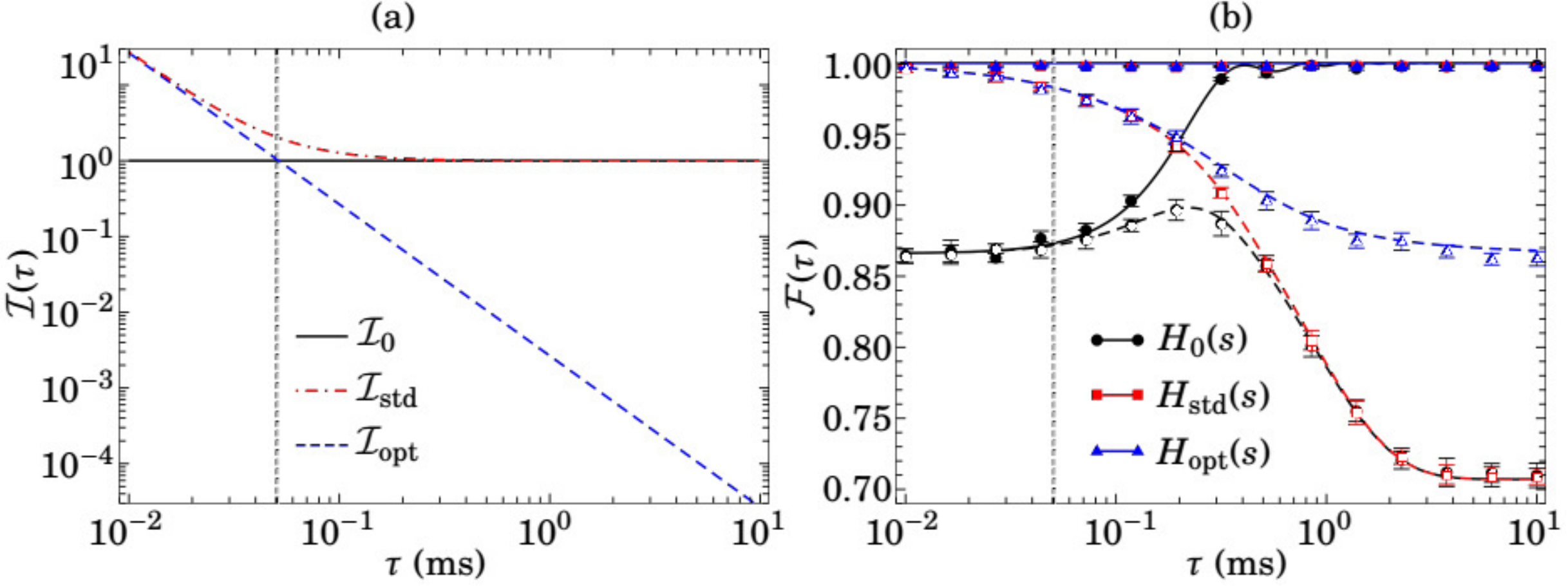}
	\caption{({\color{blue}a}) The calculated relative field intensity $\Icalb (\tau)$ for standard TQD $\Icalb_{\text{std}} (\tau)$ and optimal TQD $\Icalb_{\text{opt}} (\tau)$ as function of the total evolution time $\tau$, where the horizontal black line represents $\Icalb_{0} (\tau)$. ({\color{blue}b}) Fidelity for unitary dynamics (continuum lines) and non-unitary one (dashed lines) under dephasing. The symbols and lines represent experimental data and theoretical results, respectively. In both figures the dashed gray vertical line denotes the boundary time $\tau_{\text{B}} \approx 0.052$~ms, obtained from Eq.~\eqref{EqTauBoundary}, between the regions $\bar{I}_{\text{opt}} (\tau)>\bar{I}_{0} (\tau)$ (left hand side) and $\bar{I}_{\text{opt}} (\tau)<\bar{I}_{0} (\tau)$ (right hand side). Here we set $\vartheta \left( s\right) = \pi s /3 $, $\Delta = 2\pi \times 1$ KHz and $\Omega_{0} \left( s\right) = \Delta \tan (\pi s /3)$ is used in our experiment. The decohering rate was kept as $\gamma = 2.5$ KHz for all non-unitary dynamics. The first gray vertical line denotes the boundary time $\tau_{\text{B}} \approx 0.052$~ms obtained from Eq.~\eqref{EqTauBoundary}.}
	\label{FigFidelityTrappedIon}
\end{figure}

It is important to mention here that the energy cost as computed by the Frobenius norm of the Hamiltonian does not represents a real energy cost. In fact, it is possible to show that the norm of the Hamiltonian $H_0^{\text{LZ}}\left( s\right)$ depends on the parameter $\Delta$, but since of $\Delta$ can changed without change the microwave intensity, the norm of $H_0^{\text{LZ}}\left( s\right)$ cannot be used as a realistic cost of the dynamics in this particular case. However, if we consider this analysis, the Theorem~\ref{TheoOptmEner} guarantees that $\Sigma_{\text{std}} (\tau) > \Sigma_{\text{opt}} (\tau)$ for every $\tau$. Thus, because we spend less energy with the optimal TQD scheme, it is convenient to study the impact of this optimal TQD on the robustness of the dynamics against decoherence. To do a direct comparison between adiabatic dynamics, standard and optimal TQD, we implemented each dynamics evolving under a phase damping channel described by the non-unitary contribution given in Eq.~\eqref{EqRcalbDephasing}. 

In the experiment we consider a constant dephasing rate as where $\gamma = 2.5$~KHz is the dephasing rate and a linear interpolation for the function $\vartheta \left( s\right) = \vartheta_{0} s $ (where we get a time-independent optimal TQD Hamiltonian), with $\vartheta_{0} = \pi /3$. Therefore, we have $\Omega_{0} \left( s\right) = \Delta \tan (\pi s /3)$, where we pick the detuning parameter $\Delta$ as $2\pi \times 2$~KHz. As previously, here we quantify the robustness of the protocols from the fidelity defined in Eq.~\eqref{EqFidelityGeneral}. However, since we have a pure state as our target state $\ket{\psi_{\text{tg}}(s)}$, it is possible to proof that the Eq.~\eqref{EqFidelityGeneral} becomes~\cite{Nielsen:Book}
\begin{align}
	\Fcalb (\tau) = \sqrt{\bra{\psi_{\text{tg}}(s)}\rho\left( s\right) \ket{\psi_{\text{tg}}(s)}} \text{ , }
\end{align}
with $\rho\left( s\right)$ being solution of Schrödinger equation, obtained by numerical calculations, and $\ket{\psi_{\text{tg}}(s)}$ given by Eq.~\eqref{FundLZ}. The result for the fidelity $\Fcalb (\tau)$ as a function of $\tau$ for both free-decohering and decohering dynamics are shown in Fig.~\ref{FigFidelityTrappedIon}{\color{blue}b}.

The continuum lines, with full symbols denoting the associated experimental data, in Fig.~\ref{FigFidelityTrappedIon}{\color{blue}b} denotes the values of $\Fcalb (\tau)$ under a unitary dynamics for each protocol and the corresponding experimental data points. This is possible because the the maximum total evolution time is $\tau=10$~ms, while the coherence time is of order of $200$~ms for the qubit used here. As expected, the black curve shows that the adiabatic behavior is achieved for long total evolution time $\tau \gg \tau_{0} \approx 0.021$~ms, with $\tau_{0}$ computed from $\tau_{0} = \max_{s\in [0,1]}|\bra{E_{-}(s)} d_{s} H_{0}(s)\ket{E_{+}(s)}/g^{2}(s)|$, where $g(s) = E_{+}(s)-E_{-}(s)$ is the gap between fundamental $\ket{E_{-}(s)}$ and excited $\ket{E_{+}(s)}$ energy levels~\cite{Sarandy:04}. It is important highlight here that, for fast evolutions ($\tau < \tau_{0}$), both optimal and traditional TQD provide a high fidelity protocol, while adiabatic dynamics fails. Naturally, this high performance is accompanied by costly fields used for implementing the TQD protocols, as previously shown in Fig.~\ref{FigFidelityTrappedIon}{\color{blue}b}. On the other hand, it is important to highlight that for total evolution time $\tau > \tau_{\text{B}}$, where the optimal TQD field is smaller than both adiabatic and traditional TQD ones, the high performance of optimal TQD is kept.

By looking now at the effect of the dephasing channel (dashed curves, with empty symbols denoting their respective experimental data) from Fig.~\ref{FigFidelityTrappedIon}, we can see that optimal TQD exhibits the highest robustness for every $\tau$, with its advantage increasing as $\tau$ increases.
For fast evolution times, this high performance is again associated with costly fields in comparison with the
adiabatic fields.
However, for large evolution times, while the traditional TQD fidelity converges to adiabatic fidelity (as expected, because the counter-diabatic Hamiltonian becomes less effective than the adiabatic one~\cite{Coulamy:16}), the fidelity behavior of optimal TQD is much better than both adiabatic dynamics
and traditional TQD, with less energy resources spent in the process. This is a remarkable result,
since we can achieve both better fidelities with less energy fields involved. Table \ref{TableCompStdOptTQD} summarizes fidelity and intensity fields for the TQD models considered here.
We highlight that $\Fcalb_{\text{opt}}(\tau) \gtrapprox \Fcalb_{\text{std}}(\tau)$ and there is a range of $\tau$,
where we get $\Fcalb_{\text{opt}}(\tau) > \Fcalb_{0}(\tau)$ even when $\bar{I}_{\text{opt}}(\tau) \ll \bar{I}_{0}(\tau)$. In conclusion, the generalized TQD approach allows us for getting, at least to this case, time-independent TQD Hamiltoniansthat are more robust and less energetically costing than its adiabatic and standard counterparts.

\begin{table}[t!]
	\centering
	\caption{Intensity fields and fidelity relations for both standard and optimal TQD.}
	\label{TableCompStdOptTQD}
	\begin{tabular}{c|c|c|c}
		\hline
		Model                             & $\tau < \tau_{\text{B}}$     &   $\tau = \tau_{\text{B}}$   &  $\tau > \tau_{\text{B}}$      \\ \hline
		\multicolumn{1}{c|}{Standard} & \begin{tabular}[c]{@{}c@{}}$\Fcalb_{\text{std}}(\tau) > \Fcalb_{0}(\tau)$\\ $\bar{I}_{\text{std}}(\tau) > \bar{I}_{0}(\tau)$\end{tabular} & \begin{tabular}[c]{@{}c@{}}$\Fcalb_{\text{std}}(\tau) > \Fcalb_{0}(\tau)$\\ $\bar{I}_{\text{std}}(\tau) > \bar{I}_{0}(\tau)$\end{tabular} & \begin{tabular}[c]{@{}c@{}}$\Fcalb_{\text{std}}(\tau) \gtrapprox \Fcalb_{0}(\tau)$\\ $\bar{I}_{\text{std}}(\tau) \gtrapprox \bar{I}_{0}(\tau)$\end{tabular} \\ \hline
		\multicolumn{1}{c|}{Optimal}  & \begin{tabular}[c]{@{}c@{}}$\Fcalb_{\text{opt}}(\tau) > \Fcalb_{0}(\tau)$\\ $\bar{I}_{\text{opt}}(\tau) > \bar{I}_{0}(\tau)$\end{tabular} & \begin{tabular}[c]{@{}c@{}}$\Fcalb_{\text{opt}}(\tau) > \Fcalb_{0}(\tau)$\\ $\bar{I}_{\text{opt}}(\tau) = \bar{I}_{0}(\tau)$\end{tabular} & \begin{tabular}[c]{@{}c@{}}$\Fcalb_{\text{opt}}(\tau) > \Fcalb_{0}(\tau)$\\ $\bar{I}_{\text{opt}}(\tau) < \bar{I}_{0}(\tau)$\end{tabular} \\ \hline
	\end{tabular}
\end{table}

\section{Generalized TQD for a single-spin in a rotating magnetic field}

As a second application, we study optimal transitionless evolution of a single spin-$\frac{1}{2}$ particle in presence of a rotating magnetic field~\cite{Santos:20b}, whose Hamiltonian $H_{0}(t)$ is given by the Eq.~\eqref{EqNMRHamil}. Here we do not assume any restriction on the parameters $\omega_{0}$ and $\omega_{1}$, as done previously. As we know, the set of eigenvectors for is given in Eqs.~\eqref{EqNMREigenVectors}, then one can show that the generalized TQD Hamiltonian reads
\begin{align}
	H^{\text{gen}}_{\text{tqd}}(t) &= \frac{\hbar}{2} \theta(t)\sin \alpha (\sigma_{x}\cos \omega t + \sigma_{y}\sin \omega t )  +\frac{\hbar}{2} \left[ \omega + \theta(t)\cos \alpha \right]\sigma_{z} \text{ , }
\end{align}
where $\theta(t) = \theta_{0}(t) - \theta_{1}(t)$ and $\alpha = \arctan[ \omega_{1}/\omega_{0} ]$, with $\theta_{0}(t)$ and $\theta_{1}(t)$ being the generalized quantal phases (arbitrary). Thus, note that a transitionless quantum 
driving is not obtained by changing the angular frequency of the RF field, but rather by changing the intensity of the magnetic fields. 
In fact, as in Hamiltonian $H^{\text{nmr}}_{0}(t)$, we can write the magnetic field $\vec{B}^{\text{gen}}_{\text{tqd}}(t)$ associated with the Hamiltonian $H^{\text{gen}}_{\text{tqd}}(t) = -\gamma \vec{S}\cdot \vec{B}^{\text{gen}}_{\text{tqd}}(t)$ as
\begin{align}
	\vec{B}^{\text{gen}}_{\text{tqd}}(t) = \left[ B_{\omega} + B_{\theta}(t)\cos \alpha \right] \hat{z} + B_{\theta}(t)\sin \alpha (\cos \omega t \hat{x}+\sin \omega t \hat{y}) \text{ , } \label{GenBSA}
\end{align}
where $B_{\theta}(t) = - \theta(t)/\gamma$ is associated with the free parameter $\theta(t)$ and $B_{\omega} = - \omega/\gamma$. Therefore, the generalized transitionless theory allows us to use the phases $\theta_{n}(t)$ ($n=0,1$) to minimize the magnetic field required to implement the desired dynamics. In some cases, such a field is fixed and no free parameter can be used to optimize the field intensity. For example, if we set the parameters $\theta_{n}(t)$ in order to obtain the exact adiabatic dynamics with adiabatic phases $\theta^{\text{ad}}_{n}(t)$, we find the standard shortcut to adiabaticity, which is provided by the Hamiltonian in Eq.~\eqref{EqHTQDstd}, with the counter-diabatic term $H_{\text{cd}}(t)$ reading
\begin{align}
	H_{\text{cd}}(t) = \frac{\hbar \omega}{2} \left[ \sin^2\alpha \sigma_{z} - \frac{1}{2}\sin 2\alpha (\sigma_{x}\cos \omega t + \sigma_{y}\sin \omega t ) \right] \text{ , } \label{cc}
\end{align}
leading to the standard quantum driving Hamiltonian in the form $H^{\text{std}}_{\text{tqd}}(t) =  -\gamma \vec{S} \cdot \vec{B}^{\text{std}}_{\text{tqd}}(t)$, with the magnetic field $\vec{B}^{\text{std}}_{\text{tqd}}(t)$ to implement $H^{\text{std}}_{\text{tqd}}(t)$ reading
\begin{align}
	\vec{B}^{\text{std}}_{\text{tqd}}(t) &= \left[B_{xy} - \frac{1}{2}B_{\omega}\sin 2\alpha \right](\cos \omega t \hat{x}+\sin \omega t \hat{y}) + \left[ B_{z} + B_{\omega}\sin^2\alpha \right] \hat{z} \text{ , }
	\label{eqBTtqd}
\end{align}
where $B_{\omega}$ is the additional magnetic field.
By analyzing the energy resources to implement the desired evolution, we can find the optimal protocol for the transitionless dynamics in terms of 
the magnetic field intensities required by the Hamiltonian. In particular, it is possible to show that the optimal field is obtained by setting 
$\theta_{n}(t) = i\interpro{\dot{E}_{n}(t)}{E_{n}(t)}$, as established in Theorem~\ref{TheoOptmEner}. Therefore, the optimal transitionless counterpart 
reads
\begin{eqnarray}
	H^{\text{opt}}_{\text{tqd}}(t) =\frac{\hbar \omega}{2} \left[ \sin^2\alpha \sigma_{z} - \frac{1}{2}\sin 2\alpha (\sigma_{x}\cos \omega t + \sigma_{y}\sin \omega t ) \right] \text{ , }  \label{xx}
\end{eqnarray}
that can be written as $H^{\text{opt}}_{\text{tqd}}(t) =- \gamma \vec{S} \cdot \vec{B}^{\text{opt}}_{\text{tqd}}(t) $ with the optimal magnetic $\vec{B}^{\text{opt}}_{\text{tqd}}(t)$ identified as
\begin{eqnarray}
	\vec{B}^{\text{opt}}_{\text{tqd}}(t) = B_{\omega}\sin^2\alpha \hat{z} - \frac{1}{2}B_{\omega}\sin 2\alpha (\cos \omega t \hat{x}+\sin \omega t \hat{y}) \text{ . }
\end{eqnarray}
We can see that the norms of the fields satisfy $||\vec{B}^{\text{std}}_{\text{tqd}}(t)|| > ||\vec{B}_{0}(t)||$ and 
$||\vec{B}^{\text{std}}_{\text{tqd}}(t)|| > ||\vec{B}^{\text{opt}}_{\text{tqd}}(t)||$ for any choice of the set of parameters 
$\Omega = \{\omega,\omega_{0},\omega_{1}\}$. This means that a fast evolution based on the standard shortcut 
exactly mimicking the adiabatic phase, such as given by Eq.~(\ref{EqHTQDstd}), always requires more energy resources than the 
original adiabatic dynamics. On the other hand, the relation between $||\vec{B}^{\text{opt}}_{\text{tqd}}(t)||$ and $||\vec{B}_{0}(t)||$ 
depends on the parameters $\omega$, $\omega_{0}$ and $\omega_{1}$ used in $||\vec{B}_{0}(t)||$. In particular, 
this means that the optimal shortcut to adiabaticity, such as given by Eq.~(\ref{xx}), can be faster while spending even less resources 
than the original adiabatic dynamics. More specifically, this trade-off can be expressed through the norm relation
\begin{eqnarray}
	\frac{||\vec{B}_{0}(t)||}{||\vec{B}^{\text{opt}}_{\text{tqd}}(t)||} = \frac{B_{1}^2 + B_{0}^2}{B_{1}B_{\omega}}  =  \frac{\omega_{1}^2 + \omega_{0}^2}{\omega_{1} \omega} 
	\text{ . } \label{FieldsRel}
\end{eqnarray}
As we can see, for the case $\omega_{1}=\omega_{0}=\omega$, we get $||\vec{B}_{0}(t)|| = 2 ||\vec{B}^{\text{opt}}_{\text{tqd}}(t)||$. This shows 
the advantage of the optimal TQD approach in comparison with its adiabatic counterpart, whose implementation requires a higher magnetic field. It is worth mentioning that, since the experimental implementation of the driven Hamiltonian in NMR is obtained in the rotating frame, the dynamics here is implemented by a phase modulated pulse with constant amplitude applied off resonance. This approach is similar to the adiabatic pulse technique where large swept-frequency offsets are used to make the effective magnetic field, observed in the rotating frame,  approximately collinear with the spin magnetization~\cite{Garwood:01}. In our experiment no swept-frequency has been used. The energy cost is well defined from the transverse RF field used 
to drive the system, since the $Z$ component of the Hamiltonian is implemented by a frequency shift, with no real cost associated with it.

\subsection{Experimental implementation}

In order to demonstrate the advantage of  the TQD approach as a tool in quantum control, we have performed an experimental realization 
using a two-qubit NMR system, namely, the $^1$H and $^{13}$C  spin$-1/2$ nuclei in the Chloroform molecule, with up and down spin states denoted by $\ket{1}$ and $\ket{0}$, respectively~\cite{Sarthour:Book}. The experiment has been realized at room temperature in a Varian 500 MHz spectrometer. More specifically, we have compared the dynamics of the system driven by fields $\vec{B}_{0}(t)$,~$\vec{B}^{\text{std}}_{\text{tqd}}(t)$ and~$\vec{B}^{\text{opt}}_{\text{tqd}}(t)$. Since we need a single qubit, the $^{1}$H spin has been driven by phase modulated magnetic fields, while the $^{13}$C spin has been decoupled during the experiment. In the rotating frame, the Hamiltonian associated with the magnetic pulse is written as
\begin{eqnarray}
	H_{\text{rf}} = \frac{\hbar}{2} \left[ \omega_{0}\sigma_{z} + \omega_{1} (\sigma_{x}\cos \phi + \sigma_{y}\sin \phi ) \right] \text{ . } \label{H02}
\end{eqnarray} 
By changing the amplitude of the pulses, the offset frequency and using time phase modulation, one can reproduce the desired dynamics. The evolution has been implemented in an on-resonance condition $\omega_z\!=\!\omega_{xy}\!=\!\omega\!=\!2\pi \times 200$Hz. This means that, even though the evolution obeys the adiabatic condition, it can lead to population transfer, resulting in a violation of the expected adiabatic evolution as a resonance effect~\cite{Suter:08}. From the values set for $\omega$, $\omega_{z}$ and $\omega_{xy}$, Eq.~\eqref{FieldsRel} implies that $||\vec{B}_{0}(t)|| = 2||\vec{B}^{\text{opt}}_{\text{tqd}}(t)||$, so that $||\vec{B}^{\text{opt}}_{\text{tqd}}(t)||\!<\!||\vec{B}_{0}(t)||\!<\!||\vec{B}^{\text{std}}_{\text{tqd}}(t)||$. Notice that the optimal TQD indeed spends less energy resources as measured by the strength of the magnetic field applied. The system is then initially prepared in the ground state of $H_{0}(0)$, given by $\ket{\psi(0)} = \cos\alpha \ket{0} -\sin\alpha \ket{1}$, and evolved according to the desired Hamiltonian. The quantum state experimentally obtained is determined via quantum state tomography~\cite{Leskowitz:04} and compared to the theoretically evaluated instantaneous ground state of $H_{0}(t)$. This is performed by computing the fidelity $\Fcal(t)$ between the instantaneous ground state of $H_0(t)$, represented by the density operator $\rho_{\text{gs}}(t)$, and the dynamically evolved quantum state driven by the Hamiltonians $H_0(t)$, $H^{\text{opt}}_{\text{tqd}}(t)$, and $H^{\text{std}}_{\text{tqd}}(t)$, represented by the density operator $\rho(t)$. The fidelity $\Fcal(t)$ is defined here as the \textit{relative purity} given by~\cite{Audenaert:14,Adolfo:13-QSL,Meng:15}
\begin{equation}
	\Fcal(t) =\frac{\left \vert\trs{\rho_{\text{gs}}(t)\rho(t)}\right \vert}{\sqrt{\trs{\rho_{\text{gs}}^2(t)}}\sqrt{\trs{\rho^2(t)}}} \text{ . } \label{fidelity1}
\end{equation} 

The theoretical and experimental results are shown in Fig.~\ref{NMR-Fidelity} for the three distinct dynamics. 
Since the experiment is performed on-resonance, the evolution driven by the Hamiltonian  $H_{0}(t)$ is found to be 
nonadiabatic, oscillating as a function of time. On the other hand, the shortcuts given by $H^{\text{opt}}_{\text{tqd}}(t)$, 
and $H^{\text{std}}_{\text{tqd}}(t)$ keep a non-transitional evolution, being immune to the resonance effect. 

\begin{figure}[t!]
	%\input{Figs/FigFidelity.plt}
	%\vspace{5.25cm}
	\centering
	\includegraphics[scale=0.6]{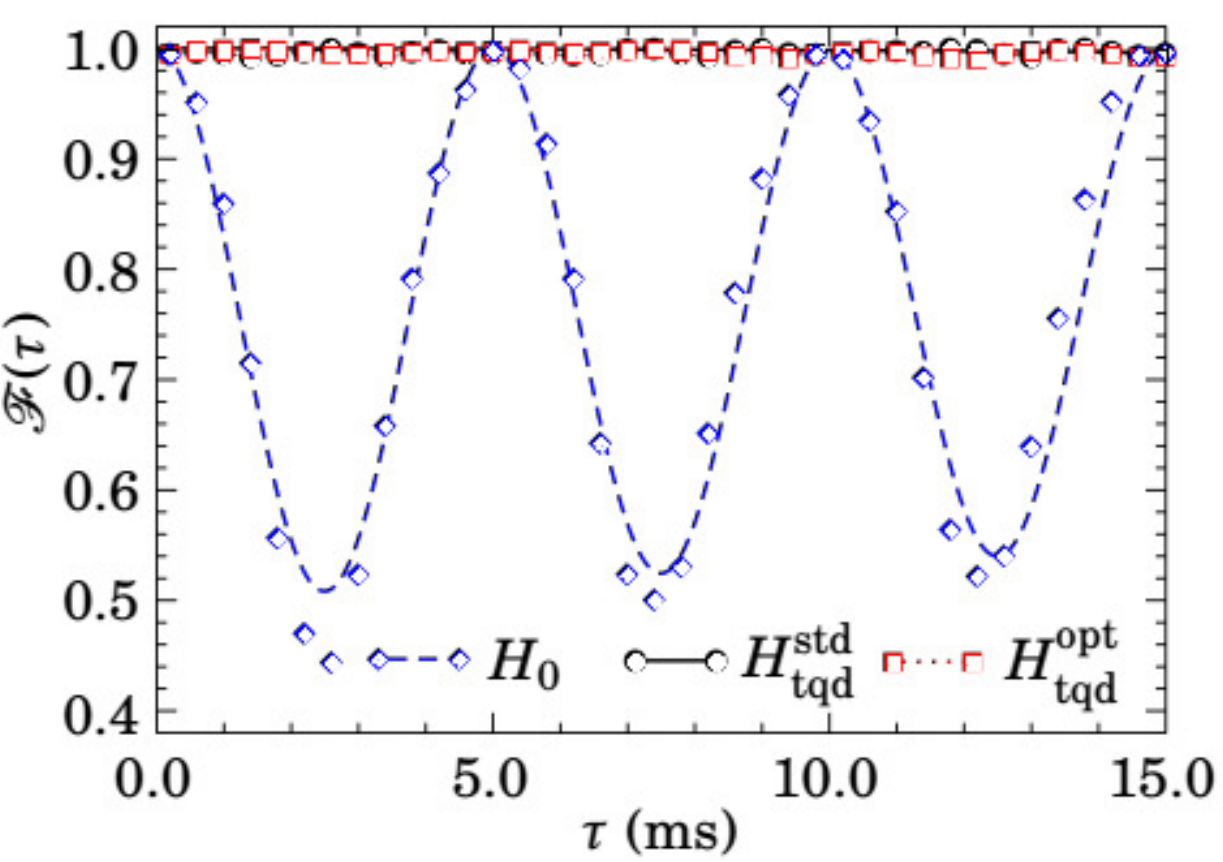}
	\caption{Fidelity $\Fcal(t)$ between the instantaneous ground state of $H_0(t)$ and the dynamically evolved quantum states driven by the Hamiltonians $H_0(t)$, $H^{\text{opt}}_{\text{tqd}}(t)$, and $H^{\text{std}}_{\text{tqd}}(t)$. The curves denote the theoretical predictions and the symbols are the experimental values. The standard $H_{0}(t)$ implementation is more sensitive to RF miscalibration, so that we have less tolerance to field inhomogeneity compared to the other cases. For this reason, we may have a fluctuation between experimental and theoretical results in some points.}
	\label{NMR-Fidelity}
\end{figure}

The fitting of the experimental data has been performed in an open quantum system scenario, which is accomplished by adjusting the parameters of the Lindblad equation whose non-unitary contribution is given by Eq.~\eqref{EqEqLindPartic}. The NMR system considered in this work is mainly affected by dephasing, which occurs due to the inhomogeneity of the static magnetic field. This variation causes the spins of all the molecules to slowly desynchronize and, therefore, lose coherence across the sample. In this case, we have a single Lindblad operator $\Gamma_{1}(t) = \gamma_{0}\sigma_{z}$, where the decohering rate $\gamma_{0}$ is given by $\gamma_{0} = \sqrt{1/T_2}$, with $T_2$ denoting the dephasing relaxation time of the system. Other non-unitary effects, such as generalized amplitude damping characterized by a relaxation scale $T_1$, are also present, but they are negligible for the time ranges considered in the experiment. Indeed, in our realization, we have $T_2 = 0.25$~s and $T_1 = 5.11$~s. Notice then that $T_1$ is two orders of magnitude 
larger than the total time of evolution, so that generalized amplitude damping does not exhibit significant effect. Hence, by considering dephasing as a main effect, the non-unitary part of the dynamics is dictated by the Eq.~\eqref{EqRcalbDephasing}. 
As we showed previously, due to the resonance configuration~\cite{Suter:08,Hu:19-b}, the adiabatic dynamics for $H_{0}(t)$ does not hold in general, being 
achieved just for some particular instants of time. Notice that the oscillation in the adiabatic fidelity is damped as a function of time. This is a by-product of dephasing 
in the system. Moreover, there are also errors induced by the implementation of RF pulses. Those errors are about $3\%$ per pulse. 
Remarkably, shortcuts to adiabaticity show considerable robustness against the decoherence and unitary errors. 

As shown in Fig.~\ref{NMR-Fidelity}, fidelity stays close to one 
throughout the evolution, which means that the states of the evolution are, indeed, kept as instantaneous eigenstates of $H_{0}(t)$. Therefore, as expected, the adiabatic behavior of the system is drastically affected by the 
resonance phenomenon. On the other hand, the standard TQD and the optimal TQD are both immune to the resonance destructive effect. In particular, 
we have shown that the optimal TQD approach provides a transitionless Hamiltonian that can be implemented with low intensity fields in comparison with the fields 
used to implement the adiabatic and the standard TQD Hamiltonians.

\section{Quantum gates by Controlled TQD}

In this section we show a useful application of TQD, namely application to quantum information processing, where we use TQD as a shortcuts to adiabaticity to speed up adiabatic quantum gates. More specifically, it has been applied to perform universal quantum computation (QC) via either counter-diabatic controlled evolutions~\cite{Santos:15} or counter-diabatic quantum teleportation~\cite{Santos:16}. As hybrid models, these approaches provide a convenient digital architecture for physical realizations while potentially keeping both the generality and some inherently robustness of analog implementations. Experimentally, digitized implementations of quantum annealing processes have been recently provided~\cite{Barends:16}, with controlled quantum gates adiabatically realized with high fidelity via superconducting qubits~\cite{Martinis:14}. In this Section, by focusing on controlled evolutions, we will now show that counter-diabatic QC can be more robust against decoherence than its adiabatic counterpart as long as the gate runtime is suitably determined within a range of evolution times.

\subsection{Quantum gates by adiabatic controlled evolutions}

In adiabatic protocols, we need to start the evolution at some eigenstate of the Hamiltonian $H_{0}(0)$, and at the end of the evolution ($t = \tau$) we should get the respective updated eigenstate of $H_{0}(\tau)$. A important point to highlight here is that the input state of the computation should not depend on the $H_{0}(0)$, and vice-versa. It is a strong requirement because we want to design a universal model of computation, so the input state needs to be considered as an \textit{unknown state}. To this end, we consider a bipartite system composed by a target subsystem $\Tcalb$ and an auxiliary subsystem $\Acalb$, whose individual Hilbert spaces $\Hcalb_\Tcalb$ and $\Hcalb_\Acalb$ have dimensions $d_\Tcalb$ and $d_\Acalb$, respectively. The auxiliary subsystem $\Acalb$ will be driven by a family of time-dependent Hamiltonians $\{\Hcal_{k}\left( s\right)\}$, with $0 \le k \le d_\Tcalb-1$. The target subsystem will be driven by a complete set $\left\{ P_{k}\right\}$ of orthogonal projectors over $\Tcalb$, which satisfy $P_{k}P_{m}=\delta _{km}P_{k}$ and $\sum_{k}P_{k}=\1$. In a controlled adiabatic dynamics, the composite system $\Tcalb\Acalb$ will be governed by a Hamiltonian in the form~\cite{Hen:15} 
\begin{equation}
	H_{0}\left( s\right) =\sum\nolimits_{k}P_{k}\otimes \Hcal_{k}\left( s\right) \text{ , } \label{EqAdHamilControlGate}
\end{equation}
with $\Hcal_{k}\left( s\right) =g\left( s\right) \Hcal_{k}^{\left( f\right)}+f\left( s\right) \Hcal^{\left( b\right) }$, where $\Hcal^{(b)}$ is the beginning Hamiltonian, $\Hcal^{(f)}_k$ is the contribution $k$ to the final Hamiltonian, and the time-dependent functions $f(s)$ and $g(s)$ satisfy the boundary conditions $f\left( 0\right) =g\left( 1\right) =1$ and $g\left( 0\right) =f\left(1\right) =0$.

Now, we prepare $\Tcalb\Acalb$ in the initial state $\ket{\Psi _{\text{init}}} =\ket{\psi} \otimes \ket{\varepsilon _{b}}$, where $\ket{\psi}$ is an arbitrary state of $\Tcalb$ and 
$\ket{\varepsilon _{b}}$ is the ground state of $\Hcal^{\left( b\right) }$. Then $\ket{\Psi _{\text{init}}}$ is the ground state of the initial Hamiltonian $\1 \otimes \Hcal^{\left( b\right) }$. By applying the adiabatic theorem~\cite{Messiah:Book,Sarandy:04}, a sufficiently slowing-varying evolution of $H_{0}(t)$ will drive the system (up to a phase) to the final state 
\begin{equation}
	\left\vert \Psi_{\text{final}}\right\rangle =\sum\nolimits_{k}P_{k}\ket{\psi} \otimes \left\vert \varepsilon _{k}\right\rangle \text{ , }
	\label{EqAdpsi-final}
\end{equation}
where $\left\vert \varepsilon_{k}\right\rangle$ is the ground state of $\Hcal_{k}^{\left( f\right) }$~\cite{Hen:15}. Note that an arbitrary projection $P_k$ over the unknown state $\ket{\psi}$ can be yielded by performing a convenient measurement over $\Acalb$. Therefore, by adjusting the set of projectors $P_{k}$ and the Hamiltonians $\Hcal_{k}\left( s\right)$, it is possible to control the output state $\ket{\psi_{\text{out}}}$ of $\Tcalb$, so that single and $N$-qubit gates can be implemented~\cite{Santos:15}. In fact, 

Let us begin by considering $\Tcalb$ as a single qubit 
and a single-qubit gate as an arbitrary 
rotation of angle $\phi$ around a direction $\hat{n}$ over the Bloch 
sphere. Under this consideration, the Hamiltonian that adiabatically implements such a single-qubit gate for an arbitrary input state 
$\ket{\psi} = a \ket{0} + b \ket{1}$, with $a,b \in \Cmath$, is given by~\cite{Hen:15}
\begin{equation}
	H_{0}^{\text{sg}}\left( s\right) = P_{+}\otimes \Hcal_{0}\left(
	s\right) + P_{-} \otimes \Hcal_{\phi}\left( s\right) \text{ ,} \label{HAdsg}
\end{equation}%
where $\{ P_{ \pm } \}$ is a complete set of orthogonal projectors over the Hilbert space of 
the target qubit. 
The projectors can be parametrized as $P_{ \pm } = ( \1 \pm \hat{n} \cdot \vec{\sigma} ) / 2$, with 
$\hat{n}$ associated with the direction of the target qubit on the Bloch sphere. In Eq.~(\ref{HAdsg}), each Hamiltonian $\Hcal_{\xi}\left(s\right)$ 
($\xi = \{ 0, \phi \}$) acts on $\Acalb$, and is given by~\cite{Hen:15}
\begin{equation}
	\Hcal_{\xi }\left( s\right) =-\omega \{ \sigma _{z}\cos (\varphi
	_{0}s)+\sin (\varphi _{0}s) [ \sigma _{x}\cos \xi +\sigma _{y}\sin \xi %
	] \} \text{ , }
	\label{Hphy}
\end{equation}
with $\varphi _{0}$ denoting an arbitrary parameter that sets the success probability 
of obtaining the desired state at the end of the evolution. This parameter plays a 
role in the energy performance of counter-diabatic QC, with probabilistic 
counter-diabatic QC ($ \varphi _{0} \neq \pi $) being energetically more favorable than its 
deterministic ($ \varphi _{0} = \pi $) counterpart~\cite{Coulamy:16}.

The projectors $\{ P_{ \pm } \}$ may be written in terms of two basis vectors $\{ \ket{n_{\pm}} \}$ in the Bloch sphere as $\{ P_{ \pm } \} = \ket{n_{\pm}} \bra{n_{\pm}}$, where
\begin{align}
	\ket{n_{+}} &= \cos (\varepsilon/2) \ket{0} + e^{i\delta} \sin (\varepsilon/2) \ket{1} \label{nmais} \text{ ,} \\
	\ket{n_{-}} &= -\sin (\varepsilon/2) \ket{0} + e^{i\delta} \cos (\varepsilon/2) \ket{1} \label{nmenos} \text{ .}
\end{align}
Thus, a quantum gate is encoded as a rotation of $\phi$ around the vector $\ket{n_{+}}$ (as we shall illustrate in Fig.~\ref{Bloch-Scheme}{\color{blue}a}). Now, by expressing the state $\ket{\psi}$ in the basis $\{ \ket{n_{\pm}} \}$, we write 
$\ket{\psi} =  \alpha \ket{n_{+}} + \beta \ket{n_{-}}$, with $\left\vert \hat{n}_{\pm }\right\rangle$ 
being a state in the direction $\hat{n}$ and $\alpha,\beta \in \Cmath$. 
We then prepare the system in the initial state $\ket{\Psi (0)}  = \ket{\psi} \ket{0}$. Then, 
assuming an adiabatic dynamics, the evolved state $\ket{\Psi (s)} $ 
is given by the superposition 
\begin{equation}
	\ket{\Psi (s)} = \alpha \ket{n_{+}} \ket{ E_{-,0} (s) } + \beta \ket{n_{-}} \ket{ E_{-,\phi} (s) } = \cos \left( \frac{\varphi_{0} s}{2} \right) \ket{\psi} \ket{0} + \sin \left( \frac{\varphi_{0} s}{2} \right) \ket{ \psi_{\text{rot}}} \ket{1} \text{ ,} \label{RefereeA1}
\end{equation}
with $\ket{ \psi_{\text{rot}}} = \alpha \ket{n_{+}} + e^{i \phi} \beta \ket{n_{-}}$ being the rotated 
desired state and the ground $\ket{ E_{-,\xi} (s) }$ and first excited $\ket{ E_{+,\xi} (s) } $ states 
of $\Hcal_{\xi }\left( s\right)$ given by
\begin{align}
	\ket{ E_{-,\xi} (s) }  &= \cos ( \varphi_{0} s / 2 ) \ket{0}  + e^{i \xi } \sin ( \varphi_{0} s / 2 ) \ket{1} \text{ ,} \label{Fund} \\
	\ket{ E_{+,\xi} (s) }  &= - \sin ( \varphi_{0} s / 2 ) \ket{0} + e^{i \xi } \cos ( \varphi_{0} s / 2 ) \ket{1} \text{ .} \label{Exc} 
\end{align}

We observe that, due to the dynamics of the auxiliary qubit through two adiabatic paths, there are quantum (adiabatic) phases $\theta^{\text{ad}}_{0}(s)$ and $\theta^{\text{ad}}_{\phi}(s)$ accompanying the evolutions associated with $\ket{ E_{-,0} (s) }$ and $\ket{ E_{-,\phi} (s) }$, respectively. Then, relative phases should in principle be considered in Eq. (\ref{RefereeA1}). However, as shown in the Ref.~\cite{Hen:15}, such phases satisfy $\theta^{\text{ad}}_{0}(s) = \theta^{\text{ad}}_{\phi}(s)$. Thus, they factorize as a global phase of the state $\ket{\Psi (s)}$. At the end of the evolution, a measurement on the auxiliary qubit yields the rotated state with probability $\sin ^2 \left( \varphi_{0} /2 \right)$ and the input state with probability $\cos ^2 \left( \varphi_{0} /2 \right)$. The computation process is therefore \textit{probabilistic}, which succeeds if the auxiliary qubit ends up in the state $\ket{1}$. Otherwise, the target system automatically returns to the input state and we simply restart the protocol. In the adiabatic scenario, the parameter $\varphi_{0}$ can then be adjusted in order to obtain the optimal fidelity $1$ by taking the limit $\varphi_{0} \rightarrow \pi$, implying in a deterministic computation. 

This model can be easily adapted to implement controlled single-qubit gates (two-qubit gates). To this end, the target system has to be increased from one qubit to two qubits, where one of these qubits will be used as control and the second one as target. Again, we need to consider the auxiliary qubit in order to provide the universality of the protocol. Here, we adopt that the single-qubit gate acts on the target register if the state of the control register is $\ket{1}$. With this convention, the Hamiltonian that implements a controlled single-qubit gate is given by
\begin{equation}
	H_{0}^{\text{cg}}\left( s\right) = (\1 - P_{1,-})\otimes \Hcal_{0}\left( s\right) + P_{1,-} \otimes \Hcal_{\phi}\left( s\right) \text{ ,} \label{HAdcg}
\end{equation}
where now the set the orthogonal projectors is given by 
$P_{k,\pm} = \ket{k} \bra{k} \otimes \ket{n_{\pm}} \bra{n_{\pm}} $, where $\ket{k}$ denotes the 
computational basis. The input state of the target system is now written as 
$\ket{ \psi_{2} } = a \ket{ 00 } + b \ket{ 01 } + c \ket{ 10 } + d \ket{ 11 } $, with $a,b,c,d \in \Cmath$ and 
$\ket{ nm }=\ket{ n } \ket{ m }$ denoting the control and target register, respectively. 
By rewriting $\ket{ \psi_{2} } $ in terms of the basis $\ket{n_{\pm}}$, we have 
$\ket{ \psi_{2} } = \alpha \ket{ 0 n_{+} } + \beta \ket{ 0 n_{-} } + \gamma \ket{ 1 n_{+} } + \delta \ket{ 1 n_{-} } $, 
with $\alpha,\beta,\gamma,\delta \in \Cmath$. 
Therefore, by assuming adiabatic evolution, the system evolves from the state 
$\ket{\Psi_{2}(0)} = \ket{ \psi_{2} } \ket{0}$, to the instantaneous state
\begin{align}
	\ket{\Psi (s)} &= \alpha \ket{ 0 n_{+} } \ket{ E_{-,0} (s) } + \beta \ket{ 0 n_{-} } \ket{ E_{-,0} (s) } + \gamma \ket{ 1 n_{+} } \ket{ E_{-,0} (s) } + \delta \ket{ 1 n_{-} } \ket{ E_{-,\phi} (s) } \nonumber \\ 
	&= \cos \left( \frac{\varphi _{0} s}{2} \right) \ket{ \psi_{2} } \ket{0} + \sin \left( \frac{\varphi_{0} s}{2} \right) \ket{ \psi_{2\text{rot}} } \ket{1} \text{ , }
\end{align}
with $\ket{ \psi_{2\text{rot}} } = \alpha \ket{ 0 n_{+} } + \beta \ket{ 0 n_{-} } + \gamma \ket{ 1 n_{+} } + e^{i \phi} \delta \ket{ 1 n_{-} } $ being the rotated desired state. We then see that the final state $\ket{\Psi (1)}  $ allows for a probabilistic interpretation for the evolution and, consequently, the computation protocol can again be taken as 
probabilistic ($\varphi_0 \neq \pi$) or deterministic ($\varphi_0 = \pi$).

\subsection{Quantum gates by standard and optimal TQD controlled evolutions}

As shown in Ref.~\cite{Santos:15}, the standard TQD Hamiltonian that implements the shortcut to controlled adiabatic evolutions reads
\begin{equation}
	H_{\text{std}}\left( s\right) = \sum_{k} P_{k} \otimes \Hcal_{k}^{\text{std}}\left(s\right) \text{ , }  \label{sce1.7}
\end{equation}
where each Hamiltonian $\Hcal_{k}^{\text{std}}\left(s\right)$ is the corresponding standard TQD Hamiltonian associated with the Hamiltonian $\Hcal_{k}\left( t\right)$ given in Eq.~\eqref{EqAdHamilControlGate}. Thus, we write $\Hcal_{k}^{\text{std}}\left(s\right) = \Hcal_{k}\left( s\right) + \Hcal_{k}^{\text{cd}}(s)$, with $\Hcal_{k}^{\text{cd}}(s)$ the corresponding counter-diabatic term of $\Hcal_{k}\left( s\right)$. For the single and two-qubit gates previously discussed, the counter-diabatic Hamiltonians $\Hcal_{\xi}^{\text{cd}}(s)$ associated with the Hamiltonians $\Hcal_{\xi }\left( s\right)$ in Eq.~\eqref{Hphy} are {\it time-independent} and given by ($\xi = \{ 0, \phi \}$)
\begin{equation}
	\Hcal_{\xi}^{\text{cd}}(s) = \Hcal_{\xi }^{\text{cd}}=\hbar \frac{\vartheta _{0}}{2\tau}\left[ \sigma _{y}\cos \xi - \sigma_{x}\sin \xi \right] \text { . } \label{sce1.7.ad}
\end{equation}

Therefore, one note that the cost of performing standard TQD requires the knowledge of the eigenvalues and eigenstates of $\Hcal_{k}\left( t\right)$. Moreover, we can show that, for an arbitrary $n$-controlled quantum gate, we get time-independent counter-diabatic Hamiltonians~\cite{Santos:15}. Therefore, the Eq.~(\ref{sce1.7.ad}) shows that the implementation of the shortcut can be achieved with a very simple assistant Hamiltonian in the quantum dynamics. Then, the standard TQD Hamiltonian that implements a single-qubit gate reads
\begin{equation}
	H_{\text{std}}^{\text{sg}}\left( s\right) = P_{+}\otimes \Hcal_{0}^{\text{std}}\left(s\right) + P_{-} \otimes \Hcal_{\phi}^{\text{std}}\left(s\right) \text{ , }  \label{EqHgatestd}
\end{equation}
with
\begin{equation}
	\Hcal_{\xi}^{\text{std}}\left(s\right) = \Hcal_{\xi}\left(s\right) + \hbar \frac{\vartheta _{0}}{2\tau}\left[ \sigma _{y}\cos \xi - \sigma_{x}\sin \xi \right] \text{ . } \label{EqHxiTQDStd}
\end{equation}

Similarly, we derive a general TQD Hamiltonian in which the quantal phases are considered as arbitrary. Remarkably, the Hamiltonian satisfy the conditions of the Theorem~\ref{TheoTimeIndep} and the parallel transport condition. Therefore, the optimal TQD reads
\begin{equation}
	H_{\text{opt}}^{\text{sg}}\left( s\right) = P_{+}\otimes \Hcal_{0}^{\text{opt}}\left(s\right) + P_{-} \otimes \Hcal_{\phi}^{\text{opt}}\left(s\right) \text{ , }  \label{EqHgateOpt1}
\end{equation}
in which each optimal Hamiltonian $\Hcal_{\xi}^{\text{opt}}\left(s\right)$ is given by $\Hcal_{\xi}^{\text{opt}}\left(s\right) = \Hcal_{\xi}^{\text{opt}} = \Hcal_{\xi}^{\text{cd}}$. Thus, the application of generalized TQD theory to quantum computation allows us to get  a optimal costing time-independent TQD Hamiltonian $H^{\text{sg}}_{\text{opt}}$ given by
\begin{equation}
	H_{\text{opt}}^{\text{sg}}  = P_{+}\otimes \Hcal_{0}^{\text{cd}} + P_{-} \otimes \Hcal_{\phi}^{\text{cd}} \text{ , }  \label{EqHgateOpt}
\end{equation}
where $\Hcal_{k }^{\text{cd}}$ is given by the Eq.~\eqref{sce1.7.ad}. To end, by using the same discussion as above, the complete set of TQD Hamiltonian to be used as to achieve an universal scheme of quantum computation is obtained as
\begin{align}
	H_{\text{std}}^{\text{cg}}\left( s\right) &= (\1 - P_{1,-})\otimes \Hcal_{0}^{\text{std}}\left(s\right) + P_{1,-} \otimes \Hcal_{\phi}^{\text{std}}\left(s\right) \text{ , } \label{EqTQDHstdcg} \\
	H_{\text{opt}}^{\text{cg}} &= (\1 - P_{1,-})\otimes \Hcal_{0}^{\text{opt}} + P_{1,-} \otimes \Hcal_{\phi}^{\text{opt}} \text{ , } \label{EqTQDHOptcg} 
\end{align}
with the projector $P_{1,-}$ defined as in Eq.~\eqref{HAdcg}, and the Hamiltonians $\Hcal_{\xi}^{\text{std}}\left(s\right)$ and $\Hcal_{\xi}^{\text{opt}}$ given as in Eqs.~\eqref{sce1.7.ad} and~\eqref{EqHxiTQDStd}, respectively.

As a important remark, it is evident that a number of strategies can be used to implement quantum gates. For example, inverse engineering of a Hamiltonian is an interesting strategy to design a class of Hamiltonians to quantum computation~\cite{Santos:18-a}, where it is shown how arbitrary single and controlled quantum gates are implementable through very simple time-independent Hamiltonians, and to speed up state engineering~\cite{Yu:18}, as well. Here we are using a different approach to achieve the same result as in Refs.~\cite{Santos:18-a,Yu:18}. However, among the infinite curves connecting the input to the output state, here we are interested in that dynamics where the path followed by the system in the Hilbert space is the adiabatic path. For this reason, these number of additional controls (and fields), auxiliary qubits and measurements (sometimes) are required.

The performance of optimal TQD approach in comparative with its adiabatic and standard counterpart is studied in Ref.~\cite{Santos:18-b}, where different scenarios of energy resource and decohering effects were considered. As main result, in any case of decoherence and energetic resource, there always exist dynamical regimes (short evolution times) for which the optimal transitionless evolutions are more robust and therefore a preferred approach in a decohering physical environment. This has been shown for both the dephasing and GAD channels acting on the eigenstate bases. It is also possible to show the advantage in other bases, such as the computational basis. The general picture is that the gain will typically occur during some finite time range, disappearing in the limit of long evolution times. These results are encouraging for the generalized transitionless approach in the open-system realm as long as local Hamiltonians are possible to be designed. In the specific case of quantum gate Hamiltonians, this approach can be applied, e.g. to derive robust local building blocks for analog implementations of quantum circuits (see, e.g., Refs.~\cite{Martinis:14,Barends:16}).

\subsection{Optimal TQD quantum computation in NMR experimental setup}

The experimental realization of quantum gates is implemented through the Chloroform molecule by taking the $^{13}$C nucleus as the target qubit and the $^1$H nucleus as the auxiliary subsystem.  In our experiment~\cite{Santos:20b}, we encode the computational states $\ket{0}$ and $\ket{1}$ into the spin states $\ketds$ and $\ketus$, respectively. The schematic representation of the quantum dynamics is illustrated in Fig.~\ref{Bloch-Scheme}{\color{blue}b}, where the composite system is initially prepared at $t=0$ and measured in $t=\tau$. The decohering time scales in our system are $T_{1\text{C}} = 7.33$~s and $T_{2\text{C}} = 4.99$~s for $^{13}$C nucleus and $T_{1\text{H}} = 14.52$~s and $T_{2\text{H}} = 0.77$~s for $^1$H nucleus.

\begin{figure}[t!]
	\centering
	\includegraphics[scale=0.6]{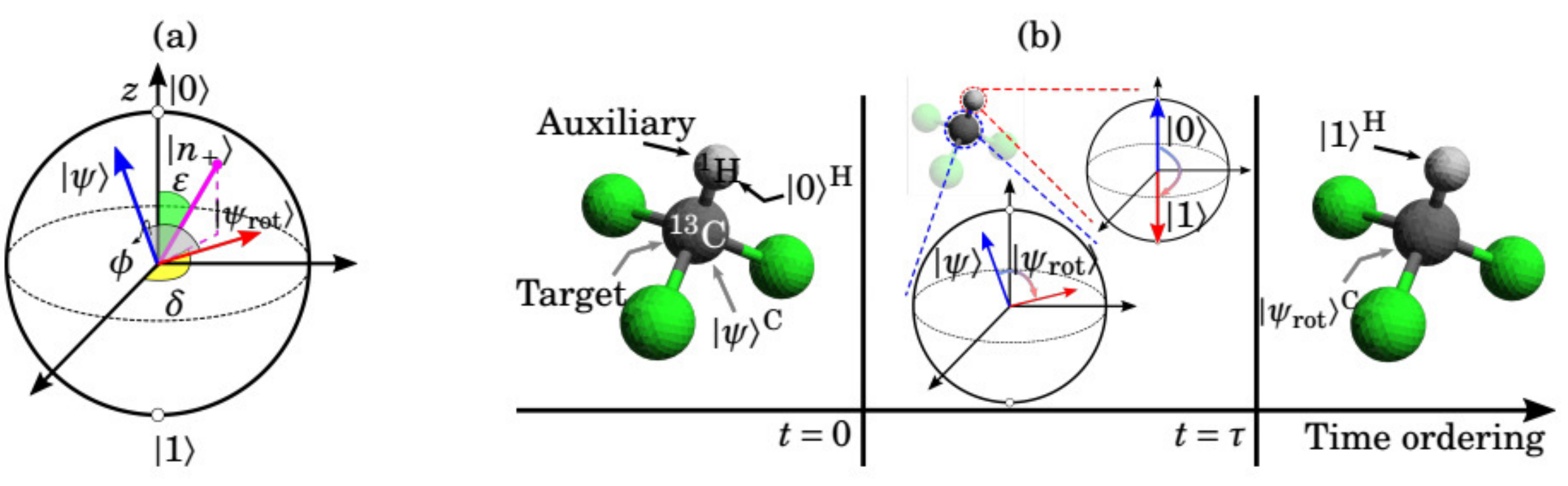}
	\caption{({\color{blue}a}) Representation of a single-qubit gate as an arbitrary rotation implemented in the Bloch sphere. ({\color{blue}b}) Schematic representation of the quantum dynamics for the Chloroform molecule, with the target and ancilla qubits encoded in the Carbon and Hydrogen nuclei, respectively.}
	\label{Bloch-Scheme}
\end{figure}

We have considered the experimental implementation of the single-qubit phase gate $Z$, which yields the Pauli matrix $\sigma^z$ applied to 
the input qubit, which is prepared in an arbitrary initial state along a direction $\hat{r}$ in the Bloch sphere. To this gate, the adiabatic, standard TQD, and optimal TQD Hamiltonians are given by (the subscript ``C" and ``H" denotes operation on Carbon and Hydrogen nucleus, respectively)
\begin{subequations}\label{EqHExpNMRQuantumGates}
	\begin{align}
		H^{\text{ph}}_{0}(t) &= - 2 \pi \nu \left[ \cos\left(\frac{\pi t}{\tau}\right) \1^{\text{C}} \otimes \sigma_{z}^{\text{H}} + \sin\left(\frac{\pi t}{\tau}\right) \sigma_{z}^{\text{C}} \otimes \sigma_{x}^{\text{H}}\right] \text{ , } \label{Hz0}  \\
		H^{\text{ph}}_{\text{std}}(t) &= H^{\text{ph}}_{0}(t) + \frac{\pi}{2\tau} \sigma_{z}^{\text{C}} \otimes \sigma_{y}^{\text{H}}\text{ , }  \label{HzTQD}
		\\
		H^{\text{ph}}_{\text{opt}} &= \frac{\pi}{2\tau} \sigma_{z}^{\text{C}} \otimes \sigma_{y}^{\text{H}}\text{ , } \label{HzOptTQD}
	\end{align}
\end{subequations}
which are obtained from Eqs.~(\ref{HAdcg}),~(\ref{EqHgatestd}), and~(\ref{EqHgateOpt}), respectively, with $\omega = 2 \pi \nu$ and $\nu$ denoting a real frequency in the experiment. It is worth highlighting that a study on the energy cost to implement adiabatic and TQD Hamiltonians has been previously discussed in Ref.~\cite{Santos:18-b} from an operator norm approach. Here, we are interested in analyzing the cost from an alternative point of view, where we associate fixed energy amounts for each pulse in a pulse sequence. Therefore, this includes the effective energy spent to implement each pulse of magnetic field, while disregarding the free evolution contributions to the quantum dynamics.

Particularly, in our experimental implementation of the phase gate $Z$, we have set $\nu = 35$~Hz. We can see that, differently from the standard TQD and adiabatic Hamiltonians, the optimal TQD protocol provides a time-independent Hamiltonian to realize the phase gate $Z$. In Fig.~\ref{Gates}{\color{blue}a}, we present the pulse sequences used to implement each Hamiltonian in Eqs.~\eqref{EqHExpNMRQuantumGates} for any input state. As an illustration, the initial state of the target qubit considered in our experiment has been taken as $\ket{\psi(0)} = \ket{+}^{\text{C}}\otimes\ket{0}^{\text{H}}$, with $\ket{+} = (1/\sqrt{2})(\ket{0}+\ket{1})$. Each pulse sequence implements the correct dynamics up to a rotation around the $Z$-axis over the auxiliary qubit. Since the final state of the auxiliary qubit is $\ket{1}$, the circuits provide the correct output up to a global phase. The pulse composition is described in details in Appendix~\ref{ApExpPuls}.

\begin{figure}[t!]
	%\input{Figs/Pulses-and-Fidelity.plt}
	%\vspace{-0.25cm}
	\centering
	\includegraphics[scale=0.56]{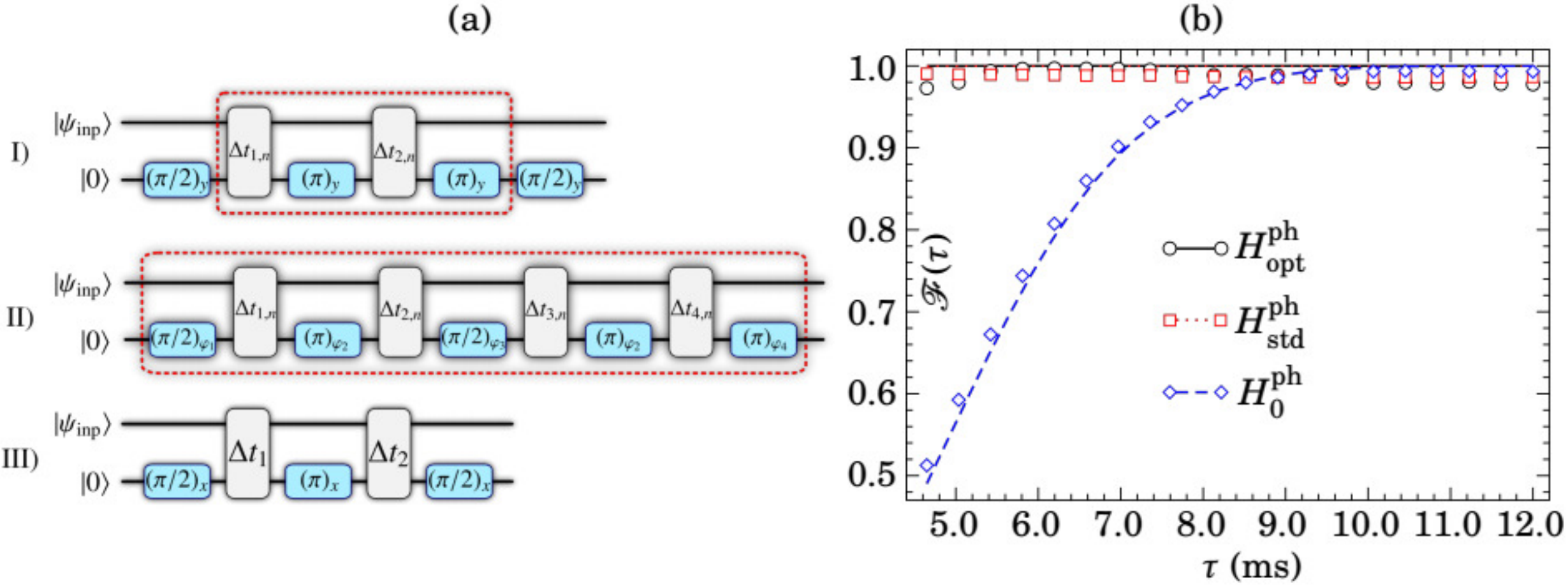}
	\caption{({\color{blue}a}) Adiabatic, standard TQD, and optimal TQD protocol pulses. Two-qubit blocks labeled with $\Delta t_{j,n}$ represent a free evolution of the chloroform molecule during a time interval $\Delta t_{j,n}$, while single-qubit blocks are 	rotations $(\vartheta)_{\varphi_{n}}$ in Hilbert space (See Appendix~\ref{ApExpPuls} for more details). ({\color{blue}b}) Fidelity for the $Z$ phase gate implementation over the initial input state $\ket{+} = (1/\sqrt{2})(\ket{0}+\ket{1})$ encoded in the Carbon nucleus, with the Hydrogen nucleus representing the auxiliary qubit in the initial state $\ket{0}$. Fidelity $\Fcal(\tau)$ is computed between the instantaneous ground state of $H^{\text{ph}}_{0}(t)$ and the dynamically evolved quantum states driven by the Hamiltonians $H^{\text{ph}}_{0}(t)$, $H^{\text{ph}}_{\text{std}}(t)$, and $H^{\text{ph}}_{\text{opt}}(t)$. Solid curves and discrete symbols represent theoretical predictions and experimental results, respectively.}
	\label{Gates}
\end{figure}

The Hamiltonians in Eq.~(\ref{Hz0}) and~(\ref{HzTQD}) are time-dependent and do not commute at different instants of time, i.e. 
$[H^{\text{ph}}_{0}(t_{1}),H^{\text{ph}}_{0}(t_{2})]\neq0$ and $[H^{\text{ph}}_{\text{std}}(t_{1}),H^{\text{ph}}_{\text{std}}(t_{2})]\neq0$ for some $t_{1} \neq t_{2}$. Then, their pulse sequences are required to implement the Dyson series for the corresponding unitaries~\cite{Nielsen:Book}. Thus, as shown in Figs.~\ref{Gates}{\color{blue}a}-I and~\ref{Gates}{\color{blue}a}-II, the $N$ repetitions are associated to a``trotterization" of the Dyson series for the adiabatic and standard TQD protocols, respectively, with the implementation of the Dyson series being exact in the limit $N\rightarrow \infty$. The repeated application of the sequence can lead to the accumilation to experimental sistematic errors, in order to avoid the rrors we have employed NMR composite pulses~\cite{Brown:04}. On the other hand, the optimal TQD Hamiltonian $H^{\text{ph}}_{\text{opt}}$ can be implemented by using a very short pulse sequence, as shown in Fig~\ref{Gates}{\color{blue}a}-III.

An immediate discussion arising from the analysis above is the amount of energy required to implement the protocol (the pulse sequence associated with the Dyson series). In order to conveniently study the energy cost of implementing the optimal TQD protocol, we need to analyze the energy cost of implementing the sequence of pulses able to reproduce the Dyson series for each dynamics. Consequently, from the pulse sequence presented in Fig~\ref{Gates}{\color{blue}a}, the energy cost evolving the optimal TQD is (at least) $N$ times less than both that associated with the adiabatic and the standard TQD protocols. In fact, if we consider that each pulse in Fig.~\ref{Gates}{\color{blue}a} [given by a rotation $(\vartheta)_{\varphi_{n}}$] is associated with an energy cost $E_{0}$, the overall energy cost for implementing the optimal TQD is $E^{\text{opt}}_{\text{tqd}} = 3E_{0}$ (disregarding the energy cost of free evolutions), while the energy cost of the adiabatic and the standard TQD are given by $E_{\text{ad}} = 2(N+1)E_{0}$ and $E^{\text{std}}_{\text{tqd}} = 5NE_{0}$, respectively. From this analysis, it follows that generalized TQD can be used to provide Hamiltonians exhibiting optimal pulse sequence.

We also consider the $\Fcal(t)$ between the instantaneous ground state of $H_Z(t)$ and the dynamically evolved quantum states driven by the Hamiltonians $H_Z(t)$, $H^{\text{opt}}_{\text{tqd}}(t)$, and $H^{\text{std}}_{\text{tqd}}(t)$. This is shown in Fig.~\ref{Gates}{\color{blue}b}. Notice that, as expected for the Hamiltonian $H_Z(t)$, the fidelity improves as the evolution time is increased. In contrast, $H^{\text{opt}}_{\text{tqd}}(t)$ and $H^{\text{std}}_{\text{tqd}}(t)$ are capable of achieving high fidelities for any evolution time, since they are designed to mimic the adiabatic evolution for arbitrary $\tau$. Notice also that, for the time scale considered in the experiment, which goes up to $\tau = 12$~ms, decoherence has little effect, since this $\tau$ is still much smaller than the dephasing relaxation time scale $T_2$.

\section{Conclusions of this chapter}

In this chapter we have developed a generalized minimal resource demanding counter-diabatic theory, which is able to yield efficient shortcuts to adiabaticity via fast transitionless evolutions. To this end, we explore the gauge freedom of the quantal phases that accompany the evolution dynamics, so that we can fix them in order to minimize some physical quantity of our interest. In particular, we have focused on minimum energy and pulse costs for implement a optimal TQD, comparing it with its associated adiabatic dynamics and traditional TQD counterparts. The performance of the theory has been experimentally investigated in two different experimental setups. Firstly, by using the trapped Ytterbium ion qubit the result are: 
i) The optimal version of generalized TQD has been shown to be a useful protocol for obtaining the optimal shortcut to adiabaticity. While adiabatic and traditional TQD require time-dependent quantum control, optimal TQD can be experimentally realized by using time-independent fields. In addition, the necessity of auxiliary fields in traditional TQD is not a requirement for implementing TQD via its optimal version.
ii) Optimal TQD is an energetically optimal protocol of shortcut to adiabaticity. By considering the average intensity fields as a measure of energy cost for implementing the protocols discussed here, we were able to show that the optimal version of generalized TQD may be energetically less demanding. This result is kept also for alternatives definitions of energy cost, e.g. taking into account the detuning contribution. 
iii) By simulating an environment-system coupling associated with the dephasing channel, we have shown that optimal TQD can be more robust than the adiabatic dynamics and the traditional TQD while at the same time spending less energy resources for a finite time range. Once the Landau-Zener Hamiltonian can be implemented in other physical systems, e.g. nuclear magnetic resonance~\cite{Nielsen:Book} and two-level systems driven by a chirped field~\cite{Demirplak:03}, it is reasonable to think that the results obtained here can be also realized in other experimental architectures.

As a second experimental verification, we investigated the generalized approach for TQD in an NMR setup. As a first application in this architecture, we considered the adiabatic dynamics of a single spin$-1/2$ particle in a resonant time-dependent rotating magnetic field. As expected, the adiabatic behavior of the system is drastically affected by the resonance phenomenon. On the other hand, the standard TQD and the optimal TQD are both immune to the resonance destructive effect. In particular, we have shown that the optimal TQD approach provides a transitionless Hamiltonian that can be implemented with low intensity fields in comparison with the fields used to implement the adiabatic and the standard TQD Hamiltonians. The second application was done through the study of the adiabatic and counter-diabatic 
implementations of single-qubit quantum gates in NMR. From the generalized approach for TQD, we have addressed the problem of the feasibility of the shortcuts to adiabaticity, as provided by TQD protocols, in the context of quantum computation via controlled evolutions~\cite{Hen:15,Santos:15}. By using the generalized TQD Hamiltonian~\cite{Santos:18-b}, we have presented the optimal solution in terms of pulse sequence and resources to implement fast quantum gates through counter-diabatic quantum computation with high fidelity. The energy-optimal protocol presented here is potentially useful for speeding up digitized adiabatic quantum computing~\cite{Hen:14}. Our study explicitly illustrates the performance of generalized TQD in terms of both energy resources and optimal pulse sequence. Since digitized adiabatic quantum computing requires the Trotterization of the adiabatic dynamics, our protocol could be useful in reducing the number of steps used in digital adiabatic quantum processes. In addition, such process is independent of the experimental approach used to digitize the adiabatic quantum evolution. 

These results are encouraging for the generalized transitionless approach in the open-system realm as long as local Hamiltonians 
are possible to be designed. In the specific case of quantum gate Hamiltonians, this approach can be applied, e.g. to derive robust local building blocks for analog implementations of quantum circuits (see, e.g., Refs.~\cite{Martinis:14,Barends:16}).

%%%%%%%%%%%%%%%%%%%%%%%%%%%%%%%%%%%%%%%%%%%%%%%%
%%%%%%%%%%%%%%%% CONCLUSIONS %%%%%%%%%%%%%%%%
%%%%%%%%%%%%%%%%%%%%%%%%%%%%%%%%%%%%%%%%%%%%%%%%

\chapter{Overview and future perspectives}

\initial{I}n this thesis we presented general results on adiabatic dynamics in closed and open systems, and shortcuts to adiabaticity by transitionless quantum driving. In the context of adiabatic dynamics in closed system we presented potentially useful theorems that can be viewed as a validation mechanism for the well-known adiabaticity conditions, where changes of frame in Schrödinger equation are required. In particular, for the closed system case, we have both theoretically and experimentally shown that
several relevant ACs, which include the traditional AC, are neither sufficient nor necessary to
describe the adiabatic behavior of a qubit in an oscillating field considered here. This result leads us to the conclusion that we do not typically have a fully applicable AC. However, from an adequate choice of the reference frame used to describe the system dynamics, sufficiency and
necessity of the ACs are fundamentally obtained under certain conditions through the non-inertial frame map. As a potential application of adiabatic dynamics, we provided a way to enhance the charging performance of quantum devices able to store energy (quantum batteries). We propose the notion of \textit{adiabatic QBs}, where the adiabatic behavior takes place in the non-inertial frame and it allows for a stable charging process of the battery, since in inertial frame the dynamics is not adiabatic. To this aim, we propose the notion of a new quantum observable that quantifies the energy transfer rate from the quantum battery to a consumption hub, the power operator. Applications of this operator have demonstrated its usefulness in quantum battery context. For open systems, our theoretical results give complementary contributions to the results known in the literature. The validity conditions used to predict the adiabatic behavior of a quantum system evolving through non-unitary dynamics are reviewed. We have shown that the adiabatic approximation for open systems holds in the asymptotic time limit $t\rightarrow \infty$ as an arbitrarily valid consequence of the one-dimensional Jordan decomposition of the Lindblad superoperator, the absence of level crossings, and the initialization of the system as a superposition of only two eigenstates of Lindbladian that drives the system. Our main result in adiabaticity in open system refers to the application in quantum thermodynamics, where we can provide some enhanced understanding on the thermal adiabatic process and dynamical adiabatic evolutions. To end, as an original contribution to the TQD theory, we have developed a generalized minimal energy demanding counter-diabatic theory, which is able to yield efficient shortcuts to adiabaticity via fast transitionless evolutions. Moreover, we have investigated the robustness of adiabatic dynamics, standard and optimal TQD under decoherence. There always exist dynamical regimes for which generalized transitionless evolutions are more robust and therefore a preferred approach in a decohering setup. 

\section*{Future perspectives}

As future perspectives, there are a number of open questions to be addressed.

\quad \textbf{1)} The study introduced here on adiabatic dynamics allows us to realize further future research as detailed below.

\quad \quad \textbf{1-a)} The consideration and influence of different reference frames in the adiabatic dynamics of the system lead us to the fundamental question of the existence of a preferable non-inertial frame. In addition, our development takes into account a single resonant field acting on the system, so that a more general approach for an arbitrary number of fields is an open question.

\quad \quad \textbf{1-b)} The experimental implementation of the two-cell adiabatic quantum battery proposed here is a desired task to be taken in consideration as future research. The charging and discharging process of the quantum battery proposed in this thesis will be investigated in a transmon qubit superconducting system in collaboration with experimental groups. The questions to be experimentally addressed here refers to stable charging/discharging of quantum batteries based on Bell states and its robustness against decohering effects.

\quad \quad \textbf{1-c)} When we use the terminology ``quantum" in quantum devices, it is required some detailed discussion on what quantum is in such devices. For example, the quantum advantage of quantum batteries may be associated with quantum entanglement~\cite{Alicki:13,PRL2013Huber,PRL2017Binder,Ferraro:18,PRL_Andolina} and it becomes evident in multi-cell quantum batteries. Therefore, the scalability of the model proposed here is of great interest, where we can verify the real role of entanglement for the power of the charging/discharging process of the battery. Moreover, since the experimental verification of the expected quantum advantage has not been addressed so far, collaboration with experimental groups will be considered in future research.

\vspace{0.3cm}

\quad \textbf{2)} As for the generalized TQD theory, we investigated its performance in minimizing the energy cost of fast evolutions, and the pulse sequence required to implement optimal TQD. However, we can think about a number of applications that can be benefited from these particular evolutions. For example, applications of generalized TQD to holonomic quantum computing~\cite{Xu:12}, superadiabatic charging of quantum batteries through minimal intensity fields~\cite{Santos:19-a,Santos:20c} and optimal superadiabatic quantum thermal engines~\cite{Lutz:18,Baris:19}.

\quad \quad \textbf{2-a)} In the general theory of optimal control~\cite{Rabitz:98}, one studies how to find a set of fields $\{\epsilon_{n}(t)\}$ (pulses, for example) that drives the system along a path $\ket{\psi(t)}$ in such way that we can maximize the quantity $\Ocal[\psi] = \bra{\psi(t)}\Ocalb\ket{\psi(t)}$, for some predefined observable $\Ocalb$~\cite{Werschnik:07}. Then, how can we use the generalized approach of TQD in a general approach of optimal control? What is the general formalism to find optimal quantal phases from an arbitrary observable $\Ocalb$?

\quad \quad \textbf{2-b)} The standard TQD has been applied to enhance performance of quantum thermal engines concerning the adiabatic version of such heat engines~\cite{Lutz:18,Baris:19,Abah:19}. However, an open question here is whether this enhancement is due to the additional fields which acts on the system. As future research we will address this question by using the generalized TQD and providing a scenario where TQD and adiabatic dynamics spend a same amount of energy. Moreover, it is known that if we speed up a thermodynamics cycle in order to enhance the power, some amount of energy is thrown away as a dissipated work due to diabatic transitions~\cite{Cakmak:16,Plastina:14}. In particular, it was shown that shortcuts to adiabaticity can be a good strategy to deal with this problem~\cite{Deng:18}. The open question is whether there is some optimal TQD Hamiltonian able to minimize undesired entropy production along some superadiabatic thermal cycle.

\quad \quad \textbf{2-c)} Given the notion of adiabatic quantum battery and the role of TQD in quantum thermodynamic tasks, the proposal of superadiabatic quantum batteries based on TQD is an exciting possibility. Although the adiabatic version of such devices allows for a stable charging process, the average power is limited by the total evolution time imposed by the adiabatic theorem. Then, standard and generalized TQD positively contribute to design a high power and stable energy transference process.

\quad \quad \textbf{2-d)} The extension for open systems of the standard version of TQD, as proposed by Demirplak and Rice, was proposed by Vacanti \textit{et al}~\cite{Vacanti:14}. However, different from closed system TQD, the open system TQD requires our ability of controlling both fields and decohering environment to drive the system along of an open system adiabatic path. Therefore, the extension of the generalized approach of TQD for open system is highly desirable as topic for future research.
% And the appendix goes here
\appendix
\chapter{Change of frame in quantum mechanics}\label{ApFrameChangeQM}

\initial{H}ere we present the calculations related with a change of frame in quantum mechanics, where we explore resonance phenomena and present auxiliary calculations on the validity conditions of adiabaticity from different reference frame representations.

\section{Spin 1/2 particle in a rotating magnetic field} \label{ApFrameChangeQM-NMR}

Let us start from below Hamiltonian
\begin{align}
	H(t) = \frac{\hbar \omega_{0}}{2} \sigma_{z} + \frac{\hbar \omega_{1}}{2} \left[ \cos(\omega t)\sigma_{x} + \sin(\omega t)\sigma_{y}\right] \text{ , }
\end{align}
which describes the spin-$\frac{1}{2}$ particle in presence of a rotating magnetic field. The above dynamics is considered concerning the laboratory inertial frame. In such frame, the unitary system dynamics is governed by Von Neumann's equation given by
\begin{align}
	\dot{\rho}(t) = \frac{1}{i\hbar} [H(t),\rho(t)] \text{ . }
\end{align}

Due the oscillatory behavior of the transverse field, it is convenient to rewrite the system dynamics in rotating frame provided by the time-dependent unitary transformation $R(t) = e^{i\frac{\omega}{2}t\sigma_{z}}$. By doing that, we can implement a rotation in above equation as
\begin{align}
	R(t)\dot{\rho}(t)R^{\dagger}(t) = \frac{1}{i\hbar} R(t)[H(t),\rho(t)]R^{\dagger}(t) = \frac{1}{i\hbar} [R(t)H(t)R^{\dagger}(t) , R(t)\rho(t)R^{\dagger}(t)] \text{ . }
\end{align}
where we can define the system rotated state as $\rho_{\text{R}}(t) = R(t)\rho(t)R^{\dagger}(t)$, so that
\begin{align}
	R(t)\dot{\rho}(t)R^{\dagger}(t) &= \dot{\rho}_{\text{R}}(t) - \dot{R}(t)\overbrace{R^{\dagger}(t)R(t)}^{\1}\rho(t)R^{\dagger}(t) - R(t)\rho(t)\overbrace{R^{\dagger}(t)R(t)}^{\1}\dot{R}^{\dagger}(t) \nonumber \\
	&= \dot{\rho}_{\text{R}}(t) - \dot{R}(t)R^{\dagger}(t)\rho_{\text{R}}(t) - \rho_{\text{R}}(t)R(t)\dot{R}^{\dagger}(t) \text{ , }
\end{align}
in which we can use $\dot{R}(t)R^{\dagger}(t) = -R(t)\dot{R}^{\dagger}(t)$ to get
\begin{align}
	R(t)\dot{\rho}(t)R^{\dagger}(t) = \dot{\rho}_{\text{R}}(t)-[\dot{R}(t)R^{\dagger}(t),\rho_{\text{R}}(t)] \text{ . }
\end{align}
Thus, we have
\begin{align}
	\dot{\rho}_{\text{R}}(t) &= \frac{1}{i\hbar} [R(t)H(t)R^{\dagger}(t) , \rho_{\text{R}}(t)] + [\dot{R}(t)R^{\dagger}(t),\rho_{\text{R}}(t)] \nonumber \\
	&= \frac{1}{i\hbar} [R(t)H(t)R^{\dagger}(t) , \rho_{\text{R}}(t)] + \frac{i\hbar}{i\hbar}[\dot{R}(t)R^{\dagger}(t),\rho_{\text{R}}(t)] \text{ , }
\end{align}
and we conclude that the dynamics of the system in rotating frame will be given by
\begin{align}
	\dot{\rho}_{\text{R}}(t) = \frac{1}{i\hbar} [H_{\text{R}}(t) , \rho_{\text{R}}(t)] \text{ , }
\end{align}
where the new Hamiltonian in rotating frame is $H_{\text{R}}(t) = R(t)H(t)R^{\dagger}(t) + i\hbar \dot{R}(t)R^{\dagger}(t)$. By computing the new Hamiltonian
\begin{align}
	H_{\text{R}}(t) &= \frac{\hbar \omega_{0}}{2} e^{i\frac{\omega }{2}t\sigma_{z}}\sigma_{z}e^{-i\frac{\omega }{2}t\sigma_{z}} + \frac{\hbar \omega_{1}}{2} e^{i\frac{\omega }{2}t\sigma_{z}}\left[ \cos(\omega t)\sigma_{x} + \sin(\omega t)\sigma_{y}\right]e^{-i\frac{\omega }{2}t\sigma_{z}} + i\hbar \frac{d}{dt} \left[ e^{i\frac{\omega }{2}t\sigma_{z}}\right]e^{-i\frac{\omega }{2}t\sigma_{z}}
	\nonumber \\
	&= \frac{\hbar \omega_{0}}{2} \sigma_{z} + \frac{\hbar \omega_{1}}{2} \left[ \cos(\omega t)e^{i\frac{\omega }{2}t\sigma_{z}}\sigma_{x}e^{-i\frac{\omega }{2}t\sigma_{z}} + \sin(\omega t) e^{i\frac{\omega }{2}t\sigma_{z}}\sigma_{y}e^{-i\frac{\omega }{2}t\sigma_{z}}\right] - \hbar\frac{\omega }{2}\sigma_{z} \text{ , }
\end{align}
now we can use that
\begin{align}
	e^{i\frac{\omega }{2}t\sigma _{z}}\sigma _{x}e^{-i\frac{\omega }{2}t\sigma _{z}} &= \cos \left( \omega t\right) \sigma_{x}-\sin \left( \omega t\right) \sigma _{y} \text{ , } \\
	e^{i\frac{\omega }{2}t\sigma _{z}}\sigma _{y}e^{-i\frac{\omega }{2}t\sigma _{z}} &= \cos \left( \omega t\right) \sigma_{y} + \sin \left( \omega t\right) \sigma _{x} \text{ , }
\end{align}
to get the time-independent Hamiltonian
\begin{align}
	H_{\text{R}}(t) = H_{\text{R}} &= \frac{\hbar \omega_{0}-\omega}{2} \sigma_{z} + \frac{\hbar \omega_{1}}{2} \sigma_{x} \text{ . }
\end{align}

Thus, the solution of the dynamics in rotating frame reads as
\begin{align}
	\rho_{\text{R}}(t) = \exp\left[ -\frac{i}{\hbar} H_{\text{R}}t \right] \rho(0) \exp\left[ \frac{i}{\hbar} H_{\text{R}}t \right]\text{ , }
\end{align}
so that in non-rotating one we get
\begin{align}
	\rho (t) = R(t)\rho_{\text{R}}(t)R^{\dagger}(t) = e^{i\frac{\omega }{2}t\sigma_{z}}e^{-\frac{i}{\hbar} H_{\text{R}}t } \rho(0) e^{\frac{i}{\hbar} H_{\text{R}}t }e^{-i\frac{\omega }{2}t\sigma_{z}} \text{ . }
\end{align}

\section{Resonance in a two-level system driven by oscillating field}\label{ApFrameChangeQM-TLSOscila}

Let us consider a Landau-Zener type Hamiltonian given by
\begin{align}
	H(t) = \frac{\hbar \omega_{0}}{2}  \sigma_{z} + \frac{\hbar \omega_{\text{x}}}{2} \sin(\omega t)\sigma_{x} \text{ .} \label{ApEqOscHamil}
\end{align}
We can rewrite Eq.~(\ref{ApEqOscHamil}) as
\begin{align}
	H(t) = \frac{\hbar \omega_{0}}{2} \left[ \sigma_{z} + \tan \theta \sin(\omega t)\sigma_{x} \right] \text{ , } \label{ApResH}
\end{align}
where $\theta = \arctan [\omega_{\text{x}}/\omega_{0}]$. It is not obvious that the Hamiltonian $H(t)$ in Eq.~\eqref{ApResH} exhibits a resonant behavior. In order to see this fact, let us define an time-dependent oscillating frame $R(t) = e^{-i\frac{\omega}{2}t\sigma_{z}}$. In this oscillating frame, the Hamiltonian is given by
\begin{align}
	H_{\text{R}}(t) = R(t) H(t) R^{\dagger}(t) + i R(t)\dot{R}^{\dagger}(t) = \hbar\frac{\omega_{0} - \omega}{2} \sigma_{z} + \hbar \frac{\sin ( \omega t) \tan \theta}{2} \vec{\omega}_{xy}(t)\cdot\vec{\sigma}_{xy} \text{ , } \label{ApEqOscHamilNonInertial}
\end{align}
with $\vec{\omega}_{xy}(t) = \omega_{0}[\cos(\omega t) \hat{x} - \sin(\omega t) \hat{y}]$ and $\vec{\sigma}_{xy} = \sigma_{x} \hat{x} + \sigma_{y} \hat{y}$.

\section{Proof of the theorems}\label{AppProofTheoremsAdiab}

Let us consider two Hamiltonians, an inertial frame Hamiltonian $H(t)$ and its non-inertial counterpart $H_{\Ocalb}(t)$, which are related by a time-dependent unitary $\Ocalb(t)$. The dynamics associated with Hamiltonians $H(t)$ and $H_{\Ocalb}(t)$ are given by
\begin{align}
	\dot{\rho}(t) &= \frac{1}{i\hbar}[H(t),\rho(t)] \text{ , } \label{ApSchodingerEq} \\
	\dot{\rho}_{\Ocalb}(t) &= \frac{1}{i\hbar}[H_{\Ocalb}(t),\rho_{\Ocalb}(t)] \text{ , } \label{ApSchodingerEqRot}
\end{align}
where $H_{\Ocalb}(t) = \Ocalb(t)H(t)\Ocalb^{\dagger}(t) + i\hbar \dot{\Ocalb}(t)\Ocalb^{\dagger}(t) $ and $\rho_{\Ocalb}(t) = \Ocalb(t)\rho(t)\Ocalb^{\dagger}(t)$. Then, the connection between the evolved states $\ket{\psi(t)}$ and $\ket{\psi_{\Ocalb}(t)}$ in inertial and non-inertial frames, respectively, is given by $\ket{\psi_{\Ocalb}(t)} = \Ocalb(t)\ket{\psi(t)}$, $\forall t \in [t_0,\tau]$. By considering the initial state in inertial frame given by a single eigenstate of $H(t)$, namely $\ket{\psi(t_0)}=\ket{E_{k}(t_0)}$, the adiabatic dynamics in this frame is written as
\begin{align}
	\ket{\psi(t)} = e^{i\int_{t_{0}}^{t} \theta_{k}(\xi)d\xi}\ket{E_{k}(t)} \text{ , }
\end{align}
where $\theta_{k}(t) = -E_k(t)/\hbar + i \langle E_k(t) | (d/dt) | E_k(t)\rangle$ is the adiabatic phase composed by the dynamical and geometrical phase, respectively. On the other hand, an adiabatic behavior is obtained in non-inertial frame if and only if
\begin{align}
	|\interpro{E^{\Ocal}_{m}(t)}{\psi_{\Ocalb}(t)}| = |\interpro{E^{\Ocalb}_{m}(t_{0})}{\psi_{\Ocalb}(t_{0})}| \text{ , } \forall m \, , \,  \forall t \in [t_0,\tau] \text{ . }
\end{align}
Therefore, we can write
\begin{align}
	|\interpro{E^{\Ocalb}_{m}(t)}{\psi_{\Ocalb}(t)}| &= |\interpro{E^{\Ocalb}_{m}(t_{0})}{\psi_{\Ocalb}(t_{0})}| \text{ , }  \nonumber \\
	|\bra{E^{\Ocalb}_{m}(t)}\Ocalb(t)\ket{\psi(t)}| &= |\bra{E^{\Ocalb}_{m}(t_{0})}\Ocalb(t_{0})\ket{\psi(t_{0})}| \text{ , } \nonumber \\
	|\bra{E^{\Ocalb}_{m}(t)}\Ocalb(t)\ket{E_{k}(t)}| &= |\bra{E^{\Ocalb}_{m}(t_{0})}\Ocalb(t_{0})\ket{E_{k}(t_{0})}| \text{ . } \label{AppCond}
\end{align}
Thus, Eq.~(\ref{AppCond})  establishes a necessary and sufficient condition to obtain an adiabatic evolution in the non-inertial frame,
assuming an adiabatic evolution in the original frame.
To conclude our proof, let us consider the converse case, where the system starts in a eigenstate of $\ket{E^{\Ocal}_{m}(t_{0})}$ in non-inertial frame. If the dynamics is adiabatic we write
\begin{align}
	\ket{\psi_{\Ocalb}(t)} = e^{i\int_{t_{0}}^{t} \theta^{\Ocalb}_{m}(\xi)d\xi}\ket{E^{\Ocalb}_{m}(t)} \text{ , }
\end{align}
where $\theta^{\Ocalb}_{m}(t)$ is the adiabatic phase collected in this frame. The dynamics will be adiabatic in the inertial frame if and only if
\begin{align}
	|\interpro{E_{m}(t)}{\psi(t)}| = |\interpro{E_{m}(t_{0})}{\psi(t_{0})}| \text{ , } \forall m \, , \,  \forall t \in [t_0,\tau] \text{ . }
\end{align}

Therefore, by using the same procedure as before, we get the condition
\begin{align}
	|\bra{E_{k}(t)}\Ocalb^{\dagger}(t)\ket{E^{\Ocalb}_{m}(t)}| &= |\bra{E_{k}(t_{0})}\Ocalb^{\dagger}(t_{0})\ket{E^{\Ocalb}_{m}(t_{0})}| \text{ , }
\end{align}
which is equivalent to Eq.~\eqref{AppCond}. This ends the proof of the Theorem~\ref{TheoAdiab}. Now, let us proof the second theorem discussed in this thesis.

To this end, we need to consider a time-independent Hamiltonian $H_{\Ocal}$ in the non-inertial frame, so that its evolution operator can be written as $U_{\Ocalb}(t,t_{0}) = e^{-\frac{i}{\hbar} H_{\Ocalb}(t-t_{0})}$. Thus, we can write the dynamics in non-inertial frame as
\begin{align}
	\ket{\psi_{\Ocalb}(t)} = e^{-\frac{i}{\hbar} H_{\Ocalb}(t-t_{0})}\ket{\psi_{\Ocalb}(t_{0})} \text{ . } \label{ApDyn}
\end{align}
Moreover, assuming adiabatic dynamics in the inertial frame, we get
\begin{align}
	| \interpro{E_{k}(t)}{\psi(t)} | = |\interpro{E_{k}(t_{0})}{\psi(t_{0})}| \text{ . }
\end{align}
By using the relationship between inertial and non-inertial frames as $\ket{\psi(t)} = \Ocalb^{\dagger}(t)\ket{\psi_{\Ocalb}(t)}$, we can write
\begin{align}
	| \bra{E_{k}(t)}\Ocalb^{\dagger}(t)e^{-\frac{i}{\hbar} H_{\Ocalb}(t-t_{0})}\ket{\psi_{\Ocalb}(t_{0})} | = |\interpro{E_{k}(t_{0})}{\psi(t_{0})}| \text{ , }
\end{align}
where we have used the Eq.~\eqref{ApDyn}. Now, by taking $\ket{\psi_{\Ocalb}(t_{0})} = \Ocalb(t_0)\ket{\psi(t_{0})}$, we obtain
\begin{align}
	| \bra{E_{k}(t)}\Ocalb^{\dagger}(t)e^{-\frac{i}{\hbar} H_{\Ocalb}(t-t_{0})}\Ocalb(t_{0})\ket{\psi(t_{0})} | = |\interpro{E_{k}(t_{0})}{\psi(t_{0})}| \label{t2af} \text{ . }
\end{align}
Thus, by inserting the initial state $\ket{\psi(t_{0})} = \ket{E_{n}(t_{0})}$ in Eq.~(\ref{t2af}), we get
\begin{align}
	| \bra{E_{k}(t)}\Ocalb^{\dagger}(t)e^{-\frac{i}{\hbar} H_{\Ocalb}(t-t_{0})}\Ocalb(t_{0})\ket{E_{n}(t_{0})} | = |\interpro{E_{k}(t_{0})}{E_{n}(t_{0})}| \text{ . }
\end{align}
This concludes the proof of Theorem~\ref{TheoAdiabTI}.

\subsection{Application to the oscillating Hamiltonian}\label{AppAppliTheorem1}

In order to show how the Theorem~\ref{TheoAdiab} allows us to predict the adiabatic behavior of the system driven by the Hamiltonian in Eq.~\eqref{ApEqOscHamil}, let us first consider a generic system under action of a single time-dependent oscillating field with characteristic
frequency $\omega$, whose Hamiltonian reads
\begin{align}
	H(\omega,t) = \hbar \omega_{0} H_{0} + \hbar \omega_{\text{T}} H_{\text{T}}(\omega,t) \text{ , } \label{ApHOrig}
\end{align}
where we consider the transverse term $\hbar \omega_{\text{T}} H_{\text{T}}(\omega,t)$ as a perturbation, so that $||\omega_{0} H_{0}|| \gg ||\omega_{\text{T}} H_{\text{T}}(\omega,t)||$, $\forall t\in [0,\tau]$. In this case, the eigenstates $\ket{E_{n}(t)}$ and energies $E_{n}(t)$ of $H(\omega,t)$ can be obtained as perturbation of eigenstates $\ket{E_{n}^{0}}$ and energies $E_{n}^{0}$ of $\hbar \omega_{0} H_{0}$ as (up to a normalization coefficient)
\begin{align}
	\ket{E_{n}(t)} &= \ket{E_{n}^{0}} + \Ocalb \left(||\hbar \omega_{\text{T}} H_{\text{T}}(\omega,t)|| \right) \text{ , } \label{ApvecnonRot}\\
	E_{n}(t) &= E_{n}^{0} + \Ocalb \left(||\hbar \omega_{\text{T}} H_{\text{T}}(\omega,t)|| \right) \text{ . } \label{ApEnenonRot}
\end{align}
On the other hand, in the non-inertial frame, we have $H_{\Ocalb}(t) = \Ocalb(t)H(t)\Ocalb^{\dagger}(t) + i\hbar \dot{\Ocalb}(t)\Ocalb^{\dagger}(t) $, which yields
\begin{align}
	H_{\Ocalb}(\omega,t) = \hbar \left( \omega_{0} - \omega \right) H_{0} + \hbar \omega_{\text{T}} H_{\Ocalb , \text{T}}(\omega,t) \text{ . } \label{ApHrot}
\end{align}
where $H_{\Ocalb , \text{T}}(\omega,t) = \Ocalb(t)H_{\text{T}}(\omega,t)\Ocalb^{\dagger}(t)$. Now, we separately consider two specific cases:

$\bullet$ \emph{Far-from resonance situation $|\omega_{0} - \omega|\gg |\omega_{\text{T}}|$}: In this case, the term $\hbar \omega_{\text{T}} H_{\Ocalb , \text{T}}(\omega,t)$ in Eq.~\eqref{ApHrot} works as a perturbation. Therefore the set of eigenvectors of $H_{\Ocalb}(\omega,t)$ reads
\begin{align}
	\ket{E^{\Ocalb}_{n}(t)} &= \ket{E_{n}^{0}} + \Ocalb \left(||\hbar \omega_{\text{T}} H_{\text{T}}(\omega,t)|| \right) \text{ , } \label{ApvecRot}
\end{align}
where we have used that the energy gaps $\tilde{E}_{n}^{0} - \tilde{E}_{k}^{0}$ of the Hamiltonian $\hbar \left( \omega_{0} - \omega \right) H_{0}$ are identical to energy gaps $E_{n}^{0} - E_{k}^{0}$ of $\hbar \omega_{0} H_{0}$ and $||\hbar \omega_{\text{T}} H_{\Ocalb , \text{T}}(\omega,t)|| = ||\hbar \omega_{\text{T}} H_{\text{T}}(\omega,t)||$. Thus,
from Eqs.~\eqref{ApvecnonRot} and~\eqref{ApvecRot} we conclude, for any eigenstate  $\ket{E_{k}{(t)}}$, that
\begin{align}
	\bra{E^{\Ocalb}_{m}(t)}\Ocalb(t)\ket{E_{k}(t)} \approx e^{i\frac{\omega}{\omega_{0}} \frac{E_{k}^{0}}{\hbar} t}\delta_{mk} \text{ , }
\end{align}
so that we get $|\bra{E^{\Ocalb}_{m}(t)}\Ocalb(t)\ket{E_{k}(t)}| = \text{constant}$, $\forall m$, $\forall t\in [t_{0},\tau]$.

$\bullet$ \emph{Resonance situation $|\omega_{0} - \omega| \ll |\omega_{\text{T}}|$}: Now, we have a more subtle situation.
Firstly, we can use Eqs.~\eqref{ApvecnonRot} and~\eqref{ApEnenonRot} to write
\begin{align}
	\Ocalb (t)\ket{E_{n}(t)} &= e^{i\frac{\omega}{\omega_{0}} \frac{E_{n}^{0}}{\hbar}t } \ket{E_{n}^{0}} + \Ocalb \left(||\hbar \omega_{\text{T}} H_{\text{T}}(\omega,t)|| \right) \text{ , } \\
	\int_{t_{0}}^{t}\theta_{n}(\xi)d\xi &= -\frac{E_{n}^{0}}{\hbar}(t-t_{0}) + \Ocalb \left(||\hbar \omega_{\text{T}} H_{\text{T}}(\omega,t)|| \right)  \text{ , }
\end{align}
so that
\begin{align}
	\bra{E^{\Ocalb}_{m}(t)}\Ocalb(t)\ket{E_{k}(t)} \approx e^{i\frac{\omega}{\omega_{0}} \frac{E_{k}^{0}}{\hbar} t} \interpro{E^{\Ocalb}_{m}(t)}{E^{0}_{k}} \text{ . }
\end{align}
Now, it is possible to see that if $|\interpro{E^{\Ocalb}_{m}(t)}{E^{0}_{k}}| = |\interpro{E^{\Ocalb}_{m}(t_{0})}{E^{0}_{k}}|$, $\forall t\in [t_{0},\tau]$, then the theorem is satisfied. By applying this results to the Hamiltonian in Eq.~\eqref{ApEqOscHamilNonInertial} we can see that for $|\omega_{0} - \omega| \ll |\omega_{\text{T}}|$ we get $H_{\Ocalb}(\omega,t) \propto \cos(\omega t) \sigma{x} - \sin(\omega t) \sigma{y}$ and
\begin{align}
	\ket{E^{\Ocalb}_{\pm}(t)} = \frac{\ket{1} \pm e^{i \omega t} \ket{0}}{\sqrt{2}} \text{ . }
\end{align}

Since $\ket{E^{0}_{k}} = \ket{k}$, with $k=\{0,1\}$, we conclude that $|\interpro{E^{\Ocalb}_{m}(t)}{E^{0}_{k}}| = |\interpro{E^{\Ocalb}_{m}(t_{0})}{E^{0}_{k}}|$.

\chapter{Experimental setup} \label{ApExperimentalSetups}

\section{The Ytterbium trapped ion experimental setup} \label{ApTrappeIon}

In this section we present more details on the first system used to experimentally study some theoretical results presented in this thesis. In particular, we are interested here in providing a detailed discussion about the main elements used in this thesis, so that the readers could better understand our results. For a complete discussion, we indicate the PhD thesis by Chang-Kang Hu~\cite{Hu:Thesis} (thesis in Chinese, read it moderately!).

\subsection{Foundations of ion trapping and qubit initialization} \label{ApTrappeIonQubit}

In general, the trapping mechanism of neutral or charged particles is dependent on the particle to be trapped, among many other factors, so here we provide a discussion on the particular system used in this thesis~\cite{Hu:Thesis}. However, independent on the system, to trap a particle (charged or not) we use electric and magnetic fields. A general result, called Earnshaw's theorem~\cite{Earnshaw:42}, allows us to identify the minimum requirement to trap a particle in a stable way, more specifically, a charged particle cannot be held in a stable equilibrium by using static fields. Even though very general, the demonstration of this result is very simple and it immediately comes from the Gauss equations for electric and magnetic static fields $\vec{\nabla} \cdot \vec{E} = 0$ and $\vec{\nabla} \cdot \vec{B} = 0$. These equations tell us that the electric/magnetic field flux going through a region (volume $V$) is null, so a charged particle under action of this field will undergoes a non-confining state.

It is reasonable to think about the trapping task for slow motion particles (less energetic), otherwise we need strong fields to trap them around a specific position. For atomic beam that comes from a heat gas, we need to reduce the beam temperature by using a procedure called Doppler cooling~\cite{Dehmelt:75,Hansch:75}. The Doppler cooling is a technique used to slow down the motion of each atom by using a laser beam tuned slightly below at the natural resonance frequency of the atom energy level. In a Ytterbium ion $^{171}$Yb$^+$, the relevant energy levels populated during the cooling procedure are shown in Fig.~\ref{FigReleEnerLev}~\cite{Hu:Thesis,Olmschenk:07}. For the ion $^{171}$Yb$^+$ the Doppler cooling is achieved through a $369.5$~nm laser beam, where the cooling is obtained from the optical transition cycle $^{2}S_{1/2}\!\rightleftarrows\!^{2}P_{1/2}$~\cite{Olmschenk:07}. This cycle transition happens because the system at state $^{2}S_{1/2}$ absorbs a photon and it becomes excited at the state $^{2}P_{1/2}$, then the atom spontaneously emits a photon in a aleatory direction and uniformly probability distribution. In particular, if the system decays to state the $^{2}S_{1/2}$, the system remain in cooling cycle. However, as shown in Fig.~\ref{FigReleEnerLev}, during the transition cycle $^{2}S_{1/2}\!\rightleftarrows\!^{2}P_{1/2}$ there is a branching ratio $R$ for population decay from $^{2}P_{1/2}$ state to the $^{2}D_{3/2}$, where the decay probability has been theoretically studied in Ref.~\cite{Migdalek:80}, with experimentally verification in Refs.~\cite{Olmschenk:07,Yu:00}. In order to sent the system back to the cooling cycle, a light at $935.2$~nm is used to promote transitions $^{2}D_{3/2}\!\rightleftarrows\!^{3}D[3/2]_{1/2}$, where the system can quickly decay from $^{3}D[3/2]_{1/2}$ to $^{2}S_{1/2}$. Therefore, by driving the system through a many cooling cycles process we can achieve the cooling limit (minimal temperature). This limit is known as \textit{Doppler cooling limit} and given it predicts the minimal temperature $T_{\text{D}} = \hbar \Gamma / 2k_{\text{B}}$, where $\Gamma$ is a characteristic parameter associated to system (transition linewidth) and $k_{\text{B}}$ is the Boltzmann's constant~\cite{Letokhov:77,Wineland:79,Chu:98,Cohen:98,Phillips:98}. For the Ytterbium ion $^{171}$Yb$^+$ we have $\Gamma = 2 \pi \times 19.6$~MHz, so that we can find $T_{\text{D}} \approx 468.5~\mu$K. Thus, since we are interested in applications in quantum information processing, after Doppler cooling step a single ion is trapped by the radio frequency electric field.

\begin{figure}
	\centering
	\includegraphics[scale = 0.6]{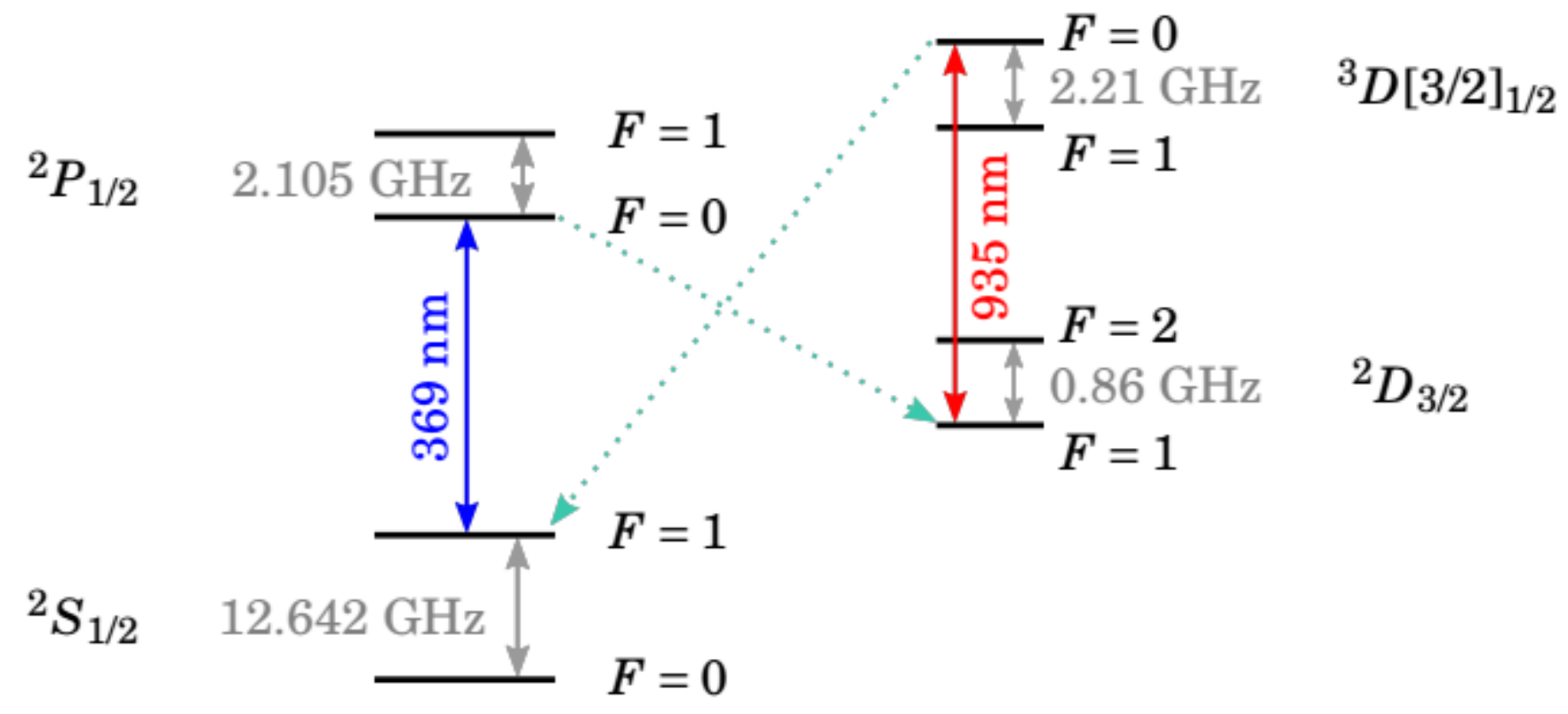}
	\caption{Relevant energy levels of the Ytterbium ion $^{171}$Yb$^+$ used in the experimental implementations in this thesis.}
	\label{FigReleEnerLev}
\end{figure}

After the cooling and trapping process, the qubit is initialized and then can be manipulated in order to process quantum information. Our qubit is encoded within energy subspace $^{2}S_{1/2}$ spanned by the states $^{2}S_{1/2}\ket{F=0,m=0}$, $^{2}S_{1/2}\ket{F=1,m=0}$ and  $^{2}S_{1/2}\ket{F=1,m=\pm 1}$, where $F$ denotes the total angular momentum of the atom and $m_{F}$ is its projection along the quantization axis. In absence of an external fields, the subspace $F=1$ is degenerate. Therefore, when we apply an external static magnetic field $\vec{B}$ the Zeeman effect in order to highlight the hyperfine structure splitting, so that an amount of energy $\delta E \propto ||\vec{B}||$ separates the energy eigenstates with $F=1$, as shown in Fig.~\ref{FigAdiabExpTrappedIon}. Thus, we identify subspace of two states that can be used as a two-level system, i.e., as our desired qubit. The Zeeman shift in transition frequency is given by $\delta_{\text{Z}} = 1.4$~MHz/G and the transition frequency of the qubit reads as $\omega_{\text{s}}^{0} \approx 12.64$~GHz. It is important to mention that the frequency $\omega_{\text{s}}^{0}$ is smoothly affected by the external magnetic field as $\omega_{\text{s}} = \omega^{0}_{\text{s}} + \delta_{\text{2z}}$, where $\delta_{\text{2z}} \approx 310.8 B^2$~Hz is the second order Zeeman shift ($B$ is the intensity of the static magnetic field in Gauss)~\cite{Hu:Thesis}. Thus, in the experiments considered in this thesis, an additional shift about $2 \pi \times 10$~KHz is added. Therefore, the qubit is encoded in states $\ket{0}\!\equiv\!^{2}S_{1/2}\ket{F=0,m=0}$ and $\ket{1}\!\equiv\!^{2}S_{1/2}\ket{F=1,m=0}$, where the energy gap reads as $E_{01} = \hbar \omega_{\text{s}}$.

\subsection{Trapped ion qubit dynamics} \label{ApManipulationQubit}

Given the natural configuration of the system and the set of states spanned by $^{2}S_{1/2}\ket{F=0,m=0}$ and $^{2}S_{1/2}\ket{F=1,m=0}$, with energies $E_{0}$ and $E_{1}$, respectively, we can define the bare Hamiltonian for our qubit as
\begin{align}
	H_{0} = E_{0} \ket{0}\bra{0} + E_{1} \ket{1}\bra{1} = \frac{E_{01}}{2} \sigma_{z} = \frac{\hbar \omega_{\text{s}}}{2} \sigma_{z} \text{ , }
\end{align}
where $\sigma_{z} = \ket{1}\bra{1} - \ket{0}\bra{0}$ is the Pauli matrix. As schematically shown in Fig.~\ref{FigAdiabExpTrappedIon}, by using an arbitrary waveform generator (AWG) we can manipulate the qubit
through a unitary dynamics by driving the system with an magnetic field $\vec{B}_{\text{un}}(t) = \vec{B}_{0}\cos( \omega t + \phi) $, one uses the dipole approximation to write the contribution of this interaction to energy of the system as $H_{1} = - \vec{\mu} \cdot \vec{B}_{\text{un}}(t)$, where $\vec{\mu} = \vec{\mu}_{01}\ket{0}\bra{1} + \vec{\mu}_{10}\ket{1}\bra{0}$~\cite{Leibfried:03,Wineland:98}. Here we use the subscript ``un" to denote the field associated with a unitary dynamics, later we shall see that the same AWG can be used to design a non-unitary dynamics through a field $\vec{B}_{\text{n-un}}(t)$. Therefore, we find
\begin{align}
	H_{1}(t) = -\vec{\mu}_{01}\cdot\vec{B}_{0}\cos( \omega t + \phi ) \ket{0}\bra{1} - \vec{\mu}_{10}\cdot\vec{B}_{0}\cos( \omega t + \phi)\ket{1}\bra{0} \text{ . } \label{ApH1Control}
\end{align}

Then, for simplicity we use that $\vec{\mu}_{01} = \vec{\mu}_{10}=\vec{\mu}$, so that the total Hamiltonian $H(t) = H_{0} + H_{1}$ reads as
\begin{align}
	H(t) = \frac{\hbar \omega_{\text{s}}}{2} \sigma_{z} + \hbar \omega_{\text{R}} \cos( \omega t + \phi) \sigma_{x}  \text{ , }
\end{align}
where we defined the Rabi frequency $\omega_{\text{R}} = - \vec{\mu}_{01}\cdot\vec{B}_{0}/\hbar$ and the Pauli matrix $\sigma_{x} = \ket{0}\bra{1}+\ket{1}\bra{0}$. Now, by using perturbation theory we can write the dynamical equation to the coefficients $c_{0}(t)$ and $c_{1}(t)$ associated with the evolved state
\begin{align}
	\ket{\psi(t)} = c_{0}(t) \ket{0} + c_{1}(t) \ket{1} \text{ . }
\end{align}

In fact, by writing $\ket{\psi_{\text{I}}(t)} = e^{iH_{0}t/\hbar} \ket{\psi(t)}$ we get the dynamics for $\ket{\psi_{\text{I}}(t)}$ as 
\begin{align}
	\ket{\dot{\psi}_{\text{I}}(t)} = i\hbar V_{\text{I}}(t) \ket{\psi_{\text{I}}(t)} \text{ , }
\end{align}
with $V_{\text{I}}(t) = e^{iH_{0}t/\hbar} H_{1}(t) e^{-iH_{0}t/\hbar}$. As a first remark, note that $[H_{1}(t_{1}),H_{1}(t_{2})]=0$ for any $(t_1,t_2)$, therefore we can also write $[V_{\text{I}}(t_{1}),V_{\text{I}}(t_{2})]=0$ for any $(t_1,t_2)$. This tells us that the solution for the above equation reads as
\begin{align}
	\ket{\psi_{\text{I}}(t)} = U_{\text{I}}(t) \ket{\psi_{\text{I}}(0)} \text{ , }
\end{align}
with
\begin{align}
	U_{\text{I}}(t) = \exp \left[ \frac{1}{i\hbar} \int_{0}^{t} V_{\text{I}}(\xi) d\xi \right] \text{ . } \label{AppUI}
\end{align}

We use that
\begin{align}
	V_{\text{I}}(t) &= e^{iH_{0}t/\hbar} H_{1}(t) e^{-iH_{0}t/\hbar} = \frac{\hbar \omega_{\text{R}}}{2} \left( e^{i(\omega t + \phi)} + e^{-i(\omega t + \phi)}\right) e^{iH_{0}t/\hbar}\sigma_{x}e^{-iH_{0}t/\hbar} \nonumber \\
	&= \frac{\hbar \omega_{\text{R}}}{2} \left( e^{i(\omega t + \phi)} + e^{-i(\omega t + \phi)}\right) e^{i\omega_{\text{s}}t/2} \left(\ket{0}\bra{1}+\ket{1}\bra{0}\right)e^{-i\omega_{\text{s}}t/2} \nonumber \\
	&= \frac{\hbar \omega_{\text{R}}}{2} \left( e^{i(\omega t + \phi)} + e^{-i(\omega t + \phi)}\right) 
	\left(e^{-i\omega_{\text{s}}t} \ket{0}\bra{1} + e^{i\omega_{\text{s}}t}\ket{1}\bra{0}\right) \nonumber \\
	&= \frac{\hbar \omega_{\text{R}}}{2} \left( e^{i(\omega_{-} t + \phi)} + e^{-i(\omega_{+} t + \phi)}\right) \ket{0}\bra{1}
	+\frac{\hbar \omega_{\text{R}}}{2} \left( e^{i(\omega_{+} t + \phi)} + e^{-i(\omega_{-} t + \phi)} \right)\ket{1}\bra{0} \text{ , }
\end{align}
where $\omega_{\pm} = \omega - \omega_{\text{s}}$. Therefore, by integrating the above expression
\begin{align}
	\int_{0}^{t} V_{\text{I}}(\xi) d\xi &= \frac{\hbar\omega_{\text{R}}}{2} \left[ 
	\frac{-ie^{i \phi} \left( e^{i\omega_{-} t} -1 \right)}{\omega_{-}} 
	+ \frac{ie^{-i \phi} \left( e^{-i\omega_{+} t} -1 \right)}{\omega_{+}} \right] 
	\ket{0}\bra{1} \nonumber \\
	&+\frac{\hbar\omega_{\text{R}}}{2}\left[\frac{-ie^{i \phi} \left( e^{i\omega_{+} t} -1 \right)}{\omega_{+}}
	+ \frac{ie^{-i \phi} \left( e^{-i\omega_{-} t} -1 \right)}{\omega_{-}}  \right] 
	\ket{1}\bra{0} \text{ . }
\end{align}

Thus, for a highly oscillating external field, we can approximate the above expression to
\begin{align}
	\int_{0}^{t} V_{\text{I}}(\xi) d\xi &\approx \frac{\hbar\omega_{\text{R}}}{2} \left[ 
	\frac{-ie^{i \phi} \left( e^{i\omega_{-} t} -1 \right)}{\omega_{-}} 
	\ket{0}\bra{1} + \frac{ie^{-i \phi} \left( e^{-i\omega_{-} t} -1 \right)}{\omega_{-}} \ket{1}\bra{0} \right] \text{ . }
\end{align}

For this reason, in our analysis we can neglect highly oscillating terms, as we shall use soon. Now, by writing the Hamiltonian $H(t)$ in rotating frame by using $R(t) = e^{i \omega t\sigma_{z}/2}$ we find
\begin{align}
	H_{R} &= R(t) H(t) R^{\dagger}(t) + i\hbar \dot{R}(t)R^{\dagger}(t) = \frac{\hbar \Delta}{2} \sigma_{z} + R(t) H_{1}(t) R^{\dagger}(t) \nonumber \\ 
	&= \frac{\hbar \Delta}{2} \sigma_{z} + \frac{\hbar \omega_{\text{R}}}{2} \left( e^{i(\omega t + \phi)} + e^{-i(\omega t + \phi)}\right) \left(e^{-i\omega t} \sigma_{-} + e^{i\omega t} \sigma_{+}\right)  \text{ , }
\end{align}
where we define $\Delta = \omega_{\text{s}} - \omega$ and we already used that $e^{i \omega t\sigma_{z}/2}\sigma_{x}e^{-i \omega t\sigma_{z}/2} = e^{-i\omega t} \sigma_{-} + e^{i\omega t} \sigma_{+}$, where $\sigma_{-} = \ket{0}\bra{1}$ and $\sigma_{+} = \sigma_{-}^{\dagger}$. Now, by neglecting the highly oscillating terms, it is possible to proof that we get
\begin{align}
	H_{R} &= \frac{\hbar \Delta}{2} \sigma_{z} + \frac{\hbar \omega_{\text{R}}}{2} \left( e^{i\phi} \sigma_{-} + e^{-i\phi} \sigma_{+}\right) \text{ . }
\end{align}

Now, by using that $\sigma_{\pm} = (1/2) (\sigma_{x} \pm i \sigma_{y})$, we find
\begin{align}
	H_{R} &= \frac{\hbar \Delta}{2} \sigma_{z} + \frac{\hbar \omega_{\text{R}}}{2} \left( \cos \phi \sigma_{x} + \sin \phi \sigma_{y}\right) \text{ . } \label{ApGenHamiltoTrapped}
\end{align}

Therefore, the above Hamiltonian can be used to implement an arbitrary dynamics of the qubit encoded as shown in Fig.~\ref{FigAdiabExpTrappedIon}.

\subsection{State detection and tomography process} \label{ApIonTomography}

To read out the outcome of some operation in the trapped ion system used in this thesis, it is used a ion fluorescence technique~\cite{Berkeland:98}. In Ytterbium trapped ion systems, this technique is applied by brightening the ion with light at $369.53$~nm, where we promote transitions from the state $^{2}S_{1/2}\ket{F=1}$ to $^{2}P_{1/2}\ket{F=0}$~\cite{Olmschenk:07}. This process allows us to determinate the population in energy level $^{2}S_{1/2}\ket{F=1}$, since the state $^{2}S_{1/2}\ket{F=0}$ is a \textit{dark state} for the laser at $369.53$~nm. After repeating this procedure many times, it is possible to accurately estimate the probability in \textit{bright state} $^{2}S_{1/2}\ket{F=1}$ and the dark one. If some population is present in state $^{2}S_{1/2}\ket{F=0}$, few photons are 
detected due to the background scattering, and the dark counting of the photomultiplier tube (a high efficiency detector used to
measure the ion fluorescence during the state detection). On the other hand, a large number of photons are scattered by the system when the system is in state $^{2}S_{1/2}\ket{F=1}$. These scattered photons are collected by a numerical aperture $NA=0.4$ objective lens. After a processing of the captured photons, it is possible to estimate that the measurement fidelity is around $99.4\%$~\cite{Hu:18}.

In terms of state tomography, the above procedure is used to estimate the values of the quantities $|\alpha|$ and $|\beta|$ of an unknown state $\ket{\psi} = \alpha \ket{0} + \beta \ket{1}$. Therefore, to measure the system state we need to follow the tomography standard process where we need to measure a number of different parameters in order to reconstruct the density matrix of the system~\cite{Nielsen:Book}. In our case, the state tomography is obtained from three parameters, since an arbitrary density matrix for a two-level system can be written as in Eq.~\eqref{EqDensiMatrixDecompGen}.

\subsection{Controllable dephasing in trapped ion systems} \label{ApTrappeIonDeco}

In order to simulate the coupling of our qubit with a dephasing reservoir, it is used frequency modulation (FM) method where a Gaussian noise is introduced by mixing the field $\vec{B}_{1}(t)$. The experimental is sketched in Fig.~\ref{FigAdiabExpComTrappedIon}. After this process, the new magnetic field reads as (the experiment considered in this thesis sets $\phi = 0$ in Eq.~\eqref{ApGenHamiltoTrapped}.)
\begin{align}
	\vec{B}_{\text{n-un}}(t) = \vec{B}_0 \cos \left[\omega t + C \eta(t) t \right] \text{ , }
\end{align}
where the noise source is encoded in the function $\eta(t)=A_{0} g(t)$, with $A$ being the average amplitude of the noise, $g(t)$ is a random
analog voltage signal and $C$ is the modulation depth supported by the commercial microwave generator E8257D. Here we want to demonstrate that this Gaussian noise can be used to implement the dynamics given by
\begin{eqnarray}
	\dot{\rho}(t) = \frac{1}{i\hbar} [H(t),\rho(t)] + \gamma(t) \left[\sigma_z \rho(t) \sigma_z - \rho(t)\right] \text{ , }
	\label{ApEqLindbaldTarget}
\end{eqnarray}
with $\gamma(t)$ being the desired decohering rate. To this end, we consider here the time-independent Hamiltonian $H^{\prime} = \hbar \tilde{\omega} \sigma_{x}$ obtained from Eq.~\eqref{ApGenHamiltoTrapped} with $\Delta = 0$, $\omega_{\text{R}}=\tilde{\omega}$ and $\phi = 0$. That is, a resonant microwave with Rabi frequency $\tilde{\omega}$. Therefore, experimentally we can proof that the dynamics is given by the above equation in two different ways: i) by accompanying the dynamics of the system and ii) by computing the \textit{process matrix} $\chi$ that describes the quantum process. Here we will consider both protocols.

By taking into account a particular dynamics of a system driven by the Hamiltonian $H^{\prime}$, from Eq.~\eqref{ApEqLindbaldTarget} we can write the Bloch equations for $\rho(t)$ as (with $\gamma(t) = \gamma_{0}$)
\begin{align}
	\dot{r}_{x}(t) &= -2\gamma_{0} r_{x}(t) \text{ , } \\
	\dot{r}_{y}(t) &= -2\gamma_{0} \left[ r_{z}(t) + r_{y}(t) \right] \text{ , } \\
	\dot{r}_{z}(t) &= 2\gamma_{0} r_{y}(t) \text{ , }
\end{align}
whose solution for the initial condition $\ket{\psi(0)}=\ket{0}$ reads as
\begin{align}
	r_{x}(t) &= 0 \text{ , } \quad r_{y}(t) = \frac{ 2\omega e^{-\gamma_{0}t} }{\sqrt{4\tilde{\omega}^2 - \gamma_{0}}} \sin \left(t\sqrt{4\tilde{\omega}^2 - \gamma_{0}}\right) \text{ , } \\
	r_{z}(t) &= -e^{-\gamma_{0}t} \left[ \cos \left(t\sqrt{4\tilde{\omega}^2 - \gamma_{0}}\right) + \frac{ \gamma_{0}\sin \left(t\sqrt{4\tilde{\omega}^2 - \gamma_{0}}\right) }{\sqrt{4\tilde{\omega}^2 - \gamma_{0}}} \right] \text{ . }
\end{align}

Therefore, by defining the projector $\Pcal_{1} = \ket{1}\bra{1}$ computing the population in state excited state $\ket{1}$ we get
\begin{align}
	p_{1}(t) = \trs{\Pcal_{1}\rho(t)} = \frac{1}{2} - \frac{1}{2}e^{-\gamma_{0}t} \left[ \cos \left(t\sqrt{4\tilde{\omega}^2 - \gamma_{0}}\right) + \frac{ \gamma_{0}\sin \left(t\sqrt{4\tilde{\omega}^2 - \gamma_{0}}\right) }{\sqrt{4\tilde{\omega}^2 - \gamma_{0}}} \right] \text{ , }
\end{align}
where in the limit $\gamma_{0} \rightarrow 0$ we recover the well-known equation $p_{n}(t)|_{\gamma_{0}=0} = \sin^2 (\tilde{\omega} t)$. As we can see, the behavior of $p_{n}(t)$ is a damped oscillating function that converges to $1/2$ in the regime $\gamma_{0}t \gg 1$, where the damping rate is dictated by the parameter $\gamma_{0}$. 

\begin{figure}[t!]
	\centering
	\includegraphics[scale=0.6]{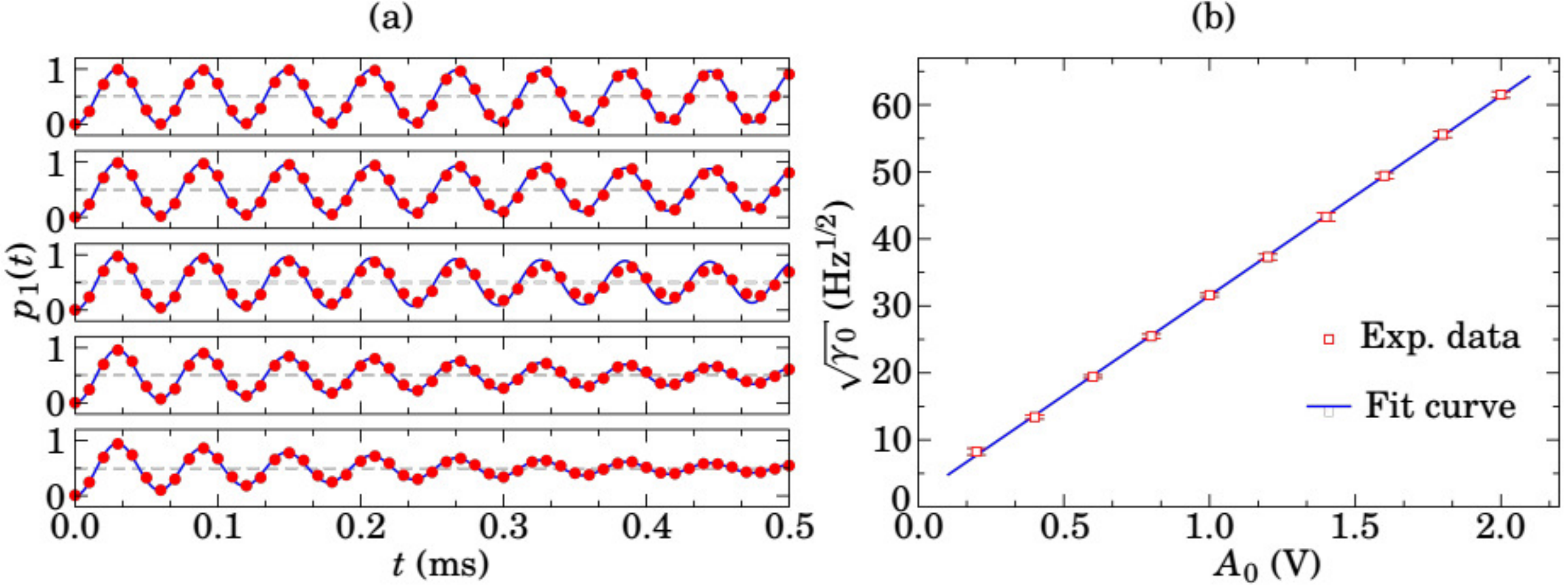}
	\caption{Dephasing rate controlled by the amplitude of noise. ({\color{blue}a}) Rabi oscillations between states $\ket{0}$ and $\ket{1}$ under  different noise intensities. From top to bottom, the noise amplitude $A_{0}$ is set to $0.4$~V, $0.8$~V, $1.2$~V, $1.6$~V and $2.0$~V, with the corresponding damping rates $\gamma$ are approximately $180$~Hz, $649$~Hz, $1391$~Hz, $2441$~Hz and $3785$~Hz, respectively. ({\color{blue}b}) Dephasing rate as a function of the noise amplitude. Points are measured data. A linear fit is obtained. Without driving noise (noise amplitude is zero), the dephasing rate of the qubit is fitted as 3.03 Hz, which is caused by the magnetic fluctuation in the laboratory.}
	\label{FigApsDecayRate}
\end{figure}

To show how to simulate the non-unitary dynamics in Eq.~\eqref{ApEqLindbaldTarget}, the experimental realization it was considered with the parameter $\tilde{\omega} = 2\pi \times 8.44$~KHz and the modulation depth $C$ as $96.00$~KHz/V, here both parameters $\tilde{\omega}$ and $C$ are kept fixed for all the experiments. Thus, by adjusting the amplitude noise $A_{0}$, different dynamics were implemented. We experimentally evaluated the quantity $p_{1}(t)$ to get the Rabi oscillations of the system for each dynamics during a total time interval of $\Delta t_{\text{exp}}=0.5$~ms. The result are shown in Fig.~\ref{FigApsDecayRate}{\color{blue}a} for the Rabi oscillations between states $\ket{0}$ and $\ket{1}$. For different values of the amplitude $A_{0}$, it is possible to see different behaviors, each one of them associated with distinct values of $\gamma_{0}$. Moreover, it is possible to experimentally evaluate the relation between the dephasing rate $\gamma$ and the noise amplitude $A_{0}$, where we used the graph in Fig.~\ref{FigApsDecayRate}{\color{blue}b} to this end. By evaluating the square root of $\gamma$ we find a linear behavior as $A_{0}$ increases, so through the linear regression method we find (approximately)
\begin{eqnarray}
	\sqrt{\gamma} = 29.806861 A_{0} + 1.739986 \text{ , } \label{EqApSqrtGamma}
\end{eqnarray}
and this show that a perfect control of decoherence rate in Eq.~\eqref{ApEqLindbaldTarget} is obtained from a precise control of amplitude noise.

As a second, and complementary, analysis of the decohering dynamics it was used the tomography process to reconstruct the process matrix $\chi$. Let us now give brief introduction to the matrix $\chi$ used to completely describe a physical process. A general physical process is adequately characterized by the equation~\cite{Nielsen:Book}
\begin{align}
	\rho (0) \rightarrow \rho (x) = \Ecalb_{x}[\rho (0)] = \sum\nolimits_{n=1}^{N} E_{n} \rho (0) E_{n}^{\dagger} \text{ , }
\end{align}
where $E_{n}$ is a complete set of $N$ operation elements that satisfy $\sum_{n=1}^{N} E_{n}^{\dagger}E_{n}=\1$, because $\Ecalb_{x}[\bullet]$ should be a trace-preserving operation. In general we have a specific set $\{E_{n}\}$ for each process considered, so for this reason it is convenient to consider a fixed basis of $K$ operators $F_{n}^{k}$ so that we can write $E_{n} = \sum_{k=1}^{K} e_{nk}\bar{E}_{k}$. Therefore, we get
\begin{align}
	\Ecalb_{x}[\rho (0)] = \sum\nolimits_{n=1}^{N} \left(\sum\nolimits_{k=1}^{K} e_{nk}\bar{E}_{k}\right) \rho (0) \left(\sum\nolimits_{l=1}^{K} e^{\ast}_{nl}\bar{E}^{\dagger}_{l}\right) = \sum\nolimits_{k,l=1}^{K} \chi_{kl} \bar{E}_{k} \rho (0) \bar{E}^{\dagger}_{l}\text{ . }
\end{align}

\begin{figure}[t!]
	\centering
	%\vspace{3.7cm}
	\includegraphics[scale=0.6]{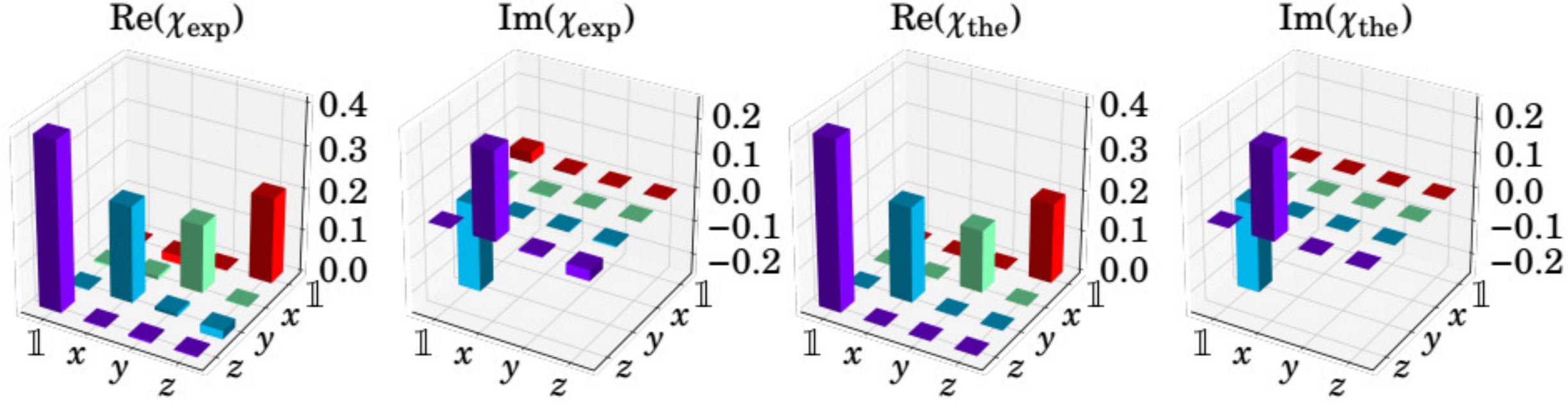}
	\caption{Process matrix obtained by process tomography, where we plot the real and imaginary parts of $\chi$ obtained from the experimental measured data ($\chi_{\text{exp}}$) and numerical simulation ($\chi_{\text{the}}$).}
	\label{ApFigTomograph}
\end{figure}

Therefore, since the set is early defined, the above equation tells that the process is characterized by a matrix with elements $\chi_{kl}$. Thus, by following this process we adopted the reference basis for a single-qubit tomography process as $\bar{E}_1 = \1$, $\bar{E}_2 = \sigma_x$, $\bar{E}_3 = \sigma_y$, $\bar{E}_4 = \sigma_z$ in the case where the dynamics is simulated with $\gamma=2.5$~KHz, where we set the total evolution time as $0.24$~ms (here, the noise amplitude is 1.62 V and the modulation depth is 96.00 KHz/V). The resulting estimated process matrix is shown in Fig.~\ref{ApFigTomograph}. Then, the fidelity between the simulated and desired process can be obtained by evaluating the fidelity between theoretical and experimental process matrices defined as
\begin{eqnarray}
	\Fcalb({\chi _{\text{exp} }},{\chi _{\text{the}}}) = \left[\text{Tr}\sqrt{ \sqrt{\chi _{\text{exp} }} \text{ } \chi _{\text{the}} \sqrt{\chi _{\text{exp} }} }\right]^2 \text{ . }
\end{eqnarray}

For example, when the amplitude of the noise is set to $1.54$~V, the process fidelities are measured as  $\Fcalb_{t_{1}} = 99.27\%$, $\Fcalb_{t_{2}} = 99.50\%$, $\Fcalb_{t_{3}} =  99.72\%$, $\Fcalb_{t_{4}} = 99.86\%$ and $\Fcalb_{t_{5}} = 99.87\%$, at times $t_{1} = 0.08$~ms, $t_{2} =0.16$~ms, $t_{3} = 0.24$~ms, $t_{4} = 0.32$~ms and $t_{5} = 0.40$~ms, respectively. Thus, the dephasing channel can be precisely controlled as desired and it can support the scheme to implement the time-dependent dephasing in experiment. A point to be highlighted here is the case where $A_{0} = 0$, situation in which the system evolves under its natural dephasing rate $\gamma_{\text{nd}}$. From Eq.~\eqref{EqApSqrtGamma} one finds that $\gamma_{\text{nd}} \approx 1.74^2 \approx 3.03$~Hz. Thus, we can see that, if we change the parameter $A$, which we can done with high controllability, the quantity $\gamma$ can be efficiently controlled. On the other hand, if we need a time-dependent rate $\gamma(t)$, we just need to consider a way to vary $A_{0}$ as a function $A(t)$. To this end, we use a second channel (CH2) of the AWG to perform amplitude modulation (AM) of the Gaussian noise. The temporal dependence of $A(t)$ is achieved by programming the channel (CH2) to change the noise amplitude along the evolution.

\section{Nuclear Magnetic Resonance experimental setup} \label{ApNuclearMagExpSet}

In NMR experimental implementations for quantum information processing, we encode the quantum bit into degree of freedom associated with the nuclear spin of atoms in molecules. The NMR qubits used in this thesis are the carbon nuclei ($^{13}$C) and the hydrogen nuclei ($^{1}$H). However, independent on the nucleon used in our discussion, it is possible to derive the Hamiltonian used to drive the system and, consequently, implement quantum gates on the qubits. By applying a strong static magnetic field $\vec{B}_{0} = B_{0}\hat{\text{z}}$ along some direction (let's say the direction $Z$), we are choosing an quantization axes and different the spin states (spin up $\ketus$ and spin down $\ketds$) of the nuclei set a energy splitting of energy given by $\Delta E = \hbar \omega_{\text{L}}$, so that we can write the internal Hamiltonian as
\begin{equation}
	H_{\text{Zee}} = \frac{\hbar \omega_{\text{L}}}{2} \left( \ketus \braus - \ketds \brads \right) = \frac{\hbar \omega_{\text{L}}}{2} \sigma_{z} = \hbar \omega_{\text{L}} I_{z} \text{ , }
\end{equation}
called Zeeman Hamiltonian due to the Zeeman interaction occurring between the magnetic dipole moment of the nuclei and the local magnetic fields. The quantity $\omega_{\text{L}}$ is the Larmor frequency and it is a characteristic parameter of the nuclei used in experiment; it is associated with the \textit{gyromagnetic ratio} of the nucleus $\gamma_{0}$ as $\omega_{\text{L}} = -\gamma_{0} B_{0}$.

Therefore, the qubit is encoded in states $\ketus$ and $\ketds$ as $\ket{0} = \ketds$ and $\ket{1} = \ketus$. The qubit is driven by an transverse rotating field called \textit{radio-frequency} (RF) field $\vec{B}_{\text{rf}}(t)$ given by
\begin{equation}
	\vec{B}_{\text{rf}}(t) = B_{1} \left[ \cos(\omega t) \hat{\text{x}} + \sin(\omega t) \hat{\text{y}} \right] \text{ , }
\end{equation}
where $\omega$ is the rotating frequency of the field $\vec{B}_{\text{rf}}(t)$. Therefore, the single-qubit Hamiltonian is a combination of the Zeeman and RF Hamiltonian, this gives
\begin{equation}
	H(t) = \hbar\omega_{0} I_{z} + \hbar\omega_{1} \left[ \cos(\omega t) I_{y} + \sin(\omega t) I_{y} \right] \text{ , }
\end{equation}
with $I_{j} = (1/2) \sigma_{j}$ and $\omega_{1} = -\gamma_{0} B_{1}$. In general, since the field $|\vec{B}_{1}|$ is much less intense than $|\vec{B}_{0}|$, it is convenient to consider the qubit dynamics in rotating frame as done in Appendix~\ref{ApFrameChangeQM-NMR}. By doing that, we find the driven Hamiltonian in rotating frame as
\begin{equation}
	H = \hbar(\omega_{0}-\omega) I_{z} + \hbar \omega_{1} I_{x} \text{ . }
\end{equation}
where the term $H_{\text{Zee}}^{\text{eff}} = \hbar(\omega_{0}-\omega) I_{z}$ is the effective Zeeman Hamiltonian. Therefore, when we set $\omega_{0} = \omega$ we find the resulting Hamiltonian in rotating frame at resonance as
\begin{equation}
	H_{x} = \hbar \omega_{1} I_{x} \text{ . }
\end{equation}

Then, by letting the system evolve under this Hamiltonian we get the evolution operator given by
\begin{equation}
	U_{x}(\Delta t) = \exp \left( -\frac{i}{\hbar} H_{x} \Delta t \right) = \exp \left( -i \omega_{1} I_{x} \Delta t \right) = R_{x} (\varpi) \text{ , }
\end{equation}
where we defined the angle $\varpi = \omega_{1} \Delta t$, with $\Delta t$ being the pulse duration of the RF field. Therefore, we can adjust the time interval of each RF pulse in order to select an arbitrary value for $\varpi$. After such pulse, the system is again governed by the effective Zeeman Hamiltonian, through a ``free evolution". In this case, the evolution operator becomes
\begin{equation}
	U_{\text{f}}(\Delta t) = \exp \left( -\frac{i}{\hbar} H_{\text{Zee}}^{\text{eff}} \Delta t \right) = \exp \left( -i (\omega_{0}-\omega) I_{z} \Delta t \right) = R_{z} (\zeta) \text{ , }
\end{equation}
where we define the angle $\zeta = (\omega_{0}-\omega)\Delta t$. Therefore, by using pulses $R_{x} (\varpi)$ and free evolution $R_{z} (\zeta)$ we can implement any single-qubit rotation~\cite{Sarthour:Book}.

In addition, the Hamiltonian that describes a two-qubit system here is obtained by considering the nuclear Zeeman effect, chemical shifts (due to the orbital motion of nearby electrons) and $J$-coupling (interaction between the
nuclear magnetic dipole moments of neighbor nuclei), which are represented by the following secular Hamiltonian (for spin-$\frac{1}{2}$ nuclei)
\begin{equation}
	\Hcalb = \hbar \omega_{0}^{(1)} I_{z}^{(1)} + \hbar \omega_{0}^{(2)} I_{z}^{(2)}+ 2\pi \hbar J_{12} I_{z}^{(1)}I_{z}^{(2)}
\end{equation}
where $I_{z}^{(n)}= (1/2) \sigma_{z}^{(n)}$ is the nuclear spin operators along the quantization axes for the $n$-th nuclei. Therefore, since the $J$-coupling can be used to implement any controlled phase gate~\cite{Santos:18-a}, it is possible to implement universal quantum computation with these number of single and two-qubit gates discussed here.

\subsection{Pulses composition for adiabatic and TQD quantum gates} \label{ApExpPuls}

From a Trotter decomposition approach for a time-dependent quantum Hamiltonian (see, e.g., Ref.~\cite{Nielsen:Book}), we use the pulse sequence 
shown in Fig.~\ref{Gates}{\color{blue}a} to implement each Hamiltonian in Eqs.~\eqref{Hz0},~\eqref{HzTQD} and~\eqref{HzOptTQD}. 
The algorithm is composed by rotations $(\vartheta)_{\varphi_{n}}$ and free evolutions. The rotations are implemented of an angle $\vartheta$ around a direction $\hat{\varphi}_{n} = \cos \varphi_{n} \hat{\text{x}}+\sin \varphi_{n} \hat{\text{y}}$, whose the rotation operator reads~\cite{Nielsen:Book}
\begin{equation}
	(\vartheta)_{\varphi_{n}} = \exp \left[ -\frac{i \vartheta}{2} \hat{\varphi}_{n}\cdot \vec{\sigma} \right] \text{ , }
\end{equation}
where $\vec{\sigma} = \sigma_{x} \hat{\text{x}} + \sigma_{y} \hat{\text{y}} + \sigma_{z} \hat{\text{z}}$, with $\sigma_{n}$ being the Pauli matrices. Each direction $\varphi_{n}$ of the rotations in Fig.~\ref{Gates}{\color{blue}a}-I are obtained from
\begin{equation}
	\varphi_{1} = \pi + \chi_{n} \text{ , } \quad \varphi_{2} = \varphi_{4} = \chi_{n} \text{ , } \quad \varphi_{3} = \frac{\pi}{2} + \chi_{n} \text{ , }
\end{equation}
with
\begin{equation}
	\chi_{n} = \sum_{\ell = n+1}^{N} \left[ -\frac{4\pi\nu\tau}{N-1} \cos \left( \pi \frac{t_{k}}{\tau} \right) + \frac{\pi}{2}\right] \text{ . }
\end{equation}

The free evolution used in circuits of the Figs.~\ref{Gates}{\color{blue}a}-I and ~\ref{Gates}{\color{blue}a}-II appear as a consequence of the natural spin evolution 
of the nuclei in our chloroform molecule, so that the evolution operator reads as
\begin{equation}
	U_{\text{free}}(\Delta t) = e^{-\frac{i}{\hbar} H_{\text{free}} \Delta t } \text{ , }
\end{equation}
where 
\begin{align}
	H_{\text{free}} = \hbar \frac{J}{2} \sigma_{z}^{(\text{C})}\sigma_{z}^{(\text{H})} \text{ . }
\end{align}
In our experiment, the interaction between the spins is constant with strength $J = 215$~Hz, so that we need to manipulate some additional parameters in order to implement different values of $\tau$ of the interaction term in Eqs.~\eqref{EqHExpNMRQuantumGates}. To this end, we use a free evolution for coupled spins so that we can map the parameter $\tau$ from two parameter $J$ and the time intervals $\Delta t$ of each free evolutions. This relation is obtained as
\begin{align}
	\Delta t_{1,n} &= \frac{\tau}{2(N - 1)} \left[ \frac{4\nu}{J} \sin \left( \pi \frac{t_{k}}{\tau} \right) + 1 \right] \text{ , } \quad \nonumber \\
	\Delta t_{2,n} &= \frac{\tau}{N - 1} - \Delta t_{1,n} \text{ , } \quad \nonumber \\
	\Delta t_{3,n} &= \frac{\tau}{2(N - 1)} \left( \frac{1}{J\tau} +1 \right) \text{ , } \quad \nonumber \\
	\Delta t_{4,n} &= \frac{\tau}{N - 1} - \Delta t_{3,n} \text{ . }
\end{align}
For Fig.~\ref{Gates}{\color{blue}a}-III, the Hamiltonian is constant 
and the time interval is then given by
\begin{equation}
	\Delta t_{1} = \frac{J\tau + 1}{2J} \quad \text{ and } \quad 
	\Delta t_{2} = \tau - \Delta t_{1} \text{ . }
\end{equation}
Therefore, we can choose the value of transitionless total evolution time $\tau$ according with the parameter $\Delta t$, for a constant value of $J$.

\chapter{Adiabatic dynamics in open system} \label{ApUAd}

\section{Evolution operator for adiabatic dynamics in open system}
Here we study some properties of the evolution operator $\Ucalb (t,t_{0})$ as derived in Eq.~\eqref{EqUAdOS}. First, given the evolution operator $\Ucalb (t,t_{0})$, let us show a necessary and sufficient condition to find $\Ucalb^{-1} (t,t_{0})$ so that $\Ucalb (t,t_{0})\Ucalb^{-1} (t,t_{0})=\1$. To this, let us define
\begin{eqnarray}
	\Ucalb^{-1} (t,t_{0}) = \sum_{\alpha = 1}^{N} \Ucalb_{\alpha}^{-1} (t,t_{0}) \text{ , }
\end{eqnarray}
with each operator $\Ucalb_{\alpha}^{-1} (t,t_{0})$ obtained from
\begin{eqnarray}
	\Ucalb_{\alpha}^{-1} (t,t_{0}) = e^{-\int_{t_{0}}^{t} \lambda_{\alpha}(\xi)d\xi} \sum _{n_{\alpha} = 1}^{N_{\alpha}} \sum _{m_{\alpha} = 1}^{N_{\alpha}} \tilde{u}_{n_{\alpha}m_{\alpha}}(t)\dket{\Dcalb_{\alpha}^{m_{\alpha}}(t_{0})}\dbra{\Ecalb_{\alpha}^{n_{\alpha}}(t)} \text{ , }
\end{eqnarray}
with parameters $\tilde{u}_{n_{\alpha}m_{\alpha}}(t)$ to be determined. This definition is convenient because we can write 
\begin{eqnarray}
	\Ucalb_{\beta} (t,t_{0})\Ucalb_{\alpha}^{-1} (t,t_{0}) = \delta_{\alpha\beta} \Ucalb_{\beta} (t,t_{0}) \Ucalb_{\alpha}^{-1} (t,t_{0}) \text{ , }
	\label{ApEqUortho}
\end{eqnarray}
where we use the orthogonality relation between right- and left-hand side quasi-eigenvectors. Now, we write
\begin{align}
	\Acalb_{1} &= \Ucalb (t,t_{0})\Ucalb^{-1} (t,t_{0}) = \sum_{\alpha = 1}^{N} \sum_{\beta = 1}^{N} \Ucalb_{\alpha} (t,t_{0}) \Ucalb_{\beta}^{-1} (t,t_{0}) = \sum_{\alpha = 1}^{N} \Ucalb_{\alpha} (t,t_{0}) \Ucalb_{\alpha}^{-1} (t,t_{0}) \text{ , }
\end{align}
where we already used the Eq.~\eqref{ApEqUortho}. Thus
\begin{align}
	\Acalb_{1} &= \sum_{\alpha = 1}^{N} \left[ \left(\sum _{n_{\alpha} = 1}^{N_{\alpha}} \sum _{m_{\alpha} = 1}^{N_{\alpha}} u_{n_{\alpha}m_{\alpha}}(t)\dket{\Dcalb_{\alpha}^{n_{\alpha}}(t)}\dbra{\Ecalb_{\alpha}^{m_{\alpha}}(t_{0})}\right) 
	\left(\sum _{j_{\alpha} = 1}^{N_{\alpha}} \sum _{k_{\alpha} = 1}^{N_{\alpha}} \tilde{u}_{j_{\alpha}k_{\alpha}}(t)\dket{\Dcalb_{\alpha}^{j_{\alpha}}(t_{0})}\dbra{\Ecalb_{\alpha}^{k_{\alpha}}(t)}\right)\right] \nonumber \\
	&= \sum_{\alpha = 1}^{N} \left[ \sum _{n_{\alpha} = 1}^{N_{\alpha}} \sum _{m_{\alpha} = 1}^{N_{\alpha}}
	\sum _{j_{\alpha} = 1}^{N_{\alpha}} \sum _{k_{\alpha} = 1}^{N_{\alpha}}
	u_{n_{\alpha}m_{\alpha}}(t)\tilde{u}_{j_{\alpha}k_{\alpha}}(t)
	\dket{\Dcalb_{\alpha}^{n_{\alpha}}(t)}
	\left( \dinterpro{\Ecalb_{\alpha}^{m_{\alpha}}(t_{0})}{\Dcalb_{\alpha}^{j_{\alpha}}(t_{0})}\right) \dbra{\Ecalb_{\alpha}^{k_{\alpha}}(t)}\right] \nonumber \\
	&= \sum_{\alpha = 1}^{N} \left[ \sum _{n_{\alpha} = 1}^{N_{\alpha}} \sum _{j_{\alpha} = 1}^{N_{\alpha}} \sum _{k_{\alpha} = 1}^{N_{\alpha}}
	u_{n_{\alpha}j_{\alpha}}(t)\tilde{u}_{j_{\alpha}k_{\alpha}}(t)
	\dket{\Dcalb_{\alpha}^{n_{\alpha}}(t)}
	\dbra{\Ecalb_{\alpha}^{k_{\alpha}}(t)}\right] \text{ . }
\end{align}

Now, computing the matrix elements $\dbra{\Ecalb_{\eta}^{p_{\eta}}(t)}\Acalb_{1}\dket{\Dcalb_{\nu}^{m_{\nu}}(t)}$, we get (to simplify the notation, from now on we will omit the time-dependence of the coefficients $u$ and $\tilde{u}$)
\begin{align}
	\dbra{\Ecalb_{\eta}^{p_{\eta}}(t)}\Acalb\dket{\Dcalb_{\nu}^{m_{\nu}}(t)} &= \sum_{\alpha = 1}^{N} \left[ \sum _{n_{\alpha} = 1}^{N_{\alpha}} \sum _{j_{\alpha} = 1}^{N_{\alpha}} \sum _{k_{\alpha} = 1}^{N_{\alpha}}
	u_{n_{\alpha}j_{\alpha}}\tilde{u}_{j_{\alpha}k_{\alpha}}
	\dinterpro{\Ecalb_{\eta}^{p_{\eta}}(t)}{\Dcalb_{\alpha}^{n_{\alpha}}(t)}
	\dinterpro{\Ecalb_{\alpha}^{k_{\alpha}}(t)}{\Dcalb_{\nu}^{m_{\nu}}(t)}\right] \nonumber \\
	&= \sum_{\alpha = 1}^{N} \left[ \sum _{n_{\alpha} = 1}^{N_{\alpha}} \sum _{j_{\alpha} = 1}^{N_{\alpha}} \sum _{k_{\alpha} = 1}^{N_{\alpha}}
	u_{n_{\alpha}j_{\alpha}}\tilde{u}_{j_{\alpha}k_{\alpha}}
	\delta_{p_{\eta}n_{\alpha}} \delta_{k_{\alpha}m_{\nu}} \delta_{\eta\alpha}
	\delta_{\alpha\nu}
	\right] \nonumber \\
	&= \sum_{\alpha = 1}^{N} \left[  \sum _{j_{\alpha} = 1}^{N_{\alpha}} 
	u_{p_{\eta}j_{\alpha}}\tilde{u}_{j_{\alpha}m_{\nu}}
	\delta_{\eta\alpha} \delta_{\alpha\nu}
	\right] = \delta_{\eta\nu} \sum _{j_{\nu} = 1}^{N_{\nu}} 
	u_{p_{\eta}j_{\nu}}\tilde{u}_{j_{\nu}m_{\nu}} \text{ . }
\end{align}

Therefore, to obtain $\Acalb_{1} = \1$ the coefficients need to satisfy
\begin{eqnarray}
	\sum _{j_{\nu} = 1}^{N_{\nu}} 
	u_{p_{\nu}j_{\nu}}\tilde{u}_{j_{\nu}m_{\nu}} = \delta_{p_{\nu}m_{\nu}} \text{ . } \label{ApEqmu1}
\end{eqnarray}

A second important property of the operator $\Ucalb (t,t_{0})$ is related with the Jordan block form of the Lindbladian. To proof it, let us consider a operator $\Acalb_{2}$ given by $\Acalb_{2} = \Ucalb^{-1} (t,t_{0}) \Lmath(t) \Ucalb (t,t_{0})$, so we have
\begin{align}
	\Acalb_{2} &= \Ucalb^{-1} (t,t_{0}) \Lmath(t) \Ucalb (t,t_{0}) = \sum_{\alpha = 1}^{N}\sum_{\beta = 1}^{N} \underbrace{\Ucalb_{\alpha}^{-1} (t,t_{0}) \Lmath(t) \Ucalb_{\beta} (t,t_{0})}_{\Acalb_{2}^{\alpha\beta}} = \sum_{\alpha = 1}^{N}\sum_{\beta = 1}^{N} \Acalb_{2}^{\alpha\beta} \text{ . }
\end{align}

Now let us analyze the each matrix $\Acalb_{2}^{\alpha\beta}$
\begin{align}
	\Acalb_{2}^{\alpha\beta} &= \Ucalb_{\alpha}^{-1} (t,t_{0}) \Lmath(t) \Ucalb_{\beta} (t,t_{0}) 
	\nonumber \\
	&= 
	e^{\int_{t_{0}}^{t} \lambda_{\beta}(\xi)-\lambda_{\alpha}(\xi)d\xi}
	\sum _{j_{\alpha} = 1}^{N_{\alpha}} \sum _{k_{\alpha} = 1}^{N_{\alpha}}
	\sum _{n_{\beta} = 1}^{N_{\beta}} \sum _{p_{\beta} = 1}^{N_{\beta}}
	\tilde{u}_{j_{\alpha}k_{\alpha}} u_{n_{\beta}p_{\beta}}
	\dbra{\Ecalb_{\alpha}^{k_{\alpha}}(t)}\Lmath(t) \dket{\Dcalb_{\beta}^{n_{\beta}}(t)}\dket{\Dcalb_{\alpha}^{j_{\alpha}}(t_{0})}\dbra{\Ecalb_{\beta}^{p_{\beta}}(t_{0})} 
	\nonumber \\
	&=
	e^{\int_{t_{0}}^{t} \lambda_{\beta}(\xi)-\lambda_{\alpha}(\xi)d\xi}
	\sum _{j_{\alpha} = 1}^{N_{\alpha}} \sum _{k_{\alpha} = 1}^{N_{\alpha}}
	\sum _{n_{\beta} = 1}^{N_{\beta}} \sum _{p_{\beta} = 1}^{N_{\beta}}
	\tilde{u}_{j_{\alpha}k_{\alpha}} u_{n_{\beta}p_{\beta}}\lambda_{\alpha}(t)
	\dinterpro{\Ecalb_{\alpha}^{k_{\alpha}}(t)}{\Dcalb_{\beta}^{n_{\beta}}(t)}
	\dket{\Dcalb_{\alpha}^{j_{\alpha}}(t_{0})}\dbra{\Ecalb_{\beta}^{p_{\beta}}(t_{0})}
	\nonumber\\
	&+
	e^{\int_{t_{0}}^{t} \lambda_{\beta}(\xi)-\lambda_{\alpha}(\xi)d\xi}
	\sum _{j_{\alpha} = 1}^{N_{\alpha}} \sum _{k_{\alpha} = 1}^{N_{\alpha}}
	\sum _{n_{\beta} = 1}^{N_{\beta}} \sum _{p_{\beta} = 1}^{N_{\beta}}
	\tilde{u}_{j_{\alpha}k_{\alpha}} u_{n_{\beta}p_{\beta}}
	\dinterpro{\Ecalb_{\alpha}^{k_{\alpha}}(t)}{\Dcalb_{\beta}^{(n_{\beta}-1)}(t)}
	\dket{\Dcalb_{\alpha}^{j_{\alpha}}(t_{0})}\dbra{\Ecalb_{\beta}^{p_{\beta}}(t_{0})} \text{ , }
	\nonumber
\end{align}
where we already used the quasi-eigenvalue equation. Thus, due orthogonality of the basis we write
\begin{align}
	\Acalb_{2}^{\alpha\beta}
	&=
	e^{\int_{t_{0}}^{t} \lambda_{\beta}(\xi)-\lambda_{\alpha}(\xi)d\xi}
	\sum _{j_{\alpha} = 1}^{N_{\alpha}} \sum _{k_{\alpha} = 1}^{N_{\alpha}}
	\sum _{n_{\beta} = 1}^{N_{\beta}} \sum _{p_{\beta} = 1}^{N_{\beta}}
	\tilde{u}_{j_{\alpha}k_{\alpha}} u_{n_{\beta}p_{\beta}}\lambda_{\alpha}(t)
	\delta_{\alpha\beta}\delta_{k_{\alpha}n_{\beta}}
	\dket{\Dcalb_{\alpha}^{j_{\alpha}}(t_{0})}\dbra{\Ecalb_{\beta}^{p_{\beta}}(t_{0})}
	\nonumber\\
	&+
	e^{\int_{t_{0}}^{t} \lambda_{\beta}(\xi)-\lambda_{\alpha}(\xi)d\xi}
	\sum _{j_{\alpha} = 1}^{N_{\alpha}} \sum _{k_{\alpha} = 1}^{N_{\alpha}}
	\sum _{n_{\beta} = 1}^{N_{\beta}} \sum _{p_{\beta} = 1}^{N_{\beta}}
	\tilde{u}_{j_{\alpha}k_{\alpha}} u_{n_{\beta}p_{\beta}}
	\delta_{\alpha \beta} \delta_{k_{\alpha}(n_{\beta}-1)}
	\dket{\Dcalb_{\alpha}^{j_{\alpha}}(t_{0})}\dbra{\Ecalb_{\beta}^{p_{\beta}}(t_{0})} \text{ , }
	\nonumber
\end{align}
and due the $\delta_{\alpha \beta}$, we can write
\begin{align}
	\Acalb_{2} &= 
	\sum_{\alpha = 1}^{N}
	\sum _{j_{\alpha} = 1}^{N_{\alpha}} \sum _{k_{\alpha} = 1}^{N_{\alpha}}
	\sum _{n_{\alpha} = 1}^{N_{\alpha}} \sum _{p_{\alpha} = 1}^{N_{\alpha}}
	\tilde{u}_{j_{\alpha}k_{\alpha}} u_{n_{\alpha}p_{\alpha}}\lambda_{\alpha}(t)
	\delta_{k_{\alpha}n_{\alpha}}
	\dket{\Dcalb_{\alpha}^{j_{\alpha}}(t_{0})}\dbra{\Ecalb_{\alpha}^{p_{\alpha}}(t_{0})}
	\nonumber\\
	&+
	\sum_{\alpha = 1}^{N}
	\sum _{j_{\alpha} = 1}^{N_{\alpha}} \sum _{k_{\alpha} = 1}^{N_{\alpha}}
	\sum _{n_{\alpha} = 1}^{N_{\alpha}} \sum _{p_{\alpha} = 1}^{N_{\alpha}}
	\tilde{u}_{j_{\alpha}k_{\alpha}} u_{n_{\alpha}p_{\alpha}}
	\delta_{k_{\alpha}(n_{\alpha}-1)}
	\dket{\Dcalb_{\alpha}^{j_{\alpha}}(t_{0})}\dbra{\Ecalb_{\alpha}^{p_{\alpha}}(t_{0})}
	\nonumber\\
	&= 
	\sum_{\alpha = 1}^{N}
	\sum _{j_{\alpha} = 1}^{N_{\alpha}}
	\sum _{n_{\alpha} = 1}^{N_{\alpha}} \sum _{p_{\alpha} = 1}^{N_{\alpha}} \left( \tilde{u}_{j_{\alpha}n_{\alpha}} u_{n_{\alpha}p_{\alpha}}\lambda_{\alpha}(t)
	+ \tilde{u}_{j_{\alpha}(n_{\alpha}-1)} u_{n_{\alpha}p_{\alpha}}\right)
	\dket{\Dcalb_{\alpha}^{j_{\alpha}}(t_{0})}\dbra{\Ecalb_{\alpha}^{p_{\alpha}}(t_{0})} \text{ . }
\end{align}

By computing the matrix elements of $\Acalb_{2}$ in basis $\{ \dket{\Dcalb_{\alpha}^{j_{\alpha}}(t_{0})},\dbra{\Ecalb_{\alpha}^{p_{\alpha}}(t_{0})} \}$, we get
\begin{align}
	\dbra{\Ecalb_{\eta}^{g_{\eta}}(t_{0})}\Acalb_{2}\dket{\Dcalb_{\nu}^{l_{\nu}}(t_{0})} &= 
	\sum_{\alpha = 1}^{N}
	\sum _{j_{\alpha} = 1}^{N_{\alpha}}
	\sum _{n_{\alpha} = 1}^{N_{\alpha}} \sum _{p_{\alpha} = 1}^{N_{\alpha}} \left( \tilde{u}_{j_{\alpha}n_{\alpha}} u_{n_{\alpha}p_{\alpha}}\lambda_{\alpha}(t)
	+ \tilde{u}_{j_{\alpha}(n_{\alpha}-1)} u_{n_{\alpha}p_{\alpha}}\right)
	\delta_{\eta\alpha}\delta_{g_{\eta}j_{\alpha}}
	\delta_{\alpha\nu}\delta_{p_{\alpha}l_{\nu}} \nonumber \\
	&= 
	\sum _{n_{\eta} = 1}^{N_{\eta}} \left( \tilde{u}_{g_{\eta}n_{\eta}} u_{n_{\eta}l_{\nu}}\lambda_{\eta}(t)
	+ \tilde{u}_{g_{\eta}(n_{\eta}-1)} u_{n_{\eta}l_{\nu}}\right)
	\delta_{\eta\nu} \text{ . }
\end{align}

As a first result, we can see that $\dbra{\Ecalb_{\eta}^{g_{\eta}}(t_{0})}\Acalb_{2}\dket{\Dcalb_{\nu}^{l_{\nu}}(t_{0})} = 0$ for any element outside some block, it means that $\Acalb_{2}$ is block diagonal in this basis. Thus, the elements of a single block are
\begin{align}
	\dbra{\Ecalb_{\nu}^{g_{\nu}}(t_{0})}\Acalb_{2}\dket{\Dcalb_{\nu}^{l_{\nu}}(t_{0})}
	&= \lambda_{\nu}(t)
	\sum _{n_{\nu} = 1}^{N_{\nu}} \tilde{u}_{g_{\nu}n_{\nu}} u_{n_{\nu}l_{\nu}}
	+ \sum _{n_{\nu} = 1}^{N_{\nu}} \tilde{u}_{g_{\nu}(n_{\nu}-1)} u_{n_{\nu}l_{\nu}} \text{ . }
\end{align}

The above equation provides us some new conditions on the coefficients which allows us to identify the elements $\dbra{\Ecalb_{\nu}^{g_{\nu}}(t_{0})}\Acalb_{2}\dket{\Dcalb_{\nu}^{l_{\nu}}(t_{0})}$ as the elements of the Jordan-Block form for $\Lmath(t)$. In fact, if we impose that
\begin{eqnarray}
	\sum _{n_{\nu} = 1}^{N_{\nu}} \tilde{u}_{l_{\nu}(n_{\nu}-1)} u_{n_{\nu}l_{\nu}} = 0 \quad \text{ and } \quad \sum _{n_{\nu} = 1}^{N_{\nu}} \tilde{u}_{l_{\nu}n_{\nu}} u_{n_{\nu}l_{\nu}} = 1 \label{ApEqmu01} \text{ , }
\end{eqnarray}
we find the diagonal elements of a single block as
\begin{align}
	\dbra{\Ecalb_{\nu}^{l_{\nu}}(t_{0})}\Acalb_{2}\dket{\Dcalb_{\nu}^{l_{\nu}}(t_{0})}
	&= \lambda_{\nu}(t) \text{ . }
\end{align}

It is important to remark that the second condition in Eq.~\eqref{ApEqmu01} does not imply in condition in Eq.~\eqref{ApEqmu1}. In case we are interested in get a Jordan block decomposition, the first element after a diagonal element should be $1$, that is, a third condition should provides
\begin{eqnarray}
	\dbra{\Ecalb_{\nu}^{l_{\nu}}(t_{0})}\Acalb_{2}\dket{\Dcalb_{\nu}^{l_{\nu}+1}(t_{0})} = 1 \quad \text{ and } \quad \dbra{\Ecalb_{\nu}^{N_{\nu}}(t_{0})}\Acalb_{2}\dket{\Dcalb_{\nu}^{N_{\nu}+1}(t_{0})} = 0 \text{ , }
\end{eqnarray}
where the last condition is associated with last element in Jordan block. However, the last equation is automatically satisfied because $\dket{\Dcalb_{\nu}^{N_{\nu}+1}(t_{0})} = 0$, for construction. Therefore, we just need to have
\begin{eqnarray}
	\lambda_{\nu}(t)
	\sum _{n_{\nu} = 1}^{N_{\nu}} \tilde{u}_{l_{\nu}n_{\nu}} u_{n_{\nu}(l_{\nu}+1)}
	+ \sum _{n_{\nu} = 1}^{N_{\nu}} \tilde{u}_{l_{\nu}(n_{\nu}-1)} u_{n_{\nu}(l_{\nu}+1)} = 1 \text{ , }
\end{eqnarray}
whose a possible solution is
\begin{eqnarray}
	\sum _{n_{\nu} = 1}^{N_{\nu}} \tilde{u}_{l_{\nu}n_{\nu}} u_{n_{\nu}(l_{\nu}+1)} = 0 \quad \text{ and } \quad \sum _{n_{\nu} = 1}^{N_{\nu}} \tilde{u}_{l_{\nu}(n_{\nu}-1)} u_{n_{\nu}(l_{\nu}+1)} = 1 \text{ . } \label{ApEqmuND01}
\end{eqnarray}

\chapter{Complementary discussion on transitionless quantum driving} \label{ApTQD}

\section{Transitionless quantum driving in closed systems} \label{ApProofTheoOptTQD}

As a first task, let us proof the Theorem~\ref{TheoOptmEner}. To this end, consider a $D$-dimensional quantum system and adopt as a measure of energy cost the Hamiltonian Hilbert-Schmidt norm, which reads
\begin{align}
	\Sigma_{\text{tqd}} \left( \tau \right) =\frac{1}{\tau }\int_{0}^{\tau }\sqrt{\text{Tr}%
		\left[ {H_\text{tqd}}^{2}\left( t\right) \right] }\text{ }dt \text{ ,} \label{CostSI}
\end{align}%
Then, we obtain
\begin{align}
	H_{\text{tqd}}^{2}\left( t\right) = \hbar ^{2}\sum_{n=1}^{D} \left[ \frac{}{}|\dot{E}_{n}(t)\rangle \langle \dot{E}_{n}(t)|+\theta
	_{n}^{2}\left( t\right) |E_{n}(t)\rangle \langle E_{n}(t)|+ i\theta _{n}\left(
	t\right) \left( |E_{n}(t)\rangle \langle \dot{E}_{n}(t)|-|\dot{E}_{n}(t)\rangle
	\langle E_{n}(t)|\right) \right] \text{ .}
	\label{HSA2}
\end{align}
By taking the trace of $H_{\text{tqd}}^{2}\left( t\right) $ in Eq.~(\ref{HSA2}), we have%
\begin{align}
	\text{Tr} \left[ H_{\text{tqd}}^{2}\left( t\right) \right] =\sum_{m=1}^{D}\langle E_{m}(t)|H_{\text{tqd}}^{2}\left( t\right) |E_{m}(t)\rangle = \hbar ^{2}\sum_{n=1}^{D} \left[ \langle \dot{E}_{n}(t)|\dot{E}_{n}(t)\rangle +\theta _{n}^{2}\left( t\right) +2i\theta _{n}\left(
	t\right) \langle \dot{E}_{n}(t)|E_{n}(t)\rangle \right] \text{ .}  \label{TrH2}
\end{align}
Then
\begin{align}
	\Sigma_{ \text{\text{tqd}} } \left( \tau \right) =\frac{1}{\tau }\int_{0}^{\tau }\sqrt{\sum_{n=1}^{D} \left[\langle \dot{E}_{n}(t)|\dot{E}_{n}(t)\rangle + \Gamma_{n}(\theta_{n})\right] }\text{ }dt \text{ ,} \label{Supercost}
\end{align}
where we have $\Gamma_{n}(\theta_{n})=\theta_{n}^{2}\left( t\right) +2i\theta _{n}\left( t\right)  \langle \dot{E}_{n}(t)|E_{n}(t) \rangle$. We can now find out the functions $\theta _{n} \left( t\right)$
that minimize the energy cost in transitionless evolutions. For this
end, we minimize the quantity $\Sigma _{\text{tqd}}\left( \tau \right) $
for the Hamiltonian $H_{\text{tqd}}\left( t\right) $ with
respect to parameters $\theta _{n}\left( t\right) $, where we will adopt it
being independents. By evaluating $\partial _{\theta _{n}}\Sigma \left( \tau
\right) $, we obtain%
\begin{align}
	\partial _{\theta _{n}}\Sigma _{\text{tqd}}\left( \tau \right) =\frac{1}{%
		2\tau }\int_{0}^{\tau }\frac{\partial _{\theta _{n}}\{\text{Tr}[H_{\text{tqd}%
		}^{2}\left( t\right) ]\}}{\sqrt{\text{Tr}\left[ {H}_{\text{tqd}}^{2}\left( t\right) %
			\right] }}\text{ }dt \text{ .} 
\end{align}
We then impose $\partial _{\theta _{n}}\{$Tr$[H_{%
	\text{tqd}}^{2}\left( t\right) ]\}=0$ for all time $t\in \left[ 0,\tau \right]$, 
which ensures $\partial _{\theta _{n}}\Sigma _{\text{tqd}}\left( \tau \right)=0$. 
Thus, by using Eq. (\ref{TrH2}), we write%
\begin{align}
	\partial _{\theta _{n}}\{\text{Tr}[H_{\text{tqd}}^{2}\left( t\right) ]\}=2\theta
	_{n}\left( t\right) +2i\langle \dot{E}_{n}(t)|E_{n}(t)\rangle = 0
	\text{ .}
\end{align}
This implies
\begin{align}
	\theta _{n}\left( t\right) = \theta ^{\text{min}}_{n}\left( t\right) = -i\langle \dot{E}_{n}(t)|E_{n}(t)\rangle \text{ .}
	\label{OptTetaAp}
\end{align}
From the second derivative analysis, it follows that the choice for $\theta _{n}\left( t\right)$ as in Eq.~(\ref{OptTetaAp})  
necessarily minimizes the energy cost, namely, $\partial ^{2} _{\theta _{n}} \Sigma _{\text{tqd}}\left( \tau \right) \vert _{\theta _{n} = \theta ^{\text{min}}_{n}}>0$, which concludes the proof.

Now, to proof the Theorem~\ref{TheoTimeIndep} we need to take the time derivative of the Hamiltonian $H_{\text{SA}}(t)$, where we obtain
\begin{align}
	\dot{H}_{\text{SA}}\left( t\right) =i  \sum_{n=1}^{D}\frac{d}{dt}\left[ \frac{%
		{}}{{}}|\dot{E}_{n}(t)\rangle \langle E_{n}(t)|+i\theta _{n}\left( t\right)
	|E_{n}(t)\rangle \langle E_{n}(t)|\frac{{}}{{}}\right] \text{ .}
\end{align}
Then, the matrix elements of $\dot{H}_{\text{SA}}\left( t\right) $ 
in the eigenbasis $\left\{ |E_{m}(t)\rangle \right\} $ of the Hamiltonian $H_0(t)$ read
\begin{align}
	\langle E_{k}(t)|\dot{H}_{\text{SA}}\left( t\right) |E_{m}(t)\rangle  
	&=i  \langle E_{k}(t)|\ddot{E}_{m}(t)\rangle +i  \sum_{n=1}^{D}\langle E_{k}(t)|\dot{E}_{n}(t)\rangle
	\langle \dot{E}_{n}(t)|E_{m}(t)\rangle  \nonumber \\
	&-  \left[ \dot{\theta}	_{k}\left( t\right) \delta _{km} 
	+\theta _{m}\left( t\right) \langle E_{k}(t)|%
	\dot{E}_{m}(t)\rangle +\theta _{k}\left( t\right) \langle \dot{E}_{k}(t)|E_{m}(t)\rangle \right] .  
\end{align}

Now, in second term of the above equation we can use that $\langle E_{k}(t)|\dot{E}_{n}(t)\rangle =-\langle \dot{E}_{k}(t)|E_{n}(t)\rangle $, to write $\langle E_{k}(t)|\dot{E}_{n}(t)\rangle \langle \dot{E}_{n}(t)|E_{m}(t)\rangle =\langle \dot{E}_{k}(t)|E_{n}(t)\rangle \langle E_{n}(t)|\dot{E}_{m}(t)\rangle $ and thus%
\begin{align}
	\langle E_{k}(t)|\dot{H}_{\text{SA}}\left( t\right) |E_{m}(t)\rangle 
	&=i  \frac{d}{d{t}}\left[ \langle E_{k}(t)|\dot{E}_{m}(t)\rangle \right]  \nonumber \\
	&- 
	\left\{ \dot{\theta}_{k}\left( t\right) \delta _{km}+\left[ \theta
	_{m}\left( t\right) -\theta _{k}\left( t\right) \right] \langle E_{k}(t)|\dot{E}_{m}(t)\rangle \right\} \text{ .} \label{finaleEq}
\end{align}%
For $k=m$ in Eq.~(\ref{finaleEq}), 
we impose the vanishing of the diagonal elements of $\dot{H}_{\text{SA}}(t)$, namely, 
$\langle E_{k}(t)|\dot{H}_{\text{SA}}\left(t\right) |E_{k}(t)\rangle =0$. This yields
\begin{align}
	\dot{\theta}_{m}\left( t\right) =i\frac{d}{dt}\left[ \langle E_{m}(t)|\dot{E}_{m}(t)\rangle \right] \text{ ,} \label{DiagSI}
\end{align}%
On the other hand, for $k \neq m$ in Eq. (\ref{finaleEq}), we now impose the vanishing 
of the off-diagonal elements of $\dot{H}_{\text{SA}}(t)$, 
namely, $\langle E_{k}(t)|\dot{H}_{\text{SA}}\left(
t\right) |E_{m}(t)\rangle =0$ $(k \ne m)$. This yields 
\begin{align}
	i\frac{d}{dt}\left[ \langle E_{k}(t)|\dot{E}_{m}(t)\rangle \right] =\left[ \theta
	_{m}\left( t\right) -\theta _{k}\left( t\right) \right] \langle E_{k}(t)|\dot{E}_{m}(t)\rangle \ \ \text{ for } k \neq m \text{ .}  \label{NonDiagSI}
\end{align}%
By taking $\langle E_m (t) |\dot{E}_{m}(t)\rangle \equiv c_{mm}$ in Eq.~(\ref{DiagSI}), with $c_{mm}$ 
denoting by hypothesis complex constants, we get 
$\theta _{m}\left( t\right) =\theta _{m}\left(0\right) \equiv \theta_{m}$, namely, $\theta _{m}\left( t\right)$ is 
a constant function $\forall$ $m$. Moreover, by using that 
$\langle E_k (t) |\dot{E}_{m}(t)\rangle \equiv c_{km}$ in Eq.~(\ref{NonDiagSI}), with $c_{km}$ denoting 
{\it nonvanishing} complex constants, we obtain $\theta _{k} = \theta _{m}$, $\forall$ $k,m$. 
If $c_{km}=0$, then $\theta _{k}$ and $\theta _{m}$ are not necessarily equal, but Eq.~(\ref{NonDiagSI}) 
will also be satisfied by this choice. Therefore, it follows that $\theta_m(t)$ can be simply taken as
\begin{align}
	\theta_m(t) =\theta_m = \theta \,\,\, \forall m \, ,
\end{align} 
with $\theta$ a single real constant. This concludes the proof.

%
% Apparently the gu{\tiny }idelines don't say anything about citations or
% bibliography styles so I guess we can use anything.

\backmatter
%\bibliography{/home/cs/MEGA/Work/Articles/Models/Bibliografia/mybib-URL.bib}
%\bibliographystyle{PRA-WithTitle}
%\clearemptydoublepage


%
% Add index
%\printindex
%   
\end{document}